\documentclass[a4paper,10pt]{article}
\usepackage{etoolbox}
\robustify{\underline}
\usepackage[latin1]{inputenc}
\usepackage{enumerate}
\usepackage[english]{babel}
\usepackage[T1]{fontenc}
\usepackage{amsmath,amssymb}
\usepackage[euler]{textgreek}
\usepackage{array,booktabs}
\usepackage{graphicx}
\usepackage{amsthm}
\usepackage{bbm}
\usepackage{bm}
\usepackage{parskip}
\usepackage{mathtools}
\usepackage{mathabx}
\usepackage{tikz}
\usepackage{upgreek}
\usepackage{soul}
\usepackage{bigints}
\usepackage{hyperref}
\usepackage{
nameref,%\nameref
cleveref,% \cref
}
\usetikzlibrary{positioning ,shadows,matrix,calc,decorations.pathmorphing,shapes,automata}
\usepackage{upgreek}
\usepackage{cmbright}
\usepackage{appendix}
\usepackage{thmtools}
\usepackage{wasysym}
\usepackage{slashed}
\usepackage{concmath}
\usepackage{bigints}
\usepackage{setspace}
\usepackage{scalerel,stackengine}
\stackMath
\newcommand\reallywidehat[1]{%
\savestack{\tmpbox}{\stretchto{%
  \scaleto{%
    \scalerel*[\widthof{\ensuremath{#1}}]{\kern-.6pt\bigwedge\kern-.6pt}%
    {\rule[-\textheight/2]{1ex}{\textheight}}%WIDTH-LIMITED BIG WEDGE
  }{\textheight}% 
}{0.5ex}}%
\stackon[1pt]{#1}{\tmpbox}%
}

\newtheorem{thm}{Theorem}
\newtheorem{theorem*}{Theorem}
\newtheorem{lemma}[thm]{Lemma}
\newtheorem{proposition}[thm]{Proposition}
\newtheorem{corollary}{Corollary}
\newtheorem{corollary*}{Corollary}
\newtheorem{defin}{Definition}
\theoremstyle{definition}

\newtheorem{remark}{Remark}
\newtheorem{remark*}{Remark}
\numberwithin{thm}{subsection}
\numberwithin{lemma}{subsection}
\numberwithin{proposition}{subsection}
\numberwithin{corollary}{subsection}
\numberwithin{remark}{subsection}
\numberwithin{defin}{subsection}
\newtheorem*{notation*}{Notation}

\crefname{thm}{theorem}{theorems}
\Crefname{thm}{Theorem}{Theorems}
\crefname{proposition}{proposition}{propositions}
\Crefname{proposition}{Proposition}{Propositions}
\crefname{corollary}{corollary}{corollaries}
\Crefname{corollary}{Corollary}{Corollaries}
\crefname{defin}{definition}{definitions}
\Crefname{defin}{Definition}{Definitions}
\crefname{remark}{remark}{remarks}
\Crefname{remark}{Remark}{Remarks}
\usepackage{lmodern}
\usepackage[margin=3cm]{geometry}
\usepackage{verbatim}
\usepackage{fullpage}
\usepackage{amsthm}
\usepackage[affil-it]{authblk}
\usepackage{mathrsfs}
\usepackage{graphicx}
\usepackage[nottoc,notlot,notlof]{tocbibind}

\newcommand{\nablau}{\Omega\slashed{\nabla}_3}
\newcommand{\nablav}{\Omega\slashed{\nabla}_4}
\newcommand{\nablagml}{\slashed{\nabla}_3}

\newcommand{\fancyd}{\slashed{\mathcal{D}}}
\newcommand{\fancydstar}{\slashed{\mathcal{D}}^*}

\newcommand{\fbar}{\underline{f}}
\newcommand{\gbar}{\underline{g}}
\newcommand{\overone}{\stackrel{\mbox{\scalebox{0.4}{(1)}}}}
\newcommand{\slashednabla}{\slashed{\nabla}}
\newcommand{\fscri}{\mathscr{I}^+}
\newcommand{\TSp}{\mathop{\mathbb{TS}^+}}
\newcommand{\TSm}{\mathop{\mathbb{TS}^-}}
\newcommand{\trx}{\overone{\Omega tr\chi}}
\newcommand{\trxbar}{\overone{\Omega tr\underline\chi}}

\numberwithin{equation}{section}
\newcommand{\invertedalpha}{\raisebox{\depth}{\scalebox{1}[-1]{$\alpha$}}}
\newcommand{\invertedpsi}{\raisebox{\depth}{\scalebox{1}[-1]{$\Psi$}}}

\def\J{\ensuremath J}

\DeclareMathOperator{\tr}{tr}
\usepackage{multirow}
\usepackage{titling}
\usepackage{breqn}
\usepackage{tensor}
\graphicspath{ {/home/hm532} }
\usepackage[backend=biber,maxbibnames=99]{biblatex}
\newcommand{\fourthorder}{\mathring{\fancydstar_2}\mathring{\fancydstar_1}\mathring{\overline{\fancyd_1}}\mathring{\fancyd_1}}
\def\equationautorefname~#1\null{(#1)\null}

\makeatletter
\patchcmd{\@maketitle}{\LARGE \@title}{\fontsize{16}{19.2}\selectfont\@title}{}{}
\makeatother

\title{A Scattering Theory for Linearised Gravity\\ on the Exterior of the Schwarzschild Black Hole I:\\
The Teukolsky Equations}
\date{July 27, 2020}
\author{Hamed Masaood}
\affil{ University of Cambridge, Department of Applied Mathematics and Theoretical Physics\\
Wilberforce Road, Cambridge CB3 0WA, United Kingdom}

\newcommand\bref[1]{(\ref{#1})}
\begin{document}

%\footnotetext{$^*$hm532@cam.ac.uk}
\setcounter{tocdepth}{3}
\maketitle

\begin{abstract}
We construct a scattering theory for the spin $\pm2$ Teukolsky equations on the exterior of the Schwarzschild spacetime, as a first step towards developing a scattering theory for the linearised Einstein equations in double null gauge. This is done by exploiting a physical-space version of the Chandrasekhar transformation used by Dafermos, Holzegel and Rodnianski in \cite{DHR16} to prove the linear stability of the Schwarzschild solution. We also address the Teukolsky--Starobinsky correspondence and construct an isomorphism between scattering data for the $+2$ and $-2$ Teukolsky equations. This will allow us to state an additional mixed scattering statement for a pair of curvature components satisfying the spin $+2$ and $-2$ Teukolsky equations and connected via the Teukolsky--Starobinsky identities, completely determining the radiating degrees of freedom of solutions to the linearised Einstein equations.
\end{abstract}
\tableofcontents
\section{Introduction and overview}\label{section 1 introduction}
Scattering theory has been an important tool in the mathematical and theoretical study of black hole solutions to the Einstein equations, which in vacuum take the form
\begin{align}\label{EVE}
    R_{ab}[g]=0
\end{align}
(setting the cosmological constant to zero). Whereas there has been extensive work on scattering for scalar, electromagnetic, fermionic fields on black hole backgrounds (see already \cite{DimockKayI}, \cite{BachelotAFMaxwell}, \cite{Nicolas}, \cite{DRSR14}, \cite{DaudeNicoleau}), in the case of the scattering of gravitational perturbations much of the historic literature has been concerned with solutions to equations governing fixed frequency modes (see \cite{Chandrasekhar}, \cite{HandlerFuttermanMatzner} for an extensive survey, and the very recent \cite{SRTdC}), and comparatively little has been said about scattering theory on black holes \textit{in physical space}. The aim of this work is to address this vacancy for the case of linearised gravitational perturbations around the Schwarzschild exterior, which in familiar coordinates has the metric \cite{Schwarzschild}:
\begin{align}\label{SchwMetric}
     g=-\left(1-\frac{2M}{r}\right)dt^2+\left(1-\frac{2M}{r}\right)^{-1}dr^2+ r^2\left(d\theta^2+\sin^2\theta d\phi^2\right).
\end{align}
The subject of scattering theory is the study of perturbations evolved on scales that are large in comparison to a characteristic scale of the perturbed system. More concretely, scattering theory is relevant when the perturbations are meant to be asymptotically free from the effects of the target. In this picture, incoming and outgoing perturbations are approximated by solutions describing "free" propagation. A mathematical description of scattering hinges on an appropriate and rigorous formulation of these ideas, and much of the value of scattering theory lies in the identification of the correct candidates for spaces of "scattering states" that describe incoming and outgoing perturbations. In these terms, a satisfactory scattering theory must provide answers to the following questions:
\begin{enumerate}[I]
\item \textit{Existence of scattering states}: Is there an interesting class of initial data that evolve to solutions which can be associated with past/future scattering states?\label{QI}
\item \textit{Uniqueness of scattering states}: Is the above association injective? Do solutions that give rise to the same scattering state coincide?\label{QII}
\item \textit{Asymptotic completeness}: Does this association exhaust the class of initial data of interest?\label{QIII}
\end{enumerate}
Because of the nonlinear nature of the Einstein equations \bref{EVE}, the study of scattering in general relativity is dependent on a thorough understanding of the perturbative behaviour of the equations. As a first step, it is useful to understand the evolution of solutions to the linearised Einstein equations, which are obtained by formally expanding a family of solutions in some smallness parameter $\epsilon$ around some fixed background, e.g.~\bref{SchwMetric}, and keeping only leading order terms in $\epsilon$ in the equations \bref{EVE}. Studying the evolution of linear equations on black hole backgrounds has its own appeal, as black holes by their very nature are immune to "direct" observation and even their existence can only be inferred by examining their effects on the propagation of wave phenomena in spacetime. The linearised Einstein equations still inherit many of the features as well as the difficulties that plague the study of the nonlinear equations.\\
\indent A foundational breakthrough in the analysis of the linearised equations was discovered by Bardeen and Press \cite{Bardeen-Press} in the case of the Schwarzschild black hole \bref{SchwMetric} and Teukolsky \cite{TeukP74} in the case of the Kerr black hole \cite{Kerr}, who showed that by casting the equations of linearised gravity in the Newman--Penrose formalism, it is possible to identify gauge-invariant components of the curvature that obey 2nd order {\em decoupled} wave equations, which on the Schwarzschild spacetime take the forms
\begin{align}\label{wave equation +}
     \Box_g \Omega^2\alpha +\frac{4}{r\Omega^2}\left(1-\frac{3M}{r}\right)\partial_u \Omega^2\alpha=V(r) \Omega^2\alpha,
\end{align}
\begin{align}\label{wave equation -}
     \Box_g \Omega^2\underline\alpha -\frac{4}{r\Omega^2}\left(1-\frac{3M}{r}\right)\partial_v \Omega^2\underline\alpha=V(r) \Omega^2\underline\alpha.
\end{align}
Here, $\Box_g$ is the d'Alembertian operator of the Schwarzschild metric $g$, $\alpha, \underline\alpha$ are symmetric traceless $S^2$-tangent 2-tensor fields, $\Omega^2=1-\frac{2M}{r}$ and $V=\frac{2(3\Omega^2+1)}{r^2}$ (see already \Cref{Chandra1}). Equations \bref{wave equation -}, \bref{wave equation +} are known as the \textbf{Teukolsky equations of spin $\bm{+2}$ and $\bm{-2}$} respectively.\\
\indent In addition to the Teukolsky equations \bref{wave equation +}, \bref{wave equation -}, the quantities $\alpha, \underline\alpha$ satisfy a closed system of equations known as the Teukolsky--Starobinsky identities:
\begin{align}
\frac{\Omega^2}{r^2}\Omega\slashed{\nabla}_3 \left(\frac{r^2}{\Omega^2}\nablau\right)^3\alpha=2r^4\slashed{\mathcal{D}}^*_2\slashed{\mathcal{D}}^*_1\overline{\slashed{\mathcal{D}}}_1\slashed{\mathcal{D}}_2 r\Omega^2{\underline\alpha}+12M\partial_t\hspace{.5mm}r\Omega^2{\underline\alpha}, \label{eq:227intro1}\\
\frac{\Omega^2}{r^2}\Omega\slashed{\nabla}_4 \left(\frac{r^2}{\Omega^2}\Omega\slashed{\nabla}_4\right)^3{\underline\alpha}=2r^4\slashed{\mathcal{D}}^*_2\slashed{\mathcal{D}}^*_1\overline{\slashed{\mathcal{D}}}_1\slashed{\mathcal{D}}_2 r\Omega^2\alpha-12M\partial_t\hspace{.5mm}r\Omega^2\alpha.\label{eq:228intro1}
\end{align}
The purpose of this paper is to study the scattering theory of the Teukolsky equations \bref{wave equation +}, \bref{wave equation -} as a prelude to studying scattering for the full system of linearised Einstein equations. This is done by first developing a scattering theory for \bref{wave equation +}, \bref{wave equation -} in particular addressing points \ref{QI}, \ref{QII}, \ref{QIII} above, and then bridging this scattering theory to the full system of linearised Einstein equations by incorporating the constraints \bref{eq:227intro1} and \bref{eq:228intro1}. A complete treatment of the full system will appear in the forthcoming \cite{M2050}. \\
\indent To elaborate on the ideas involved we go through a quick survey of the history of the subject. In \Cref{RedshiftScalar} we review known scattering theory for the scalar wave equation highlighting the role of redshift as a feature of scattering on black hole backgrounds. \Cref{LinearisedGravity} is a survey of the difficulties encountered in the study of scattering for the (linearised) Einstein equations, and will motivate and introduce the main results. \Cref{IntroResults} contains a preliminary statement of the results of this paper. \Cref{subsection 1.4 outline} contains an outline of the structure of the paper.
\subsection{Scattering for the scalar wave equation and the redshift effect}\label{RedshiftScalar} It is clear that understanding scattering for the scalar wave equation
\begin{align}\label{wave equation}
    \Box_{g} \phi=0
\end{align}
on a fixed Schwarzschild background \bref{SchwMetric} is a necessary prerequisite for our scattering problem, and already at this level we see many of the difficulties that characterise the evolution of perturbations to black holes. Much of the historical literature on scattering for \bref{wave equation} concerns the Schr\"odinger-like equation that results from a formal separation of \bref{wave equation} and governs the radial part. While this leads to important insights, it does not lead on its own to a satisfactory answer to points \ref{QI}, \ref{QII}, \ref{QIII} above. \\
\indent The first result on physical-space scattering for \bref{wave equation} on \bref{SchwMetric} goes back to Dimock and Kay \cite{DimockKayI}, who applied the Lax--Philips scattering theory to the scalar wave equation on the Schwarzschild spacetime. In \cite{Friedlander}, Friedlander's use of the radiation field at null infinity to describe future scattering states initiated an alternative method from the Lax--Philips formalism to a more geometric treatment of the notion of scattering states, and subsequent works have largely adhered to this point of view, see the discussion by Nicolas \cite{Nicolas}. The state of the art in this area is the work of Dafermos, Rodnianski and Shlapentokh-Rothman \cite{DRSR14}, where a complete understanding of scattering for the wave equation \bref{wave equation} on the Kerr exterior is laid out. The scattering problem for the scalar wave equation \bref{wave equation} on the extremal Reissner--Nordstr\"om background was definitively resolved in \cite{AAG19}. In the case of asymptotically de-Sitter black holes, we note the result \cite{HafnerGerardGeorgescu} on asymptotic completeness for the Klein--Gordon equation restricting to solution of fixed azimuthal modes against a very slowly rotating Kerr--de-Sitter black hole. Scattering for \bref{wave equation} has also been considered on the interior of the Reissner--Nordstr\"om black hole by Kehle and Shlapentokh-Rothman \cite{KSR18}.\\
\indent What leads to the rich theory available to \bref{wave equation} is the fact that it comes with a natural Lagrangian structure with which we can associate conservation laws encoded in the energy-momentum tensor:
\begin{align}
    T_{\mu\nu}[\phi]=\partial_\mu \phi \; \partial_\nu \phi-\frac{1}{2}g_{\mu\nu}\;\partial_\alpha\phi \;\partial^\alpha \phi,
\end{align}
which satisfies $\nabla_\mu T^\mu{}^\nu[\phi]=0$. Since the vector field $T:=\partial_t$ generates an isometry, classical scattering theory immediately suggests the class of solutions of finite $T$-energy, defined as the flux on a spacelike or null hypersurface of the quantity
\begin{align}
     n^\mu J^T_\mu[\phi],
\end{align}
where $J^X[\phi]_\mu=T_{\mu\nu}[\phi]X^\nu$ and $n^\mu$ is the vector field normal to the hypersurface, as this flux is non-negative definite and conserved. Solutions to \bref{wave equation} arising from suitable Cauchy data have sufficiently tame asymptotics to induce smooth radiation fields on $\mathscr{I}^+$ and $\mathscr{H}^+$. The conservation of $T$-energy allows us to resolve the scattering problem by constructing an isomorphism between the space of Cauchy data of finite energy and the corresponding space of radiation fields. With this, the answer to the questions \ref{QI},  \ref{QII}, \ref{QIII} of scattering theory for equation \bref{wave equation} is in the affirmative. \\
\indent At the same time, the fact that the vector field $T$ becomes null on the event horizon points to a deficiency, since the  $T$-energy density then loses control over some derivatives and the norm on the event horizon defined by the $T$-energy,
\begin{align}\label{horizon energy}
    \int_{\mathscr{H}^+}J^T_\mu[\phi]n^\mu_{\mathscr{H}^+},
\end{align}
is degenerate. The energy density observed along a horizon-penetrating timelike curve is better described by $J^N_\mu[\phi]$ for a timelike vector field $N$, but such a vector field cannot be Killing everywhere. The flux of this quantity is therefore not conserved and new issues appear, paramount among which is the \textit{redshift effect}.\\
\indent An intuitive hint of the role played by the redshift effect is the exponential decay in frequency that affects signals originating near the event horizon by the time they reach late-time observers, which relates to the divergence of outgoing null geodesics near the event horizon towards the future. It turns out that this effect can be exploited to produce nondegenerate energies useful for evolution in the future direction, precisely by choosing a timelike $N$ to be  a time-translation invariant vector field measuring the separation of null geodesics near the event horizon, see \cite{DR05}. In addition to using $N$ as a multiplier $X=N$, key to this method is the fact that commuting the wave equation \bref{wave equation} with such $N$ produces terms of lower order derivatives that come with a good sign when estimating the solution forwards. This can be traced to the positivity of the surface gravity; the fact that on $\mathscr{H}^+$, $\nabla_T T=\kappa T$ with $\kappa>0$. See \cite{DR08} for a detailed exposition.\\
\indent Unfortunately, when it comes to backwards evolution the technique described above does not work, as the redshift effect in the forwards evolution problem turns to a deleterious blueshift effect when evolving towards the past, and it is not possible to use the energy associated with $N$ to bound the solution in the backwards direction. Furthermore, it can be shown that there exists a large class of scattering data having a finite $N$-energy on the future event horizon $\mathscr{H}^+$ whose $N$-energy blows up evolving backwards, see \cite{DSR17}.\\
\indent Note that in the case of the Kerr exterior $(a\neq0)$ there is no obvious analogue of the $T$-energy scattering theory, as the stationary Killing vector field becomes spacelike in the ergoregion and therefore its flux no longer has a definite sign. Therefore, superradiance features as an additional aspect of scattering theory. One cannot hope for a unitary map, but one can still hope for a bounded invertible map. In view of the above discussion, the $N$-energy space is not appropriate however. One of the difficulties is indeed identifying the correct notion of energy.  See \cite{DRSR14} for the detailed treatment.
\subsection{Linearised gravity and the Teukolsky equations}\label{LinearisedGravity}
The above discussion involves linear \textit{scalar} perturbations only, i.e.~solutions to \bref{wave equation}, and little is known about the scattering theory of the Einstein equations even when linearised, see \cite{Chandrasekhar} and \cite{HandlerFuttermanMatzner} for a survey. Indeed, a comprehensive study of scattering under the Einstein equations \bref{EVE} on black hole exteriors involves and subsumes major aspects of the study of black hole stability. To date, full nonlinear stability for an asymptotically flat spacetime has only been satisfactorily proven for Minkowski space, see \cite{Ch-K}, \cite{Lin-Rod} for instance. For asymptotically flat black holes, stability results against generic perturbations exist only for the linearised Einstein equations, see \cite{DHR16} for the case of the Schwarzschild spacetime, \cite{DHR18}, \cite{Ma}, \cite{AnderssonKerr} and \cite{HVH19} for the case of very slowly rotating Kerr black holes, and \cite{SRTdC} for the general subextremal case. For the case of asymptotically de-Sitter black holes, results concerning the nonlinear stability of black hole solutions with positive cosmological constant do exist, see \cite{Hintz2018}.
\subsubsection{The Bianchi equations and the lack of a Lagrangian structure}\label{subsubsection 1.2.1 no lagrangian}
 In a spacetime satisfying the Einstein equations \bref{EVE} with a vanishing cosmological constant, the components of the Weyl curvature tensor satisfy the \textit{Bianchi equations}
\begin{align}\label{Bianchi}
    \nabla^a W_{abcd}=0.
\end{align} 
These equations, along with the equations defining the connection components, comprise the evolutionary content of the Einstein equations \bref{EVE}. Importantly, the Bel--Robinson tensor
\begin{align}
Q_{abcd}=W_{aecf} W_b{}^e{}_d{}^f + {}^*W_{aecf} {}^*W_b{}^e{}_d{}^f
\end{align}
acts as an energy-momentum tensor for the Bianchi equations. Upon linearising these equations against the background of Minkowski space, this structure survives in the linearised equations and allows to estimate the curvature components using the vector field method in the same way that it was applied to study the scalar wave equation, as was done in \cite{Ch-K-linear}. In fact, the vector field method applied using the Bel--Robinson tensor was key to the proof of nonlinear stability of the Minkowski spacetime by Christodoulou and Klainerman in \cite{Ch-K}, and it is possible to use this strategy to study scattering for small perturbations to the Minkowski spacetime evolving according to the nonlinear Einstein equations \bref{EVE}.\\
\indent Unfortunately, this structure is lost in the process of linearising around black holes, where the connection components couple to the curvature in a way that destroys the Lagrangian structure of the equations \bref{Bianchi}: in terms of a formal expansion of perturbed quantities of the form
\begin{align}
    \bm{g}=g\;+\stackrel{\;\;\mbox{\scalebox{0.4}{(1)}}}{\epsilon g}, \qquad \bm{\Gamma}=\Gamma+\stackrel{\;\;\mbox{\scalebox{0.4}{(1)}}}{\epsilon\; \Gamma}, \qquad \bm{R}=R+\stackrel{\;\;\mbox{\scalebox{0.4}{(1)}}}{\epsilon R},
\end{align}
the linearised version of equations \bref{Bianchi} have the schematic form
\begin{align}\label{coupling}
   \stackrel{\;\;\;\;\mbox{\scalebox{0.4}{(1)}}}{\nabla \; W}+\stackrel{\mbox{\scalebox{0.4}{(1)}}\;\;\;\;\;\;}{\Gamma\; W}=0.
\end{align}
Therefore, it is not possible to directly use the Bianchi equations alone to prove boundedness and decay results for curvature components independently of the connection components. See the discussion in \cite{DHR16}, \cite{DHR17}. 
\subsubsection{Double null gauge}\label{subsubsection 1.2.2 double null gauge}
It is important to note that the formulation of the problem depends crucially on the choice of gauge. It turns out that working with a \textit{double null gauge} is particularly useful to manifest a special structure in the linearised Einstein equations that reveals an alternative method to control curvature. This gauge leads to a well-posed reduction of the linearised Einstein equations around Schwarzschild, arising from a well-posed reduction of the full Einstein equations (see \cite{DHR16} and \cite{Ch-K}).\\
\indent A double null gauge is a coordinate system $(\bm{u},\bm{v},\bm{\theta}^A)$ that foliates spacetime with two families of ingoing and outgoing null hypersurfaces. In this gauge we decompose the curvature and connection components in terms of $\mathcal{S}_{\bm{u},\bm{v}}$-tangent tensor fields, where $\mathcal{S}_{\bm{u},\bm{v}}$ is the compact 2-dimensional manifold where the null hypersurfaces of constant $\bm{u}, \bm{v}$ intersect (see already \Cref{section 2 preliminaries} and \Cref{Appendix B Double null guage}). On the exterior of the Schwarzschild spacetime, the Eddington--Finkelstein null coordinates $(u,v,\theta^A)$ provide an example of this gauge (where $\mathcal{S}_{u,v}$ are just standard spheres). \\
\indent For an example of the resulting equations, the linearised curvature components $\overone{\alpha}_{AB}=\overone{W}_{A4B4}$ and $\overone{\beta}_A=\overone{W}_{A434}$ obey the transport equations
\begin{align}\label{example1}
    \frac{1}{\Omega}\nablagml r\Omega^2\overone{\alpha}\;=-2r\fancydstar_2 \Omega\overone{\beta} +\frac{6M}{r^2}\Omega\overone{\hat{\chi}}, \qquad\qquad \nablav r^4\Omega\overone{\beta}-2M r^2\Omega\overone{\beta}\;= r\slashed{div}\;r^3\Omega^2\overone{\alpha},
\end{align}
where $\Omega^2=\left(1-\frac{2M}{r}\right)$,  $\slashednabla_4,\slashednabla_3$ denote the projections of the null covariant derivatives to $S^2_{u,v}$ and $\overone{\hat{\chi}}$ denotes the linearised outgoing shear. The coupling to the connection components means we must simultaneously consider the connection components like $\overone{\hat\chi}$, which satisfy transport equations of a similar form, for example:
\begin{align}\label{example2}
    \nablav\; r\Omega \overone{\hat{\chi}}+\Big(1-\frac{4M}{r}\Big)\Omega \overone{\hat{\chi}}=-r\Omega^2\overone{\alpha}.
\end{align}
\indent We note that in this formulation, we can see the presence of a \textit{blueshift} effect in the linearised Einstein equations by observing that the second equation of \bref{example1} above carries a lower order term with a sign that forces the solution to grow exponentially when evolved forward in a neighborhood of the horizon. This appears to be an essential feature of working with tensorial quantities decomposed using null frames.
\subsubsection{The Teukolsky equations}\label{subsubsection 1.2.3 Teukolsky}
A quick glance at \bref{example1}, \bref{example2} reveals that we can derive a decoupled equation for $\overone{\alpha}$ alone by acting on the first equation of \bref{example1} with $\nablav$ and following through the remaining equations to discover that $\overone{\alpha}$ obeys the $+2$ Teukolsky equation \bref{wave equation +}. The linearisation of the component $\bm{\underline\alpha}_{AB}=W_{A4B4}$ can be shown to obey \bref{wave equation -} by a similar logic, see \Cref{subsection 2.2 Linearised Einstein equations in double null gauge} for the full list of the linearised Einstein equations around the Schwarzschild background.\\
\indent The derivation of \bref{wave equation +}, \bref{wave equation -} by Bardeen and Press \cite{Bardeen-Press} for perturbations around Schwarzschild and their extension to the Kerr black holes by Teukolsky \cite{Teu73} (using the Newman--Penrose formalism) was a game changer in the study of linearised gravity. If one can estimate solutions to the Teukolsky equations (i.e.~equations \bref{wave equation +}, \bref{wave equation -} on Schwarzschild), one can hope to make use of the hierarchical nature of the linearised Einstein equations in double null gauge (as manifest in \bref{example1}, \bref{example2} for example) to estimate the remaining components. \\
\indent Unfortunately, however, having arrived at the decoupled wave equations \bref{wave equation +}, \bref{wave equation -} for the components $\overone{\alpha}, \overone{\underline\alpha}$, the essential difficulty in dealing with the linearised Einstein equations is still inherited by the Teukolsky equations \bref{wave equation +}, \bref{wave equation -}, in the sense that equations \bref{wave equation +}, \bref{wave equation -}, taken in isolation, also suffer from the lack of a variational principle, and neither \bref{wave equation +} nor \bref{wave equation -} has its own energy-momentum tensor. This is related to the 1st order null derivative term on the left hand side of \bref{wave equation +}, \bref{wave equation -}. These first order terms are reminiscent of the wave equation \bref{wave equation} when commuted with the redshift vector field $N$ (note in particular that the 1st order term in the $-2$ Teukolsky equation \bref{wave equation -} has a redshift sign near $\mathscr{H}^+$, while the $+2$ has a 1st order term with a blueshift sign near $\mathscr{H}^+$). This issue meant that the Teukolsky equations \bref{wave equation +}, \bref{wave equation -}, despite their decoupling, have remained immune to known methods for a long time. 
\subsubsection{Chandrasekhar-type transformations in physical space}\label{subsubsection 1.2.4 DHR}
In \cite{DHR16}, Dafermos, Holzegel and Rodnianski succeed in deriving boundedness and decay estimates for \bref{wave equation +} and \bref{wave equation -} and they subsequently prove the linear stability of the Schwarzschild solution in double null gauge. Key to their work is the exploitation of a physical space version of a trick due to Chandrasekhar \cite{Chandrasekhar}, which works by commuting derivatives in the null directions past the equations. This commutation removes the first order derivative terms and reduces the equations \bref{wave equation +}, \bref{wave equation -} to a familiar form:
\begin{align}\label{RWintro}
    \nablau\nablav\overone\Psi-\Omega^2\slashed{\Delta}\overone\Psi+V(r)\overone\Psi=0,
\end{align}
where $V(r)=\frac{\Omega^2(3\Omega^2+1)}{r^2}$ and
\begin{align}\label{transport}
    \overone\Psi=\left(\frac{r^2}{\Omega^2}\nablau\right)^2r\Omega^2\overone\alpha.
\end{align}
The same applies to $\overone{\underline\alpha}$ by differentiating in the $4$- direction instead and we obtain a quantity $\overone{\underline\Psi}$ satisfying \bref{RWintro} via
\begin{align}\label{transport 2}
    \overone{\underline\Psi}=\left(\frac{r^2}{\Omega^2}\nablav\right)^2r\Omega^2\overone{\underline\alpha}.
\end{align}
\indent Equation \bref{RWintro} is the well-known Regge--Wheeler equation, which first appeared in the context of the theory of metric perturbations studied by Regge and Wheeler \cite{ReggeWheeler}, Vishveshwara \cite{Vishveshwara}, and Zerilli \cite{Zerilli} to describe gauge invariant combinations of the metric perturbations. The Regge--Wheeler equation \bref{RWintro} has a very similar structure to the equation that governs the radiation field of the scalar wave equation \bref{wave equation}, and in particular the vector field method can be adapted to study \bref{RWintro}. This is what was done in \cite{DHR16} to obtain boundedness and decay estimates for solutions of \bref{RWintro}. These estimates for \bref{RWintro} can in turn be used to estimate $\overone\alpha, \overone{\underline\alpha}$ \textit{by regarding \bref{transport} and its $\overone{\underline\alpha}$ counterpart as transport equations for $\overone\alpha, \overone{\underline\alpha}$}. For this to work, it was fundamental that a sufficiently strong decay statement is available for solutions of \bref{RWintro} for a nondegenerate energy (i.e.~the analogue of the $N$-energy above).\\
\indent Note that in the case of the Kerr spacetime $a\neq 0$, the strategy outlined above suffers from the fact that the analogues of \bref{RWintro} are coupled to $\overone{\alpha}, \overone{\underline\alpha}$ via $a$. Nevertheless, it is possible to apply the same strategy to obtain boundedness and decay results for solutions to the Teukolsky equations, see \cite{DHR18} and \cite{Ma} for the case of the very slowly rotating Kerr exterior $|a|\ll M$ and the very recent \cite{SRTdC} for the full subextremal range $|a|<M$. For the case of the extremal Kerr exterior $a=M$, see \cite{Rita2019}, \cite{Lucietti_2012}.\\
\indent The first preliminary goal of our work will be to analyse the Regge--Wheeler equation \bref{RWintro} from the point of view of scattering. The fact that the conservation of the $T$-energy leads to a scattering theory for the scalar wave equation \bref{wave equation} means one can expect to prove an analogous statement for the Regge--Wheeler equation using analogous methods. This will be the content of \textbf{Theorem 1} (see \Cref{subsubsection 1.3.1 scattering for RW}). 
\subsubsection{Reconstructing curvature from the Regge--Wheeler equation}\label{subsubsection 1.2.5 RW}
\indent Starting from such a scattering theory for the Regge--Wheeler equation \bref{RWintro}, one can hope to apply the strategy used in \cite{DHR16} to construct a scattering theory for the Teukolsky equations \bref{wave equation +} and \bref{wave equation -} via the transport relations \bref{transport} and \bref{transport 2}. It is however far from clear that the transport equations \bref{transport}, \bref{transport 2} can lead to a suitable scattering theory, in particular one that could in turn lead to a scattering theory for the linearised Einstein equations. The central question we aim to address is whether the $T$-energy obtained via the Regge--Wheeler equation could define a Hilbert space of scattering states for solutions to \bref{wave equation +}, \bref{wave equation -}, for which the central questions of scattering theory (points \ref{QI}, \ref{QII}, \ref{QIII} above) could be answered. \\
\indent Adapting the strategy above to a scattering setting based on $T$-energies, we succeed in constructing such a scattering theory for the Teukolsky equations answering \ref{QI}, \ref{QII}, \ref{QIII} in the affirmative. This will lead to \textbf{Theorem 2} of this paper (see \Cref{subsubsection 1.3.2 scattering for teukolsky}).
\subsubsection{The Teukolsky--Starobinsky correspondence}\label{subsubsection 1.2.6 TS}
Finally, we treat what is known as the Teukolsky--Starobinsky correspondence. The Teukolsky--Starobinsky correspondence is the study of the relationship between $\overone{\alpha}, \overone{\underline\alpha}$ using \bref{eq:227intro1}, \bref{eq:228intro1} and the Teukolsky equations \bref{wave equation +}, \bref{wave equation -}, independently of the remaining components of a solution to the linearised Einstein system.
The idea that knowing either $\overone{\alpha}$ or $\overone{\underline\alpha}$ uniquely determines the other via \bref{eq:227intro1}, \bref{eq:228intro1} permeates the literature on the Einstein equations since the appearance of the constraints in \cite{TeukP74}, \cite{StarC}, but little has been done in the way of a systematic study of the combined system consisting of the Teukolsky equations \bref{wave equation +}, \bref{wave equation -} and the constraints \bref{eq:227intro1}, \bref{eq:228intro1}, governing a pair $\overone\alpha, \overone
{\underline\alpha}$. \\
\indent The constraints \bref{eq:227intro1}, \bref{eq:228intro1} provide a bridge between the scattering theory we construct for equations \bref{wave equation +}, \bref{wave equation -} and the full linearised Einstein equations. This is because scattering for the linearised Einstein equations would involve scattering data for the metric components, from which data for only one of $\overone\alpha$ or $\overone{\underline\alpha}$ could be constructed from the scattering data for the metric on each component of the asymptotic boundary. One can hope to use the identities \bref{eq:227intro1}, \bref{eq:228intro1} to obtain scattering data for either $\overone\alpha$ or $\overone{\underline\alpha}$ out of the other, but it is entirely unclear whether we would obtain scattering data that are compatible with the scattering theory constructed here for \bref{wave equation +}, \bref{wave equation -}, or even whether the system consisting of \bref{wave equation +}, \bref{wave equation -}, \bref{eq:227intro1}, \bref{eq:228intro1} is well-posed. In the context of scattering, we are specifically interested in whether the operators involved on each side of the identities \bref{eq:227intro1}, \bref{eq:228intro1} are invertible on the spaces of scattering states, and we would like to know whether, given scattering data for $\overone\alpha, \overone{\underline\alpha}$ related via \bref{eq:227intro1}, \bref{eq:228intro1}, the ensuing solutions to \bref{wave equation +}, \bref{wave equation -} would  in turn satisfy \bref{eq:227intro1}, \bref{eq:228intro1}.\\
\indent Interestingly, it turns out that the study of constraints \bref{eq:227intro1}, \bref{eq:228intro1} is much more transparent when done via scattering rather than directly via the Cauchy problem, and combining this with asymptotic completeness will answer the question of well-posedness for the system \bref{wave equation +}, \bref{wave equation -}, \bref{eq:227intro1}, \bref{eq:228intro1}. We also find that it is only in the context where solutions to \bref{wave equation +}, \bref{wave equation -} are studied on the entirety of the exterior region that the constraints \bref{eq:227intro1}, \bref{eq:228intro1} are sufficient to determine $\overone\alpha$ completely from $\overone{\underline\alpha}$ and vice versa. Scattering necessarily involves considering solutions globally on the exterior. These considerations are the subject of \textbf{Theorem 3}. \\
\indent A corollary to our main results is that one may formulate a scattering statement for a combined pair $(\overone\alpha,\overone{\underline\alpha})$ satisfying the Teukolsky equations \bref{wave equation +}, \bref{wave equation -} and the constraints \bref{eq:227intro1}, \bref{eq:228intro1} (this is \textbf{Corollary 1}, see \Cref{subsection 4.4 Corollary 1: mixed scattering}). One can then hope that such a scattering statement would provide a bridge towards scattering for the full linearised Einstein equations, taking into account \Cref{example2} relating $\overone{\alpha}$ to $\overone{\hat{\chi}}$ and counterpart equation relating $\overone{\underline\alpha}$ to $\overone{\hat{\underline\chi}}$. We will immediately remark at the end of this introduction on how to formally derive a conservation law at the level of the shears $\overone{\hat{\chi}}$, $\overone{\hat{\underline\chi}}$ which excludes the possibility of superradiant reflection (see \bref{conservation law} of \Cref{subsubsection 1.3.4 corollary}). This will be treated in detail again in the upcoming \cite{M2050} as part of a complete scattering theory for the linearised Einstein equation in double null gauge.
\subsection{Scattering maps}\label{IntroResults}
The following are preliminary statements of the results of this work, with detailed statements to follow in the body of the paper (see \cref{section 4 main theorems}).
\subsubsection{Scattering for the Regge--Wheeler equation}\label{subsubsection 1.3.1 scattering for RW}
We begin by stating the result for the Regge--Wheeler equation \bref{RWintro} (we omit the superscript $\overone{{}}$ in what follows). We show that a solution arising from Cauchy data with initially finite $T$-energy gives rise to a set of radiation fields in the limit towards $\mathscr{I}^+, \mathscr{H}^+$, from which the solution can be recovered. The choice of the Cauchy surface does not affect the fact that the flux of the $T$-energy defines a Hilbert space norm on Cauchy data. For the surface $\overline\Sigma=\{t=0\}$, this flux is given by
\begin{align}\label{RWfluxT}
    \big\|(\Psi|_{\overline\Sigma},\slashednabla_{n_{\overline{\Sigma}}}\Psi|_{\overline\Sigma})\big\|^2_{\mathcal{E}^T_{\overline\Sigma}}=\int_{\overline\Sigma}dr\sin\theta d\theta d\phi\;|\slashednabla_{n_{\overline\Sigma}}\Psi|^2+\Omega^2|\slashednabla_r\Psi|^2+|\slashednabla\Psi|^2+\frac{3\Omega^2+1}{r^2}|\Psi|^2.
\end{align}
Conservation of the $T$-energy suggests Hilbert space norms on $\mathscr{I}^+, \mathscr{H}^+$:
\begin{align}\label{RWfluxIH}
    \|\bm{\uppsi}_{\mathscr{I}^+}\|_{\mathcal{E}^T_{\mathscr{I}^+}}^2=\int_{\mathscr{I}^+}du\sin\theta d\theta d\phi\; |\partial_u\bm{\uppsi}_{\mathscr{I}^+}|^2,\qquad\qquad \|\Psi_{{\mathscr{H}^+}}\|_{\mathcal{E}^T_{\mathscr{H}^+}}^2=\int_{\mathscr{H}^+}du\sin\theta d\theta d\phi\;|\partial_v\bm{\uppsi}_{\mathscr{H}^+}|^2.
\end{align}
The Hilbert spaces $\mathcal{E}^T_{\overline{\Sigma}}, \mathcal{E}^T_{\overline{\mathscr{H}^+}},\mathcal{E}^T_{\overline{\mathscr{I}^+}}$ are defined to be the completion of smooth, compactly supported data under the norms defined in \bref{RWfluxT}, \bref{RWfluxIH} and the spaces $\mathcal{E}^T_{\mathscr{H}^-}, \mathcal{E}^T_{\mathscr{I}^-}$ are defined analogously.
\begin{theorem*}\label{Theorem 1}
Forward evolution under the Regge--Wheeler equation \bref{RWintro} extends to a unitary Hilbert space isomorphism
\begin{align}
    \mathscr{F}^+:\mathcal{E}^T_{\overline\Sigma} \longrightarrow \mathcal{E}^T_{\overline{\mathscr{H}^+}}\oplus\mathcal{E}^T_{\mathscr{I}^+}.
\end{align}
A similar statement holds for scattering towards $\mathscr{H}^-, \mathscr{I}^-$. As a corollary, we obtain the Hilbert space isomorphism
\begin{align}
    \mathscr{S}:\mathcal{E}^T_{\mathscr{H}^-}\oplus\mathcal{E}^T_{\mathscr{I}^-}\longrightarrow\mathcal{E}^T_{\mathscr{H}^+}\oplus\mathcal{E}^T_{\mathscr{I}^+}.
\end{align}
\end{theorem*}
The precise statement of this result is contained in \Cref{forwardRW,,backwardRW,,RW isomorphisms} of \Cref{subsection 4.1 Theorem 1}.\\
\indent Note that Theorem 1 can be applied to the study of scattering for the linearised Einstein equations in the Regge--Wheeler gauge, see also the recent \cite{TruongConformal}.
\subsubsection{Scattering for the Teukolsky equations}\label{subsubsection 1.3.2 scattering for teukolsky}
Given $\alpha$ or $\underline\alpha$ solving the Teukolsky equations \bref{wave equation +}, \bref{wave equation -}, the weighted null derivatives $\Psi, \underline\Psi$ defined by \bref{transport}, \bref{transport 2} satisfy the Regge--Wheeler equation \bref{RWintro}, so we can try to use Theorem 1 to construct a scattering theory for $\alpha, \underline\alpha$ using the spaces of scattering states associated to \bref{RWintro}: \\
\indent Let $(\upalpha,\upalpha')$, $(\underline\upalpha,\underline\upalpha')$ be Cauchy data for \bref{wave equation +}, \bref{wave equation -} respectively on $\overline{\Sigma}$ and define
\begin{align}
    \|(\upalpha,\upalpha')\|^2_{\mathcal{E}^{T,+2}_{\overline\Sigma}}:=\|(\Psi,\slashednabla_{n_{\overline\Sigma}}\Psi)\|^2_{\mathcal{E}^T_{\overline{\Sigma}}},\qquad\qquad\|(\underline\upalpha,\underline\upalpha')\|^2_{\mathcal{E}^{T,-2}_{\overline\Sigma}}:=\|(\underline\Psi,\slashednabla_{n_{\overline\Sigma}}\underline\Psi)\|^2_{\mathcal{E}^T_{\overline{\Sigma}}}.
\end{align}
The expressions $ \|\;\|^2_{\mathcal{E}^{T,+2}_{\overline\Sigma}}, \|\;\|^2_{\mathcal{E}^{T,-2}_{\overline\Sigma}}$ turn out indeed to be norms on smooth, compactly supported data sets on $\overline\Sigma$ and thus they define Hilbert space norms on the completions of such data. Note that the values on $\overline\Sigma$ of $\Psi,\underline\Psi$ and their derivatives can be computed locally using the Teukolsky equations \bref{wave equation +}, \bref{wave equation -}, out of higher order derivatives of the initial data $(\upalpha,\upalpha')$, $(\underline\upalpha,\underline\upalpha')$ on $\overline\Sigma$.\\
\indent As mentioned earlier, the energies defining the Hilbert spaces of scattering states for the Teukolsky equations stem from the $T$-energy associated to the Regge--Wheeler equations. Remarkably, on $\mathscr{I}^\pm, \mathscr{H}^\pm$, the radiation fields of $\Psi, \underline\Psi$ are related to those of $\alpha, \underline\alpha$ by tangential derivatives, and it is possible to find meaningful expressions for the corresponding norms on $\mathscr{I}^\pm, \overline{\mathscr{H}^\pm}$ directly in terms of the radiation fields of $\alpha, \underline\alpha$.
\begin{theorem*}\label{Theorem 2}
For the Teukolsky equations \bref{wave equation +}, \bref{wave equation -} of spins $\pm2$, evolution from smooth, compactly supported data on a Cauchy surface extends to unitary Hilbert space isomorphisms:
\begin{align}
    {}^{(+2)}\mathscr{F}^{+}:\mathcal{E}^{T,+2}_{\overline\Sigma}\longrightarrow \mathcal{E}^{T,+2}_{\mathscr{I}^+}\oplus\mathcal{E}^{T,+2}_{\overline{\mathscr{H}^+}},\qquad\qquad{}^{(-2)}\mathscr{F}^{+}:\mathcal{E}^{T,-2}_{\overline\Sigma}\longrightarrow \mathcal{E}^{T,-2}_{\mathscr{I}^+}\oplus\mathcal{E}^{T,-2}_{\overline{\mathscr{H}^+}},\\
    {}^{(+2)}\mathscr{F}^{-}:\mathcal{E}^{T,+2}_{\overline\Sigma}\longrightarrow \mathcal{E}^{T,+2}_{\mathscr{I}^-}\oplus\mathcal{E}^{T,+2}_{\overline{\mathscr{H}^-}},\qquad\qquad{}^{(-2)}\mathscr{F}^{-}:\mathcal{E}^{T,-2}_{\overline\Sigma}\longrightarrow \mathcal{E}^{T,-2}_{\mathscr{I}^-}\oplus\mathcal{E}^{T,-2}_{\overline{\mathscr{H}^-}}.
\end{align}
The spaces of past/future scattering states  $\mathcal{E}^{T,\pm2}_{\mathscr{I}^\pm},\mathcal{E}^{T,\pm2}_{\mathscr{H}^\pm},$ are the Hilbert spaces obtained by completing suitable smooth, compactly supported data on $\mathscr{I}^\pm, \mathscr{H}^\pm$ under the corresponding norms in the following:

\begin{changemargin}{-1cm}{2cm}
\begin{center}
\setstretch{1.5}
\begin{tikzpicture}[scale=0.6,on grid]
\node (I)    at ( 0,0)   {};

\path % Four conners of the right diamond (no labels this time)
   (I) +(90:4)  coordinate (Itop) coordinate[label=90:$i^+$]
       +(-90:4) coordinate (Ibot) coordinate[label=-90:$i^-$]
       +(180:4) coordinate (Ileft)
       +(0:4)   coordinate (Iright) coordinate[label=0:$i^0$]
       ;
% No text this time in the next diagram
\draw  (Ileft) --  node[align=center,yshift=15,xshift=15]{$\Big\|(\mathring{\slashed{\Delta}}-2)(\mathring{\slashed{\Delta}}-4)\left(2M\int^{\infty}_v d\bar{v}e^{\frac{1}{2M}({v}-\bar{v})}\Omega^2\alpha\right)\Big\|^2_{L^2(\overline{\mathscr{H}^+})}\qquad\qquad\qquad\qquad\qquad\qquad\qquad\qquad\qquad\qquad\qquad\qquad$\\$+\Big\|6M\partial_v\left(2M\int^{\infty}_v d\bar{v}e^{\frac{1}{2M}({v}-\bar{v})}\Omega^2\alpha\right)\Big\|_{L^2(\overline{\mathscr{H}^+})}^2\qquad\qquad\qquad\qquad\qquad\qquad\qquad\qquad\qquad\qquad\qquad\;\qquad$} node[rotate=45,below]{$\overline{\mathscr{H}^+}$} (Itop) ;
\draw  (Ileft) -- node[yshift=-15,xshift=15]{$\Big\|2M\left(-2(2M\partial_u)+3(2M\partial_u)^2-(2M\partial_u)^3\right)2M\Omega^{-2}\alpha\Big\|^2_{L^2(\overline{\mathscr{H}^-})}\qquad\qquad\qquad\qquad\qquad\qquad\qquad\qquad\qquad\qquad\qquad\qquad\qquad$} node[rotate=-45,above]{$\overline{\mathscr{H}^-}$} (Ibot) ;
\draw[dash dot dot] (Ibot) -- node[align=center][yshift=-10,xshift=-15]{$\qquad\qquad\qquad\qquad\qquad\qquad\qquad\qquad\qquad\Big\|6M\upalpha_{\mathscr{I}^-}\Big\|^2_{L^2(\mathscr{I}^-)}$\\[1mm] $\qquad\qquad\qquad\qquad\qquad\qquad\qquad\qquad\qquad+\left\|(\mathring{\slashed{\Delta}}-2)(\mathring{\slashed{\Delta}}-4)\left(\int^{v}_{-\infty}\upalpha_{\mathscr{I}^-} d\bar{v}\right)\right\|^2_{L^2(\mathscr{I}^-)}$} node[rotate=45,above]{$\mathscr{I}^-$}(Iright) ;
\draw[dash dot dot] (Iright) -- node[yshift=10,xshift=-10]{$\qquad\qquad\;\;\;\;\;\;\;\;\;\;\;\;\;\;\;\;\;\;\;\;\qquad\left\|(\partial_u)^3\upalpha_{\mathscr{I}^+}\right\|^2_{L^2(\mathscr{I}^+)}$} node[rotate=-45,below]{$\mathscr{I}^+$}(Itop) ;

\filldraw[white] (Itop) circle (3pt);
\draw[black] (Itop) circle (3pt);

\filldraw[white] (Ibot) circle (3pt);
\draw[black] (Ibot) circle (3pt);
\draw[black] (Ileft) circle (3pt);
\filldraw[black] (Ileft) circle (3pt);
\filldraw[white] (Iright) circle (3pt);
\draw[black] (Iright) circle (3pt);
\end{tikzpicture}

\end{center}

%%%%%%%%%%%%%%%%%%%%%%%%%%
\end{changemargin}
\begin{changemargin}{-1.4cm}{2cm}

\begin{center}
\begin{tikzpicture}[scale=0.6]
\node (I)   at ( 0,0)   {};

\path % Four conners of the right diamond (no labels this time)
   (I) +(90:4)  coordinate (Itop) coordinate[label=90:$i^+$]
       +(-90:4) coordinate (Ibot) coordinate[label=-90:$i^-$]
       +(180:4) coordinate (Ileft)
       +(0:4)   coordinate (Iright) coordinate[label=0:$i^0$]
       ;
% No text this time in the next diagram
\draw  (Ileft) --  node[yshift=20,xshift=25]{$\Big\|2M\left(2(2M\partial_v)+3(2M\partial_v)^2+(2M\partial_v)^3\right)2M\Omega^{-2}\underline\alpha\Big\|^2_{L^2(\overline{\mathscr{H}^+})}\qquad\qquad\qquad\qquad\qquad\qquad\qquad\qquad\qquad\qquad\qquad\qquad\qquad\qquad$} node[rotate=45,below]{$\overline{\mathscr{H}^+}$} (Itop) ;
\draw  (Ileft) -- node[align=center][yshift=-10,xshift=10]{$\Big\|6M\partial_u\left(2M\int^{u}_{-\infty}d\bar{u}e^{\frac{1}{2M}(u-\bar{u})}\Omega^2\underline\alpha\right)\Big\|^2_{L^2(\overline{\mathscr{H}^-})}\qquad\qquad\qquad\qquad\qquad\qquad\qquad\qquad\qquad\qquad\qquad$\\[1mm] $+\left\|(\mathring{\slashed{\Delta}}-2)(\mathring{\slashed{\Delta}}-4)\left(2M\int^{u}_{-\infty}d\bar{u}e^{\frac{1}{2M}(u-\bar{u})}\Omega^2\underline\alpha\right)\right\|^2_{L^2(\overline{\mathscr{H}^-})}\qquad\qquad\qquad\qquad\qquad\qquad\qquad\qquad\qquad\qquad\qquad$} node[rotate=-45,above]{$\overline{\mathscr{H}^-}$} (Ibot) ;
\draw[dash dot dot] (Ibot) -- node[yshift=-12,xshift=-12]{$\qquad\qquad\;\;\;\;\;\;\;\;\;\;\;\;\;\;\qquad\qquad\left\|(\partial_v)^3\underline\upalpha_{\mathscr{I}^-}\right\|_{L^2(\mathscr{I}^-)}^2$} node[rotate=45,above]{$\mathscr{I}^-$}(Iright) ;
\draw[dash dot dot] (Iright) -- node[align=center][yshift=10,xshift=-20]{ $\qquad\qquad\qquad\qquad\qquad\qquad\qquad\qquad\qquad\qquad\left\|(\mathring{\slashed{\Delta}}-2)(\mathring{\slashed{\Delta}}-4)\left(\int^{u}_{-\infty}\underline\upalpha_{\mathscr{I}^+} d\bar{u}\right)\right\|^2_{L^2(\mathscr{I}^+)}$ \\$\qquad\qquad\qquad\qquad\qquad\qquad\qquad\qquad\qquad+\Big\|6M\underline\upalpha_{\mathscr{I}^+}\Big\|^2_{L^2(\mathscr{I}^+)}$} node[rotate=-45,below]{$\mathscr{I}^+$}(Itop) ;

\filldraw[white] (Itop) circle (3pt);
\draw[black] (Itop) circle (3pt);

\filldraw[white] (Ibot) circle (3pt);
\draw[black] (Ibot) circle (3pt);

\filldraw[white] (Iright) circle (3pt);
\draw[black] (Iright) circle (3pt);
\filldraw[black] (Ileft) circle (3pt);
\draw[black] (Iright) circle (3pt);
\end{tikzpicture}

\end{center}
\end{changemargin}
The maps ${}^{(\pm2)}\mathscr{F}^{\pm}$ lead to the Hilbert-space isomorphisms
\begin{align}
\begin{split}
    &\mathscr{S}^{+2}: \mathcal{E}^{T,+2}_{\mathscr{I}^+}\oplus\mathcal{E}^{T,+2}_{\overline{\mathscr{H}^+}}\longrightarrow \mathcal{E}^{T,+2}_{\mathscr{I}^-}\oplus\mathcal{E}^{T,+2}_{\overline{\mathscr{H}^-}},\\
    &\mathscr{S}^{-2}: \mathcal{E}^{T,-2}_{\mathscr{I}^+}\oplus\mathcal{E}^{T,-2}_{\overline{\mathscr{H}^+}}\longrightarrow \mathcal{E}^{T,-2}_{\mathscr{I}^-}\oplus\mathcal{E}^{T,-2}_{\overline{\mathscr{H}^-}}.
\end{split}
\end{align}
\end{theorem*}
\begin{remark*}
The scattering maps of Theorem 2 answer the questions \ref{QI}, \ref{QII}, \ref{QIII} posed at the beginning of the introduction. In particular, the issue of asymptotic completeness is answered in the sense that the spaces $\mathcal{E}^{T,\pm2}_{\overline{\Sigma}}$ include all smooth, compactly supported Cauchy data for \bref{wave equation +}, \bref{wave equation -} as dense subspaces.
\end{remark*}
\begin{remark*}\label{introduction regular frame norm}
As the Eddington--Finkelstein coordinate system degenerates at the bifurcation sphere $\mathcal{B}$, it is necessary to use a regular coordinate system, such as the Kruskal coordinates $U=e^{-\frac{u}{2M}}, V=e^{\frac{v}{2M}}$. In this coordinate system we see that $\overone{W}_{AVBV}\sim V^{-2}\Omega^2\alpha \sim U^2\Omega^{-2}\alpha$ and $\overone{W}_{AUBU}\sim  V^{2}\Omega^{-2}\alpha \sim U^{-2}\Omega^{2}\alpha$ extend regularly to the bifurcation sphere. The integrands defining $\mathcal{E}^{T,\pm2}_{\mathscr{H}^\pm}$ also extend regularly to the bifurcation sphere $\mathcal{B}$. For example,
\begin{align}\label{introduction regular frame expression H-}
     &-2(2M\partial_u)+3(2M\partial_u)^2-(2M\partial_u)^3\Omega^{-2}\alpha=U\partial_U^3 U^{2}\Omega^{-2}\alpha,
\end{align}
\begin{align}\label{introduction regular frame expression I-}
    \int^{\infty}_v e^{\frac{1}{2M}({v}-\bar{v})}\Omega^2\alpha\;d\bar{v}=V\int_{V}^{\infty} \overline{V}^{-2}\Omega^2\alpha\; d\overline{V}.
\end{align}
We take $L^2(\overline{\mathscr{H}^+})$ to be defined with respect to the measure $dv\sin\theta d\theta d\phi$, and we define $ L^2({\mathscr{I}^+})$ via the measure $du\sin\theta d\theta d\phi$. Analogous statements apply to $\mathscr{I}^-, \overline{\mathscr{H}^-}$.\\
\indent The detailed statement of Theorem 2 is contained in \Cref{+2 future forward scattering,,+2 future backward scattering,,+2 past forward scattering,,scatteringthm+2} of \Cref{subsubsection 4.2.1 scattering for the +2 equation}, and \Cref{-2 future forward scattering,,-2 future backward scattering,,-2 past forward scattering,,scatteringthm-2} of \Cref{subsubsection 4.2.2 Scattering for the -2 equation}.
\end{remark*}
\subsubsection{Teukolsky--Starobinsky correspondence}\label{subsubsection 1.3.3 TS}
\indent Finally, concerning the Teukolsky--Starobinsky correspondence relating $\alpha, \underline\alpha$, we may summarise our result as follows:
\begin{theorem*}\label{Theorem 3}
The constraints \bref{eq:227intro1}, \bref{eq:228intro1} can be used to define unitary Hilbert space isomorphisms:
\begin{align}
    \mathcal{TS}_{\mathscr{I}^+}:\mathcal{E}^{T,+2}_{\mathscr{I}^+}\longrightarrow\mathcal{E}^{T,-2}_{\mathscr{I}^+},\qquad\qquad\mathcal{TS}_{\mathscr{H}^+}:\mathcal{E}^{T,+2}_{\overline{\mathscr{H}^+}}\longrightarrow\mathcal{E}^{T,-2}_{\overline{\mathscr{H}^+}},
\end{align}
\begin{align}
    \mathcal{TS}=\mathcal{TS}_{\mathscr{H}^+}\oplus\mathcal{TS}_{\mathscr{I}^+}: \mathcal{E}^{T,+2}_{\overline{\mathscr{H}^+}}\oplus\mathcal{E}^{T,+2}_{\mathscr{I}^+}\longrightarrow\mathcal{E}^{T,-2}_{\overline{\mathscr{H}^+}}\oplus\mathcal{E}^{T,-2}_{\mathscr{I}^+}.
\end{align}
Applying $\mathcal{TS}$ to scattering data, one can associate to a solution to the $+2$ Teukolsky equation \bref{wave equation +} arising from smooth scattering data in $\mathcal{E}^{T,+2}_{\mathscr{I}^+}\oplus\mathcal{E}^{T,+2}_{\overline{\mathscr{H}^+}}$ a unique solution $\underline\alpha$ of the $-2$ Teukolsky equation \bref{wave equation -} with smooth scattering data in $\mathcal{E}^{T,-2}_{\mathscr{I}^+}\oplus\mathcal{E}^{T,-2}_{\overline{\mathscr{H}^+}}$ such that \bref{eq:227intro1}, \bref{eq:228intro1} are satisfied everywhere on the exterior region of Schwarzschild. 
\end{theorem*}
The map $\mathcal{TS}_{\mathscr{I}^+}$ is realised by taking the limit of constraint \bref{eq:227intro1} near $\mathscr{I}^+$ and inverting either side of the constraint on smooth, compactly supported scattering data, which are by definition dense subsets of  $\mathcal{E}^{T,\pm2}_{\mathscr{I}^+}$. The map $\mathcal{TS}_{\mathscr{H}^+}$ is obtained analogously by studying constraint \bref{eq:228intro1} near $\overline{\mathscr{H}^+}$. Note that in order to obtain a unique smooth radiation field $\upalpha_{\mathscr{H}^+}$ for the +2 Teukolsky equation \bref{wave equation +} on the event horizon starting from a radiation field  $\underline\upalpha_{\mathscr{H}^+}$ for the $-2$ equation \bref{wave equation -}, it is necessary to specify $\underline\upalpha_{\mathscr{H}^+}$ on the entirety of $\overline{\mathscr{H}^+}$, and vice versa for $\mathscr{I}^+$. Thus the isomorphisms $\mathcal{TS}_{\mathscr{I}^+}$, $\mathcal{TS}_{\mathscr{H}^+}$ can only be defined on spaces of scattering data that determine solutions to \bref{wave equation +}, \bref{wave equation -} \textit{globally} on the Schwarzschild exterior.\\
\indent In particular, note that spacetimes of Robinson--Trautman type are excluded from our scattering theory, see \Cref{section 9 TS correspondence} and  \Cref{Appendix A Robinson--Trautman}. The Robinson--Trautman spacetimes have the property that one of $\overone\alpha$ or $\overone{\underline\alpha}$ is non-trivial while the other is vanishing, and as such they pose a counterexample to the Teukolsky--Starobinsky correspondence if not properly formulated. We show that this possibility is eliminated when finite-energy scattering is considered globally on the entirety of the exterior of the Schwarzschild solution.\\
\indent The detailed statement of Theorem 3 is contained in \Cref{Theorem 3 detailed statement} of \Cref{subsection 4.3 the Teukolsky--Starobinsky identities}. See \Cref{section 9 TS correspondence} for the detailed treatment.
\subsubsection{A preview of scattering for the full linearised Einstein equations}\label{subsubsection 1.3.4 corollary}
In reference to Theorem 2, Theorem 3 allows us to bridge the scattering theory we build for the Teukolsky equations to develop scattering for the full system of linearised Einstein equations in double null gauge via the following corollary:
\begin{corollary*}\label{Corollary 1}
Given a smooth, compactly supported $\upalpha_{\mathscr{I}^-}$ on $\mathscr{I}^-$ such that $\int_{-\infty}^\infty d\bar{v} \; \upalpha_{\mathscr{I}^-}=0$, and an $\underline\upalpha_{\mathscr{H}^-}$ such that $U^{-2}\underline\upalpha_{\mathscr{H}^-}$ is smooth, compactly supported on $\overline{\mathscr{H}^-}$, there exists a unique smooth pair $(\alpha, \underline\alpha)$ on the exterior region of Schwarzschild, satisfying equations \bref{wave equation +}, \bref{wave equation -} respectively, where $\alpha$ realises $\upalpha_{\mathscr{H}^+}$ as its radiation field on $\overline{\mathscr{H}^+}$, $\underline\alpha$ realises $\underline\upalpha_{\mathscr{I}^+}$ as its radiation field on ${\mathscr{I}^+}$, such that constraints \bref{eq:227intro1} and \bref{eq:228intro1} are satisfied. Moreover, $\alpha, \underline\alpha$ induce smooth radiation fields $\underline{\upalpha}_{\mathscr{I}^+}, \upalpha_{\mathscr{H}^+}$ in $\mathcal{E}^{T,-2}_{{\mathscr{I}^+}}, \mathcal{E}^{T,+2}_{\overline{\mathscr{H}^+}}$ respectively. This extends to a unitary Hilbert-space isomorphism:
\begin{align}
    \mathscr{S}^{-2,+2}:\mathcal{E}^{T,+2}_{\mathscr{I}^-}\oplus\mathcal{E}^{T,-2}_{\overline{\mathscr{H}^-}}\longrightarrow \mathcal{E}^{T,-2}_{\mathscr{I}^+}\oplus\mathcal{E}^{T,+2}_{\overline{\mathscr{H}^+}}.
\end{align}
\end{corollary*}
\begin{center}
\begin{tikzpicture}[->,scale=0.7, arrow/.style={
            color=black,
            draw=blue,thick,
            -latex,
                font=\fontsize{8}{8}\selectfont},
        ]
\node (I)    at ( 0,0)   {};

\path % Four corners of the right diamond (no labels this time)
   (I) +(90:4)  coordinate (Itop) coordinate [label={$i^+$}]
       +(180:4) coordinate (Ileft) coordinate [label=180:{$\mathcal{B}\;$}]
       +(0:4)   coordinate (Iright) coordinate [label=0:{$\;i^0$}]
       +(270:4) coordinate (Ibot) coordinate [label=-90:{$i^-$}]
       ;
%\draw[arrow] ($(Itop)+(-135:2.5cm)$) to [out=-25,in=90]($(Itop)+(-100:3.6cm)$);
\draw[arrow] ($(Itop)+(-90:3.6cm)$) to [in=-25,out=90] ($(Itop)+(-135:2.5cm)$);

\draw[arrow] ($(Itop)+(-90:3.6cm)$) to [in=205,out=90]($(Itop)+(-45:2.5cm)$);

%\draw[arrow] ($(Ibot)+(100:3.6cm)$) to [in=45,out=-90] ($(Ibot)+(135:2.5cm)$);
\draw[arrow] ($(Ibot)+(135:2.7cm)$) to [out=10,in=-90] ($(Ibot)+(90:3.6cm)$);

\draw[arrow] ($(Ibot)+(45:2.7cm)$) to [out=170,in=-90] ($(Ibot)+(90:3.6cm)$);

%\draw[arrow] ($(Itop)+(-45:2.5cm)$) to [out=180,in=45]($(Itop)+(-75:3.9cm)$);

% No text this time in the next diagram
\draw  (Ileft) -- node[yshift=4mm,xshift=-1mm]{$\upalpha_{\mathscr{H}^+}$} (Itop) ;
\draw[dash dot dot] (Iright) --  node[yshift=4mm,xshift=1.mm]{$\underline\upalpha_{\mathscr{I}^+}$}(Itop) ;
\node[draw] at ($(Itop)+(-90:4cm)$) {$(\alpha,\underline\alpha)$};
\draw[dash dot dot] (Iright) --  node[yshift=-4mm,xshift=1.mm]{$\upalpha_{\mathscr{I}^-}$}(Ibot) ;
\draw  (Ileft) -- node[yshift=-4mm,xshift=-1mm]{$\underline\upalpha_{\mathscr{H}^-}$} (Ibot) ;

%\draw ($(Ileft)$)  to[out=0, in=180, edge node={node [below] {$\overline{\Sigma}$}}] ($(Iright)$);

\filldraw[white] (Itop) circle (3pt);
\draw[black] (Itop) circle (3pt);

\filldraw[white] (Ibot) circle (3pt);
\draw[black] (Ibot) circle (3pt);

\filldraw[white] (Iright) circle (3pt);
\draw[black] (Iright) circle (3pt);
\filldraw[black] (Ileft) circle (3pt);
\draw[black] (Ileft) circle (3pt);
\end{tikzpicture}
\end{center}
\Cref{Corollary 1} is stated again as \Cref{corollary to be proven} of \Cref{subsection 4.4 Corollary 1: mixed scattering}. The proof is contained in \Cref{subsection 9.4 mixed scattering}.\\
\indent To apply this result to scattering for the linearised Einstein equations, the strategy will be to start from data for the metric on $\mathscr{H}^-, \mathscr{I}^-$ (or $\mathscr{H}^+, \mathscr{I}^+$), obtain data for the shears $\overone{\hat{\chi}}$ and hence $\overone{\alpha}$ on $\mathscr{H}^+$, $\overone{\hat{\underline{\chi}}}$ and hence $\overone{\underline\alpha}$ on $\mathscr{I}^+$, then use \Cref{Corollary 1} to obtain scattering data and solutions to \cref{wave equation +} and \cref{wave equation -}, and conclude by constructing the remaining quantities using the linearised Bianchi and null structure equations. This will be the subject of a forthcoming sequel to this paper \cite{M2050}.\\
\indent We can give a preview of the scattering results of the full system: assume we have a solution to the linearised Einstein equations defined on the whole of the exterior region (see \Cref{subsection 2.2 Linearised Einstein equations in double null gauge} for a full list of equations), such that $\overone\alpha, \overone{\underline\alpha}$ induce radiation fields $\overone{\upalpha}_{\mathscr{I}^+}\in\mathcal{E}^{T,-2}_{\mathscr{I}^+}$, $\overone{\upalpha}_{\mathscr{I}^-}\in\mathcal{E}^{T,+2}_{\mathscr{I}^-}$, $\overone{\upalpha}_{{\mathscr{H}^+}}\in\mathcal{E}^{T,+2}_{\overline{\mathscr{H}^+}}$, $\overone{\underline\upalpha}_{{\mathscr{H}^-}}\in\mathcal{E}^{T,-2}_{\overline{\mathscr{H}^-}}$. Using \bref{example2} and its counterpart in the 4-direction, we can assert that the radiation fields belonging to the linearised shears $\overone{\hat\chi}, \overone{\hat{\underline\chi}}$ must satisfy
\begin{align}\label{boundedness shears}
    \begin{split}
        &\left\|\left(\mathring{\slashed{\Delta}}-2\right)\left(\mathring{\slashed{\Delta}}-4\right)\overone{\hat{\underline\upchi}}_{\mathscr{I}^+}\right\|_{L^2(\mathscr{I}^+)}^2+\left\|6M\partial_u\overone{\hat{\underline\upchi}}_{\mathscr{I}^+}\right\|_{L^2(\mathscr{I}^+)}^2\\
        &\qquad\qquad +\left\|\left(\mathring{\slashed{\Delta}}-2\right)\left(\mathring{\slashed{\Delta}}-4\right)\overone{\hat{\upchi}}_{\mathscr{H}^+}\right\|_{L^2(\overline{\mathscr{H}^+})}^2+\left\|6M\partial_v\overone{\hat{\upchi}}_{\mathscr{H}^+}\right\|_{L^2(\overline{\mathscr{H}^+})}^2\\
        &\qquad\qquad\qquad\qquad\qquad=\left\|\left(\mathring{\slashed{\Delta}}-2\right)\left(\mathring{\slashed{\Delta}}-4\right)\overone{\hat{\upchi}}_{\mathscr{I}^-}\right\|_{L^2(\mathscr{I}^-)}^2+\left\|6M\partial_v\overone{\hat{\upchi}}_{\mathscr{I}^-}\right\|_{L^2(\mathscr{I}^-)}^2\\
        &\qquad\qquad\qquad\qquad\qquad\qquad\qquad+\left\|\left(\mathring{\slashed{\Delta}}-2\right)\left(\mathring{\slashed{\Delta}}-4\right)\overone{\hat{\underline\upchi}}_{\mathscr{H}^-}\right\|_{L^2(\overline{\mathscr{H}^-})}^2+\left\|6M\partial_u\overone{\hat{\underline\upchi}}_{\mathscr{H}^-}\right\|_{L^2(\overline{\mathscr{H}^-})}^2.
    \end{split}
\end{align}
%\begin{changemargin}{-0.7cm}{2cm}
%\begin{align}\label{boundedness shears}
%    \scalebox{0.5}{
%    \begin{split}
%        &\left\|\left(\mathring{\slashed{\Delta}}-2\right)\left(\mathring{\slashed{\Delta}}-4\right)\overone{\hat{\underline\upchi}}_{\mathscr{I}^+}\right\|_{L^2(\mathscr{I}^+)}^2+\left\|6M\partial_u\overone{\hat{\underline\upchi}}_{\mathscr{I}^+}\right\|_{L^2(\mathscr{I}^+)}^2+\left\|\left(\mathring{\slashed{\Delta}}-2\right)\left(\mathring{\slashed{\Delta}}-4\right)\overone{\hat{\upchi}}_{\mathscr{H}^+}\right\|_{L^2(\overline{\mathscr{H}^+})}^2+\left\|6M\partial_v\overone{\hat{\upchi}}_{\mathscr{H}^+}\right\|_{L^2(\overline{\mathscr{H}^+})}^2\\
%        &=\left\|\left(\mathring{\slashed{\Delta}}-2\right)\left(\mathring{\slashed{\Delta}}-4\right)\overone{\hat{\upchi}}_{\mathscr{I}^-}\right\|_{L^2(\mathscr{I}^-)}^2+\left\|6M\partial_v\overone{\hat{\upchi}}_{\mathscr{I}^-}\right\|_{L^2(\mathscr{I}^-)}^2+\left\|\left(\mathring{\slashed{\Delta}}-2\right)\left(\mathring{\slashed{\Delta}}-4\right)\overone{\hat{\underline\upchi}}_{\mathscr{H}^-}\right\|_{L^2(\overline{\mathscr{H}^-})}^2+\left\|6M\partial_u\overone{\hat{\underline\upchi}}_{\mathscr{H}^-}\right\|_{L^2(\overline{\mathscr{H}^-})}^2.
%    \end{split}
%    }%
%\end{align}
%\end{changemargin}
\indent The fact that time translation and angular momentum operators commute with $\Box_g$ means that we can project scattering data on individual azimuthal modes and consider solutions in frequency space. Since $\overone{\hat{\chi}}, \overone{\hat{\underline\chi}}$ are supported on $\ell\geq2$, and in view of the unitarity of \bref{boundedness shears}, we can translate \bref{boundedness shears} in terms of fixed frequency, fixed azimuthal mode solutions to the following statement:
\begin{align}
    \Big\|\overone{\hat{\upchi}}_{\mathscr{H}^+,\;\omega,m,\ell}\Big\|_{L^2_\omega }^2+\Big\|\overone{\hat{\underline\upchi}}_{\mathscr{I}^+,\;\omega,m,\ell}\Big\|^2_{L^2_\omega }\;=\;\Big\|\overone{\hat{\underline\upchi}}_{\mathscr{H}^-,\;\omega,m,\ell}\Big\|_{L^2_\omega }^2+\Big\|\overone{\hat{\upchi}}_{\mathscr{I}^-,\;\omega,m,\ell}\Big\|^2_{L^2_\omega }.
\end{align}
Resumming in $\ell_{m,\ell}^2$ and using Plancherel, we obtain the identity 
\begin{align}\label{conservation law}
     \Big\|\overone{\hat{\upchi}}_{\mathscr{H}^+}\Big\|_{L^2(\overline{\mathscr{H}^+})}^2+\Big\|\overone{\hat{\underline\upchi}}_{\mathscr{I}^+}\Big\|_{L^2(\mathscr{I}^+)}^2\;=\; \Big\|\overone{\hat{\underline\upchi}}_{\mathscr{H}^-}\Big\|^2_{L^2(\overline{\mathscr{H}^-})}+\Big\|\overone{\hat{\upchi}}_{\mathscr{I}^-}\Big\|^2_{L^2(\mathscr{I}^-)}.
\end{align}
The statement \bref{conservation law} above ties up with the work by Holzegel \cite{Holzegel_2016}, where a set of conservation laws are derived for the full system of linearised Einstein equations on the Schwarzschild exterior \bref{SchwMetric} (using purely physical-space methods).\\
\indent Note that in particular, for past scattering data that is vanishing on $\overline{\mathscr{H}^-}$, the identity  \bref{conservation law} has the interpretation that the energy of the gravitational energy radiated to $\mathscr{I}^+$ is bounded \underline{with constant 1} by the incoming gravitational energy radiated from $\mathscr{I}^-$, i.e.~there is no superradiant amplification of reflected gravitational radiation on the Schwarzschild exterior.
\subsection{Outline of the paper}\label{subsection 1.4 outline}
This paper is organised as follows: We review the linearised Einstein equations in double null gauge around the Schwarzschild spacetime in \Cref{section 2 preliminaries}. In \Cref{TRW} we introduce the Teukolsky equations, the Regge--Wheeler equations and derive important identities connecting the equations. Detailed statements of the results of this work are presented in \Cref{section 4 main theorems}, and then the scattering theory of the Regge--Wheeler equations is studied in \Cref{section 5 scattering theory for RW}. We develop scattering for the Teukolsky equations by first working out the necessary estimates to understand the asymptotic behaviour in forward evolution for both equations in \Cref{section 6} and \Cref{section 7}. Backwards scattering for both equations is treated in \Cref{section 8 constructing the scattering maps}, followed by the study of the constraints \bref{eq:227intro1} and \bref{eq:228intro1} in \Cref{section 9 TS correspondence}. \Cref{Appendix A Robinson--Trautman} is concerned with Robinson--Trautman spacetimes, and \Cref{Appendix B Double null guage} is a brief review of the double null gauge.
\setcounter{tocdepth}{2}

\subsection*{Acknowledgements} The author would like to express his gratitude to his supervisors Mihalis Dafermos, Claude Warnick and  Malcolm J. Perry for introducing him to the fascinating area of scattering theory in general relativity, and for their unwavering support throughout the undertaking of this project. The author would like to thank Dejan Gajic, Leonhard Kehrberger, Rita Teixeira da Costa and Yakov Shlapentokh-Rothman for stimulating discussions and helpful remarks. The author acknowledges support by the EPSRC grant  EP/L016516/1. 

\section{Preliminaries}\label{section 2 preliminaries}

\subsection{The Schwarzschild exterior in a double null gauge}\label{subsection 2.1 schwarzschild in dng}

\indent Denote by $\mathscr{M}$ the exterior of the maximally extended Schwarzschild spacetime. Using Kruskal coordinates, this is the manifold with corners
\begin{align}
    \mathscr{M}=\{(U,V,\theta^A)\in(-\infty,0]\times [0,\infty)\times S^2\}
\end{align}
equipped with the metric
\begin{align}\label{metric Kruskal}
    ds^2=-\frac{32M^3}{r(U,V)}e^{-\frac{r(U,V)}{2M}}dUdV+r(U,V)^2\gamma_{AB}d\theta^A d\theta^B.
\end{align}
The function $r(U,V)$ is determined by $-UV=\left(\frac{r}{2M}-1\right)e^{\frac{r}{2M}}$, $(\theta^A)$ is a coordinate system on $S^2$ and $\gamma_{AB}$ is the standard metric on the unit sphere $S^2$. The time-orientation of $\mathscr{M}$ is defined by the vector field $\partial_U+\partial_V$. 

The boundary of $\mathscr{M}$ consists of the two null hypersurfaces
\begin{align}
    \mathscr{H}^+&=\{0\}\times(0,\infty)\times S^2,\\
    \mathscr{H}^-&=(-\infty,0)\times \{0\}\times S^2,
\end{align}
and the 2-sphere $\mathcal{B}$ where $\mathscr{H}^+$ and $\mathscr{H}^-$ bifurcate:
\begin{align}
    \mathcal{B}=\{U,V=0\}\cong S^2 .
\end{align}
We define $\overline{\mathscr{H}^+}=\mathscr{H}^+\cup \mathcal{B}$, $\overline{\mathscr{H}^-}=\mathscr{H}^-\cup \mathcal{B}$. \\
\indent The interior of $\mathscr{M}$ can be covered with the familiar Schwarzschild coordinates $(t,r,\theta^A)$ and the metric takes the form \bref{SchwMetric}, i.e.
\begin{align}
    ds^2=-\left(1-\frac{2M}{r}\right)dt^2+\left(1-\frac{2M}{r}\right)^{-1}dr^2+r^2\gamma_{AB}d\theta^Ad\theta^B.
\end{align}
Let $\Omega^2=\left(1-\frac{2M}{r}\right)$. It will be convenient to work instead in Eddington--Finkelstein coordinates
\begin{align}\label{EF null coordinates}
    u=\frac{1}{2}(t-r_*),\qquad\qquad\qquad v=\frac{1}{2}(t+r_*),
\end{align}
where $r_*$ is defined up to a constant by $\frac{dr_*}{dr}=\frac{1}{\Omega^2}$. The coordinates $(u,v,\theta^A)$ also define a double null foliation (see Appendix B) of the interior of $\mathscr{M}$ since the metric takes the form
\begin{align}
    ds^2=-4\left(1-\frac{2M}{r}\right)dudv+r(u,v)^2(d\theta^2+\sin^2\theta d\phi^2).
\end{align}
In particular the null frame defined by the coordinates \bref{EF null coordinates} is given by (see Appendix B):
\begin{align}
    e_3=\frac{1}{\Omega}\partial_u,\qquad\qquad e_4=\frac{1}{\Omega}\partial_v.
\end{align}
We may relate $U,V$ to $u,v$ after fixing the residual freedom in defining $t,r_*$ by
\begin{align}\label{Kruskal}
    U=-e^{-\frac{u}{2M}},\qquad\qquad V=e^{\frac{v}{2M}},
\end{align}
Note that the intersections of null hypersurfaces of constant $u,v$ are spheres with metric $\slashed{g}_{AB}:=r^2\gamma_{AB}$. We denote these spheres by $S^2_{u,v}$.\\
\indent The $(u,v)$-coordinate system degenerates on $\overline{\mathscr{H}^+}$ and $\overline{\mathscr{H}^-}$ where $u=\infty,v=-\infty$ respectively. To compensate for this we can use the Kruskal coordinates to introduce weighted quantities in the coordinates $(u,v,\theta^A)$ that are regular on $\mathscr{H}^\pm$. We note already at this stage that the regularity of $\partial_U,\partial_V$ on the event horizons implies that $\frac{1}{\Omega}e_3, \Omega e_4$ are regular on $\mathscr{H}^+$ and $\frac{1}{\Omega}e_4, \Omega e_3$ are regular on $\mathscr{H}^-$ (but not $\overline{\mathscr{H}^\pm}$, which include $\mathcal{B}$).\\
\indent We denote by $\mathscr{C}_{u^*}$ the ingoing null hypersurface of constant $u=u^*$, and similarly $\underline{\mathscr{C}}_{v^*}$ denotes the outgoing null hypersurface $v=v^*$; define $\mathscr{C}_{u^*}\cap[v_1,v_2]$ to be the subset of $\mathscr{C}_{u^*}$ for which $v\in[v_1,v_2]$, $\underline{\mathscr{C}}_v\cap[u_1,u_2]$ denotes the subset of  $\underline{\mathscr{C}}_v$ for which $u\in[u_-,u_+]$. Let ${\Sigma}$ be the spacelike surface $\{t=0\}$ and let $\overline{\Sigma}=\Sigma\cup\mathcal{B}$ be the topological closure of $\Sigma$ in $\mathscr{M}$. $\overline{\Sigma}$ is a smooth Cauchy surface for $\mathscr{M}$ which connects $\mathcal{B}$ with "spacelike infinity"; in Kruskal coordinates it is given by $\{U+V=0\}$. We also work with a spacelike hypersurface $\Sigma^*$
intersecting $\mathscr{H}^+$ to the future of $\mathcal{B}$, defined as follows: let 
\begin{align}
    t^*=t+2M\log\left(\frac{r}{2M}+1\right).
\end{align}
The function $t^*$ can be extended to $\mathscr{H}^\pm$ to define a smooth function on all of $\mathscr{M}$, and we define $\Sigma^*$ by
\begin{align}
    \Sigma^*=\{t^*=0\}
\end{align}
\noindent Note that $\Sigma^*$ intersects $\mathscr{H}^+$ at $v=0$ and asymptotes to spacelike infinity. Define $\mathscr{H}^+_{\geq 0}:=\mathscr{H}^+\cap  J^+(\Sigma^*)$. We will occasionally use the notation $x:=1-\frac{1}{\Omega^2}$.
We denote the spacetime region bounded by $\mathscr{C}_{u_0}\cap[v_0,v_1], \mathscr{C}_{u_1}\cap[v_0,v_1], \underline{\mathscr{C}}_{v_0}\cap[u_0,u_1], \underline{\mathscr{C}}_{v_1}\cap[u_0,u_1]$ by $\mathscr{D}^{u_1,v_1}_{u_0,v_0}$. We also denote the spacetime region bounded by $\mathscr{C}_u,\underline{\mathscr{C}}_v, \Sigma^*$ by $\mathscr{D}^{u,v}_{\Sigma^*}$.
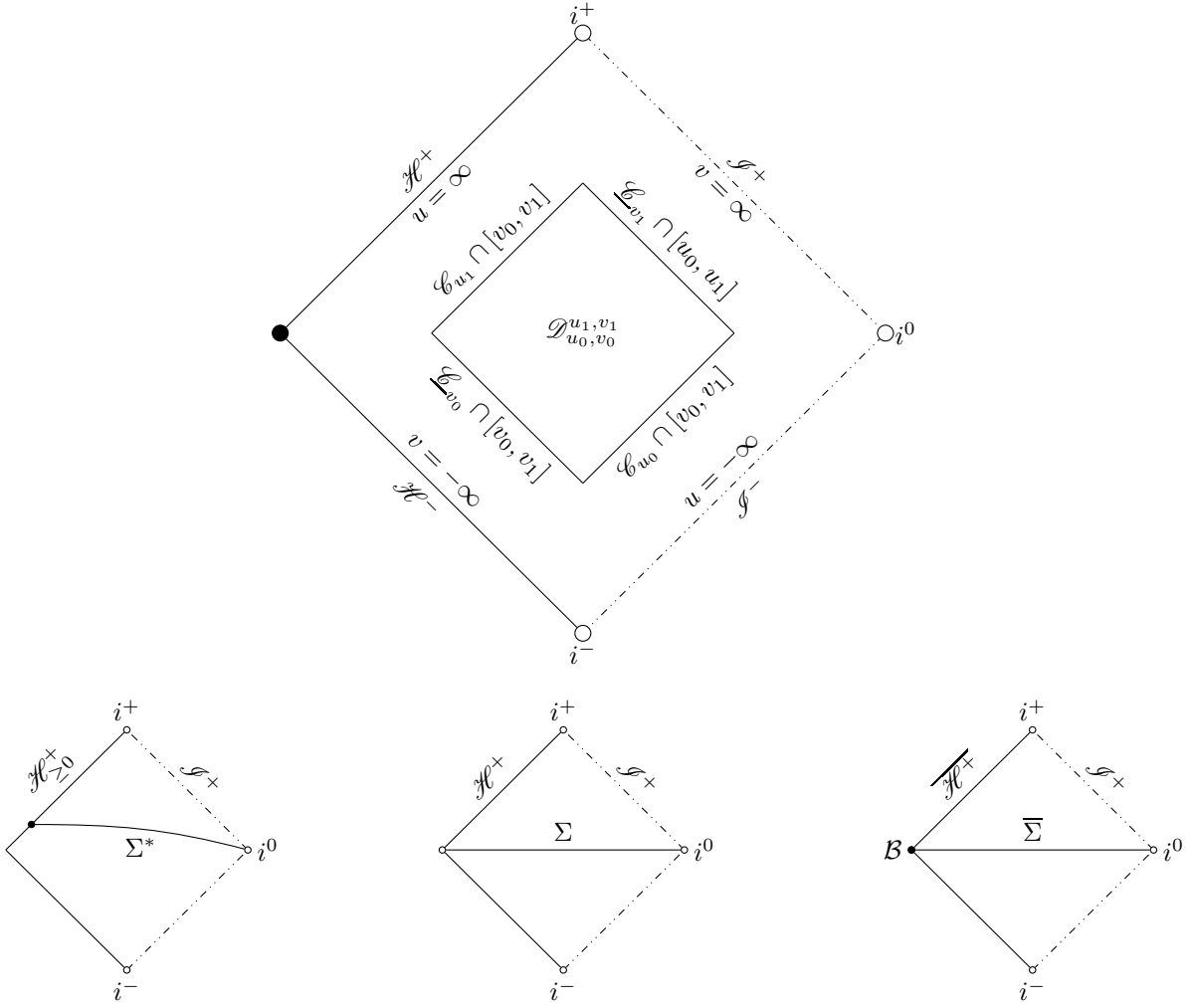
\begin{center}
\begin{tikzpicture}[scale=1]
\node (I)    at ( 0,0)   {$\mathscr{D}^{u_1,v_1}_{u_0,v_0}$};

\path % Four corners of the right diamond (no labels this time)
   (I) +(90:4)  coordinate (Itop) coordinate[label=90:$i^+$]
       +(-90:4) coordinate (Ibot) coordinate[label=-90:$i^-$]
       +(180:4) coordinate (Ileft)
       +(0:4)   coordinate (Iright) coordinate[label=0:$i^0$]
       ;

\path % Four corners of the right diamond (no labels this time)
   (I) +(90:2)  coordinate (Ictop)
       +(-90:2) coordinate (Icbot)
       +(180:2) coordinate (Icleft)
       +(0:2)   coordinate (Icright)
       ;

% No text this time in the next diagram
\draw  (Ileft) -- node[rotate=45,below] {$u=\infty$} node[rotate=45,above]{$\mathscr{H}^+$} (Itop) ;
\draw  (Ileft) -- node[rotate=-45,above] {$v=-\infty$} node[rotate=-45,below]{$\mathscr{H}^-$}(Ibot) ;
\draw[dash dot dot] (Ibot) -- node[rotate=45,above] {$u=-\infty$} node[rotate=45,below]{$\mathscr{I}^-$}(Iright) ;
\draw[dash dot dot] (Iright) -- node[rotate=-45,below] {$v=\infty$} node[rotate=-45,above]{$\mathscr{I}^+$}(Itop) ;

%\path (Icright)+ (45:2) coordinate (Iccright)

\draw(Icleft) --node[rotate=45,above] {$\mathscr{C}_{u_1}\cap[v_0,v_1]$} (Ictop);
\draw(Ictop) -- node[rotate=-45,above] {$\underline{\mathscr{C}}_{v_1}\cap[u_0,u_1]$}(Icright);
\draw(Icright) -- node[rotate=45,below] {$\mathscr{C}_{u_0}\cap[v_0,v_1]$}(Icbot);
\draw(Icbot) -- node[rotate=-45,below] {$\underline{\mathscr{C}}_{v_0}\cap[v_0,v_1]$}(Icleft);

%\draw(Ileft) --node[above] {$\overline{\Sigma}$} (Iright) ;

\filldraw[white] (Itop) circle (3pt);
\draw[black] (Itop) circle (3pt);

\filldraw[white] (Ibot) circle (3pt);
\draw[black] (Ibot) circle (3pt);

\filldraw[white] (Iright) circle (3pt);
\draw[black] (Iright) circle (3pt);
\filldraw[black] (Ileft) circle (3pt);

\end{tikzpicture}
\end{center}

\begin{center}
\begin{tikzpicture}[scale=0.4]
\node (I)    at ( 0,0)   {};

\path % Four corners of the right diamond (no labels this time)
   (I) +(90:4)  coordinate (Itop) coordinate[label=90:$i^+$]
       +(-90:4) coordinate (Ibot) coordinate[label=-90:$i^-$]
       +(180:4) coordinate (Ileft)
       +(0:4)   coordinate (Iright) coordinate[label=0:$i^0$]
       ;

% No text this time in the next diagram
\draw  (Ileft) --  node[rotate=45,above]{$\mathscr{H}^+_{\geq0}$} (Itop) ;
\draw  (Ileft) -- (Ibot) ;
\draw[dash dot dot] (Ibot) -- (Iright) ;
\draw[dash dot dot] (Iright) --  node[rotate=-45,above]{$\mathscr{I}^+$}(Itop) ;

% node[rotate=45,below]{$\mathscr{I}^-$}      node[rotate=-45,below]{$\mathscr{H}^-$}

\draw ($(Ileft)+(45:1.2)$)  to[out=-0, in=165, edge node={node [below] {$\Sigma^*$}}] ($(Iright)$);

%\draw ($(Itop)+(-75:1)$) to ($(Itop)+(-75:1)+(-135:3.1)$);

%\draw ($(Itop)+(-75:1)$) to ($(Itop)+(-75:1)+(-45:3.9)$);

%\node (II) at ($(Itop)+(-75:1.2)+(-95:1.2)$) {$\mathscr{D}^{u,v}_{\Sigma^*}$};

\filldraw[white] (Itop) circle (3pt);
\draw[black] (Itop) circle (3pt);

\filldraw[white] (Ibot) circle (3pt);
\draw[black] (Ibot) circle (3pt);

\filldraw[white] (Iright) circle (3pt);
\draw[black] (Iright) circle (3pt);
\filldraw[black]($(Ileft)+(45:1.2)$) circle (3pt);
\end{tikzpicture}\hspace{2cm}\begin{tikzpicture}[scale=0.4]
\node (I)    at ( 0,0)   {};

\path % Four corners of the right diamond (no labels this time)
   (I) +(90:4)  coordinate (Itop) coordinate[label=90:$i^+$]
       +(-90:4) coordinate (Ibot) coordinate[label=-90:$i^-$]
       +(180:4) coordinate (Ileft)
       +(0:4)   coordinate (Iright) coordinate[label=0:$i^0$]
       ;

% No text this time in the next diagram
\draw  (Ileft) --  node[rotate=45,above]{$\mathscr{H}^+$} (Itop) ;
\draw  (Ileft) --(Ibot) ;
\draw[dash dot dot] (Ibot) -- (Iright) ;
\draw[dash dot dot] (Iright) --  node[rotate=-45,above]{$\mathscr{I}^+$}(Itop) ;

\draw ($(Ileft)$)  to[out=0, in=180, edge node={node [above] {$\Sigma$}}] ($(Iright)$);

\filldraw[white] (Itop) circle (3pt);
\draw[black] (Itop) circle (3pt);

\filldraw[white] (Ibot) circle (3pt);
\draw[black] (Ibot) circle (3pt);

\filldraw[white] (Iright) circle (3pt);
\draw[black] (Iright) circle (3pt);
\filldraw[white] (Ileft) circle (3pt);
\draw[black] (Ileft) circle (3pt);
\end{tikzpicture}\hspace{2cm}\begin{tikzpicture}[scale=0.4]
\node (I)    at ( 0,0)   {};

\path % Four corners of the right diamond (no labels this time)
   (I) +(90:4)  coordinate (Itop) coordinate[label=90:$i^+$]
       +(-90:4) coordinate (Ibot) coordinate[label=-90:$i^-$]
       +(180:4) coordinate (Ileft) coordinate[label=180:$\mathcal{B}$]
       +(0:4)   coordinate (Iright) coordinate[label=0:$i^0$]
       ;

% No text this time in the next diagram
\draw  (Ileft) --  node[rotate=45,above]{$\overline{\mathscr{H}^+}$} (Itop) ;
\draw  (Ileft) --(Ibot) ;
\draw[dash dot dot] (Ibot) -- (Iright) ;
\draw[dash dot dot] (Iright) --  node[rotate=-45,above]{$\mathscr{I}^+$}(Itop) ;

\draw ($(Ileft)$)  to[out=0, in=180, edge node={node [above] {$\overline{\Sigma}$}}] ($(Iright)$);

\filldraw[white] (Itop) circle (3pt);
\draw[black] (Itop) circle (3pt);

\filldraw[white] (Ibot) circle (3pt);
\draw[black] (Ibot) circle (3pt);

\filldraw[white] (Iright) circle (3pt);
\draw[black] (Iright) circle (3pt);
\filldraw[black] (Ileft) circle (3pt);
\draw[black] (Ileft) circle (3pt);
\end{tikzpicture}
\end{center}

\subsubsection*{Null infinity $\mathscr{I}^\pm$}

We define the notion of null infinity by directly attaching it as a boundary to $\mathscr{M}$. Define $\mathscr{I}^+,\mathscr{I}^-$ to be the manifolds
\begin{align}
\mathscr{I}^+,\mathscr{I}^-:=\mathbb{R}\times S^2
\end{align}
and define $\overline{\mathscr{M}}$ to be the extension 
\begin{align}
\overline{\mathscr{M}}=\mathscr{M}\cup\mathscr{I}^+\cup\mathscr{I}^-.
\end{align}
For sufficiently large $R$ and any open set $\mathcal{O}\subset\mathbb{R}\times S^2$, declare the sets $\mathcal{O}^+_R=(R,\infty]\times\mathcal{O}$ to be open in $\overline{\mathscr{M}}$, identifying $\mathscr{I}^+$ with the points $(u,\infty,\theta,\phi)$. To the set $\mathcal{O}_R^+$ we assign the coordinate chart $(u,s,\theta,\phi)\in \mathbb{R}\times[0,1)\times S^2$ via the map
\begin{align}
(u,v,\theta,\phi)\longrightarrow(u,\frac{R}{v},\theta,\phi),
\end{align}
where $(u,v,\theta,\phi)$ are the Eddington--Finkelstein coordinates we defined earlier. The limit $\lim_{v\longrightarrow\infty} (u,v,\theta,\phi)$ exists and is unique, and we use it via the above charts to fix a coordinate system $(u,\theta,\phi)$ on $\mathscr{I}^+$. The same can be repeated to define an atlas attaching $\mathscr{I}^-$ as a boundary to $\overline{\mathscr{M}}$.

\subsubsection{$S^2_{u,v}$-projected connection and angular derivatives}\label{D1D2}
We will be working primarily with tensor fields that are everywhere tangential to the $S^2_{u,v}$ spheres foliating $\mathscr{M}$. By this we mean any tensor fields of type $(k,l)$, $\digamma\in \mathcal{T}^{(k,l)}\mathscr{M}$ on $\mathscr{M}$ such that for any point $q=(u,v,\theta^A)\in\mathscr{M}$ we have $\digamma|_q\in \mathcal{T}^{(k,l)}_{(\theta^A)}S^2_{u,v}$. (Note that a vector $X^A\in \mathcal{T}_{(\theta^A)}S^2_{u,v}$ is canonically identified with a vector $X^a\in\mathcal{T}_q\mathscr{M}$ via the inclusion map, whereas we make the identification of a 1-form $\eta_A\in\mathcal{T}^*_{(\theta^A)}\mathscr{M}$ as an element in the cotangent bundle of $\mathscr{M}$ by declaring that $\eta(X)=0$ if $X$ is in the orthogonal complement of $\mathcal{T}S^2_{u,v}$ under the spacetime metric \bref{metric Kruskal}.) We will refer to such tensor fields as "$S^2_{u,v}$-tangent" tensor fields in the following. It will also be convenient to work with an "$S^2_{u,v}$ projected" version of the covariant derivative belonging to the Levi-Civita connection of the metric \bref{SchwMetric}. We define these notions as follows:\\ 
\indent We denote by $\slashed{\nabla}_A$ (or sometimes simply $\slashed{\nabla}$) the covariant derivative on $S^2_{u,v}$ with the metric $\slashed{g}_{AB}$. Note that $r\slashed{\nabla}=\slashed{\nabla}_{\mathbb{S}^2}$ which we also denote by $\mathring{\slashed{\nabla}}$.\\
\indent For an $S^2_{u,v}$-tangent 1-form $\xi$, define $\fancyd_1 \xi$ to be the pair of functions
\begin{align}
    \fancyd_1{\xi}=(\slashed{\text{div}}\xi,\slashed{\text{curl}}\xi),
\end{align}
where $\slashed{\text{div}}\xi=\slashed{\nabla}^A \xi_A$ and $\slashed{\text{curl}}\xi=\slashed{\epsilon}^{AB} \slashednabla_A\xi_B$. For an $S^2_{u,v}$-tangent symmetric 2-tensor $\Xi_{AB}$ we define $\fancyd_2 \theta$ to be the 1-form given by
\begin{align}
    (\fancyd_2 \theta)_A=(\slashed{\text{div}}\theta)_A=\slashednabla^B \Xi_{BA}.
\end{align}
\indent We define the operator $\fancydstar_1 $ to be the $L^2({S^2_{u,v}})$-dual to $\fancyd_1$. For scalars $(f,g)$ the 1-form $\fancydstar_1(f,g)$ is given by
\begin{align}
    \fancydstar_1(f,g)=-\slashednabla_A f +\epsilon_{AB}\slashednabla^B g.
\end{align}
\indent Similarly we denote by $\fancydstar_2$ the $L^2_{S^2_{u,v}}$-dual to $\fancyd_2$. For an $S^2_{u,v}$-tangent 1-form $\xi$ this is given by
\begin{align}
    (\fancydstar_2\xi)_{AB}=-\frac{1}{2}\left(\slashednabla_A \xi_B+\slashednabla_B\xi_A-\slashed{g}_{AB}\slashed{\text{div}}\xi\right).
\end{align}
We also use the notation
\begin{align}
\begin{split}
    \mathring{\fancyd}_1:=r\fancyd_1,\qquad\qquad\mathring{\fancydstar_1}:=r\fancydstar_1,\\
    \mathring{\fancyd}_2:=r\fancyd_2,\qquad\qquad\mathring{\fancydstar_2}:=r\fancydstar_2.
\end{split}
\end{align}
For example, if $\xi$ is a 1-form on $S^2_{u,v}$ then
\begin{align}
    \mathring{\fancydstar_2}\xi=-\frac{1}{2}\left(\mathring{\slashednabla}_A \xi_B+\mathring{\slashednabla}_B\xi_A-\slashed{g}_{AB}\mathring{\slashednabla}_C\xi^C\right).
\end{align}
and so on. Let $\xi$ be an $S^2_{u,v}$-tangent tensor field. We denote by $D\xi$ and $\underline{D}\xi$ the projected Lie derivative of $\xi$ in the 3- and 4-directions respectively.  In EF coordinates we have
\begin{align}
    (D\xi)_{A_1 A_2...A_n}=\partial_u(\xi_{A_1 A_2 ... A_n})\qquad (\underline{D}\xi)_{A_1 A_2...A_n}=\partial_v(\xi_{A_1 A_2 ... A_n})
\end{align}
Similarly, we define $\slashed{\nabla}_3 \xi$ and $\slashed{\nabla}_4 \xi$ to be the projections of the covariant derivatives $\nabla_3 \xi$ and $\nabla_4 \xi$ to $S^2_{u,v}$. 

\subsubsection{Elliptic estimates on $S^2_{u,v}$}\label{subsubsection 2.1.2 Elliptic estimates on S2}

For a $k$-covariant $S^2_{u,v}$-tangent tensor field $\theta$ on $\mathscr{M}$, define 
\begin{align}
    |\theta|_{S^2}=\sqrt{\gamma^{A_1B_1}\gamma^{A_2B_2}\cdot\cdot\cdot\gamma^{A_pB_p}\Xi_{A_1...A_p}\Xi_{B_1...B_p}},\qquad |\theta|=r^{-k}|\theta|_{S^2}
\end{align}
The following is a summary of Section 4.4 of \cite{DHR16}. Given scalars $(f,g)$ we can define an $S^2_{u,v}$ 1-form by $\xi=r\fancydstar_1(f,g)$. In turn, given a 1-form $\xi$ we can define a symmetric traceless 2-form $\theta$ via $\theta=r\fancydstar_2\xi$. It turns out that these representations span the space of such $\xi$ and $\theta$:

\begin{proposition}
Let $\xi$ be an $S^2_{u,v}$-tangent 1-form. Then there exist scalars $f,g$ such that 
\begin{align}
\xi=r\fancydstar_1(f,g).
\end{align}
Let $\Xi$ be $S^2_{u,v}$-tangent symmetric traceless 2-form. Then there exist scalars $f,g$ such that
\begin{align}
    \Xi=r^2\fancydstar_2\fancydstar_1(f,g)
\end{align}
\end{proposition}
Note that when considering the decomposition of $f,g$ into their spherical harmonic modes, the operation of acting by $\fancydstar_1$ annihilates their $\ell=0$ modes and the action of $\fancydstar_2$ annihilates their $\ell=1$ modes. Thus in the case of a 1-form $f,g$ can be taken to have vanishing $\ell=0$ modes, in which case $f,g$ are unique. Similarly, for a symmetric traceless $S^2_{u,v}$ 2-tensor there exist a unique pair $f,g$ with vanishing $\ell=0,1$ such that $\theta$ is given by the expression above.
\begin{remark}
The operators $\fancyd_1,\fancyd_2,\fancydstar_1,\fancydstar_2$ defined in \Cref{D1D2} can be combined to give
\begin{align}
\begin{split}
    &-2r^2\fancydstar_2\fancyd_2=\mathring{\slashed{\Delta}}-2\qquad\qquad-r^2\fancydstar_1\fancyd_1=\mathring{\slashed{\Delta}}-1\\
    &-2r^2\fancyd_2\fancydstar_2=\mathring{\slashed{\Delta}}+1\qquad\qquad -r^2\fancyd_1\fancydstar_1=\mathring{\slashed{\Delta}}.
\end{split}
\end{align}
The operator $\mathring{\slashed\Delta}$ is the Laplacian on the unit 2-sphere $S^2$.
\end{remark}
\begin{proposition}
Let $\Xi$ be a smooth symmetric traceless $S^2_{u,v}$ 2-tensor. We have the following identities:
\begin{align}
    \int_{S^2_{u,v}}\sin\theta d\theta d\phi\left[ |\slashednabla\Xi|^2+2K|\Xi|^2\right]=2\int_{S^2_{u,v}} \sin\theta d\theta d\phi |\fancyd_2 \Xi|^2,
\end{align}
\begin{align}
    \int_{S^2_{u,v}}\sin\theta d\theta d\phi\left[\frac{1}{4}|\slashed{\Delta}\Xi|^2+K|\Xi|^2 +K^2|\slashednabla\Xi|^2\right]=\int_{S^2_{u,v}} \sin\theta d\theta d\phi |\fancydstar_2\fancyd_2 \Xi|^2,
\end{align}
where $K=\frac{1}{r^2}$ is the Gaussian curvature of $S^2_{u,v}$. 
\end{proposition}

We also note the following Poincar\'e inequality:
\begin{proposition}\label{poincaresection}
Let $\Xi$ be a smooth symmetric traceless $S^2_{u,v}$ 2-tensor, then we have
\begin{align}\label{poincare}
    2K\int_{S^2_{u,v}}\sin\theta d\theta d\phi|\Xi|^2\leq \int_{S^2_{u,v}} \sin\theta d\theta d\phi|\slashednabla \Xi|^2
\end{align}
\end{proposition}
\begin{remark}
We will be using the notation
\begin{align}
    \mathcal{A}_2:=-2r^2\fancydstar_2\fancyd_2=\mathring{\slashed{\Delta}}-2.
\end{align}
\end{remark}

\subsubsection{Asymptotics of $S^2_{u,v}$-tensor fields}\label{subsubsection 2.1.3 Asymptotics of S2 tensor fields}

Let $\digamma$ be a $k$-covariant $S^2_{u,v}$-tangent tensor field on $\mathscr{M}$. We say that $\digamma$ converges to $F=F_{A_1A_2...A_p}(u)$ as $v\longrightarrow\infty$ if $r^{-k}\digamma\longrightarrow F$ in the norm $|\;\;|_{S^2}$. We may write
\begin{align}
\begin{split}
    \left|\frac{1}{r^k}\digamma(u,v,\theta^A)-F(u,\theta^A)\right|_{S^2}&=\left|\int_{v}^\infty d\bar{v} \frac{d}{dv}\frac{1}{r^k}\digamma \right|_{S^2}\leq \int_{v}^\infty d\bar{v} \left|\frac{d}{dv}\frac{1}{r^k}\digamma\right|_{S^2}
    \\&=\int_{v}^{\infty}d\bar{v}\left|r^k\frac{d}{dv}\frac{1}{r^k}\digamma\right|=\int_{v}^\infty d\bar{v}|\nablav \digamma|.
\end{split}
\end{align}
Therefore, if $\nablav\digamma$ is integrable in $L^1_vL^2_{S^2_{u,v}}$ then $\digamma$ has a limit towards $\mathscr{I}^+$. It is easy to see that if $\{\digamma_n\}_n^\infty$ is a Cauchy sequence in $|\;\;|$ then $\digamma_n$ converges in the sense of this definition. The above extends to tensors of rank $(k,\ell)$, where $r^{-k}$ is replaced by $r^{-k+\ell}$. Similar considerations apply when taking the limit towards $\mathscr{I}^-$.
In particular, for a symmetric tensor $\Psi$ of rank $(2,0)$, it will be simpler to work with $\Psi^{A}{}_B$. Note that $\nablav\Psi^A{}_B=\partial_v\Psi^A{}_B$, $\nablau\Psi^A{}_B=\partial_u \Psi^A{}_B$. Unless otherwise indicated, we work with $S^2_{u,v}$-tangent $(1,1)$-tensors throughout.

\subsection{Linearised Einstein equations in a double null gauge}\label{subsection 2.2 Linearised Einstein equations in double null gauge}

When linearising the Einstein equations \bref{EVE} against the Schwarzschild background in a double null gauge, the quantities governed by the resulting equations can be organised into a collection of $S^2_{u,v}$-tangent tensor fields:
\begin{itemize}
    \item The linearised metric components 
    \begin{align}\label{linearised metric}
        \overone{\hat{\slashed{g}}}\;,\; \overone{b}\;,\;\overone{\sqrt{\slashed{g}}}\;,\; \overone{\Omega}\;,
    \end{align}
    \item the linearised connection coefficients
    \begin{align}\label{linearised connection}
        \overone{\hat{\chi}}\;,\; \overone{\hat{\underline\chi}}\;, \overone{\eta}\;,\; \overone{\underline\eta}\;,\; \overone{(\Omega \tr\chi)}\;, \;\overone{(\Omega \tr\underline\chi)}\;,\;\overone{\omega}\;,\; \overone{\underline{\omega}}\;,\; 
    \end{align}
    \item the linearised curvature components
    \begin{align}\label{linearised curvature}
        \overone{\alpha}\;,\; \overone{\underline\alpha}\;,\; \overone{\beta}\;,\;\overone{\underline\beta}\;,\; \overone{\rho}\;,\; \overone{\sigma}\;,\; \overone{K}.
    \end{align}
\end{itemize}

See Appendix B and \cite{DHR16} for the details of linearising the vacuum Einstein equations \bref{EVE} in a double null gauge. We now state the linearised vacuum Einstein equations around the Schwarzschild black hole in a double null gauge:

\begin{itemize}
    \item The equations governing the linearised metric components \bref{linearised metric}:
    \begin{align}
    \partial_v \overone{\sqrt{\slashed{g}}}\;=\;2(\overone{\Omega\tr\chi})-2\;\slashed{div}\overone{b}&,\qquad\nablav \overone{\hat{\slashed{g}}}\;=\;2\Omega\overone{\hat{\chi}}+2\fancydstar_2\overone{b},\\
    \partial_u \overone{\sqrt{\slashed{g}}}\;=\;2(\overone{\Omega\tr\underline\chi})&,\qquad \nablau\overone{\hat{\slashed{g}}}\;=\;2\Omega\underline{\overone{\hat{\chi}}}.\\
    \partial_u\overone{b}\;=\;&2\Omega^2(\overone{\eta}-\overone{\underline\eta}),\\
    \partial_v\left(\frac{\overone{\Omega}}{\Omega}\right)\;=\;\overone{\omega},\qquad\qquad\partial_u\left(\frac{\overone{\Omega}}{\Omega}\right)\;&=\;\overone{\underline\omega},\qquad\qquad\overone{\eta}_A+\overone{\underline{\eta}}_A\;=\;2\slashednabla_A \left(\frac{\overone{\Omega}}{\Omega}\right).\label{omega omegabar eta etabar}
\end{align}
    \item The equations governing the linearised connection coefficients \bref{linearised connection}:
    \begin{equation} \label{start of full system}
\Omega\slashed\nabla_4\; r\trxbar=2\Omega^2\left(\slashed{div}\; r\overone{\underline\eta}+2r\overone{\rho} -\frac{4M}{r^2}\frac{\stackrel{\mbox{\scalebox{0.45}{(1)}}}{\Omega}}{\Omega}\right)+\Omega^2\trx,
\end{equation}
\begin{equation}\label{D3TrChiBar}
\Omega\slashed\nabla_3\; r\trx=2\Omega^2\left(\slashed{div}\; r\overone{\eta}+2r\overone{\rho}-\frac{4M}{r^2}\frac{\stackrel{\mbox{\scalebox{0.45}{(1)}}}{\Omega}}{\Omega}\right) -\Omega^2 \trxbar,
\end{equation}
\begin{equation}\label{D4TrChi}
\Omega\slashed\nabla_4\frac{r^2}{\Omega^2}\trx=4r\overset{\mbox{\scalebox{0.4}{(1)}}}{\omega},\qquad\qquad \Omega\slashed\nabla_3\frac{r^2}{\Omega^2}\trxbar=-4r\overset{\mbox{\scalebox{0.4}{(1)}}}{\underline\omega},
\end{equation}
\begin{equation}\label{D4Chihat}
\Omega\slashed\nabla_4\frac{r^2 \overone{\hat{\chi}}}{\Omega}=-r^2\overone{\alpha},\qquad\qquad \Omega\slashed\nabla_3\frac{r^2\overone{\underline{\hat{\chi}}}}{\Omega}=-r^2\overone{\underline\alpha},
\end{equation}
\begin{equation}\label{D3Chihat}
\Omega\slashed\nabla_3\; r\Omega\overone{\hat{\chi}}=-2r\slashed{\mathcal{D}}^*_2 \Omega^2 \overone{\eta}-\Omega^2 \left(\Omega \overone{\underline{\hat{\chi}}}\right),
\end{equation}
\begin{equation}\label{D4Chihatbar}
\Omega\slashed\nabla_4\;r\Omega\overone{\underline{\hat{\chi}}}=-2r\slashed{\mathcal{D}}^*_2\Omega^2\overone{\underline\eta}+\Omega^2\left(\Omega\overone{\hat{\chi}}\right),
\end{equation}

\begin{equation}\label{D3etabar}
\Omega\slashed\nabla_3r\overone{\underline\eta}=r\Omega\overone{\underline\beta}-\Omega^2\overone{\eta},\qquad\qquad \Omega\slashed\nabla_4r\overone{\eta}=-r\Omega\overone{\beta}+\Omega^2\overone{\underline\eta},
\end{equation}
\begin{equation}\label{D4etabar}
\Omega\slashed\nabla_4 r^2\overone{\underline\eta}=2r^2\slashed\nabla_A\overset{\mbox{\scalebox{0.4}{(1)}}}{\omega}+r^2\Omega\overone{\beta},\qquad\qquad \Omega\slashed\nabla_3r^2\overone{\eta}=2r^2\slashed\nabla_A\underline{\overset{\mbox{\scalebox{0.4}{(1)}}}{\omega}}-r^2\Omega\overone{\underline\beta},
\end{equation}
\item The equations governing the curvature components \bref{linearised curvature}:
\begin{equation}\label{Bianchi +2}
\Omega\slashed\nabla_3 \;r\Omega^2\overone{\alpha}=-2r\slashed{\mathcal{D}}^*_2 \Omega^2 \Omega \overone{\beta}+\frac{6M\Omega^2}{r^2}\Omega\overone{\hat{\chi}},\qquad\quad \Omega\slashed\nabla_4\;r\Omega^2\overone{\underline\alpha}=2r\slashed{\mathcal{D}}^*_2\Omega^2\Omega\overone{\underline\beta}+\frac{6M\Omega^2}{r^2}\Omega\overone{\underline{\hat{\chi}}},\;
\end{equation}
\begin{equation}\label{Bianchi +1a}
\Omega\slashed\nabla_4 \frac{r^4\overone{\beta}}{\Omega}=r\slashed{div}\;r^3\overone{\alpha},\qquad\qquad \Omega\slashed\nabla_3\frac{r^4\overone{\underline\beta}}{\Omega}=-r\slashed{div}\; r^3\overone{\underline\alpha},
\end{equation}
\begin{equation}\label{Bianchi +1b}
\Omega\slashed\nabla_4r^2\Omega\overone{\underline\beta}=r\slashed{\mathcal{D}}^*_1(r\Omega^2\overone{\rho},r\Omega^2\overone{\sigma})+\frac{6M\Omega^2}{r}\overone{\underline\eta},\qquad\quad\Omega\slashed\nabla_3 r^2\Omega\overone{\beta}=r\slashed{\mathcal{D}}^*_1(-r\Omega^2\overone{\rho},r\Omega^2\overone{\sigma})-\frac{6M\Omega^2}{r}\overone{\eta},
\end{equation}
\begin{equation}\label{Bianchi 0}
\Omega\slashed\nabla_4\; r^3\overone{\rho}=r\slashed{div}\;r^2\Omega\overone{\beta}+3M \trx,\qquad\qquad \Omega\slashed\nabla_3\;r^3 \overone{\rho}=-r\slashed{div}\;r^2\Omega\overone{\underline\beta}+3M\trxbar,
\end{equation}
\begin{equation}\label{Bianchi 0*}
\Omega\slashed\nabla_4\; r^3\overone{\sigma}=-r\slashed{curl} \;r^2\Omega\overone{\beta},\qquad\qquad \Omega\slashed\nabla_3\; r^3\sigma=-r\slashed{curl}\;r^2\Omega \overone{\underline\beta}.
\end{equation}
\end{itemize}
\begin{remark}\label{regular}
The degeneration of the Eddington--Finkelstein (EF) frame near $\overline{\mathscr{H}^+}$ carries over to a degeneration of the quantities governed by equations \bref{start of full system}--\bref{Bianchi 0*}, as these quantities were derived via the EF frame (see Appendix B). By switching to a regular frame, e.g.~the Kruskal frame, it can be shown that these quantities extend regularly to $\overline{\mathscr{H}^+}$ when supplied with the appropriate weights in $U,V$. In particular, note that 
\begin{align}
    \tilde{\alpha}=V^{-2}\Omega^2\alpha,\qquad\underline{\widetilde{\alpha}}=U^2\Omega^{-2}\underline\alpha,
\end{align}
extend regularly to $\overline{\mathscr{H}^+}$, including $\mathcal{B}$.
\end{remark}

\section{The Teukolsky equations, the Teukolsky--Starobinsky identities and the Regge--Wheeler equations}\label{TRW}
\subsection{The Teukolsky equations and their well-posedness}\label{Chandra1}
Let $\overone\alpha$, $\overone{\underline\alpha}$ belong to a solution to the linearised Einstein equations \bref{start of full system}--\bref{Bianchi 0*}. It turns out that the linearised fields $\overone\alpha$, $\overone{\underline\alpha}$ obey decoupled 2nd order hyperbolic equations, the well-known Teukolsky equations.\\
\indent Take $\overone\alpha$ and multiply by $\frac{r^4}{\Omega^4}$:
\begin{equation}
\frac{r^4}{\Omega^4}\Omega\slashed\nabla_3\;r\Omega^2\overone\alpha=-2r\slashed{\mathcal{D}}^*_2 \frac{r^4\overone\beta}{\Omega}+6M\frac{r^2\overone{\hat{\chi}}}{\Omega}.
\end{equation}
Now differentiate in the $\Omega e_4$ direction and multiply by $\frac{\Omega^2}{r^2}$ to obtain the \textbf{Spin +2 Teukolsky equation}:
\begin{equation}\label{T+2}
\frac{\Omega^2}{r^2}\Omega\slashed\nabla_4 \;\frac{r^4}{\Omega^4}\Omega \slashed\nabla_3\; r\Omega^2\overone\alpha=-2r^2\slashed{\mathcal{D}}^*_2\slashed{\mathcal{D}}_2r\Omega^2\overone\alpha-\frac{6M}{r}r\Omega^2\overone\alpha.
\end{equation}
We note that:
\begin{equation}
\slashed{\mathcal{D}}_2^*\slashed{\mathcal{D}}_2=-\frac{1}{2}\slashed{\Delta}+\frac{1}{r^2},\qquad\qquad \Omega\slashed\nabla_4 \frac{r^2}{\Omega^2}=-\Omega\slashed\nabla_3 \frac{r^2}{\Omega^2}=r(x+2).
\end{equation}
We may rewrite the equation as:
\begin{equation}\label{T+2d}
-\frac{r^2}{\Omega^2}\Omega\slashed\nabla_3\Omega\slashed\nabla_4\;r\Omega^2\overone\alpha+r^2\slashed\Delta\;r\Omega^2\overone\alpha-2r(x+2)\Omega\slashed\nabla_3 r\Omega^2\overone\alpha+(3\Omega^2-5)r\Omega^2\overone\alpha=0.
\end{equation}
\indent An analogous procedure produces the \textbf{Spin }$\bm{-2}$\textbf{ Teukolsky equation}
\begin{equation}\label{T-2}
\frac{\Omega^2}{r^2}\Omega\slashed\nabla_3 \;\frac{r^4}{\Omega^4}\Omega \slashed\nabla_4\; r\Omega^2\overone{\underline\alpha}=-2r^2\slashed{\mathcal{D}}^*_2\slashed{\mathcal{D}}_2r\Omega^2\overone{\underline\alpha}-\frac{6M}{r}r\Omega^2\overone{\underline\alpha},
\end{equation}
which we may rewrite as
\begin{equation}\label{T-2d}
-\frac{r^2}{\Omega^2}\Omega\slashed\nabla_3\Omega\slashed\nabla_4\;r\Omega^2\overone{\underline\alpha}+r^2\slashed\Delta\;r\Omega^2\overone{\underline\alpha}+2r(x+2)\Omega\slashed\nabla_4 r\Omega^2\overone{\underline\alpha}+(3\Omega^2-5)r\Omega^2\overone{\underline\alpha}=0.
\end{equation} 
We now state well-posedness theorems which are standard for linear second-order hyperbolic equations of the type that \cref{T+2}, \cref{T-2} fall under.  Taking into account \Cref{regular}, we start with the future evolution of $\Omega^2\alpha$ and $\Omega^{-2}\underline\alpha$. \\
\indent Having derived the Teukolsky equations \bref{T+2}, \bref{T-2}, we can study these equations in isolation. Since the following theorems do not pertain to the linearised Einstein equations, we drop the superscript $\overone{{}}$.
\begin{proposition}\label{WP+2Sigma*}
Prescribe on $\Sigma^*$ a pair of smooth symmetric traceless $S^2_{u,v}$ 2-tensor fields $(\upalpha,\upalpha')$. Then there exists a unique smooth symmetric traceless $S^2_{u,v}$ 2-tensor field $\Omega^2\alpha$ that satisfies \bref{T+2} on $J^+(\Sigma^*)$, with $\Omega^2\alpha|_{\Sigma^*}=\upalpha, \slashednabla_{n_{\Sigma^*}}\Omega^2\alpha|_{\Sigma^*}=\upalpha'$.
\end{proposition}
\begin{proposition}\label{WP-2Sigma*}
Prescribe on $\Sigma^*$ a pair of smooth symmetric traceless $S^2_{u,v}$ 2-tensor fields $(\underline\upalpha,\underline\upalpha')$. Then there exists a unique smooth symmetric traceless $S^2_{u,v}$ 2-tensor field $\Omega^{-2}\underline\alpha$ that satisfies \bref{T-2} on $J^+(\Sigma^*)$, with $\Omega^{-2}\underline\alpha|_{\Sigma^*}=\underline\upalpha, \slashednabla_{n_{\Sigma^*}}\Omega^2\underline\alpha|_{\Sigma^*}=\underline\upalpha'$.
\end{proposition}

The same applies replacing $\Sigma^*$ with any other $\mathscr{H}^+$-penetrating spacelike surface ending at $i^0$.\\
\indent The degeneration of the EF frame discussed in \Cref{regular} is inherited by \bref{T+2}, \bref{T-2}, and we must work with $\widetilde{\alpha}=V^{-2}\Omega^2\alpha, \widetilde{\underline\alpha}=U^2\Omega^{-2}\underline\alpha$ in order to study the Teukolsky equations with data on $\overline{\Sigma}$. The weighted quantities $\widetilde{\alpha}, \widetilde{\underline\alpha}$ satisfy the following equations:

\begin{align}\label{T+2B}
    \frac{1}{\Omega^2}\nablau\nablav r\widetilde{\alpha}+\frac{1}{M}(4-3\Omega^2)\nablau r\widetilde{\alpha}-\frac{1}{r}(3\Omega^2-5)\widetilde{\alpha}-\slashed{\Delta}r\widetilde{\alpha}=0,
\end{align}
\begin{align}\label{T-2B}
    \frac{1}{\Omega^2}\nablau\nablav r\widetilde{\underline\alpha}-\frac{1}{M}(4-3\Omega^2)\nablav r\widetilde{\underline\alpha}-\frac{1}{r}(3\Omega^2-5)\widetilde{\underline\alpha}-\slashed{\Delta}r\widetilde{\underline\alpha}=0.
\end{align}
Equations (\ref{T+2B}) and (\ref{T-2B}) do not degenerate near $\mathcal{B}$ and we can make the following well-posedness statement:
\begin{proposition}\label{WP+2Sigmabar}
Prescribe a pair of smooth symmetric traceless $S^2_{U,V}$ 2-tensor fields $(\widetilde{\upalpha},\widetilde{\upalpha}')$ on $\overline{\Sigma}$. Then there exists a unique smooth symmetric traceless $S^2_{u,v}$ 2-tensor field $\Omega^2{\alpha}$ that satisfies (\ref{T+2}) on $ J^+(\overline{\Sigma})$ with $V^{-2}\Omega^2\alpha|_{\overline{\Sigma}}=\widetilde{\upalpha}$ and $\slashednabla_{n_{\overline{\Sigma}}}V^{-2}\Omega^2\alpha|_{\overline{\Sigma}}=\widetilde{\upalpha}'$.
\end{proposition}
\begin{proposition}\label{WP-2Sigmabar}
Prescribe a pair of smooth symmetric traceless $S^2_{U,V}$ 2-tensor fields $(\widetilde{\underline\upalpha},\widetilde{\underline\upalpha}')$ on $\overline{\Sigma}$. Then there exists a unique smooth symmetric traceless $S^2_{u,v}$ 2-tensor field $\Omega^{-2}{\underline\alpha}$ that satisfies (\ref{T-2}) on $ J^+(\overline{\Sigma})$ with $V^{2}\Omega^{-2}\underline\alpha|_{\overline{\Sigma}}=\widetilde{\underline\upalpha}$ and $\slashednabla_{n_{\overline{\Sigma}}}V^{2}\Omega^{-2}\underline\alpha|_{\overline{\Sigma}}=\widetilde{\underline\upalpha}'$.
\end{proposition}
Analogous statements to the above apply to past development from $\overline{\Sigma}$ with $U,\Omega^2$ switching places with $V,\Omega^{-2}$ respectively.\\
\indent In developing backwards scattering we will use the following well-posedness statement for the past development of a mixed initial-characteristic value problem:
\begin{proposition}\label{WP+2backwards}
Let $u_+<\infty, v_+<v_*<\infty$. Let $\widetilde{\Sigma}$ be a spacelike hypersurface connecting $\mathscr{H}^+$ at $v_*$ to $\mathscr{I}^+$ at $u_+$ and let $\underline{\mathscr{C}}=\underline{\mathscr{C}}_{v_*}\cap J^-(\widetilde{\Sigma})\cap J^+(\overline\Sigma)$. Prescribe a pair of symmetric traceless $S^2_{u,v}$ 2-tensor fields:
\begin{itemize}
\item $\alpha_{{\mathscr{H}^+}}$ on ${\mathscr{H}^+}\cap\{v\leq v_+\}$ vanishing in a neighborhood of $\mathscr{H}^+\cap\{v=v_+\}$, such that $V^{-2}\alpha_{\overline{\mathscr{H}^+}}$ extends smoothly to $\mathcal{B}$,
\item $\alpha_{0,in}$ on $\underline{\mathscr{C}}$ vanishing in a neighborhood of $\underline{\mathscr{C}}\cap\widetilde{\Sigma}$.
\end{itemize}
Then there exists a unique smooth symmetric traceless $S^2_{u,v}$ 2-tensor $\alpha$ on $D^-\left(\overline{\mathscr{H}^+}\cup\widetilde{\Sigma}\cup\underline{\mathscr{C}}\right)\cap J^+(\overline{\Sigma})$ such that $V^{-2}\Omega^2\alpha|_{\overline{\mathscr{H}^+}}=V^{-2}\alpha_{\overline{\mathscr{H}^+}}$, $\alpha|_{\underline{\mathscr{C}}}=\alpha_{0,in}$ and $\left(\Omega^2\alpha|_{\widetilde{\Sigma}},\slashednabla_{n_{\widetilde{\Sigma}}}\Omega^2\alpha|_{\widetilde{\Sigma}}\right)=(0,0)$.
\end{proposition}
\begin{proposition}\label{WP-2backwards}
Let $u_+<\infty, v_+<v_*<\infty$. Let $\widetilde{\Sigma}$ be a spacelike hypersurface connecting $\mathscr{H}^+$ at $v_+$ to $\mathscr{I}^+$ at $u_+$ and let $\underline{\mathscr{C}}=\underline{\mathscr{C}}_{v_*}\cap J^+(\widetilde{\Sigma})\cap\{t\geq0\}$. Prescribe a pair of symmetric traceless $S^2_{u,v}$ 2-tensor fields:
\begin{itemize}
\item $\underline\alpha_{{\mathscr{H}^+}}$ on ${\mathscr{H}^+}\cap\{v<v_+\}$ vanishing in a neighborhood of $v_+$, such that $V^{2}\underline\alpha_{{\mathscr{H}^+}}$ extends smoothly to $\mathcal{B}$,
\item $\underline\alpha_{0,in}$ on $\underline{\mathscr{C}}$ vanishing in a neighborhood of $\underline{\mathscr{C}}\cap\mathscr{H}^+$.
\end{itemize}
Then there exists a unique smooth symmetric traceless $S^2_{u,v}$ 2-tensor $\underline\alpha$ on $D^-\left(\overline{\mathscr{H}^+}\cup\widetilde{\Sigma}\cup\underline{\mathscr{C}}\right)\cap J^+(\overline{\Sigma})$ such that $V^{2}\Omega^{-2}\underline\alpha|_{\overline{\mathscr{H}^+}}=V^{2}\underline\alpha_{\overline{\mathscr{H}^+}}$, $\underline\alpha|_{\underline{\mathscr{C}}}=\underline\alpha_{0,in}$ and $\left(\Omega^{-2}\underline\alpha|_{\widetilde{\Sigma}},\slashednabla_{n_{\widetilde{\Sigma}}}\Omega^{-2}\underline\alpha|_{\widetilde{\Sigma}}\right)=(0,0)$.
\end{proposition}
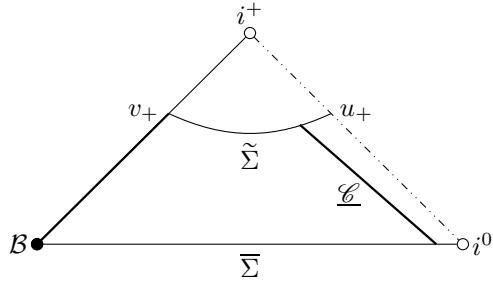
\begin{center}
\begin{tikzpicture}[scale=0.7]
\node (I)    at ( 0,0)   {};

\path % Four corners of the right diamond (no labels this time)
   (I) +(90:4)  coordinate (Itop) coordinate[label=90:$i^+$]
      % +(-90:4) coordinate (Ibot) 
       +(180:4) coordinate (Ileft) coordinate[label=180:$\mathcal{B}$]
       +(0:4)   coordinate (Iright) coordinate[label=0:$i^0$]
       ;

% No text this time in the next diagram coordinate[label=-90:$i^-$]
\draw  (Ileft) --  node[yshift=1mm,above]{$v_+$} (Itop) ;
%\draw  (Ileft) --(Ibot) ;
%\draw[dash dot dot] (Ibot) -- (Iright) ;
\draw[dash dot dot] (Iright) --  node[yshift=1mm,above]{$u_+$}(Itop) ;
\draw [line width=0.3mm]($(Ileft)$)--($(Ileft)+(45:3.5cm)$);
\draw [line width=0.3mm]($(Iright)+(180:0.5cm)$)--node[below,xshift=-0.2cm,yshift=0.15cm]{$\underline{\mathscr{C}}\;$}($(Iright)+(135:3.2cm)+(180:0.8cm)$);
\draw ($(Ileft)+(45:3.5cm)$)  to[out=-25, in=205, edge node={node [below] {$\widetilde{\Sigma}$}}] ($(Iright)+(135:3.5cm)$);
%\draw ($(Iright)+(135:3.5cm)+(180:0.8cm)$) to[out=15, in=25]  ($(Iright)+(135:3.8cm)$);

\draw ($(Ileft)$)  to[out=0, in=180, edge node={node [below] {$\overline{\Sigma}$}}] ($(Iright)$);

\filldraw[white] (Itop) circle (3pt);
\draw[black] (Itop) circle (3pt);

%\filldraw[white] (Ibot) circle (3pt);
%\draw[black] (Ibot) circle (3pt);

\filldraw[white] (Iright) circle (3pt);
\draw[black] (Iright) circle (3pt);
\filldraw[black] (Ileft) circle (3pt);
\draw[black] (Ileft) circle (3pt);
\end{tikzpicture}
\end{center}

We will also need

\begin{proposition}\label{backwards wellposedness +2}
Let $\tilde{\alpha}_{\mathscr{H}^+}$ be a smooth symmetric traceless $S^2_{\infty,v}$ 2-tensor on $\overline{\mathscr{H}^+}\cap J^-(\Sigma^*)$, $(\widetilde{\upalpha}_{\Sigma^*},\widetilde{\upalpha}_{\Sigma^*}')$ be a pair of smooth symmetric traceless $S^2_{\infty,v}$ 2-tensors on $\Sigma^*$. Then there exists a unique solution $\widetilde{\alpha}$ to \bref{T+2B} in $J^+(\overline{\Sigma})\cap\{t^*\leq 0\}$ such that $\widetilde{\alpha}|_{\overline{\mathscr{H}^+}}=\widetilde{\alpha}_{\overline{\mathscr{H}^+}}$,  $(\widetilde{\alpha}|_{\Sigma^*},\slashednabla_{n_{\Sigma^*}}\widetilde{\alpha}|_{\Sigma^*})=(\widetilde{\upalpha}_{\Sigma^*},\widetilde{\upalpha}_{\Sigma^*}')$. 
\end{proposition}
\begin{proposition}\label{backwards wellposedness -2}
An analogous statement to \Cref{backwards wellposedness +2} holds for \cref{T-2B}.
\end{proposition}
Analogous statements apply for the "finite" backwards scattering problem from the past of $\overline{\Sigma}$, with $U$ replacing $V$ and $\Omega^2$ switching places with $\Omega^{-2}$. 
\begin{remark}[\textbf{Time inversion}]\label{time inversion}
Under the transformation $t\longrightarrow-t$, $u\longrightarrow -v$ and $v\longrightarrow -u$ and thus $\alpha(u,v,\theta^A)\longrightarrow\alpha(-v,-u,\theta^A)=:\invertedalpha(u,v,\theta^A)$ and $\underline\alpha(u,v,\theta^A)\longrightarrow\underline\alpha(-v,-u,\theta^A)=:\underline\invertedalpha(u,v,\theta^A)$.\\
\indent It is clear $\invertedalpha(u,v,\theta^A)$ satisfies the $-2$ Teukolsky equation, i.e.~the equation satisfied by $\underline\alpha$. Similarly,  $\underline\invertedalpha(u,v,\theta^A)$ satisfies the $+2$ Teukolsky equation, i.e.~the equation satisfied by $\alpha$. This observation means that the asymptotics of $\alpha$ towards the future are identical to those of $\underline\alpha$ towards the past, i.e.~determining the asymptotics of both $\underline\alpha$ and $\alpha$ towards the future is enough to determine the asymptotics of either $\alpha$ or $\underline\alpha$ in both the past and future directions. We will use this fact to obtain bijective scattering maps from studying the forward evolution of the fields $\alpha,\underline\alpha$. In particular, this prescription is sufficient to obtain well-posedness statements for the equations (\ref{T-2}) and (\ref{T+2}) for past development.
\end{remark}

\subsection{Derivation of the Teukolsky--Starobinsky identities}\label{derivation of the Teukolsky--Starobinsky identities}
We now return to the full system \bref{start of full system}--\bref{Bianchi 0*} to derive the Teukolsky--Starobinsky identities \bref{eq:227intro1}, \bref{eq:228intro1}.\\
\indent Let $\overone\alpha$ belong to a solution of the linearised Einstein equations. \Cref{Bianchi +2} implies:
\begin{align}
    \frac{r^2}{\Omega^2}\nablau r\Omega^2\overone{\alpha}=-2r\fancydstar_2r^2\Omega\overone{\beta}+6M\Omega\overone{\hat{\chi}}.
\end{align}
Using \bref{Bianchi +1a} and \bref{D3Chihat} we obtain
\begin{equation}
\left(\frac{r^2}{\Omega^2}\nablau\right)^2r\Omega^2\overone\alpha=-2r^2\slashed{\mathcal{D}}^*_2\slashed{\mathcal{D}}^*_1\left(-r^3\overone\rho,r^3\overone\sigma\right)+6M(r\Omega\overone{\hat{\chi}}-r\Omega\overone{\underline{\hat\chi}}).
\end{equation}
We now apply $\frac{r^2}{\Omega^2}\nablau$ to both sides and use equations \bref{Bianchi 0}, \bref{Bianchi 0*}, \bref{D3Chihat} and the second equation of \bref{D4Chihat} to deduce

\begin{align}
\left(\frac{r^2}{\Omega^2}\nablau\right)^3r\Omega^2\overone\alpha=-2r^2\slashed{\mathcal{D}}_2^*\slashed{\mathcal{D}}^*_1\overline{\slashed{\mathcal{D}}}_1\left(\frac{r^4\overone{\underline\beta}}{\Omega}\right)+6M\left[r^2\slashed{\mathcal{D}}_2^*\slashed{\mathcal{D}}^*_1\left(\frac{r^2}{\Omega^2}\overone{\underline{f}},0\right)+ r^3\overone{\underline\alpha}-(3\Omega^2-1)\frac{r^2\overone{\underline{\hat\chi}}}{\Omega}-2r\slashed{\mathcal{D}}_2^*r^2\overone\eta\right].
\end{align}
Now we apply $\nablau$ once again and use \bref{D3TrChiBar}, the second equation of \bref{D4Chihat} and the second equations of \bref{D4etabar}:
\begin{align}
\begin{split}
\Omega\slashed{\nabla}_3 \left(\frac{r^2}{\Omega^2}\nablau\right)^3r\Omega^2\overone\alpha&=-2r^3\slashed{\mathcal{D}}^*_2\slashed{\mathcal{D}}^*_1\overline{\slashed{\mathcal{D}}}_1(-r\slashed{\mathcal{D}}_2r^3\overone{\underline\alpha})+6M\Bigg[r^2\slashed{\mathcal{D}}^*_2\slashed{\mathcal{D}}^*_1\left(-4r\overone{\underline\omega},0\right)-(3\Omega^2-1)r^2\overone{\underline\alpha}\\&\;\;+\frac{r^2}{\Omega^2}\Omega\slashed{\nabla}_3r\Omega^2\overone\alpha+6M\frac{r^2}{\Omega^2} \frac{r^2\overone{\underline{\hat\chi}}}{\Omega}-(3\Omega^2-1)(-r^2\overone{\underline\alpha})-2r\slashed{\mathcal{D}}^*_2(2r\slashed{\nabla}r\overone{\underline\omega}-r^2\Omega\overone{\underline\beta})\Bigg]
\\&=2r^4\slashed{\mathcal{D}}^*_2\slashed{\mathcal{D}}^*_1\overline{\slashed{\mathcal{D}}}_1\slashed{\mathcal{D}}_2 r^3\overone{\underline\alpha}+6M\frac{r^2}{\Omega^2}\left[\Omega\slashed{\nabla}_4+\Omega\slashed{\nabla}_3\right]r\Omega^2\overone{\underline\alpha}.
\end{split}
\end{align}
Finally, we have
\begin{align}\label{eq:TS1}
\frac{\Omega^2}{r^2}\Omega\slashed{\nabla}_3 \left(\frac{r^2}{\Omega^2}\nablau\right)^3r\Omega^2\overone\alpha=2r^4\slashed{\mathcal{D}}^*_2\slashed{\mathcal{D}}^*_1\overline{\slashed{\mathcal{D}}}_1\slashed{\mathcal{D}}_2 r\Omega^2\overone{\underline\alpha}+6M\left[\Omega\slashed{\nabla}_4+\Omega\slashed{\nabla}_3\right]r\Omega^2\overone{\underline\alpha}.
\end{align}
\indent An entirely analogous procedure starting from the equation for $\overone{\underline\alpha}$ in \bref{Bianchi +2} leads to 
\begin{align}\label{eq:TS2}
\frac{\Omega^2}{r^2}\Omega\slashed{\nabla}_4 \left(\frac{r^2}{\Omega^2}\nablav\right)^3r\Omega^2\overone{\underline\alpha}=2r^4\slashed{\mathcal{D}}^*_2\slashed{\mathcal{D}}^*_1\overline{\slashed{\mathcal{D}}}_1\slashed{\mathcal{D}}_2 r\Omega^2\overone{\alpha}-6M\left[\Omega\slashed{\nabla}_4+\Omega\slashed{\nabla}_3\right]r\Omega^2\overone{\alpha}.
\end{align}

Equation \bref{eq:TS2} is the constraint \bref{eq:228intro1}.
\subsection{Physical-space Chandrasekhar transformations and the Regge--Wheeler equation}\label{Chandra}
The Regge--Wheeler equation for a symmetric traceless $S^2_{u,v}$ 2-tensor $\Psi$ is given by
\begin{align}\label{RW}
    \nablav\nablau\Psi-\Omega^2\slashed{\Delta}\Psi+\frac{\Omega^2}{r^2}(3\Omega^2+1)\Psi=0.
\end{align}
\indent Suppose the field $\alpha$ satisfies the +2 Teukolsky equation. Define the following hierarchy of fields
\begin{align}\label{hier+}
\begin{split}
    &r^3\Omega \psi:=\frac{r^2}{\Omega^2}\nablau r\Omega^2\alpha,\\
    &\Psi:=\frac{r^2}{\Omega^2}\nablau r^3\Omega \psi=\left(\frac{r^2}{\Omega^2}\nablau\right)^2 r\Omega^2\alpha.
\end{split}
\end{align}
We have the following commutation relation:
\begin{align}\label{commutation relation}
\begin{split}
\Bigg[&-\frac{r^2}{\Omega^2}\Omega\slashed\nabla_3\Omega\slashed\nabla_4-(k+xk^')r\Omega\slashed\nabla_3+a\Omega^2+bx+c\Bigg]\frac{r^2}{\Omega^2}\Omega\slashed\nabla_3
\\&=\frac{r^2}{\Omega^2}\Omega\slashed\nabla_3\left[-\frac{r^2}{\Omega^2}\Omega\slashed\nabla_3\Omega\slashed\nabla_4-\left(k+2+x(k^'+1)\right)r\Omega\slashed\nabla_3+(a+2k+2k^')\Omega^2+bx+c-k-2k^'\right]\\&+2M(a+2k+2k^'),
\end{split}
\end{align}
where $a,b,c,k,k'$ are integers. We commute the operator $\left(\frac{r^2}{\Omega^2}\Omega\slashed\nabla_3\right)^2$ past the Regge--Wheeler operator:
\begin{align*}
\begin{split}
\left[-\frac{r^2}{\Omega^2}\Omega\slashed\nabla_3\Omega\slashed\nabla_4+r^2\slashed\Delta-3\Omega^2-1\right]\left(\frac{r^2}{\Omega^2}\Omega\slashed\nabla_3\right)^2=\Bigg\{\frac{r^2}{\Omega^2}\Omega\slashed\nabla_3&\Bigg[-\frac{r^2}{\Omega^2}\Omega\slashed\nabla_3\Omega\slashed\nabla_4+r^2\slashed\Delta-(2+x)r\Omega\slashed\nabla_3\\
&-3\Omega^2-1\Bigg]-6M\Bigg\}\frac{r^2}{\Omega^2}\Omega\slashed\nabla_3
\end{split}
\end{align*}
\begin{align}
\begin{split}
&=\frac{r^2}{\Omega^2}\Omega\slashed\nabla_3\Bigg\{\left[-\frac{r^2}{\Omega^2}\Omega\slashed\nabla_3\Omega\slashed\nabla_4+r^2\slashed\Delta-(2+x)r\Omega\slashed\nabla_3-3\Omega^2-1\right]\frac{r^2}{\Omega^2}\Omega\slashed\nabla_3-6M\Bigg\}
\\&=\left(\frac{r^2}{\Omega^2}\Omega\slashed\nabla_3\right)^2\Bigg\{\left[-\frac{r^2}{\Omega^2}\Omega\slashed\nabla_3\Omega\slashed\nabla_4+r^2\slashed\Delta-2(2+x)r\Omega\slashed\nabla_3+3\Omega^2-5\right]-6M+6M\Bigg\}\label{commutator}
\end{split}
\end{align}
This shows that if $\alpha$ satisfies the +2 Teukolsky equation then $\Psi$ satisfies the Regge--Wheeler equation (\ref{RW}).\\ \indent Analogously, with the following hierarchy of fields
\begin{align}\label{hier-}
\begin{split}
    &r^3\Omega \underline\psi:=\frac{r^2}{\Omega^2}\nablav r\Omega^2\underline\alpha,\\
    &\underline\Psi:=\frac{r^2}{\Omega^2}\nablav r^3\Omega \underline\psi=\left(\frac{r^2}{\Omega^2}\nablav\right)^2 r\Omega^2\underline\alpha,
\end{split}
\end{align}
we have
\begin{align}\label{commutation relation 2}
\begin{split}
\Bigg[&-\frac{r^2}{\Omega^2}\Omega\slashed\nabla_3\Omega\slashed\nabla_4+(l+xl^')r\Omega\slashed\nabla_4+a\Omega^2+bx+c\Bigg]\frac{r^2}{\Omega^2}\Omega\slashed\nabla_4
\\&=\frac{r^2}{\Omega^2}\Omega\slashed\nabla_4\left[-\frac{r^2}{\Omega^2}\Omega\slashed\nabla_3\Omega\slashed\nabla_4+\left(l+2+x(l^'+1)\right)r\Omega\slashed\nabla_4+(a+2l+2l^')\Omega^2+bx+c-l-2l^'\right]\\&+6M(a+2l+2l^'),
\end{split}
\end{align}
where $a,b,c,l,l'$ are integers. Thus, if $\underline\alpha$ satisfies the $-2$ Teukolsky equation then $\underline\Psi$ also satisfies the Regge--Wheeler equation.\\
\indent We state a standard well-posedness result for (\ref{RW}):
\begin{proposition}\label{RWwpCauchy}
For any pair $(\uppsi,\uppsi')$ of smooth symmetric traceless $S^2_r$ 2-tensor fields on $\Sigma^*$, there exists a unique smooth symmetric traceless $S^2_{u,v}$ 2-tensor field $\Psi$ which solves \cref{RW} in $ J^+(\Sigma^*)$ such that $\Psi|_{\Sigma^*}=\uppsi$ and $\slashednabla_{n_{\Sigma^*}} \Psi|_{\Sigma^*}=\uppsi'$. The same applies when data are posed on $\Sigma$ or $\overline{\Sigma}$.
\end{proposition}
In contrast to the Teukolsky equations \bref{T+2}, \bref{T-2}, the Regge--Wheeler equation \bref{RW} does not suffer from additional regularity issues near $\mathcal{B}$, as can be seen by rewriting \cref{RW} in Kruskal coordinates:
\begin{align}
    \slashednabla_U\slashednabla_V\Psi-\mathring{\slashed{\Delta}}+\frac{3\Omega^2+1}{r^2}\Psi=0.
\end{align}
If $\Psi$ is related to a field $\alpha$ that satisfies \bref{T+2}, then it is related to $\widetilde{\alpha}$ by
\begin{align}
    \Psi=\left(\frac{r^2}{\Omega^2}\nablau\right)^2r\Omega^2\alpha=\left(2Mr^2f(r)\slashednabla_U\right)^2r\tilde{\alpha}.
\end{align}
\begin{proposition}\label{RWwpSigmabar}
\Cref{RWwpCauchy} is valid replacing $\Sigma^*$ with $\overline{\Sigma}$ everywhere.
\end{proposition}
For backwards scattering we will need the following well-posedness statement:
\begin{proposition}\label{RWwpBackwards}
Let $u_+<\infty, v_+<v_*<\infty$. Let $\widetilde{\Sigma}$ be a spacelike hypersurface connecting $\mathscr{H}^+$ at $v=v_+$ to $\mathscr{I}^+$ at $u=u_+$ and let $\underline{\mathscr{C}}=\underline{\mathscr{C}}_{v_*}\cap J^+(\widetilde{\Sigma})\cap\{t\geq0\}$. Prescribe a pair of smooth symmetric traceless $S^2_{u,v}$ 2-tensor fields:
\begin{itemize}
\item $\Psi_{{\mathscr{H}^+}}$ on ${\overline{\mathscr{H}^+}}\cap\{v<v_+\}$ vanishing in a neighborhood of $\widetilde{\Sigma}$,
\item $\Psi_{0,in}$ on $\underline{\mathscr{C}}$ vanishing in a neighborhood of $\widetilde{\Sigma}$.
\end{itemize}
Then there exists a unique smooth symmetric traceless $S^2_{u,v}$ 2-tensor $\Psi$ on $D^-\left(\overline{\mathscr{H}^+}\cup\widetilde{\Sigma}\cup\underline{\mathscr{C}}\right)\cap J^+(\overline{\Sigma})$ such that $\Psi|_{\overline{\mathscr{H}^+}}=\Psi_{{\mathscr{H}^+}}$, $\Psi|_{\underline{\mathscr{C}}}=\Psi_{0,in}$ and $\left(\Psi|_{\widetilde{\Sigma}},\slashednabla_{n_{\widetilde{\Sigma}}}\Psi|_{\widetilde{\Sigma}}\right)=(0,0)$.
\end{proposition}
We will also need
\begin{proposition}\label{RWwp local statement near B}
Let $(\uppsi,\uppsi')$ be smooth symmetric traceless $S^2_{u,v}$ 2-tensor fields on $\Sigma^*$, $\uppsi_{\mathscr{H}^+}$ be a smooth symmetric traceless $S^2_{\infty,v}$ 2-tensor field on $\overline{\mathscr{H}^+}\cap\{t^*\leq0\}$. Then there exists a unique smooth symmetric traceless $S^2_{u,v}$ 2-tensor field $\Psi$ on $J^-(\Sigma^*)$ such that $\Psi|_{\overline{\mathscr{H}^+}\cap\{t^*\leq0\}}=\uppsi_{\mathscr{H}^+}, \left(\Psi|_{\Sigma^*},\slashednabla_{n_{\Sigma^*}}\Psi|_{\Sigma^*}\right)=(\uppsi,\uppsi')$.
\end{proposition}

\begin{remark}\label{time inversion of RW}
Unlike the Teukolsky equations \bref{T+2}, \bref{T-2}, the Regge--Wheeler equation \bref{RW} is invariant under time inversion. If $\Psi(u,v)$ satisfies \bref{RW}, then $\invertedpsi(u,v):=\Psi(-v,-u)$ also satisfies \bref{RW}.
\end{remark}
\subsection{Further constraints among $\alpha,\Psi$ and $\underline\alpha,\underline\Psi$}\label{constraint derivation}
We can apply the same ideas as in \Cref{Chandra} to transform solutions of the Regge--Wheeler equation into solutions of the +2 Teukolsky equation. Let $\Psi$ satisfy \Cref{RW}, then using \bref{commutation relation 2} we can show that
\begin{align}\label{alpha to Psi}
    \frac{\Omega^2}{r^2}\nablav\frac{r^2}{\Omega^2}\nablav\Psi
\end{align}
satisfies \Cref{T+2}. \\
\indent Now suppose $\alpha$ satisfies \Cref{T+2} and $\Psi$ is the solution to \Cref{RW} related to $\alpha$ by \Cref{hier+}. We can evaluate the expression \bref{alpha to Psi} using \Cref{T+2}: we apply $\Omega\slashed{\nabla}_4$ and substitute using the $+2$ equation only (we drop the superscript $\overone{}$):
\begin{align*}
\begin{split}
\Omega\slashed{\nabla}_4\Psi&=\Omega\slashed{\nabla}_4\left(\frac{r^2}{\Omega^2}\Omega\slashed{\nabla}_3\right)^2r\Omega^2\alpha
\\&=r(x+2)\Omega\slashed{\nabla}_3\frac{r^2}{\Omega^2}\Omega\slashed{\nabla}_3r\Omega^2\alpha+\frac{r^2}{\Omega^2}\Omega\slashed{\nabla}_4\Omega\slashed{\nabla}_3\frac{r^2}{\Omega^2}\Omega\slashed{\nabla}_3r\Omega^2\alpha
\\&=\frac{3\Omega^2-1}{r}\Psi+\frac{r^2}{\Omega^2}\Omega\slashed{\nabla}_3\Omega\slashed{\nabla}_4\frac{r^2}{\Omega^2}\Omega\slashed{\nabla}_3r\Omega^2\alpha
\\&=\frac{3\Omega^2-1}{r}\Psi+\frac{r^2}{\Omega^2}\Omega\slashed{\nabla}_3\Omega\slashed{\nabla}_4\frac{r^4}{\Omega^4}\frac{\Omega^2}{r^2}\Omega\slashed{\nabla}_3r\Omega^2\alpha
\\&=\frac{3\Omega^2-1}{r}\Psi+\frac{r^2}{\Omega^2}\Omega\slashed{\nabla}_3\left[\left(-\frac{\Omega^4}{r^4}r(x+2)\right)\frac{r^4}{\Omega^4}\Omega\slashed{\nabla}_3r\Omega^2\alpha\right]
\\&\;\;+\frac{r^2}{\Omega^2}\Omega\slashed{\nabla}_3\frac{\Omega^2}{r^2}\Omega\slashed{\nabla}_4\frac{r^4}{\Omega^4}\Omega\slashed{\nabla}_3r\Omega^2\alpha
\\&=\frac{3\Omega^2-1}{r}\Psi-\frac{r^2}{\Omega^2}\Omega\slashed{\nabla}_3\left[\frac{\Omega^2}{r^2}r(x+2)\frac{r^2}{\Omega^2}\Omega\slashed{\nabla}_3r\Omega^2\alpha\right]+\frac{r^2}{\Omega^2}\Omega\slashed{\nabla}_3\mathcal{T}^{+2}_N r\Omega^2\alpha
\\&=-2(3\Omega^2-2)\frac{r^2}{\Omega^2}\Omega\slashed{\nabla}_3r\Omega^2\alpha+\frac{r^2}{\Omega^2}\Omega\slashed{\nabla}_3\mathcal{T}^{+2}_Nr\Omega^2\alpha\\
\end{split}
\end{align*}
\begin{align}\label{eq:d4Psi}
    =-2r^2\slashed{\mathcal{D}}^*_2\slashed{\mathcal{D}}_2 \frac{r^2}{\Omega^2}\Omega\slashed{\nabla}_3r\Omega^2\alpha-6Mr\Omega^2\alpha-(3\Omega^2-1)\frac{r^2}{\Omega^2}\Omega\slashed{\nabla}_3r\Omega^2\alpha,
\end{align}
i.e.,
\begin{align}
\begin{split}
\frac{r^2}{\Omega^2}\Omega\slashed{\nabla}_4\Psi=-2r^2\slashed{\mathcal{D}}_2^*\slashed{\mathcal{D}}_2 \frac{r^4}{\Omega^4}\Omega\slashed{\nabla}_3r\Omega^2\alpha-(3\Omega^2-1)\frac{r^4}{\Omega^4}\Omega\slashed{\nabla}_3r\Omega^2\alpha-6M\frac{r^2}{\Omega^2} r\Omega^2\alpha.
\end{split}
\end{align}
We act on both sides with $\nablav$ again:
\begin{align}
\begin{split}
\Omega\slashed{\nabla}_4 \frac{r^2}{\Omega^2}\Omega\slashed{\nabla}_4 \Psi &=-2r^2\slashed{\mathcal{D}}_2^*\slashed{\mathcal{D}}_2\left[\frac{r^2}{\Omega^2}\left(-2r^2\slashed{\mathcal{D}}_2^*\slashed{\mathcal{D}}_2 r\Omega^2\alpha-\frac{6M}{r}r\Omega^2\alpha\right)\right]\qquad\qquad\qquad\qquad\qquad\\
&\;\;\;\;\;\;-6M\left[\frac{r^2}{\Omega^2}\left(\Omega\slashed{\nabla}_3+\Omega\slashed{\nabla}_4\right)r\Omega^2\alpha+r(x+2)r\Omega^2\alpha\right] \\&\qquad-\left[-2r^2\fancydstar_2\fancyd_2-\frac{6M}{r}\right]\left[\frac{r^2}{\Omega^2}(3\Omega^2-1)r\Omega^2\alpha\right]
\\&=-2r^2\slashed{\mathcal{D}}_2^*\slashed{\mathcal{D}}_2\left[\frac{r^2}{\Omega^2}\left(-2r^2\slashed{\mathcal{D}}_2^*\slashed{\mathcal{D}}_2 r\Omega^2\alpha-2r\Omega^2\alpha\right)\right]-6M\left[\frac{r^2}{\Omega^2}\left(\Omega\slashed{\nabla}_3+\Omega\slashed{\nabla}_4\right)r\Omega^2\alpha\right].
\end{split}
\end{align}
We finally arrive at
\begin{align}\label{eq:d4d4Psi}
\begin{split}
\frac{\Omega^2}{r^2}\Omega\slashed{\nabla}_4 \frac{r^2}{\Omega^2}\Omega\slashed{\nabla}_4 \Psi&=-2r^2\slashed{\mathcal{D}}_2^*\slashed{\mathcal{D}}_2\left[-2r^2\slashed{\mathcal{D}}_2^*\slashed{\mathcal{D}}_2 r\Omega^2\alpha-2r\Omega^2\alpha\right]-6M\left[\left(\Omega\slashed{\nabla}_3+\Omega\slashed{\nabla}_4\right)r\Omega^2\alpha\right].
\end{split}
\end{align}
\indent We record the same for $\underline\Psi$: Using only the Teukolsky equation \bref{T-2} we obtain the analogue of \bref{eq:d4Psi}
\begin{align}\label{eq:d3psibar}
\Omega\slashed{\nabla}_3\underline\Psi=-(3\Omega^2-1)\frac{r^2}{\Omega^2}\Omega\slashed{\nabla}_4r\Omega^2\underline\alpha+6Mr\Omega^2\underline\alpha-2r^2\slashed{\mathcal{D}}^*_2\slashed{\mathcal{D}}_2\frac{r^2}{\Omega^2}\Omega\slashed{\nabla}_4r\Omega^2\underline\alpha,
\end{align}
and the analogue of \bref{eq:d4d4Psi}
\begin{align}\label{eq:d3d3psibar}
\frac{\Omega^2}{r^2}\Omega\slashed{\nabla}_3\frac{r^2}{\Omega^2}\Omega\slashed{\nabla}_3\underline\Psi=+6M\left[\Omega\slashed{\nabla}_4+\Omega\slashed{\nabla}_3\right]r\Omega^2\underline\alpha+\left[-2r^2\slashed{\mathcal{D}}_2^*\slashed{\mathcal{D}}_2-2\right]\left(-2r^2\slashed{\mathcal{D}}_2^*\slashed{\mathcal{D}}_2r\Omega^2\underline\alpha\right).
\end{align}
In the remainder of this paper we focus exclusively on the Teukolsky equations \bref{T+2}, \bref{T-2}, the Teukolsky--Starobinsky identities \bref{eq:227}, \bref{eq:228} and the Regge--Wheeler equation \bref{RW}. In particular, we do not refer to the linearised Einstein equations \bref{start of full system}--\bref{Bianchi 0*} and as such, we drop the superscript $\overone{{}}$.\\
\indent Throughout this paper we will we distinguish between solutions arising from data on $\Sigma^*, \Sigma$ or $\overline{\Sigma}$, and we subsequently construct separate scattering statements for each of these cases, in particular distinguishing between spaces of scattering states on $\mathscr{H}^+_{\geq0}, \mathscr{H}^\pm$ and $\overline{\mathscr{H}^\pm}$. It will be easiest to work with data $\Sigma^*$ first, and then the results for the remaining cases would follow easily.
\section{Main theorems}\label{section 4 main theorems}
We define in this section the spaces of scattering states and provide a precise statement of the results. In what follows, $L^2$ spaces on $\mathscr{I}^\pm, \mathscr{H}^+_{\geq0}, \mathscr{H}^\pm,\overline{\mathscr{H}^\pm}$ are defined with respect to the measures $du\sin\theta d\theta d\phi$, $dv\sin\theta d\theta d\phi$ induced by the Eddington--Finkelstein coordinates.
\begin{notation*}
    For a spherically symmetric submanifold $\mathcal{S}$ of $\overline{\mathscr{M}}$, denote by $\Gamma(\mathcal{S})$ the space of smooth symmetric traceless $S^2_{u,v}$ 2-tensor fields on $\mathcal{S}$. The space of such fields that are compactly supported is denoted by $\Gamma_c (\mathcal{S})$. We use the same notation for smooth fields on $\mathscr{I}^\pm, \mathscr{H}^\pm,\overline{\mathscr{H}^\pm}$.  
\end{notation*}
\noindent In particular, note that $A\in\Gamma(\Sigma^*)$ says that $A$ is smooth up to and including $\Sigma^*\cap\mathscr{H}^+$.
\subsection{Theorem 1: Scattering for the Regge--Wheeler equation}\label{subsection 4.1 Theorem 1}
\begin{defin}\label{RWscatteringsigma}
Let $(\uppsi,\uppsi')\in\Gamma_c (\Sigma^*)\oplus\Gamma_c(\Sigma^*)$ be Cauchy data on $\Sigma^*$ for \bref{RW} of compact support. Define the space $\mathcal{E}^{T}_{\Sigma^*} $ to be the completion of $\Gamma_c (\Sigma^*)$ data under the norm
\begin{align}\label{this22222}
    \|(\uppsi,\uppsi')\|^2_{\mathcal{E}^T_{\Sigma^*}}=\int_{\Sigma^*} dr\sin\theta d\theta d\phi\; (2-\Omega^2)|\slashednabla_{T^*}\Psi|^2+\Omega^2|\slashednabla_{R}\Psi|^2+|\slashednabla\Psi|^2+\frac{3\Omega^2+1}{r^2}|\Psi|^2,
\end{align}
where $\Psi$ is smooth and satisfies $\Psi|_{\Sigma^*}=\uppsi, \slashednabla_{n_{\Sigma^*}}\Psi|_{\Sigma^*}=\uppsi'$. The space $\mathcal{E}^T_{\Sigma}$ is similarly defined with the norm
\begin{align}\label{this2222}
    \|(\uppsi,\uppsi')\|^2_{\mathcal{E}^T_{\Sigma}}=\int_{\Sigma} dr\sin\theta d\theta d\phi \;|\slashednabla_{n_{\Sigma}} \Psi|^2+\Omega^2|\slashednabla_R\Psi|^2+|\slashednabla\Psi|^2+\frac{3\Omega^2+1}{r^2}|\Psi|^2.
\end{align}
Define the space $\mathcal{E}^T_{\Sigma}$ to be the completion of $\Gamma_{c}(\Sigma)$ data under the norm \bref{this2222}. The space $\mathcal{E}^{T}_{\overline{\Sigma}}$ and the norm $\|\;\|_{\mathcal{E}^T_{\overline{\Sigma}}}$ are similarly defined.
\end{defin}
\begin{remark}\label{RW enough to be in space}
The kernel of $\|\;\;\|_{\mathcal{E}^T_{\Sigma^*}}$ has trivial intersection with $\Gamma(\Sigma^*)$. It suffices for a smooth data set $(\uppsi,\uppsi')$ to satisfy $\|(\uppsi,\uppsi')\|_{\mathcal{E}^T_{\Sigma^*}}<\infty$ to have $(\uppsi,\uppsi')\in\mathcal{E}^T_{\Sigma^*}$, so $\|\;\|_{\mathcal{E}^T_{\Sigma*}}, \|\;\|_{\mathcal{E}^T_{\Sigma}}, \|\;\|_{\mathcal{E}^T_{\overline\Sigma}}$ and \bref{this2222} define normed spaces that can be extended to Hilbert spaces.
\end{remark}
\begin{defin}\label{RW def of rad at H}
    Define the space $\mathcal{E}_{\mathscr{H}^+}^{T}$ to be the completion of $\Gamma_c (\mathscr{H}^+_{\geq 0})$ under the norm
    \begin{align}\label{RW def rad flux at H}
        ||\Psi||_{\mathcal{E}_{\mathscr{H}^+_{\geq 0}}^{T}}^2=\int_{\mathscr{H}^+_{\geq v_0}}|\partial_v\Psi|^2\sin\theta d\theta d\phi dv.
    \end{align}
The spaces $\mathcal{E}^T_{\mathscr{H}^+}$, $\mathcal{E}^T_{\overline{\mathscr{H}^+}}$ are analogously defined.
\end{defin}
\begin{remark}\label{Subspace of L2}
\begin{enumerate}
   \item  The energy $\|\;\|_{\mathcal{E}_{\mathscr{H}^+_{\geq0}}^{T}}$ indeed defines a norm on $\Gamma_c(\mathscr{H}^+_{\geq0})$, which thus extends to a Hilbert space $\mathcal{E}^{T}_{\mathscr{H}^+_{\geq0}}$ when completed under $\|\;\|_{\mathcal{E}_{\mathscr{H}^+_{\geq0}}^{T}}$.\\
   \item The space $\mathcal{E}^{T}_{\mathscr{H}^+_{\geq0}}$ can be realised as the subset $\Psi_{\mathscr{H}^+}\in L^2_{loc}(\mathscr{H}^+_{\geq0})$ such that 
   \begin{itemize}
       \item $\nablav\Psi_{\mathscr{H}^+}\in L^2(\mathscr{H}^+_{\geq0})$,
       \item $\lim_{v\longrightarrow\infty} \|\Psi_{\mathscr{H}^+}\|_{L^2(S_{\infty,v}^2)}=0$.
   \end{itemize}
   Note that Hardy's inequality holds on elements of this space and we have
   \begin{align}\label{weighted L2 statement}
       \int_{\mathscr{H}^+_{\geq0}} dv \sin\theta d\theta d\phi \frac{|\Psi_{\mathscr{H}^+}|^2}{v^2+1}\lesssim\|\xi\|^2_{\mathcal{E}^{T}_{\mathscr{H}^+_{\geq0}}}<\infty.
   \end{align}
\end{enumerate}
\end{remark}
\begin{defin}
    Define the space $\mathcal{E}^T_{\mathscr{I}^+}$ to be the completion of $\Gamma_c(\mathscr{I}^+)$ under the norm
    \begin{align}
        \|\Psi\|_{\mathcal{E}_{\mathscr{I}^+}^{T}}^2=\int_{\mathscr{I}^+}|\partial_u\Psi|^2\sin\theta d\theta d\phi du.
    \end{align}
\end{defin}
\begin{defin}
    Define the space $\mathcal{E}_{\mathscr{H}^-}^{T}$ to be the completion of $\Gamma_c (\mathscr{H}^-)$ under the norm
    \begin{align}
        ||\Psi||_{\mathcal{E}_{\mathscr{H}^-}^{T}}^2=\int_{\mathscr{H}^-}|\partial_u\Psi|^2\sin\theta d\theta d\phi du.
    \end{align}
The space $\mathcal{E}^T_{\overline{\mathscr{H}^-}}$ is similarly defined.
\end{defin}
\begin{defin}
    Define the space $\mathcal{E}_{\mathscr{I}^-}^T$ to be the completion of $\Gamma_c(\mathscr{I}^-)$ under the norm
    \begin{align}
        \|\Psi\|_{\mathcal{E}^T_{\mathscr{I}^-}}^2=\int_{\mathscr{I}^-} |\partial_v \Psi|^2 dv\sin\theta d\theta d\phi.
    \end{align}
\end{defin}
\begin{remark} Similar statements to \Cref{Subspace of L2} apply to the norms $\|\;\;\|_{\mathcal{E}^T_{\mathscr{H}^\pm}}, \|\;\;\|_{\mathcal{E}^T_{\overline{\mathscr{H}^\pm}}}, \|\;\;\|_{\mathcal{E}^T_{\mathscr{I}^\pm}}$; they are positive-definite on smooth, compactly supported data on the respective regions of $\overline{\mathscr{M}}$, thus they define normed spaces which extend to Hilbert spaces $\mathcal{E}^T_{\mathscr{H}_{\geq0}}, \mathcal{E}^T_{\mathscr{H}^\pm}, \mathcal{E}^T_{\overline{\mathscr{H}^\pm}}, \mathcal{E}^T_{\mathscr{I}^\pm}$ upon completion. Elements of these spaces can be identified with tensor fields in $L^2_{loc}(\mathscr{H}^-)$ for which a similar statement to \bref{weighted L2 statement} applies.
\end{remark}
\begin{thm}\label{forwardRW}
Let $(\uppsi,\uppsi')\in\Gamma_c(\Sigma^*)\times\Gamma_c(\Sigma^*)$. Then the corresponding unique solution $\Psi$ to \bref{RW} given by \Cref{RWwpCauchy} on $J^+(\Sigma^*)$ induces smooth radiation fields $(\bm{\uppsi}_{\mathscr{H}^+},\bm{\uppsi}_{\mathscr{I}^+})\in \Gamma(\mathscr{H}^+_{\geq0})\oplus\Gamma(\mathscr{I}^+)$ as in definitions \ref{RW future rad field scri} and \ref{RWonH}, 
with $\bm{\uppsi}_{\mathscr{I}^+}, \Psi_{\mathscr{H}^+_{\geq0}}$ satisfying
    \begin{align}
        \left|\left|(\uppsi,\uppsi')\right|\right|_{\mathcal{E}^T_{\Sigma^*}}^2=\left|\left|\bm{\uppsi}_{\mathscr{I}^+}\right|\right|_{\mathcal{E}^T_{\mathscr{I}^+}}^2+\left|\left|\bm{\uppsi}_{\mathscr{H}^+}\right|\right|_{\mathcal{E}^T_{\mathscr{H}^+}}^2.
    \end{align}
    This extends to a map
    \begin{align}
        \mathscr{F^+}: \mathcal{E}^T_{\Sigma^*}\longrightarrow \mathcal{E}^T_{\mathscr{H}^+_{\geq0}}\oplus \mathcal{E}^T_{\mathscr{I}^+}.
    \end{align}
Analogously, forward evolution from smooth compactly supported data on $\Sigma$ or $\overline{\Sigma}$ extends to the maps,
\begin{align}
    \mathscr{F}^+:\mathcal{E}^T_\Sigma \longrightarrow \mathcal{E}^T_{\mathscr{{H}^+}} \oplus \mathcal{E}^T_{\mathscr{I}^+},\\
    \mathscr{F}^+:\mathcal{E}^T_{\overline{\Sigma}} \longrightarrow \mathcal{E}^T_{\overline{\mathscr{{H}^+}}} \oplus \mathcal{E}^T_{\mathscr{I}^+}.
\end{align}
\end{thm}
\begin{thm}\label{backwardRW}
Let $\bm{\uppsi}_{\mathscr{I}^+}\in \Gamma_c (\mathscr{I}^+), \bm{\uppsi}_{\mathscr{H}^+} \in \Gamma_c (\mathscr{H}^+_{\geq0})$. Then there exists a unique solution $\Psi$ to \cref{RW} in $J^+(\Sigma^*)$ which is smooth, such that
\begin{align}
    \lim_{v\longrightarrow\infty} \Psi(u,v,\theta^A)=\bm{\uppsi}_{\mathscr{I}^+},\qquad\qquad \Psi\big|_{\mathscr{H}^+_{\geq0}}=\bm{\uppsi}_{\mathscr{H}^+}.
\end{align}
with $\left|\left|(\Psi|_{\Sigma^*},\slashednabla_{n_{\Sigma^*}}\Psi|_{\Sigma^*})\right|\right|_{\mathcal{E}^T_{\Sigma^*}}^2=\left|\left|\bm{\uppsi}_{\mathscr{I}^+}\right|\right|_{\mathcal{E}^T_{\mathscr{I}^+}}^2+\left|\left|\bm{\uppsi}_{\mathscr{H}^+}\right|\right|_{\mathcal{E}^T_{\mathscr{H}^+}}^2$. 
This extends to a map 
\begin{align}
    \mathscr{B}^-: \mathcal{E}^T_{\mathscr{H}^+_{\geq0}}\oplus \mathcal{E}^T_{\mathscr{I}^+}\longrightarrow \mathcal{E}^T_{{\Sigma^*}} ,
\end{align}
which inverts the map $\mathscr{F}^+$ of \Cref{forwardRW}. Thus $\mathscr{F}^+, \mathscr{B}^+$ are unitary Hilbert space isomorphisms and
\begin{align}
\mathscr{B}^-\circ\mathscr{F}^+=\mathscr{F}^+\circ\mathscr{B}^+=Id.
\end{align}
Similar statements apply to produce maps
\begin{align}
    \mathscr{B}^-:  \mathcal{E}^T_{\mathscr{{H}^+}} \oplus \mathcal{E}^T_{\mathscr{I}^+}\longrightarrow \mathcal{E}^T_\Sigma,\\
    \mathscr{B}^-: \mathcal{E}^T_{\overline{\mathscr{{H}^+}}} \oplus \mathcal{E}^T_{\mathscr{I}^+} \longrightarrow \mathcal{E}^T_{\overline{\Sigma}}.
\end{align}
\end{thm}
\begin{thm}\label{RW isomorphisms}
    Analogously to \cref{forwardRW,backwardRW}, there exist bounded maps 
    \begin{align}
        \mathscr{F}^-:\mathcal{E}^T_\Sigma\longrightarrow \mathcal{E}^T_{\mathscr{H}^-}\oplus \mathcal{E}^T_{\mathscr{I}^-},\qquad\qquad\qquad \mathscr{B}^+:\mathcal{E}^T_{\mathscr{H}^-}\oplus \mathcal{E}^T_{\mathscr{I}^-}\longrightarrow \mathcal{E}^T_\Sigma,
    \end{align}
     \begin{align}
        \mathscr{F}^-:\mathcal{E}^T_{\overline{\Sigma}}\longrightarrow \mathcal{E}^T_{\overline{\mathscr{H}^-}}\oplus \mathcal{E}^T_{\mathscr{I}^-},\qquad\qquad\qquad \mathscr{B}^+:\mathcal{E}^T_{\overline{\mathscr{H}^-}}\oplus \mathcal{E}^T_{\mathscr{I}^-}\longrightarrow \mathcal{E}^T_{\overline\Sigma},
    \end{align}
such that $\mathscr{F}^-\circ\mathscr{B}^+=\mathscr{B}^+\circ\mathscr{F}^-=Id$ on the respective domains. The maps
    \begin{align}
        \mathscr{S}=\mathscr{F}^+\circ\mathscr{B}^+:\mathcal{E}^T_{\mathscr{H}^-}\oplus \mathcal{E}^T_{\mathscr{I}^-}\longrightarrow \mathcal{E}^T_{\mathscr{H}^+}\oplus\mathcal{E}^T_{\mathscr{I}^+},\\
        \mathscr{S}=\mathscr{F}^+\circ\mathscr{B}^+:\mathcal{E}^T_{\overline{\mathscr{H}^-}}\oplus \mathcal{E}^T_{\mathscr{I}^-}\longrightarrow \mathcal{E}^T_{\overline{\mathscr{H}^+}}\oplus\mathcal{E}^T_{\mathscr{I}^+}
    \end{align}
constitute unitary Hilbert space isomorphism with inverses
\begin{align}
        \mathscr{S}=\mathscr{F}^-\circ\mathscr{B}^-:\mathcal{E}^T_{\mathscr{H}^+}\oplus \mathcal{E}^T_{\mathscr{I}^+}\longrightarrow \mathcal{E}^T_{\mathscr{H}^-}\oplus\mathcal{E}^T_{\mathscr{I}^-},\\
        \mathscr{S}=\mathscr{F}^-\circ\mathscr{B}^-:\mathcal{E}^T_{\overline{\mathscr{H}^+}}\oplus \mathcal{E}^T_{\mathscr{I}^+}\longrightarrow \mathcal{E}^T_{\overline{\mathscr{H}^-}}\oplus\mathcal{E}^T_{\mathscr{I}^-}
    \end{align}
%$\mathscr{F}^-\circ\mathscr{B}^-$ 
on the respective domains.
\end{thm}
\begin{remark}\label{RW distinct spaces}
We emphasise that the spaces $\mathcal{E}^T_{\Sigma}$ and $\mathcal{E}^T_{\overline{\Sigma}}$ are different and $\mathcal{E}^T_{\Sigma}\subsetneq\mathcal{E}^T_{\overline{\Sigma}}$. Similarly, $\mathcal{E}^T_{\mathscr{H}^+}\subsetneq\mathcal{E}^T_{\overline{\mathscr{H}^+}}$. Our prescription in distinguishing between these spaces is consistent in the sense that elements of $\mathcal{E}^T_{\Sigma}$ are mapped into $\mathcal{E}^T_{\mathscr{H}^+}$ and vice versa. Our point of view is that the spaces $\mathcal{E}^T_{\overline{\Sigma}}, \mathcal{E}^T_{\overline{\mathscr{H}}^\pm}$ are the natural spaces to consider, since in these spaces scattering data are not restricted to vanish at the bifurcation sphere $\mathcal{B}$. It is however useful to have the statements involving $\mathcal{E}^T_{{\Sigma}}, \mathcal{E}^T_{{\mathscr{H}}^\pm}$. In particular, solutions arising from past scattering data identically vanishing on $\mathscr{H}^-$ will lie in these spaces.
\end{remark}
\subsection{Theorem 2: Scattering for the Teukolsky equations of spins $\pm2$}\label{subsection 4.2 scattering for the teukolsky equations of spins +,-2}
\subsubsection{Scattering for the +2 Teukolsky equation}\label{subsubsection 4.2.1 scattering for the +2 equation}
\begin{defin}\label{+2 norm on Sigma}
    Let $(\upalpha,\upalpha')\in\Gamma_c(\Sigma^*)\oplus\Gamma_c(\Sigma^*)$ be Cauchy data for $\bref{T+2}$ on $\Sigma^*$ giving rise to a solution $\alpha$.
    Define the space $\mathcal{E}^{T,+2}_{\Sigma^*}$ to be the completion of $\Gamma_c(\Sigma^*)\oplus\Gamma_c(\Sigma^*)$ under the norm 
    \begin{align}
    ||(\upalpha,\upalpha')||_{\mathcal{E}^{T,+2}_{\Sigma^*}}^2=||(\Psi,\slashednabla_{n_{\Sigma^*}}\Psi)||_{\mathcal{E}^{T}_{\Sigma^*}}^2,
    \end{align}
    where $\Psi$ is the weighted second derivative $\Psi=\left(\frac{r^2}{\Omega^2}\nablau\right)^2r\Omega^2\alpha$ of $\alpha$. The spaces $\mathcal{E}^{T,+2}_{\Sigma}$, $\mathcal{E}^{T,+2}_{\overline{\Sigma}}$ are similarly defined.
\end{defin}
We immediately note the following:
\begin{proposition}\label{+2 norm on Sigma is coercive}
$\|\;\|_{\mathcal{E}^{T,+2}_{\Sigma}}$ indeed defines a norm on $\Gamma_c(\Sigma)\times\Gamma_c(\Sigma)$. Similar statements hold for $\|\;\|_{\mathcal{E}^{T,+2}_{\Sigma^*}}, \|\;\|_{\mathcal{E}^{T,+2}_{\overline{\Sigma}}}$.
\end{proposition}
\begin{proof}
It suffices to check that $\|(\upalpha,\upalpha')\|_{\mathcal{E}^{T,+2}_{\Sigma}}=0$ for a  smooth, compactly supported pair $(\upalpha,\upalpha')$ implies that $(\upalpha,\upalpha')=(0,0)$. Let $\alpha$, $\Psi$ be as in \Cref{+2 norm on Sigma}. It is clear that $\Psi=0$, and (\ref{eq:d4d4Psi}) implies:
\begin{align}\label{424242}
    \slashednabla_T\alpha=\frac{1}{12M}\mathcal{A}_2(\mathcal{A}_2-2)\alpha.
\end{align}
\Cref{eq:d4Psi} implies that on $\Sigma$
\begin{align}
    \left(\mathcal{A}_2-2+\frac{6M}{r}\right)\left(\frac{1}{12M}\mathcal{A}_2(\mathcal{A}_2-2)-\slashednabla_{R^*}\right)r\Omega^2\alpha-6M\frac{\Omega^2}{r^2}r\Omega^2\alpha=0.
\end{align}
Take $F=\left(\mathcal{A}_2-2+\frac{6M}{r}\right)r\Omega^2\alpha$, then the above says $\slashednabla_{R^*}F=\frac{1}{12M} \mathcal{A}_2\left(\mathcal{A}_2-2\right)F-12M\frac{\Omega^2}{r^2}r\Omega^2\alpha$. We integrate over the region $R_0<r<R$ on $\Sigma$:
\begin{align}
\begin{split}
    \|F\|^2_{S^2,r=R}=\|F\|^2_{S^2,r=R_0}+\int_{\Sigma\cap\{R_0<r<R\}} \frac{1}{6M}&\left\{|\mathcal{A}_2F|^2+2|\mathring{\slashednabla}F|^2+4|F|^2\right\}\\&+24M\frac{\Omega^2}{r^2}\left\{|\mathring{\slashednabla}r\Omega^2\alpha|^2+\left(4-\frac{6M}{r}\right)|r\Omega^2\alpha|^2\right\}.
\end{split}
\end{align}
This implies $\|F\|^2_{S^2,r=R}\geq \|F\|^2_{S^2,r=R_0}$ (notice that the integral on the right hand side remains positive by Poincar\'e's inequality). If the data are compactly supported then $F$ must vanish everywhere on $\Sigma$, and the vanishing of $F$ implies the vanishing of $\Omega^2\alpha$ for smooth $\alpha$ since the operator $\mathcal{A}_2-2+\frac{6M}{r}$ is uniformly elliptic on the set of symmetric, traceless 2-tensor field on $S^2$. This in turn implies the vanishing of $\slashednabla_T\Omega^2\alpha$ by \bref{424242}. We can repeat this argument for data on $\Sigma^*, \overline{\Sigma}$.
\end{proof}
\begin{defin}
 Define the space of future scattering states $\mathcal{E}^{T,+2}_{\mathscr{H}^+_{\geq 0}}$ on $\mathscr{H}^+$ to be the completion of $\Gamma_c (\mathscr{H}^+_{\geq 0})$ under the norm
\begin{align}\label{+2 scattering norm on H+}   
\begin{split}
    &\|A\|_{\mathcal{E}^{T,+2}_{\mathscr{H}^+_{\geq0}}}=\left\|\mathcal{A}_2(\mathcal{A}_2-2)\left(\int_v^\infty d\bar{v}\; e^{\frac{1}{2M}(v-\bar{v})}A\right)\right\|^2_{L^2(\mathscr{H}^+_{\geq0})}+\left\|6M\partial_v \left(\int_v^\infty d\bar{v}\; e^{\frac{1}{2M}(v-\bar{v})}A\right)\right\|^2_{L^2(\mathscr{H}^+_{\geq0})}\\&+\int_{S^2}\sin\theta d\theta d\phi \left(\left|\mathring{\slashed{\Delta}}\int_{\bar{v}=0}^\infty d\bar{v}\; e^{\frac{1}{2M}(v-\bar{v})}A\right|^2+6\left|\mathring{\slashednabla}\int_{\bar{v}=0}^\infty d\bar{v}\; e^{\frac{1}{2M}(v-\bar{v})}A\right|^2+8\Big|\int_{\bar{v}=0}^\infty d\bar{v}\; e^{\frac{1}{2M}(v-\bar{v})}A\Big|^2\right).
\end{split}
\end{align}
Define the space $\mathcal{E}^{T,+2}_{\mathscr{H}^+}$ to be the completion of $\Gamma_c(\mathscr{H}^+)$ under the norm
\begin{align}
    \|A\|_{\mathcal{E}^{T,+2}_{\mathscr{H}^+}}=\left\|\mathcal{A}_2(\mathcal{A}_2-2)\left(\int_v^\infty d\bar{v}\; e^{\frac{1}{2M}(v-\bar{v})}A\right)\right\|^2_{L^2(\mathscr{H}^+)}+\left\|6M\partial_v\left(\int_v^\infty d\bar{v}\; e^{\frac{1}{2M}(v-\bar{v})}A\right)\right\|^2_{L^2(\mathscr{H}^+)}.
\end{align}
Define the space $\mathcal{E}^{T,+2}_{\overline{\mathscr{H}^+}}$ to be the completion of the space consisting of symmetric traceless $S^2_{\infty,v}$ 2-tensor fields $A$ on $\overline{\mathscr{H}^+}$ such that $V^{-2}A\in \Gamma_c \left(\overline{\mathscr{H}^+}\right)$, under the same norm above evaluated over $\overline{\mathscr{H}^+}$.
\end{defin}
\begin{remark}\label{+2 norm is norm on H+}
Let $A\in\Gamma_c(\mathscr{H}_{\geq0}^+)$. If $\|A\|_{\mathcal{E}^{T,+2}_{\mathscr{H}^+_{\geq0}}}=0$ then $\int_v^\infty d\bar{v}\; e^{\frac{1}{2M}(v-\bar{v})}A=0$ for all $v$, which  implies that $A$ must vanish if it is smooth. Thus $\|\;\|_{\mathcal{E}^{T,+2}_{\mathscr{H}^+_{\geq0}}}$ defines a norm on $\Gamma_c(\mathscr{H}^+_{\geq0})$, which then extends to the Hilbert space $\mathcal{E}^{T,+2}_{\mathscr{H}^+_{\geq0}}$. The same applies to $\mathcal{E}^{T,+2}_{\mathscr{H}^+_{\geq0}}$,  $\mathcal{E}^{T,+2}_{\overline{\mathscr{H}^+}}$.
\end{remark}
\begin{defin}
    Define the space of future scattering states $\mathcal{E}^{T,+2}_{\mathscr{I}^+}$ on $\mathscr{I}^+$ to be the completion of $\Gamma_c (\mathscr{I}^+)$ under the norm
    \begin{align}
       \|A\|_{\mathcal{E}_{\mathscr{I}^+}^{T,+2}}=\left|\left|\partial_u^3A\right|\right|_{L^2(\mathscr{I}^+)}.
    \end{align}
\end{defin}
\begin{remark}\label{+2 norm is norm on scri+}
    The energy $\|\;\|_{\mathcal{E}_{\mathscr{I}^+}^{T,+2}}$ indeed defines a norm on $\Gamma_c(\mathscr{I}^+)$, which thus extends to a Hilbert space $\mathcal{E}^{T,+2}_{\mathscr{I}^+}$ when completed under $\|\;\|_{\mathcal{E}_{\mathscr{I}^+}^{T,+2}}$.
    We can identify $\mathcal{E}^{T,+2}_{\mathscr{I}^+}$ as the subset $A\in L^2_{loc}(\mathscr{I}^+)$ whose elements satisfy
    \begin{itemize}
        \item $\partial_u^3 A\in L^2(\mathscr{I}^+)$,
        \item $\lim_{u\longrightarrow\infty} \|A\|_{L^2(S^2)}=0$.
    \end{itemize}
    Hardy's inequality holds and we have on this subset
    \begin{align}
        \int_{\mathscr{H}^+_{\geq0}} dv \sin\theta d\theta d\phi \frac{|A|^2}{v^6+1}\lesssim\|\xi\|^2_{\mathcal{E}^{T,+2}_{\mathscr{H}^+_{\geq0}}} <\infty.
    \end{align}
\end{remark}
\begin{defin}\label{+2 backwards scattering H}
 Define the space of past scattering states $\mathcal{E}^{T,+2}_{\mathscr{H}^-}$ on $\mathscr{H}^-$ to be the completion of $\Gamma_c (\mathscr{H}^-)$ under the norm
 \begin{align}\label{this 2424}
\|A\|_{\mathcal{E}^{T,+2}_{\mathscr{H}^-}}=\left|\left|2(2M\partial_u)A-3(2M\partial_u)^2A+(2M\partial_u)^3A\right|\right|_{L^2(\mathscr{H}^-)}.
\end{align}
Define the space $\mathcal{E}^{T,+2}_{\overline{\mathscr{H}^-}}$ to be the closure of the space consisting of symmetric traceless $S^2_{u,-\infty}$ 2-tensor fields $A$ on $\overline{\mathscr{H}^-}$ such that $U^2A\in \Gamma_c \left(\overline{\mathscr{H}^-}\right)$, under the same norm above evaluated over $\overline{\mathscr{H}^-}$.
\end{defin}
\begin{remark}\label{+2 norm is norm on H-}
As mentioned in \Cref{introduction regular frame norm} of Section 1.3.2 of the introduction, the energy defined in \bref{this 2424} can be written using the Kruskal frame as
\begin{align}
    \|A\|_{\mathcal{E}^{T,+2}_{\mathscr{H}^-}}=\|U^{1/2}\partial_U^3U^2A\|_{L^2_UL^2(S^2)}.
\end{align}
This defines a norm on $\Gamma_c(\mathscr{H}^-)$, which then extends to the Hilbert space $\mathcal{E}^{T,+2}_{\mathscr{H}^-}$. It is possible to represent the elements of $\mathcal{E}^{T,+2}_{\mathscr{H}^-}$ as the subset $A\in L^2_{loc}(\mathscr{H}^-)$ whose elements satisfy
\begin{itemize}
    \item $\partial_uA$, $\partial_u^2 A$, $\partial_u^3 A \in L^2(\mathscr{H}^-)$,
    \item $\lim_{u\longrightarrow -\infty} \|A\|_{L^2(S^2)}=0$
\end{itemize}
Hardy's inequality holds on this space we have
\begin{align}
       \int_{\mathscr{H}^-} du \sin\theta d\theta d\phi \frac{|A|^2}{u^2+1}\lesssim\|\xi\|^2_{\mathcal{E}^{T,+2}_{\mathscr{H}^-}}<\infty.
   \end{align}
\end{remark} 
\begin{defin}\label{+2 backwards scattering scri}
    Define the space of past scattering states $\mathcal{E}^{T,+2}_{\mathscr{I}^-}$ on $\mathscr{I}^-$ to be the completion of the space
\begin{align}
    A\in\Gamma(\mathscr{I}^-): \int_{-\infty}^\infty dv\;A=0
\end{align}    
under the norm
\begin{align}
     \|A\|^2_{\mathcal{E}^{T,+2}_{\mathscr{I}^-}}=&\int_{\mathscr{I}^-} d\bar{v}\sin\theta d\theta d\phi\left[ 6M|A|^2+\left|\mathcal{A}_2(\mathcal{A}_2-2)\int_{\bar{v}}^\infty A\right| ^2\right].
\end{align}
\end{defin}
\begin{remark}\label{+2 norm is norm on scri-}
Let $A\in\Gamma_c(\mathscr{I}^-)$. If  $\|A\|^2_{\mathcal{E}^{T,+2}_{\mathscr{I}^-}}=0$ then $A=0$. Thus $\|\;\|^2_{\mathcal{E}^{T,+2}_{\mathscr{I}^-}}$ defines a norm on $\Gamma_c(\mathscr{I}^-)$ which then extends to the Hilbert space $\mathcal{E}^{T,+2}_{\mathscr{I}^-}$.
\end{remark}
\begin{thm}\label{+2 future forward scattering}
Forward evolution under the $+2$ Teukolsky equation \bref{T+2} from smooth, compactly supported data $(\upalpha,\upalpha')$ on $\Sigma^*$ gives rise to smooth radiation fields $(\upalpha_{\mathscr{H}^+},\upalpha_{\mathscr{I}^+})\in \mathcal{E}^{T,+2}_{\mathscr{H}^+_{\geq0}}\oplus\mathcal{E}^{T,+2}_{\mathscr{I}^+}$ where
\begin{enumerate}
\item $\upalpha_{\mathscr{H}^+}=2M\Omega^{2}{\alpha}\big|_{\mathscr{H}^+} \in \Gamma(\mathscr{H}^+)$,
\item $\upalpha_{\mathscr{I}^+}=\lim_{v\longrightarrow \infty} r^5\alpha(v,u,\theta^A)$, with $\upalpha_{\mathscr{I}^+}\in \Gamma(\mathscr{I}^+)$,
\end{enumerate}
with $\upalpha_{\mathscr{I}^+}, \alpha_{\mathscr{H}^+_{\geq0}}$ satisfying
\begin{align}
    \left|\left|(\upalpha,\upalpha')\right|\right|_{\mathcal{E}^{T,+2}_{\Sigma^*}}^2=\left|\left|\upalpha_{\mathscr{I}^+}\right|\right|_{\mathcal{E}^{T,+2}_{\mathscr{I}^+}}^2+\left|\left|\upalpha_{\mathscr{H}^+}\right|\right|_{\mathcal{E}^{T,+2}_{\mathscr{H}^+_{\geq0}}}^2.
\end{align}
This extends to a unitary map
\begin{align}
    {}^{(+2)}\mathscr{F^+}: \mathcal{E}^{T,+2}_{\Sigma^*}\longrightarrow \mathcal{E}^{T,+2}_{\mathscr{H}^+_{\geq0}}\oplus \mathcal{E}^{T,+2}_{\mathscr{I}^+}.
\end{align}
The same conclusions apply when replacing $\Sigma^*$ with $\Sigma$ and $\mathscr{H}^+_{\geq0}$ with $\mathscr{H}^+$, or when replacing with $\overline\Sigma$ and $\overline{\mathscr{H}^+}$. In the latter case, $(\upalpha,\upalpha')$ must be consistent with the well-posedness statement $\Cref{WP+2Sigma*}$ and consequently we obtain that $V^{-2}\upalpha_{{\mathscr{H}^+}}\in \Gamma(\overline{\mathscr{H}^+})$.
\end{thm}
\begin{thm}\label{+2 future backward scattering}
Let $\upalpha_{\mathscr{I}^+}\in \Gamma_c (\mathscr{I}^+), \upalpha_{\mathscr{H}^+} \in \Gamma_c (\mathscr{H}^+_{\geq0})$. Then there exists a unique solution $\alpha$ to \cref{T+2} in $J^+(\Sigma^*)$ which is smooth, such that
\begin{align}
    \lim_{v\longrightarrow\infty} r^5\alpha(u,v,\theta^A)=\upalpha_{\mathscr{I}^+},\qquad\qquad \Omega^{2}\alpha\big|_{\mathscr{H}^+_{\geq0}}=\upalpha_{\mathscr{H}^+},
\end{align}
with $(\Omega^2\alpha|_{\Sigma^*},\slashednabla_{n_{\Sigma^*}}\Omega^2\alpha|_{\Sigma^*})\in \mathcal{E}^{T,+2}_{\Sigma^*} $ and $
        \left\|(\Omega^2\alpha|_{\Sigma^*},\slashednabla_{n_{\Sigma^*}}\Omega^2\alpha|_{\Sigma^*})\right\|_{\mathcal{E}^{T,+2}_{\Sigma^*}}^2=\left\|\upalpha_{\mathscr{I}^+}\right\|_{\mathcal{E}^{T,+2}_{\mathscr{I}^+}}^2+\left|\left|\upalpha_{\mathscr{H}^+}\right|\right|_{\mathcal{E}^{T,+2}_{\mathscr{H}^+}}^2$. 
This extends to a unitary map 
\begin{align}
    {}^{(+2)}\mathscr{B}^-: \mathcal{E}^{T,+2}_{\mathscr{H}^+_{\geq0}}\oplus \mathcal{E}^{T,+2}_{\mathscr{I}^+}\longrightarrow \mathcal{E}^{T,+2}_{{\Sigma^*}},
\end{align}
which inverts the map ${}^{(+2)}\mathscr{F}^+$ of \Cref{+2 future forward scattering}
\begin{align}
{}^{(+2)}\mathscr{B}^-\circ{}^{(+2)}\mathscr{F}^+={}^{(+2)}\mathscr{F}^+\circ{}^{(+2)}\mathscr{B}^-=Id.
\end{align}
The same conclusions apply when replacing $\Sigma^*$ with $\Sigma$ and $\mathscr{H}^+_{\geq0}$ with $\mathscr{H}^+$, or when replacing with $\overline\Sigma$ and $\overline{\mathscr{H}^+}$. In the latter case, we require that $V^{-2}\alpha_{{\mathscr{H}^+}}\in \Gamma(\overline{\mathscr{H}^+})$ and with that $(\alpha|_{\Sigma^*},\slashednabla_{n_{\Sigma^*}}\alpha|_{\Sigma^*})$ is consistent with \Cref{WP+2Sigma*}.
\end{thm}
\begin{thm}\label{+2 past forward scattering}
Evolution from $(\upalpha,\upalpha')\in\Gamma_c(\Sigma)\times\Gamma_c(\Sigma)$ to $J^-(\Sigma)$ gives rise to radiation fields on $\mathscr{H}^-,\mathscr{I}^-$ analogously to \Cref{+2 future forward scattering}, where the radiation fields are defined by
\begin{align}
    \lim_{v\longrightarrow\infty} r\alpha(u,v,\theta^A)=\upalpha_{\mathscr{I}^-},\qquad\qquad 2M\Omega^{-2}\alpha\big|_{\mathscr{H}^-}=\upalpha_{\mathscr{H}^-}.
\end{align}
This extends to a unitary map
\begin{align}
    {}^{(+2)}\mathscr{F^-}: \mathcal{E}^{T,+2}_{\Sigma}\longrightarrow \mathcal{E}^{T,+2}_{\mathscr{H}^-}\oplus \mathcal{E}^{T,+2}_{\mathscr{I}^-},
\end{align}
with inverse ${}^{(+2)}\mathscr{B}^+:\mathcal{E}^{T,+2}_{\mathscr{H}^-}\oplus \mathcal{E}^{T,+2}_{\mathscr{I}^-}\longrightarrow \mathcal{E}^{T,+2}_{\Sigma}$. The same conclusions apply when replacing $\Sigma$ with $\overline{\Sigma}$ and $\mathscr{H}^-$ with $\overline{\mathscr{H}^-}$. In this case, we require that  $(U^{2}\Omega^{-2}\upalpha,U^{2}\Omega^{-2}\upalpha')$ are smooth up to and including $\mathcal{B}$, and consequently we obtain that $U^{2}\upalpha_{{\mathscr{H}^-}}\in \Gamma(\overline{\mathscr{H}^-})$.
\end{thm}
\begin{thm}\label{scatteringthm+2}
    The maps
    \begin{align}
        {}^{(+2)}\mathscr{S}&={}^{(+2)}\mathscr{F}^+\circ{}^{(+2)}\mathscr{B}^+:\mathcal{E}^{T,+2}_{\mathscr{H}^-}\oplus \mathcal{E}^{T,+2}_{\mathscr{I}^-}\longrightarrow \mathcal{E}^{T,+2}_{\mathscr{H}^+}\oplus\mathcal{E}^{T,+2}_{\mathscr{I}^+},\\
            {}^{(+2)}\mathscr{S}&={}^{(+2)}\mathscr{F}^+\circ{}^{(+2)}\mathscr{B}^+:\mathcal{E}^{T,+2}_{\overline{\mathscr{H}^-}}\oplus \mathcal{E}^{T,+2}_{\mathscr{I}^-}\longrightarrow \mathcal{E}^{T,+2}_{\overline{\mathscr{H}^+}}\oplus\mathcal{E}^{T,+2}_{\mathscr{I}^+}
    \end{align}
constitute unitary Hilbert space isomorphism with inverses %${}^{(+2)}\mathscr{F}^-\circ{}^{(+2)}\mathscr{B}^+$ 
\begin{align}
    {}^{(+2)}\mathscr{S}^-={}^{(+2)}\mathscr{F}^-\circ{}^{(+2)}\mathscr{B}^-:\mathcal{E}^{T,+2}_{\mathscr{H}^+}\oplus \mathcal{E}^{T,+2}_{\mathscr{I}^+}\longrightarrow \mathcal{E}^{T,+2}_{\mathscr{H}^-}\oplus\mathcal{E}^{T,+2}_{\mathscr{I}^-}\\
        {}^{(+2)}\mathscr{S}^-={}^{(+2)}\mathscr{F}^-\circ{}^{(+2)}\mathscr{B}^-:\mathcal{E}^{T,+2}_{\overline{\mathscr{H}^+}}\oplus \mathcal{E}^{T,+2}_{\mathscr{I}^+}\longrightarrow \mathcal{E}^{T,+2}_{\overline{\mathscr{H}^-}}\oplus\mathcal{E}^{T,+2}_{\mathscr{I}^-}
\end{align}
on the respective domains.
\end{thm}
\subsubsection{Scattering for the $-2$ Teukolsky equation}\label{subsubsection 4.2.2 Scattering for the -2 equation}
\begin{defin}\label{-2 norm on Sigma*}
    Let $(\underline\upalpha,\underline\upalpha')\in\Gamma_c(\Sigma^*)\oplus\Gamma_c(\Sigma^*)$ be Cauchy data for $\bref{T+2}$ on $\Sigma^*$ giving rise to a solution $\underline\alpha$. 
    Define the space $\mathcal{E}^{T,-2}_{\Sigma^*}$ to be the completion of $\Gamma_c(\Sigma^*)\oplus\Gamma_c(\Sigma^*)$ under the norm 
    \begin{align}\label{equivnorm-2}
    ||(\underline\upalpha,\underline\upalpha')||_{\mathcal{E}^{T,-2}_{\Sigma^*}}^2=||(\underline\Psi,\slashednabla_{n_{\Sigma^*}}\underline\Psi)||_{\mathcal{E}^{T}_{\Sigma^*}}^2,
    \end{align}
    where $\underline\Psi$ is the weighted second derivative $\underline\Psi=\left(\frac{r^2}{\Omega^2}\nablav\right)^2r\Omega^2\alpha$ of $\underline\alpha$. The spaces $\mathcal{E}^{T,-2}_{\Sigma}$, $\mathcal{E}^{T,-2}_{\overline{\Sigma}}$ are similarly defined.
\end{defin}
\begin{proposition}\label{-2 norm on Sigma is coercive}
$\|\;\|_{\mathcal{E}^{T,-2}_{\Sigma}}$ indeed defines a norm on $\Gamma_c(\Sigma)\times\Gamma_c(\Sigma)$.
\end{proposition}
\begin{proof}
It suffices to check that $\|(\underline\upalpha,\underline\upalpha')\|_{\mathcal{E}^{T,-2}_{\Sigma}}=0$ implies $(\underline\upalpha,\underline\upalpha')=(0,0)$. Let $\underline\alpha$ and $\underline\Psi$ be as in \Cref{-2 norm on Sigma*}. It is clear that $\|(\underline\upalpha,\underline\upalpha')\|_{\mathcal{E}^{T,-2}_{\Sigma}}=0$ implies $\Psi=0$. \Cref{eq:d3d3psibar} implies that
\begin{align}
    \slashednabla_T r\Omega^2\underline\alpha=-\frac{1}{12M}\mathcal{A}_2(\mathcal{A}_2-2)r\Omega^2\underline\alpha.
\end{align}
\Cref{eq:d3psibar} then gives us
\begin{align}\label{this2323}
    \left[\mathcal{A}+2-\frac{6M}{r}\right]\left(\frac{1}{12M}\mathcal{A}_2(\mathcal{A}_2-2)-\slashednabla_{R^*}\right)r\Omega^2\underline\alpha+6M\frac{\Omega^2}{r^2}r\Omega^2\underline\alpha=0.
\end{align}
Let $\underline{F}=\left(\mathring{\slashed{\Delta}}-\frac{6M}{r}\right)r\Omega^2\underline\alpha$, then \bref{this2323} above implies
that $\slashednabla_{R^*}\underline{F}=\frac{1}{12M}\mathcal{A}_2(\mathcal{A}_2-2)\underline{F}$. The result follows similarly to \Cref{+2 norm on Sigma is coercive}.
\end{proof}
\begin{defin}
    Define the space of future scattering states $\mathcal{E}^{T,-2}_{\mathscr{H}^+_{\geq 0}}$ on $\mathscr{H}^+_{\geq0}$ to be the completion of $\Gamma_c(\mathscr{H}^+_{\geq0})$ under the norm
  \begin{align}
\|\underline{A}\|_{\mathcal{E}^{T,-2}_{\mathscr{H}^+_{\geq0}}}=(2M)^2\left|\left|2(2M\partial_v)\underline{A}+3(2M\partial_v)^2\underline{A}+(2M\partial_v)^3\underline{A}\right|\right|_{L^2(\mathscr{H}^+_{\geq0})}.
\end{align}
    The space $\mathcal{E}^{T,-2}_{\mathscr{H}^+}$ is defined by the same norm taken over $\mathscr{H}^+$. Define and $\mathcal{E}^{T,-2}_{\overline{\mathscr{H}^+}}$ to be the closure of the space consisting of symmetric traceless $S^2_{\infty,v}$ 2-tensor fields $\underline{A}$ on $\overline{\mathscr{H}^+}$ such that $V^{2}\underline{A}\in \Gamma_c \left(\overline{\mathscr{H}^+}\right)$, under the same norm above evaluated over $\overline{\mathscr{H}^+}$.
\end{defin}
\begin{remark}\label{-2 norm is norm on H+}
As with \Cref{+2 norm is norm on H-} on $\|\;\|_{\mathcal{E}^{T,+2}_{\mathscr{H}^-}}$, the energy $\|\;\|_{\mathcal{E}^{T,-2}_{\mathscr{H}^+}}$ indeed defines a norm on $\Gamma_c(\overline{\mathscr{H}^+})$, which then extends to the Hilbert space $\mathcal{E}^{T,-2}_{\overline{\mathscr{H}^+}}$. It is possible to represent the elements of $\mathcal{E}^{T,-2}_{\mathscr{H}^+_{\geq0}}$ as the subset $\underline{A}\in L^2_{loc}(\mathscr{H}^+_{\geq0})$ whose elements satisfy
\begin{itemize}
    \item $\partial_v \underline{A}$, $\partial_v^2 \underline{A}$, $\partial_v^3 A \in L^2(\mathscr{H}^+_{\geq0})$,
    \item $\lim_{v\longrightarrow \infty} \|\underline{A}\|_{L^2(S^2)}=0$
\end{itemize}
Hardy's inequality holds on this space we have
\begin{align}
       \int_{\mathscr{H}^+_{\geq0}} dv \sin\theta d\theta d\phi \frac{|\underline{A}|^2}{v^2+1}\lesssim\|\underline{A}\|^2_{\mathcal{E}^{T,-2}_{\mathscr{H}^+_{\geq0}}}<\infty.
   \end{align}
Similar statements apply to $\mathcal{E}^{T,-2}_{{\mathscr{H}^+}}$, $\mathcal{E}^{T,-2}_{\overline{\mathscr{H}^+}}$.
\end{remark}
\begin{defin}
Define the space of future scattering states $\mathcal{E}^{T,-2}_{\mathscr{I}^+}$ on $\mathscr{I}^+$ to be the completion of the space
\begin{align}
        \underline{A}\in \Gamma_c (\mathscr{I}^-): &\int_{-\infty}^{\infty} du\underline{A}=0
\end{align}
 under the norm
 \begin{align}\label{-2 tricky norm at scri}
 \|\underline{A}\|^2_{\mathcal{E}^{T,-2}_{\mathscr{I}^+}}=&\int_{\mathscr{I}^+} d{u}\sin\theta d\theta d\phi\left[ (6M)^2|\underline{A}|^2+\left|\mathcal{A}_2(\mathcal{A}_2-2)\int_{\bar{u}}^\infty d\bar{u}\; \underline{A}\right| ^2\right].
\end{align}
\end{defin}
\begin{remark}\label{-2 norm on scri+}
As with $\|\;\|_{\mathcal{E}^{T,+2}_{\mathscr{I}^-}}$ and \Cref{+2 norm is norm on scri-}, the energy $\|\;\|_{\mathcal{E}^{T,-2}_{\mathscr{I}^+}}$ indeed defines a norm on $\Gamma_c({\mathscr{I}^+})$, which then extends to the Hilbert space $\mathcal{E}^{T,-2}_{{\mathscr{I}^+}}$.
\end{remark}
\begin{defin}
Define the space $\mathcal{E}^{T,-2}_{\mathscr{H}^-}$ to be the completion of $\Gamma_c(\mathscr{H}^-)$ under the norm
\begin{align}
    \|\underline{A}\|_{\mathcal{E}^{T,-2}_{\mathscr{H}^-}}=\left\|\mathcal{A}_2(\mathcal{A}_2-2)\left(\int^u_{-\infty} d\bar{u}\; e^{\frac{1}{2M}(u-\bar{u})}\underline{A}\right)\right\|^2_{L^2(\mathscr{H}^-)}+\left\|6M\partial_u\left( \int^u_{-\infty} d\bar{u}\; e^{\frac{1}{2M}(u-\bar{u})}\underline{A}\right)\right\|^2_{L^2(\mathscr{H}^-)}.
\end{align}
Define the space $\mathcal{E}^{T,-2}_{\overline{\mathscr{H}^-}}$ to be the completion of the space consisting of symmetric traceless $S^2_{u,-\infty}$ 2-tensor fields $\underline{A}$ on $\overline{\mathscr{H}^-}$ such that $U^{-2}A\in \Gamma_c \left(\overline{\mathscr{H}^-}\right)$, under the same norm above evaluated over $\overline{\mathscr{H}^-}$.
\end{defin}
\begin{remark}\label{-2 norm is norm on H-}
As with $\|\;\|_{\mathcal{E}^{T,-2}_{\overline{\mathscr{H}^-}}}$ and \Cref{-2 norm is norm on H-}, the energy $\|\;\|_{\mathcal{E}^{T,-2}_{\overline{\mathscr{H}^-}}}$ indeed defines a norm on $\Gamma_c(\overline{\mathscr{H}^-})$, which then extends to the Hilbert space $\mathcal{E}^{T,-2}_{\overline{\mathscr{H}^-}}$.
\end{remark}
\begin{defin}\label{-2 norm on scri-}
    Define the space of future scattering states $\mathcal{E}^{T,-2}_{\mathscr{I}^-}$ on $\mathscr{I}^-$ to be the completion of $\Gamma_c (\mathscr{I}^-)$ under the norm
    \begin{align}
       \|\underline{A}\|_{\mathcal{E}_{\mathscr{I}^-}^{T,-2}}=\left|\left|\partial_v^3\underline{A}\right|\right|_{L^2(\mathscr{I}^-)}.
    \end{align}
\end{defin}
\begin{remark}\label{-2 norm is norm on scri-}
    The energy $\|\;\|_{\mathcal{E}_{\mathscr{I}^-}^{T,-2}}$ indeed defines a norm on $\Gamma_c(\mathscr{I}^-)$, which thus extends to a Hilbert space $\mathcal{E}^{T,-2}_{\mathscr{I}^-}$ when completed under $\|\;\|_{\mathcal{E}_{\mathscr{I}^-}^{T,-2}}$.
    We can identify $\mathcal{E}^{T,-2}_{\mathscr{I}^-}$ as the subset $A\in L^2_{loc}(\mathscr{I}^-)$ whose elements satisfy
    \begin{itemize}
        \item $\partial_v \underline{A}$, $\partial_v^2 \underline{A}$, $\partial_v^3 \underline{A}\in L^2(\mathscr{I}^-)$,
        \item $\lim_{v\longrightarrow-\infty} \|\underline{A}\|_{L^2(S^2)}=0$.
    \end{itemize}
    Hardy's inequality holds and we have on this subset
    \begin{align}
        \int_{\mathscr{I}^-} dv \sin\theta d\theta d\phi \frac{|\underline{A}|^2}{v^6+1}\lesssim\|\underline{A}\|^2_{\mathcal{E}^{T,-2}_{\mathscr{I}^-}} <\infty.
    \end{align}
\end{remark}
\begin{thm}\label{-2 future forward scattering}
Forward evolution under the $-2$ Teukolsky equation \bref{T-2} from smooth, compactly supported data $(\underline\upalpha,\underline\upalpha')$ on $\Sigma^*$ gives rise to smooth radiation fields $(\underline\upalpha_{\mathscr{H}^+},\underline\upalpha_{\mathscr{I}^+})\in \mathcal{E}^{T,-2}_{\mathscr{H}^+_{\geq0}}\oplus\mathcal{E}^{T,-2}_{\mathscr{I}^+}$ where
\begin{enumerate}
\item $\underline\upalpha_{\mathscr{H}^+}=2M\Omega^{-2}{\underline\alpha}\big|_{\mathscr{H}^+} \in \Gamma(\mathscr{H}^+)$,
\item $\underline\upalpha_{\mathscr{I}^+}=\lim_{v\longrightarrow \infty} r\underline\alpha(v,u,\theta^A)$, with $\underline\upalpha_{\mathscr{I}^+}\in \Gamma(\mathscr{I}^+)$,
\end{enumerate}
with $\underline\upalpha_{\mathscr{I}^+}, \underline\upalpha_{\mathscr{H}^+}$ satisfying
\begin{align}
    \left|\left|(\underline\upalpha,\underline\upalpha')\right|\right|_{\mathcal{E}^{T,-2}_{\Sigma^*}}^2=\left|\left|\upalpha_{\mathscr{I}^+}\right|\right|_{\mathcal{E}^{T,-2}_{\mathscr{I}^+}}^2+\left|\left|\upalpha_{\mathscr{H}^+}\right|\right|_{\mathcal{E}^{T,-2}_{\mathscr{H}^+_{\geq0}}}^2.
\end{align}
This extends to a unitary map
\begin{align}
    {}^{(-2)}\mathscr{F^+}: \mathcal{E}^{T,-2}_{\Sigma^*}\longrightarrow \mathcal{E}^{T,-2}_{\mathscr{H}^+_{\geq0}}\oplus \mathcal{E}^{T,-2}_{\mathscr{I}^+}.
\end{align}
The same conclusions apply when replacing $\Sigma^*$ with $\Sigma$ and $\mathscr{H}^+_{\geq0}$ with $\mathscr{H}^+$, or when replacing with $\overline\Sigma$ and $\overline{\mathscr{H}^+}$. In the latter case, $(\underline\upalpha,\underline\upalpha')$ must be consistent with the well-posedness statement $\Cref{WP-2Sigma*}$ and consequently we obtain that $V^{2}\underline\upalpha_{{\mathscr{H}^+}}\in \Gamma(\overline{\mathscr{H}^+})$.
\end{thm}
\begin{thm}\label{-2 future backward scattering}
Let $\underline\upalpha_{\mathscr{I}^+}\in \Gamma_c (\mathscr{I}^+), \underline\upalpha_{\mathscr{H}^+} \in \Gamma_c (\mathscr{H}^+_{\geq0})$ with $\int_{-\infty}^\infty d\bar{u}\; \underline\upalpha_{\mathscr{I}^+}=0$. Then there exists a unique solution $\underline\alpha$ to \cref{T-2} in $J^+(\Sigma^*)$ which is smooth, such that
\begin{align}
    \lim_{v\longrightarrow\infty} r\underline\alpha(u,v,\theta^A)=\underline\upalpha_{\mathscr{I}^+},\qquad\qquad 2M\Omega^{-2}\underline\alpha\big|_{\mathscr{H}^+_{\geq0}}=\underline\upalpha_{\mathscr{H}^+},
\end{align}
with $(\underline\alpha|_{\Sigma^*},\slashednabla_{n_{\Sigma^*}}\underline\alpha|_{\Sigma^*})\in \mathcal{E}^{T,-2}_{\Sigma^*} $ and $
        \left|\left|(\underline\alpha|_{\Sigma^*},\slashednabla_{n_{\Sigma^*}}\underline\alpha|_{\Sigma^*})\right|\right|_{\mathcal{E}^{T,-2}_{\Sigma^*}}^2=\left|\left|\underline\upalpha_{\mathscr{I}^+}\right|\right|_{\mathcal{E}^{T,-2}_{\mathscr{I}^+}}^2+\left|\left|\underline\upalpha_{\mathscr{H}^+}\right|\right|_{\mathcal{E}^{T,-2}_{\mathscr{H}^+}}^2$. 
This extends to a unitary map 
\begin{align}
    {}^{(-2)}\mathscr{B}^-: \mathcal{E}^{T,-2}_{\mathscr{H}^+_{\geq0}}\oplus \mathcal{E}^{T,-2}_{\mathscr{I}^+}\longrightarrow \mathcal{E}^{T,-2}_{{\Sigma^*}},
\end{align}
which inverts the map ${}^{(-2)}\mathscr{F}^+$ of \Cref{-2 future forward scattering}
\begin{align}
{}^{(-2)}\mathscr{B}^-\circ{}^{(-2)}\mathscr{F}^+={}^{(-2)}\mathscr{F}^+\circ{}^{(-2)}\mathscr{B}^-=Id.
\end{align}
The same conclusions apply when replacing $\Sigma^*$ with $\Sigma$ and $\mathscr{H}^+_{\geq0}$ with $\mathscr{H}^+$, or when replacing with $\overline\Sigma$ and $\overline{\mathscr{H}^+}$. In the latter case, we require that $V^2\underline\upalpha_{{\mathscr{H}^+}}\in \Gamma(\overline{\mathscr{H}^+})$ and with that $(\underline\alpha|_{\Sigma^*},\slashednabla_{n_{\Sigma^*}}\underline\alpha|_{\Sigma^*})$ is consistent with $\Cref{WP-2Sigma*}$
\end{thm}
\begin{thm}\label{-2 past forward scattering}
Evolution from $(\underline\upalpha,\underline\upalpha')\in\Gamma_c(\Sigma)\times\Gamma_c(\Sigma)$ to $J^-(\Sigma)$ gives rise to radiation fields on $\mathscr{H}^-,\mathscr{I}^-$ analogously to \Cref{+2 future forward scattering}, where the radiation fields are defined by
\begin{align}
    \lim_{v\longrightarrow\infty} r^5\underline\alpha(u,v,\theta^A)=\underline\upalpha_{\mathscr{I}^-}\qquad\qquad 2M\Omega^{2}\underline\alpha\big|_{\mathscr{H}^-}=\underline\upalpha_{\mathscr{H}^-}
\end{align}
This extends to a unitary map
\begin{align}
    {}^{(-2)}\mathscr{F^-}: \mathcal{E}^{T,-2}_{\Sigma}\longrightarrow \mathcal{E}^{T,-2}_{\mathscr{H}^-}\oplus \mathcal{E}^{T,-2}_{\mathscr{I}^-}
\end{align}
with inverse ${}^{(-2)}\mathscr{B}^+:\mathcal{E}^{T,-2}_{\mathscr{H}^-}\oplus \mathcal{E}^{T,-2}_{\mathscr{I}^-}\longrightarrow \mathcal{E}^{T,-2}_{\Sigma}$. The same conclusions apply when replacing $\Sigma$ with $\overline{\Sigma}$ and $\mathscr{H}^-$ with $\overline{\mathscr{H}^-}$. In this case, we require that  $(U^{-2}\Omega^2\underline\alpha,U^{-2}\Omega^2\underline\alpha')$ are smooth up to and including $\mathcal{B}$, and consequently we obtain that $U^{-2}\underline\alpha_{{\mathscr{H}^+}}\in \Gamma(\overline{\mathscr{H}^+})$
\end{thm}
\begin{thm}\label{scatteringthm-2}
    The maps
    \begin{align}
        &{}^{(-2)}\mathscr{S}^+={}^{(-2)}\mathscr{F}^+\circ{}^{(-2)}\mathscr{B}^+:\mathcal{E}^{T,-2}_{\mathscr{H}^-}\oplus \mathcal{E}^{T,-2}_{\mathscr{I}^-}\longrightarrow \mathcal{E}^{T,-2}_{\mathscr{H}^+}\oplus\mathcal{E}^{T,-2}_{\mathscr{I}^+},\\
        &{}^{(-2)}\mathscr{S}^+={}^{(-2)}\mathscr{F}^+\circ{}^{(-2)}\mathscr{B}^+:\mathcal{E}^{T,-2}_{\overline{\mathscr{H}^-}}\oplus \mathcal{E}^{T,-2}_{\mathscr{I}^-}\longrightarrow \mathcal{E}^{T,-2}_{\overline{\mathscr{H}^+}}\oplus\mathcal{E}^{T,-2}_{\mathscr{I}^+}
    \end{align}
constitute unitary Hilbert space isomorphism with inverses
\begin{align}
    {}^{(-2)}\mathscr{S}^-={}^{(-2)}\mathscr{F}^-\circ{}^{(-2)}\mathscr{B}^-:\mathcal{E}^{T,-2}_{\mathscr{H}^+}\oplus \mathcal{E}^{T,-2}_{\mathscr{I}^+}\longrightarrow \mathcal{E}^{T,-2}_{\mathscr{H}^-}\oplus\mathcal{E}^{T,-2}_{\mathscr{I}^-}\\
        {}^{(-2)}\mathscr{S}^-={}^{(-2)}\mathscr{F}^-\circ{}^{(-2)}\mathscr{B}^-:\mathcal{E}^{T,-2}_{\overline{\mathscr{H}^+}}\oplus \mathcal{E}^{T,-2}_{\mathscr{I}^+}\longrightarrow \mathcal{E}^{T,-2}_{\overline{\mathscr{H}^-}}\oplus\mathcal{E}^{T,-2}_{\mathscr{I}^-}
\end{align}
on the respective domains.
\end{thm}
\begin{remark}
We emphasise that the spaces $\mathcal{E}^{T,\pm2}_{\Sigma}$ and $\mathcal{E}^{T,\pm2}_{\overline{\Sigma}}$ are different and $\mathcal{E}^{T,\pm2}_{\Sigma}\subsetneq\mathcal{E}^{T,\pm2}_{\overline{\Sigma}}$. Similarly, $\mathcal{E}^{T,\pm2}_{\mathscr{H}^+}\subsetneq\mathcal{E}^{T,\pm2}_{\overline{\mathscr{H}^+}}$. Our prescription in distinguishing between these spaces is consistent in the sense that elements of $\mathcal{E}^{T,\pm2}_{\Sigma}$ are mapped into $\mathcal{E}^{T,\pm2}_{\mathscr{H}^+}$ and vice versa. As mentioned for the Regge--Wheeler equation \bref{RW} in \Cref{RW distinct spaces}, our point of view is that the spaces $\mathcal{E}^{T,\pm2}_{\overline{\Sigma}}, \mathcal{E}^{T,\pm2}_{\overline{\mathscr{H}}^\pm}$ are the more natural spaces to consider, but as we make the distinction between these spaces, we additionally face the issue that the inclusion of the bifurcation sphere $\mathcal{B}$ in the domains of the scattering data requires studying both the equations \bref{T+2}, \bref{T-2} and their unknowns in a different frame near $\mathcal{B}$. 
\end{remark}
\subsection{Theorem 3: The Teukolsky--Starobinsky correspondence}\label{subsection 4.3 the Teukolsky--Starobinsky identities}
\begin{thm}\label{Theorem 3 detailed statement}
Let $\upalpha_{\mathscr{I}^+}\in\Gamma_c(\mathscr{I}^+)$. There exists a unique $\underline\upalpha_{\mathscr{I}^+}\in\Gamma(\mathscr{I}^+)$ such that $\|\upalpha_{\mathscr{I}^+}\|_{\mathcal{E}^{T,+2}_{{\mathscr{I}^+}}}=\|\underline\upalpha_{\mathscr{I}^+}\|_{\mathcal{E}^{T,-2}_{{\mathscr{I}^+}}}$ and
\begin{align}\label{constraint null infinity section 4}
    \partial_u^4\upalpha_{\mathscr{I}^+}=\Big[2\fourthorder+6M\partial_u\Big]\underline{\alpha}_{\mathscr{I}^+}.
\end{align}
An analogous statement applies starting from $\underline\upalpha_{\mathscr{I}^+}\in\Gamma_c(\mathscr{I}^+)$ to obtain $\upalpha_{\mathscr{I}^+}\in\Gamma(\mathscr{I}^+)$ with $\|\underline\upalpha_{\mathscr{I}^+}\|_{\mathcal{E}^{T,-2}_{{\mathscr{I}^+}}}=\|\upalpha_{\mathscr{I}^+}\|_{\mathcal{E}^{T,+2}_{{\mathscr{I}^+}}}$ satisfying \bref{constraint null infinity section 4}.\\ 
\indent Let $\underline\upalpha_{{\mathscr{H}^+}}$ be such that $V^2\underline\upalpha_{{\mathscr{H}^+}}\in\Gamma_c(\overline{\mathscr{H}^+})$. There exists a unique $\upalpha_{\mathscr{H}^+}\in\Gamma(\overline{\mathscr{H}^+})$ such that $\|\upalpha_{\mathscr{H}^+}\|_{\mathcal{E}^{T,+2}_{\overline{\mathscr{H}^+}}}=\|\underline\upalpha_{\mathscr{H}^+}\|_{\mathcal{E}^{T,-2}_{\overline{\mathscr{H}^+}}}$ and
\begin{align}\label{constraint horizon section 4}
     \partial_V^4 V^2\underline\upalpha_{\mathscr{H}^+}=\Big[2\fourthorder-3V\partial_V-6\Big]V^{-2}\upalpha_{\mathscr{H}^+}.
\end{align}
An analogous statement applies starting from $\upalpha_{\mathscr{H}^+}$ such that $V^{-2}\upalpha_{\mathscr{H}^+}\in\Gamma_c(\overline{\mathscr{H}^+})$ to obtain $\underline\upalpha_{\mathscr{H}^+}\in\Gamma(\overline{\mathscr{H}^+})$ with $\|\upalpha_{\mathscr{H}^+}\|_{\mathcal{E}^{T,+2}_{\overline{\mathscr{H}^+}}}=\|\underline\upalpha_{\mathscr{H}^+}\|_{\mathcal{E}^{T,-2}_{\overline{\mathscr{H}^+}}}$ satisfying \bref{constraint horizon section 4}. 
\indent The statements above give rise to unitary Hilbert space isomorphisms
\begin{align}
    \mathcal{TS}_{\mathscr{I}^+}:\mathcal{E}^{T,+2}_{\mathscr{I}^+}\longrightarrow\mathcal{E}^{T,-2}_{\mathscr{I}^+},\qquad\qquad\mathcal{TS}_{\mathscr{H}^+}:\mathcal{E}^{T,+2}_{\overline{\mathscr{H}^+}}\longrightarrow\mathcal{E}^{T,-2}_{\overline{\mathscr{H}^+}}.
\end{align}
\begin{align}
    \mathcal{TS}^+=\mathcal{TS}_{\mathscr{H}^+}\oplus\mathcal{TS}_{\mathscr{I}^+}: \mathcal{E}^{T,+2}_{\overline{\mathscr{H}^+}}\oplus\mathcal{E}^{T,+2}_{\mathscr{I}^+}\longrightarrow\mathcal{E}^{T,-2}_{\overline{\mathscr{H}^+}}\oplus\mathcal{E}^{T,-2}_{\mathscr{I}^+}.
\end{align}
Let $\alpha$ be a solution to the $+2$ Teukolsky equation \bref{T+2} arising from scattering data $\upalpha_{\mathscr{I}^+}\in\Gamma_c(\mathscr{I}^+)$, $\upalpha_{\mathscr{H}^+}$ be such that $V^{-2}\upalpha_{\mathscr{H}^+}\in \Gamma_c(\overline{\mathscr{H}^+})$. Using $\mathcal{TS}^+_{\mathscr{I}^+}, \mathcal{TS}^+_{\mathscr{H}^+}$ we can find a unique set of smooth scattering data $\underline\upalpha_{\mathscr{I}^+}, \underline\upalpha_{\mathscr{H}^+}$ on $\mathscr{I}^+, \mathscr{H}^+$ with $V^2\underline\upalpha_{\mathscr{H}^+}$ regular on $\overline{\mathscr{H}^+}$, giving rise to a solution $\underline\alpha$ to the $-2$ Teukolsky equation \bref{T-2} such that the constraints
\begin{align}
    \frac{\Omega^2}{r^2}\Omega\slashed{\nabla}_3 \left(\frac{r^2}{\Omega^2}\nablau\right)^3\alpha-2r^4\slashed{\mathcal{D}}^*_2\slashed{\mathcal{D}}^*_1\overline{\slashed{\mathcal{D}}}_1\slashed{\mathcal{D}}_2 r\Omega^2{\underline\alpha}-6M\left[\Omega\slashed{\nabla}_4+\Omega\slashed{\nabla}_3\right]r\Omega^2{\underline\alpha}=0,\label{theorem constraint 1}\\
\frac{\Omega^2}{r^2}\Omega\slashed{\nabla}_4 \left(\frac{r^2}{\Omega^2}\Omega\slashed{\nabla}_4\right)^3{\underline\alpha}-2r^4\slashed{\mathcal{D}}^*_2\slashed{\mathcal{D}}^*_1\overline{\slashed{\mathcal{D}}}_1\slashed{\mathcal{D}}_2 r\Omega^2\alpha+6M\left[\Omega\slashed{\nabla}_4+\Omega\slashed{\nabla}_3\right]r\Omega^2\alpha=0.\label{theorem constraint 2}
\end{align}
are satisfied by $\alpha, \underline\alpha$ on $\overline{\mathscr{M}}$. The data satisfy
\begin{align}\label{unitarity}
    \|\upalpha_{\mathscr{I}^+}\|_{\mathcal{E}^{T,+2}_{\mathscr{I}^+}}^2=\|\underline\upalpha_{\mathscr{I}^+}\|_{\mathcal{E}^{T,-2}_{\mathscr{I}^+}}^2,\qquad\qquad \|\upalpha_{\mathscr{H}^+}\|_{\mathcal{E}^{T,+2}_{\mathscr{H}^+}}^2=\|\underline\upalpha_{\mathscr{H}^+}\|_{\mathcal{E}^{T,-2}_{\mathscr{H}^+}}^2.
\end{align}
\indent Analogously, let $\underline\alpha$ be a solution to the $-2$ Teukolsky equation \bref{T-2} arising from scattering data $\underline\upalpha_{\mathscr{I}^+}\in\Gamma_c(\mathscr{I}^+)$, $\underline\upalpha_{\mathscr{H}^+}$ be such that $V^{2}\underline\upalpha_{\mathscr{H}^+}\in \Gamma_c(\overline{\mathscr{H}^+})$. Then there exist unique smooth scattering data $\upalpha_{\mathscr{I}^+}, \upalpha_{\mathscr{H}^+}$ on $\mathscr{I}^+, \mathscr{H}^+$ with $V^{-2}\upalpha_{\mathscr{H}^+}$ regular on $\overline{\mathscr{H}^+}$, giving rise to a solution $\alpha$ to the +2 Teukolsky equation \bref{T+2} such that $\alpha, \underline\alpha$ satisfy the constraints \bref{theorem constraint 1}, \bref{theorem constraint 2}.\\
\indent An analogous statement applies to scattering from $\mathscr{I}^-, \mathscr{H}^-$ and we have the isomorphism
\begin{align}
    \mathcal{TS}^-=\mathcal{TS}_{\mathscr{H}^-}\oplus\mathcal{TS}_{\mathscr{I}^-}: \mathcal{E}^{T,+2}_{\overline{\mathscr{H}^-}}\oplus\mathcal{E}^{T,+2}_{\mathscr{I}^-}\longrightarrow\mathcal{E}^{T,-2}_{\overline{\mathscr{H}^-}}\oplus\mathcal{E}^{T,-2}_{\mathscr{I}^-}.
\end{align}
\end{thm}
\subsection{Corollary 1: A mixed scattering statement for combined ($\alpha,\underline\alpha$)}\label{subsection 4.4 Corollary 1: mixed scattering}
Importantly, we have the following corollary:
\begin{corollary}\label{corollary to be proven}
Let $\upalpha_{\mathscr{I}^-}\in\Gamma_c(\mathscr{I}^-)$, $\underline\upalpha_{\mathscr{H}^+}$ be such that $V^{2}\underline\upalpha_{\mathscr{H}^-}\in\Gamma_c(\overline{\mathscr{H}^-})$. Then there exists a unique smooth pair $(\alpha, \underline\alpha)$ on $\mathscr{M}$, such that $\alpha$ solves \bref{T+2}, $\underline\alpha$ solves \bref{T-2}, $\alpha, \underline\alpha$ satisfy \bref{theorem constraint 1}, \bref{theorem constraint 2} and $\underline\alpha$ realises $\underline\upalpha_{\mathscr{H}^-}$ as its radiation field on $\overline{\mathscr{H}^-}$, $\alpha$ realises $\upalpha_{\mathscr{I}^-}$ as its radiation field on $\mathscr{I}^-$. Moreover, the radiation fields of $\alpha$ and $\underline\alpha$ on $\overline{\mathscr{H}^+}, \mathscr{I}^+$ are such that 
\begin{align}
    \|\upalpha_{\mathscr{H}^+}\|_{\mathcal{E}^{T,+2}_{\overline{\mathscr{H}^+}}}^2\;+\;\|\underline\upalpha_{\mathscr{I}^+}\|^2_{\mathcal{E}^{T,-2}_{\mathscr{I}^+}}=\|\upalpha_{\mathscr{I}^-}\|_{\mathcal{E}^{T,+2}_{{\mathscr{I}^-}}}^2\;+\;\|\underline\alpha_{{\mathscr{H}^-}}\|^2_{\mathcal{E}^{T,-2}_{\overline{\mathscr{H}^-}}}.
\end{align}
This extends to a unitary Hilbert-space isomorphism
\begin{align}
    \mathscr{S}^{-2,+2}:\mathcal{E}^{T,-2}_{\overline{\mathscr{H}^-}}\oplus\mathcal{E}^{T,+2}_{\mathscr{I}^-}\longrightarrow \mathcal{E}^{T,+2}_{\overline{\mathscr{H}^+}}\oplus\mathcal{E}^{T,-2}_{\mathscr{I}^+}.
\end{align}
\end{corollary}
\section{Scattering theory of the Regge--Wheeler equation}\label{section 5 scattering theory for RW}
This section is devoted to proving \Cref{Theorem 1} in the introduction, whose detailed statement is contained in \Cref{forwardRW,,backwardRW,,RW isomorphisms}.\\
\indent We will first study in \Cref{subsection 5.2 subsection Radiation fields} the behaviour of future radiation fields belonging to solutions of the Regge--Wheeler equation \bref{RW} that arise from smooth, compactly supported data on $\Sigma^*$ using the estimates gathered in \Cref{subsection 5.1 Basic integrated boundedness and decay estimates}, and this will justify the definitions of radiation fields and spaces of scattering states made in \Cref{subsection 4.1 Theorem 1}. We will first prove \Cref{forwardRW} (in \Cref{subsection 5.3 the forwards scattering map}) and \Cref{backwardRW,,RW isomorphisms} (in \Cref{subsubsection 5.4 the backwards scattering map}) for the case of data on $\Sigma^*$, and most of what follows applies to $\Sigma$ and $\overline{\Sigma}$ unless otherwise stated. \Cref{subsection 5.5 auxiliary results} contain additional results on backwards scattering that will become important later on in the study of the Teukolsky--Starobinsky identities in \Cref{section 9 TS correspondence}.
\subsection{Basic integrated boundedness and decay estimates}\label{subsection 5.1 Basic integrated boundedness and decay estimates}
Here we collect basic boundedness and decay results for (\ref{RW}) proven in \cite{DHR16}. In what follows $(\uppsi,\uppsi')$ is a smooth data set for \cref{RW} as in \Cref{RWwpCauchy}. 

\noindent $\bullet$ \emph{\textbf{Energy boundedness}} Let $X=T:=\Omega e_3+\Omega e_4$, multiply (\ref{RW}) by $\slashed{\nabla}_X$ and integrate by parts over $S^2$ to obtain
\begin{align}\label{T derivative identity}
    \nablau\left[|\nablav\Psi|^2+\Omega^2|\slashed\nabla \Psi|^2+V|\Psi|^2\right]+\nablav\left[|\nablau\Psi|^2+\Omega^2|\slashed\nabla \Psi|^2+V|\Psi|^2\right]\stackrel{S^2}{\equiv}0.
\end{align}
For an outgoing null hypersurface $\mathscr{N}$ define
\begin{align}
F_{\mathscr{N}}^T[\Psi]:=\int_{\mathscr{N}}\sin\theta d\theta d\phi dv\left[|\Omega\slashed\nabla_4\Psi|^2+\Omega^2|\slashed\nabla\Psi|^2+V|\Psi|^2\right].
\end{align}
Similarly for an ingoing null hypersurface $\underline{\mathscr{N}}$ we define
\begin{align}
\underline{F}_{\underline{\mathscr{N}}}^T[\Psi]:=\int_{\underline{\mathscr{N}}}\sin\theta d\theta d\phi du\left[|\Omega\slashed\nabla_3\Psi|^2+\Omega^2|\slashed\nabla\Psi|^2+V|\Psi|^2\right].
\end{align}
Denote $F_{u}^T[\Psi](v_0,v)=F_{\mathscr{C}_{u}\cap\{\bar{v}\in[v_0,v]\}}^T[\Psi]$, $\underline{F}_{v}^T[\Psi](u_0,u)=\underline{F}_{\underline{\mathscr{C}}_{v}\cap\{\bar{u}\in[u_0,u]\}}^T[\Psi]$. Integrating \bref{T derivative identity} over the region $\mathscr{D}^{u,v}_{u_0,v_0}$ yields
\begin{align}
    F^T_u[\Psi](v_0,v)+\underline{F}^T_v[\Psi](u_0,u)= F^T_{u_0}[\Psi](v_0,v)+\underline{F}^T_{v_0}[\Psi](u_0,u).
\end{align}
Similarly, integrating \bref{T derivative identity} over $J^+(\Sigma^*)\cap J^-(\mathscr{C}_u)\cap J^-(\underline{\mathscr{C}}_v)$ yields
\begin{align}
     F^T_u[\Psi](v_0,v)+\underline{F}^T_v[\Psi](u_0,u)=\mathbb{F}_{\Sigma^*\cap J^-(\mathscr{C}_u)\cap J^-(\underline{\mathscr{C}}_v)}[\Psi],
\end{align}
where $\mathbb{F}_{\Sigma^*}[\Psi]$ is given by
\begin{align}
   \mathbb{F}_{\Sigma^*}[\Psi]= \int_{\Sigma^*}dr\sin\theta d\theta d\phi\; (2-\Omega^2)|\slashednabla_{T^*}\Psi|^2+\Omega^2|\slashednabla_R\Psi|^2+|\slashednabla\Psi|^2+(3\Omega^2+1)\frac{|\Psi|^2}{r^2},
\end{align}
and $\mathbb{F}_{\mathcal{U}}$ for a subset $\mathcal{U}\in\Sigma^*$ being defined analogously.\\
\indent Integrating \bref{T derivative identity} over $J^+(\Sigma)\cap J^-(\mathscr{C}_u)\cap J^-(\underline{\mathscr{C}}_v)$ instead yields a similar identity:
\begin{align} 
 F^T_u[\Psi](v_0,v)+\underline{F}^T_v[\Psi](u_0,u)=\mathbb{F}_{\Sigma\cap J^-(\mathscr{C}_u)\cap J^-(\underline{\mathscr{C}}_v)}[\Psi],
\end{align}
with
\begin{align}
    \mathbb{F}^T_{\Sigma}[\Psi]=\int_{\Sigma} \sin\theta dr d\theta d\phi \;\frac{1}{\Omega^2}|\slashednabla_T\Psi|^2+\Omega^2|\slashednabla_R \Psi|^2+|\slashednabla\Psi|^2+(3\Omega^2+1)\frac{|\Psi|^2}{r^2}.
\end{align}
and similarly for $\overline{\Sigma}$.\\
\indent All of the energies defined here so far become degenerate at $\overline{\mathscr{H}^+}$. We can compensate for that for energies defined over hypersurfaces do not intersect the bifurcation sphere $\mathcal{B}$, and we do this by modifying $X$ with a multiple of $\frac{1}{\Omega^2}T$ and repeating the procedure above as in \cite{DR08}, making crucial use of the positivity of the surface gravity of $\mathscr{H}^+$. We then obtain the so called 'redshift' estimates:
\begin{defin}
Define the following nondegenerate energies
    \begin{align}
    F_{\mathscr{N}}[\Psi]=\int_{\mathscr{N}}\sin\theta d\theta d\phi dv \left[|\Omega\slashed\nabla_4\Psi|^2+|\slashed\nabla\Psi|^2+\frac{1}{r^2}|\Psi|^2\right],
\end{align}
\begin{align}
    \underline{F}_{\underline{\mathscr{N}}}[\Psi]=\int_{\underline{\mathscr{N}}}\sin\theta d\theta d\phi du\Omega^2 \left[|\Omega^{-1}\slashed\nabla_3\Psi|^2+|\slashed\nabla\Psi|^2+\frac{1}{r^2}|\Psi|^2\right],
\end{align}
\begin{align}
    \mathbb{F}_{\Sigma^*}[\Psi]=\int_{\Sigma^*} \sin\theta dr d\theta d\phi\left[|\slashednabla_{T^*}\Psi|^2+|\slashednabla_R\Psi|^2+\frac{1}{r^2}|\Psi|^2+|\slashednabla\Psi|^2\right],
\end{align}
and their higher order versions
\begin{align}
    F_{\mathscr{N}}^{n,T,\slashed{\nabla}}[\Psi]=\sum_{i+|\alpha|\leq n}F_{\mathscr{N}}[T^i(r\slashed{\nabla})^\alpha \Psi](v_0,v),
\end{align}
\begin{align}
    \underline{F}_{\underline{\mathscr{N}}}^{n,T,\slashed{\nabla}}[\Psi]=\sum_{i+|\alpha|\leq n} \underline{F}_{\underline{\mathscr{N}}}[T^i(r\slashed{\nabla})^\alpha \Psi](u_0,u),
\end{align}
\begin{align}
    F_{\mathscr{N}}[\Psi]=\sum_{i+j+|\alpha|\leq n} F_{\mathscr{N}}[(\Omega^{-1}\nablagml )^i(r\nablav)^j(r\slashed{\nabla})^\alpha\Psi](v_0,v),
\end{align}
\begin{align}
    \underline{F}_{\underline{\mathscr{N}}}[\Psi]=\sum_{i+j+|\alpha|\leq n} \underline{F}_{\underline{\mathscr{N}}}[(\Omega^{-1}\nablagml )^i(r\nablav)^j(r\slashed{\nabla})^\alpha\Psi](v_0,v),
\end{align}
\begin{align}
    \mathbb{F}^{n,T,\slashednabla}_{\Sigma^*}[\Psi]=\sum_{i+|\alpha|\leq n} \mathbb{F}_{\Sigma^*}[T^i(r\slashed{\nabla})^\alpha \Psi],
\end{align}
\begin{align}
     \mathbb{F}^{n}_{\Sigma^*}[\Psi]=\sum_{i_1+i_2+|\alpha|\leq n}\mathbb{F}_{\Sigma^*}\left[\slashednabla_T^{i_1}\left(\Omega^{-1}\slashednabla_3\right)^{i_2}(r\slashednabla)^\alpha\Psi\right].
\end{align}
\end{defin}
\begin{proposition}\label{RWredshift}
Let $\Psi$ be a solution to (\ref{RW}) arising from data as in \Cref{RWwpCauchy}, then we have
    \begin{align}
    F_u[\Psi](v_0,\infty)+\underline{F}_v[\Psi](u_0,\infty)\lesssim \mathbb{F}_{\Sigma^*}[\Psi].
\end{align}
Similar statements hold for $F^{n,T,\slashed{\nabla}}_u[\Psi](v_0,v), \underline{F}^{n,T,\slashed{\nabla}}_v[\Psi](u_0,u), F^n_u[\Psi](v_0,v)$ and $\underline{F}^n_v[\Psi](u_0,u)$.
\end{proposition}
\noindent $\bullet$ \emph{\textbf{Integrated local energy decay}} We have the following Morawetz-type integrated decay estimate:
\begin{proposition}\label{RWILED}
Let $\Psi$ be a solution to (\ref{RW}) arising from data as in \Cref{RWwpCauchy}, $\mathscr{D}_{\Sigma^*}^{u,v}= J^+(\Sigma^*)\cap J^-(\mathscr{C}_u\cup\underline{\mathscr{C}}_v)$ and define
\begin{align}\label{RWILEDestimate}
\begin{split}
    \mathbb{I}_{deg}^{u,v}[\Psi]= \int_{\mathscr{D}_{\Sigma^*}^{u,v}}d\bar{u}d\bar{v} \sin\theta &d\theta d\phi \Omega^2 \Bigg[\frac{1}{r^2}|\slashed{\nabla}_{R^*}\Psi|^2+\frac{1}{r^3}|\Psi|^2\\
    &+\frac{1}{r}\left(1-\frac{3M}{r}\right)^2\left(|\slashed{\nabla}\Psi|^2+\frac{1}{r^2}|\nablav\Psi|^2+\frac{\Omega^2}{r^2}|\Omega^{-1}\nablagml\Psi|^2\right)\Bigg].
\end{split}
\end{align}
then we have
\begin{align*}
    \begin{split}
        \mathbb{I}_{deg}^{u,v}[\Psi]\lesssim \mathbb{F}_{\Sigma^*}[\Psi].
    \end{split}
\end{align*}
A similar statement holds for
\begin{align}
\mathbb{I}_{deg}^{u,v,n,T,\slashed{\nabla}}[\Psi]=\sum_{i+|\alpha|\leq n} \mathbb{I}_{deg}^{u,v}[T^i(r\slashed{\nabla})^\alpha \Psi]
\end{align}
and 
\begin{align}
\mathbb{I}_{deg}^{u,v,n}[\Psi]=\sum_{i+j+|\alpha|\leq n}\mathbb{I}_{deg}^{u,v,n}[(\Omega^{-1}\nablagml)^i(r\nablav)^j(r\slashed{\nabla})^\alpha\Psi].
\end{align}
\end{proposition}
\noindent $\bullet$ \emph{\textbf{$r^p$-hierarchy of estimates near $\mathscr{I}^+$}} If we multiply (\ref{RW}) by $r^p\Omega^{-2k}\nablav\Psi$ and integrate by parts on $S^2$ we obtain the following 
\begin{align}
    \begin{split}
&\Omega\slashed\nabla_3\left[r^p\Omega^{-2k}|\Omega\slashed\nabla_4\Psi|^2\right]+\Omega\slashed\nabla_4\left[r^p\Omega^{-2k}(\Omega^2|\slashed\nabla\Psi|^2+V|\Psi|^2)\right]\\
&+r^{p-1}\Omega^{-2k}\Bigg\{(p+kx)|\Omega\slashed\nabla_4\Psi|^2-\left[\frac{4\Omega^2}{r^2}+V(p-3+x(k-1))\right]\Omega^2|\Psi|^2
\\&\qquad\qquad\qquad-(p-2+x(k-1))\Omega^4|\slashed\nabla\Psi|^2\Bigg\}\stackrel{S^2}{\equiv}0.
    \end{split}
\end{align}
We can ensure that the bulk term is non-negative by taking  $p=0,k=0$ or $p=2, 1\leq k\leq2$ or $p\in(0,2)$ and restricting to large enough $r$. Integrating in a region $\mathscr{D}^{u,v}_{\Sigma^*}\cap \{r>R\}$ yields (after averaging in $R$ and using \Cref{RWILED} to deal with the timelike boundary term)
\begin{proposition}\label{RWrp}
Let $\Psi$ be a solution to (\ref{RW}) arising from data as in \Cref{RWwpCauchy}, and define
\begin{align}
    {\mathbb{I}_p}_{u_0,v_0}^{u,v}[\Psi]=\int_{\mathscr{D}_{\Sigma^*}^{u,v}\cap\{r>R\}} dudv\sin\theta d\theta d\phi r^{p-1}\left[p|\nablav\Psi|^2+(2-p)|\slashednabla\Psi|^2+{r^{-2}}|\Psi|^2\right],
\end{align}
then we have for $p\in [0,2]$
\begin{align}\label{RW rp estimate}
\begin{split}
 \int_{\mathscr{C}_u\cap\{r>R\}}  dv \sin\theta d\theta d\phi r^p |\nablav \Psi|^2+ {\mathbb{I}_p}_{u_0,v_0}^{u,v}[\Psi]\lesssim \mathbb{F}_{\Sigma^*}[\Psi]+\int_{\Sigma^*\cap\{r>R\}} r^p|\nablav\Psi|^2 dr\sin\theta d\theta d\phi.
\end{split}
\end{align}

A similar statement holds for
\begin{align}
{{\mathbb{I}_p}_{u_0,v_0}^{u,v}}^{n,T,\slashed{\nabla}}[\Psi]=\sum_{i+|\alpha|\leq n}{\mathbb{I}_p}_{u_0,v_0}^{u,v}[T^i (r\slashed{\nabla})^\alpha\Psi]
\end{align}
and for 
\begin{align}
{{\mathbb{I}_p}_{u_0,v_0}^{u,v}}^{n,k}[\Psi]=\sum_{i+j+|\alpha|\leq n}{\mathbb{I}_p}_{u_0,v_0}^{u,v}[(\Omega^{-1}\nablagml)^i(r^k\nablav)^j(r\slashed{\nabla})^\alpha\Psi]
\end{align}
if $0\leq k\leq2$.
\end{proposition}
%, while commuting with $\nablau$ and using \ref{RW} produces lower order terms which allows to establish the estimates inductively.
We sketch how to establish higher order versions of the estimates of \Cref{RWrp}. Commuting with $r^h\nablav$ for $0\leq h \leq 2$ or $r\slashednabla$ produces terms with favorable signs and we can close the argument by appealing to Hardy and Poincar\'e estimates. Consider for example $\frac{r^2}{\Omega^2}\nablav\Psi:=\Phi^{(1)}$, which satisfies 
\begin{align}\label{Phi 1 transport equation}
\begin{split}
    \nablau \Phi^{(1)}+\frac{3\Omega^2-1}{r}\Phi^{(1)}=\mathring{\slashed{\Delta}}\Psi-(3\Omega^2+1)\Psi.
\end{split}
\end{align}
Applying $\nablav$ and using \bref{RW} we obtain
\begin{align}\label{Phi 1 wave equation}
    \nablau\nablav\Phi^{(1)}+\frac{3\Omega^2-1}{r}\nablav\Phi^{(1)}-\frac{\Omega^2}{r^2}(3\Omega^2-5)\Phi^{(1)}-\Omega^2\slashed\Delta\Phi^{(1)}=-6M\frac{\Omega^2}{r^2}\Psi.
\end{align}
We see that the new $\nablav\Phi^{(1)}$ term has a good sign, so that we when we multiply by $r^p\Omega^{-2k}\nablav\Phi^{(1)}$, integrate by parts over $S^2$ and use Cauchy--Schwarz we get:
\begin{align}
    \begin{split}
        &\nablau\left[r^p\Omega^{-2k}|\nablav\Phi^{(1)}|^2\right]+\nablav\left[r^p\Omega^{-2k}\left(\Omega^2|{\slashednabla}\Phi^{(1)}|^2+(5-3\Omega^2)\frac{\Omega^2}{r^2}|\Phi^{(1)}|^2\right)\right]+{r^{p-1}}{\Omega^{-2(k-1)}}\times\\
        &\Bigg\{(p+4+x(k+2)-\epsilon)|\nablav\Phi^{(1)}|^2+(p-2+x(k-1))\Omega^2|\slashednabla\Phi^{(1)}|^2+\left[\frac{6M}{r}+(5-3\Omega^2)(p-1+x(k-1))\right]\frac{\Omega^2}{r^2}|\Phi^{(1)}|^2\Bigg\}\\
        &\stackrel{S^2}{\lesssim} r^{p-3}\Omega^{2(k-1)}|\Psi|^2,
    \end{split}
\end{align}
where $\epsilon>0$ is sufficiently small. Integrating over $\mathscr{D}^{u,v}_{\Sigma^*}\cap\{r>R\}$ for large enough $R$ and using \Cref{RWrp} for $p\in[0,2]$ we get (using $d\omega=\sin\theta d\theta d\phi$):
\begin{align}
\begin{split}
    &\int_{{\mathscr{C}}_u\cap\{r>R\}}d\bar{v}d\omega\; r^p|\nablav\Phi^{(1)}|^2+\int_{\mathscr{D}^{u,v}_{\Sigma^*}\cap\{r>R\}} d\bar{u}d\bar{v}d\omega\;r^{p-1}\left[(p+4)|\nablav\Phi^{(1)}|^2+(2-p)|\slashednabla\Phi^{(1)}|^2+r^{-2}|\Phi^{(1)}|^2\right]
    \\ &+\int_{\mathscr{I}^+\cap\{\bar{u}\in[u_0,u]\}}d\bar{u}d\omega\; |\mathring{\slashednabla}\Phi^{(1)}|^2+2|\Phi^{(1)}|^2\lesssim \int_{\Sigma^*\cap\{r>R\}}drd\omega\; r^p|\nablav\Phi^{(1)}|^2+\int_{r=R}dtd\omega\; r^p\left[|\slashednabla\Phi^{(1)}|^2+r^{-2}|\Phi^{(1)}|^2\right]\\&+\int_{\mathscr{D}^{u,v}_{\Sigma^*}\cap\{r>R\}}d\bar{u}d\bar{v}d\omega\;r^{p-3}|\Psi|^2.
\end{split}
\end{align}
We control the second term by averaging in $R$ and appealing to \Cref{RWILED} commuted with $\nablav$, and we deal with the last term using the lower order estimate for $\Psi$ from \Cref{RWrp}. Thus
\begin{align}\label{RWrp k=1}
\begin{split}
     &\int_{{\mathscr{C}}_u\cap\{r>R\}}d\bar{v}d\omega\; r^p|\nablav\Phi^{(1)}|^2+\int_{\mathscr{D}^{u,v}_{\Sigma^*}\cap\{r>R\}} d\bar{u}d\bar{v}d\omega\;r^{p-1}\left[(p+4)|\nablav\Phi^{(1)}|^2+(2-p)|\slashednabla\Phi^{(1)}|^2+r^{-2}|\Phi^{(1)}|^2\right]
    \\ &+\int_{\mathscr{I}^+\cap\{\bar{u}\in[u_0,u]\}}d\bar{u}d\omega\; |\mathring{\slashednabla}\Phi^{(1)}|^2+|\Phi^{(1)}|^2 \lesssim \int_{\Sigma^*\cap\{r>R\}}d\bar{v}d\omega\; r^p\left[|\nablav\Phi^{(1)}|^2+|\nablav\Psi|^2\right]+\mathbb{F}^1_{\Sigma^*}[\Psi].
\end{split}
\end{align}
We could do this again for $\left(\frac{r^2}{\Omega^2}\nablav\right)^2\Psi:=\Phi^{(2)}$ and get a similar estimate following the same steps:
\begin{align}\label{RWrp k=2}
    \begin{split}
        &\int_{{\mathscr{C}}_u\cap\{r>R\}} d\bar{v}d\omega\; r^p|\nablav\Phi^{(2)}|^2+\int_{\mathscr{D}^{u,v}_{\Sigma^*}\cap\{r>R\}} d\bar{u}d\bar{v}d\omega\; r^{p-1}\left[(p+8)|\nablav\Phi^{(2)}|^2+(2-p)|\slashednabla\Phi^{(2)}|^2+r^{-2}|\Phi^{(2)}|^2\right]
    \\ &+\int_{\mathscr{I}^+\cap\{\bar{u}\in[u_0,u]\}}d\bar{u}d\omega\;|\mathring{\slashednabla}\Phi^{(2)}|^2-|\Phi^{(2)}|^2\lesssim \int_{\Sigma^*\cap\{r>R\}}d\bar{v}d\omega\; r^p\left[|\nablav\Phi^{(2)}|^2+|\nablav\Phi^{(1)}|^2+|\nablav\Psi|^2\right]+\mathbb{F}^2_{\Sigma^*}[\Psi].
    \end{split}
\end{align}
Note that the integral on $\mathscr{I}^+$ on the right hand side is positive by Poincar\'e's inequality. See \cite{AAG16a}, \cite{AAG16b}, \cite{Mos18} for more about this method, applied to the scalar wave equation.
\subsection{Radiation fields}\label{subsection 5.2 subsection Radiation fields}
In this section we establish the properties of future radiation fields belonging to solutions that arise from smooth, compactly supported data on $\Sigma^*$
%In what follows we assume $\Psi$ arises from smooth, compactly supported data $(\uppsi,\uppsi')$ on $\Sigma^*$.
\subsubsection{Radiation on $\mathscr{H}^+$}\label{subsubsection 5.2.1 radiation on H+}
\begin{defin}\label{RWonH}
Let $\Psi$ be a solution to (\ref{RW}) arising from smooth data $(\uppsi,\uppsi')$ on $\Sigma^*, \Sigma$ or $\overline{\Sigma}$ as in \Cref{RWwpCauchy}. The radiation field $\bm{\uppsi}_{\mathscr{H}^+}$ is defined to be the restriction of $\Psi$ to $\mathscr{H}^+_{\geq0}, \mathscr{H}^+$ or $\overline{\mathscr{H}}^+$ respectively. 
\end{defin}
\begin{remark}
We will be using the same notation for the radiation field on $\mathscr{H}^+_{\geq0}, \mathscr{H}^+$ or $\overline{\mathscr{H}^+}$.
\end{remark}
As a corollary to \Cref{RWwpCauchy} we have
\begin{corollary}
The radiation field $\bm{\uppsi}_{\mathscr{H}^+}$ as in \Cref{RWonH} is smooth on $\mathscr{H}^+_{\geq0}$. The same applies to  $(\Omega^{-1}\nablagml)^k\Psi$ for arbitrary $k$.
\end{corollary}
The integrated local energy decay statement of \Cref{RWILED} gives a quick way to show slow decay for $\bm{\uppsi}_{\mathscr{H}^+}$ and its derivatives:
\begin{proposition}\label{RWdecayfixedR}
    For a solution $\Psi$ to \cref{RW} arising from smooth data of compact support on $\Sigma^*$, $\left|\Psi|_{\{r=R\}}\right|$ decays as $t\longrightarrow\infty$.
\end{proposition}
\begin{proof}
Commuting \bref{RWILEDestimate} with $\mathcal{L}_T$ twice and using the redshift estimate of \Cref{RWredshift} give us for any $R<\infty$ 
\begin{align}
\int_{v_0}^\infty d\bar{v}\;\left[ \underline{F}_{\underline{\mathscr{C}}_v\cap\{r<R\}}[\Psi]+ \underline{F}_{\underline{\mathscr{C}}_v\cap\{r<R\}}[\slashednabla_T\Psi]\right]<\infty.
\end{align}
This in turn implies energy decay in a neighborhood of $\mathscr{H}^+$:
\begin{align*}
    \lim_{v\longrightarrow\infty} \underline{F}_v[\Psi](u_{R},\infty)=0,
\end{align*}
where $v-u_R=R^*$. Commuting with $\Omega^{-1} e_3$ and  ${\mathcal{L}_{\Omega^i}}$ and using \Cref{RWredshift} again gives
\begin{align*}
    \lim_{v\longrightarrow \infty} \sup_{u\in[u_R,\infty]} |\Psi|_{v}=0.
\end{align*}
\end{proof}
\begin{remark} 
The preceding argument works to show that $(\Omega^{-1}\nablagml)^k\Psi$ decays on any hypersurface $r=R$. See also Section 8.2 of \cite{DRSR14}.
\end{remark}
\begin{proposition}
For a solution $\Psi$ to \cref{RW} arising from smooth data of compact support on $\Sigma^*$, The energy flux on $\mathscr{H}^+$ is equal to
\begin{align*}
    F^T_{\mathscr{H}^+}=\int_{\mathscr{H}^+} |\partial_v\Psi|^2 dv \sin\theta d\theta d\phi.
\end{align*}
\end{proposition}
\begin{proof}
This follows from the regularity of $\Psi$ and its angular derivatives on $\mathscr{H}^+$ together with energy conservation.
\end{proof}
\subsubsection{Radiation on $\mathscr{I}^+$}\label{subsubsection 5.2.2 radiation field on I+}
An $r^p$-estimate like \Cref{RWrp} implies the existence of radiation field on $\mathscr{I}^+$ as a "soft" corollary.
\begin{proposition}\label{RWradscri}
    For a solution $\Psi$ to \cref{RW} arising from smooth data of compact support on $\Sigma^*$,
    \begin{align}
        \bm{\uppsi}_{\mathscr{I}^+}(u,\theta^A)=\lim_{v\longrightarrow\infty}\Psi(u,v,\theta^A)
    \end{align}
     exists and belongs to $\Gamma(\mathscr{I}^+)$. Moreover,
    \begin{align}\label{RW limit of energy at null infinity}
        \lim_{v\longrightarrow \infty}\int_{\mathscr{C}_v\cap\{u\in[u_0,u_1]\}}dud\omega\; |\Omega\slashed{\nabla}_3\Psi|^2+\Omega^2|\slashed{\nabla}\Psi|^2+V|\Psi|^2=\int_{\mathscr{I}^+\cap \{u\in[u_0,u_1]\}}dud\omega\; |\partial_u\bm{\uppsi}_{\mathscr{I}^+}|^2.
    \end{align}
\end{proposition}
\begin{proof}
Let $r_2>r_1>8M$, fix $u$ and set $v(r_2,u)\equiv v_2, v(r_1,u)\equiv v_1$. The Sobolev embedding on the sphere $W^{3,1}(S^2)\hookrightarrow L^\infty(S^2)$ and the fundamental theorem of calculus give us:
\begin{align}\label{first order}
\begin{split}
|\Psi(u,v_2,\theta,\phi)-\Psi(u,v_1,\theta,\phi)|^2\leq& B\left[\sum_{|\gamma|\leq 3} \int_{S^2}d\omega\; |\slashed{\mathcal{L}}^\gamma_{S^2} (\Psi(u,v_2,\theta,\phi)-\Psi(u,v_1,\theta,\phi))|\right]^2\\
&=  B\left[\sum_{|\gamma|\leq 3} \int_{S^2}d\omega\int_{v_1}^{v_2}dv |\slashed{\mathcal{L}}^\gamma_{S^2} \Omega\slashed\nabla_4\Psi|\right]^2
\end{split}
\end{align}
Cauchy--Schwarz gives:
\begin{align}
|\Psi(u,v_2,\theta,\phi)-\Psi(u,v_1,\theta,\phi)|^2 \leq \frac{B}{r_1}\Bigg[\sum_{|\gamma|\leq 3} \int_{\mathscr{C}_u\cap\{v>v_1\}}dvd\omega\; r^2|\slashed{\mathcal{L}}^\gamma_{S^2}\Omega\slashed{\nabla}_4\Psi|^2dv\sin \theta d\theta d\phi\Bigg].
\end{align}
where $\slashed{\mathcal{L}}^\gamma_{S^2}=\mathcal{L}_{\Omega_1}^{\gamma_1}\mathcal{L}_{\Omega_2}^{\gamma_2}\mathcal{L}_{\Omega_3}^{\gamma_3}$ denotes Lie differentiation on $S^2$ with respect to its $so(3)$ algebra of Killing fields. This says that $\Psi(u,v,\theta,\phi)$ converges in $L^\infty(\mathscr{I}^+\cap\{u>u_0\})$ for some $u_0>-\infty$ as $v\longrightarrow\infty$. Using higher order $r^p$-estimates we can repeat this argument to show
\begin{align}\label{second order}
    \left|\frac{r^2}{\Omega^2}\nablav\Psi(u,v_2,\theta,\phi)-\frac{r^2}{\Omega^2}\nablav\Psi(u,v_1,\theta,\phi)\right|^2\lesssim \frac{1}{r_1}\Bigg[\sum_{|\gamma|\leq 3} \int_{\mathscr{C}_u\cap\{v>v_1\}} \left|r^2\slashed{\mathcal{L}}^\gamma_{S^2}\nablav\left(r^2\nablav\right)\Psi\right|^2dv\sin \theta d\theta d\phi\Bigg].
\end{align}
Commuting \bref{first order} with $T$ and $\Omega^i$ and using \bref{second order} gives that $\Psi|_{\mathscr{I}^+}=\lim_{v\longrightarrow \infty} \Psi(u,v,\theta,\phi)$ is differentiable on $\mathscr{I}^+$. We can repeat this argument with higher order $r^p$-estimates to find that $\bm{\uppsi}_{\mathscr{I}^+}$ is smooth and $\lim_{v\longrightarrow\infty}\nablau^i(r\slashednabla)^\gamma \Psi=\partial_u^i\mathring{\slashednabla}{}^\gamma\bm{\uppsi}_{\mathscr{I}^+}$ for any index $i$ and multiindex $\gamma$. \Cref{RW limit of energy at null infinity} follows immediately.
\end{proof}
In the following, define  $\Phi^{(k)}:=\left(\frac{r^2}{\Omega^2}\nablav\right)^k\Psi$.
\begin{corollary}\label{RW transverse derivatives converge}
Under the assumptions of \Cref{RWradscri}, the limit
\begin{align}
    \bm{\upphi}^{(k)}_{\mathscr{I}^+}(u,\theta^A)=\lim_{v\longrightarrow\infty}\Phi^{(k)}(u,v,\theta^A)
\end{align}
exists and defines an element of $\Gamma(\mathscr{I}^+)$.
\end{corollary}
\begin{proof}
Let $R,u_1$ be such that $\Psi$ vanishes on $\mathscr{C}_{u}\cap\{r>R\}$ for $u\leq u_0$. We can integrate \cref{RW} from a point $(u_0,v,\theta^A)$ to $(u,v,\theta^A)$ where $r(u_0,v)>R$ to find
\begin{align}
    \Phi^{(1)}(u,\theta^A)=\frac{r^2}{\Omega^2}(u,v)\int_{u_0}^u\frac{\Omega^2}{r^2}\left[\mathring{\slashed{\Delta}}\Psi-(3\Omega^2+1)\Psi\right].
\end{align}
The right hand side converges as $v\longrightarrow\infty$ by \Cref{RWradscri} and Lebesgue's bounded convergence theorem. An inductive argument works to show the same for higher order derivatives.
\end{proof}
\begin{defin}\label{RW future rad field scri}
    Let $\Psi$ be a solution to \cref{RW} arising from smooth data of compact support on $\Sigma^*, \Sigma$ or $\overline{\Sigma}$. The future radiation field on $\mathscr{I}^+$ is defined to be the limit of $\Psi$ towards $\mathscr{I}^+$
    \begin{align*}
        \bm{\uppsi}_{\mathscr{I}^+}(u,\theta,\phi)=\lim_{v\longrightarrow\infty}\Psi(u,v,\theta,\phi). 
    \end{align*}
\end{defin}
\begin{remark}
Note that a solution $\Psi$ arising from compactly supported data on $\overline{\Sigma}$ necessarily corresponds to compactly supported data on $\Sigma^*$.
\end{remark}
\noindent The $r^p$-estimates of \Cref{RWrp} further imply that $\bm{\uppsi}_{\mathscr{I}^+}$ decays as $u\longrightarrow\infty$:
\begin{proposition}\label{RWdecayscri}
    Let $\Psi,\bm{\uppsi}_{\mathscr{I}^+}$ be as in \Cref{RWradscri}. Then $\bm{\uppsi}_{\mathscr{I}^+}$ decays along $\mathscr{I}^+$.
\end{proposition}
\begin{proof}
The fundamental theorem of calculus, Cauchy--Schwarz and a Hardy estimate give us:
\begin{align}
\begin{split}
    \int_{S^2_{u,\infty}}|\Psi_{\fscri}|^2\leq&\int_{S^2_{u,v(r=R)}}|\Psi_{r=R}|^2+\int_{\mathscr{C}_u}\frac{1}{r^2}|\Psi|^2\times \int_{\mathscr{C}_u}r^2|\nablav\Psi|^2\\
    \lesssim&\int_{S^2_{u,v(r=R)}}|\Psi_{r=R}|^2+\int_{\mathscr{C}_u}|\nablav\Psi|^2\times \int_{\mathscr{C}_u}r^2|\nablav\Psi|^2.
\end{split}
\end{align}
\Cref{RWrp} applied to $\Psi$ and $\slashednabla_T\Psi$ implies the decay of $\int_{\mathscr{C}_u\cap\{r>R\}}|\nablav\Psi|^2$ and the boundedness of $\int_{\mathscr{C}_u\cap\{r>R\}}r^2|\nablav\Psi|^2$, and the result follows considering \Cref{RWdecayfixedR}.
\end{proof}
We can in fact compute $\bm{\upphi}_{\mathscr{I}^+}^{(k)}$ out of $\bm{\uppsi}_{\mathscr{I}^+}$ for $k=1,2$:
\begin{corollary}\label{Phi 1 forward}
For a solution $\Psi$ to \cref{RW} arising from smooth data of compact support on $\Sigma^*$, we have
\begin{align}
    \bm{\upphi}^{(1)}_{\mathscr{I}^+}(u,\theta^A)=-\int_u^\infty d\bar{u}\left[\mathcal{A}_2-2\right]\bm{\uppsi}_{\mathscr{I}^+}(\bar{u},\theta^A).
\end{align}
\end{corollary}
\begin{proof}
    Let $-\infty<u_1<u_2<\infty$, $v$ such that $(u,v,\theta^A)\in J^+(\Sigma^*)$. We integrate \cref{Phi 1 transport equation} on $\underline{\mathscr{C}}_v$ between $u_1,u_2$ and use the fact that $\Phi^{(1)}$ has a finite limit $\bm{\upphi}^{(1)}_{\mathscr{I}^+}$ towards $\mathscr{I}^+$ to get
    \begin{align}
    \bm{\upphi}^{(1)}_{\mathscr{I}^+}(u_1,\theta^A)-\bm{\upphi}^{(1)}_{\mathscr{I}^+}(u_2,\theta^A)=-\int_{u_1}^{u_2}d\bar{u}\left[\mathcal{A}_2-2\right]\bm{\uppsi}_{\mathscr{I}^+}(\bar{u},\theta^A).
\end{align}
Since $\bm{\upphi}^{(1)}_{\mathscr{I}^+}$ is uniformly bounded, we have that $\left[\mathcal{A}_2-2\right]\bm{\uppsi}_{\mathscr{I}^+}$ is integrable over $\mathscr{I}^+$. The result follows since $\bm{\upphi}^{(1)}_{\mathscr{I}^+}(u,\theta^A)$ decays as $u\longrightarrow\infty$.
\end{proof}
\begin{lemma}
If $\Psi$ satisfies \bref{RW} then
\begin{align}\label{eq:191}
    \left(\frac{r^2}{\Omega^2}\nablau\right)^2\frac{\Omega^2}{r^2}\nablav\frac{r^2}{\Omega^2}\nablav\Psi=\left[\mathcal{A}_2(\mathcal{A}_2-2)-12M\slashednabla_T\right]\Psi.
\end{align}
\end{lemma}
\begin{proof}
Straightforward computation using \cref{RW}.
\end{proof}
Following the same steps as in the proof of \Cref{Phi 1 forward} we find
\begin{corollary}\label{Phi 2 forward}
For a solution $\Psi$ to \cref{RW} arising from smooth data of compact support on $\Sigma^*$, then $\bm{\upphi}^{(2)}_{\mathscr{I}^+}(u,\theta^A)$ satisfies
\begin{align}
    \bm{\upphi}^{(2)}_{\mathscr{I}^+}(u,\theta^A)=\int_{u}^\infty\int_{u_1}^\infty du_1 du_2 \left[\mathcal{A}(\mathcal{A}_2-2)-6M\partial_u\right]\bm{\uppsi}_{\mathscr{I}^+}(u_2,\theta^A).
\end{align}
\end{corollary}
\begin{corollary}
    For a solution $\Psi$ to \cref{RW} arising from smooth data of compact support on $\Sigma^*$, then 
the radiation field $\bm{\uppsi}_{\mathscr{I}^+}$ satisfies
\begin{align}
    \int_{-\infty}^\infty du_1 \bm{\uppsi}_{\mathscr{I}^+}= \int_{-\infty}^\infty \int_{u_1}^\infty du_1 du_2 \bm{\uppsi}_{\mathscr{I}^+}=0.
\end{align}
\end{corollary}
\subsection{The forwards scattering map}\label{subsection 5.3 the forwards scattering map}
This section combines the results of \Cref{subsection 5.2 subsection Radiation fields} above to prove \Cref{forwardRW}.
\begin{proposition}\label{RWfcp}
    Solutions to (\ref{RW}) arising from smooth data on $\Sigma^*$ of compact support give rise to smooth radiation fields $\uppsi_{\mathscr{I}^+}\in\mathcal{E}_{\mathscr{I}^+}^{T}$ on $\mathscr{I}^+$ and $\uppsi_{\mathscr{H}^+}\in\mathcal{E}_{\mathscr{H}^+_{\geq0}}^{T}$ on $\mathscr{H}^+_{\geq0}$, such that 
    \begin{align}\label{818181}
   ||\bm{\uppsi}_{\mathscr{I}^+}||_{\mathcal{E}^T_{\mathscr{I}^+}}^2+||\bm{\uppsi}_{\mathscr{H}^+}||_{\mathcal{E}^T_{\mathscr{H}^+_{\geq0}}}^2=||(\Psi|_{\Sigma^*},\slashednabla_{n_{\Sigma^*}}\Psi|_{\Sigma^*}) ||_{\mathcal{E}^T_{\Sigma^*}}^2 .
\end{align}
\end{proposition}
\begin{proof}
For data of compact support, Propositions \ref{RWwpCauchy} and \ref{RWradscri} give us the existence of smooth radiation fields $\bm{\uppsi}_{\mathscr{I}^+}$ and $\bm{\uppsi}_{\mathscr{H}^+}$, and by Propositions \ref{RWdecayfixedR}, \ref{RWdecayscri}, $\bm{\uppsi}_{\mathscr{I}^+}$ decays towards $\mathscr{I}^+_+$ and $\bm{\uppsi}_{\mathscr{H}^+}$ decays towards $\mathscr{H}^+$. Let $R$ be sufficiently large and let $v_+,u_+$ be such that $v_+-u_+=R^*$, $v_++u_+>0$. A $T$-energy estimate on the region bounded by $\Sigma^*$, $\mathscr{H}^+_{\geq0}\cap\{v\leq v_+\}$, $\mathscr{I}^+\cap\{u\leq u_+\}$ and $\mathscr{C}_{u_+}\cap\{r\geq R\}, \underline{\mathscr{C}}_{v_+}\cap\{r\leq R\}$ gives
\begin{align}
   \underline{F}^T_{v_+}[\Psi](u_+,\infty)+F^T_{u_+}[\Psi](v_+,\infty)+ \int_{\mathscr{H}^+_{\geq0}\cap \{v\leq v_+\}}dvd\omega\;|\partial_v\Psi|^2+\int_{\mathscr{I}^+\cap\{u\leq u_+\}}dud\omega\;|\partial_u\Psi|^2=||\Psi||_{\mathcal{E}^T_{\Sigma^*}}^2.
\end{align}
The integrated local energy decay statement of \Cref{RWILED} commuted with $\slashednabla_T$, along with the estimate \bref{RW rp estimate} of \Cref{RWrp} for $p=1$ commuted with $\slashednabla_T$, imply that $\underline{F}^T_{v_+}[\Psi](u_+,\infty)+F^T_{u_+}[\Psi](v_+,\infty)$ decay as $u_+\longrightarrow\infty$. This gives us that $\bm{\uppsi}_{\mathscr{I}^+}\in\mathcal{E}^T_{\fscri}$ and $\bm{\uppsi}_{\mathscr{H}^+}\in\mathcal{E}^T_{\mathscr{H}^+_{\geq0}}$ and that $\bm{\uppsi}_{\mathscr{I}^+}, \bm{\uppsi}_{\mathscr{H}^+}$ satisfy \bref{818181}.
\end{proof}
\begin{corollary}\label{RWfcpSigma}
    Solutions to (\ref{RW}) arising from data on ${\Sigma}$ of compact support give rise to smooth radiation fields in $\mathcal{E}_{\mathscr{I}^+}^{T}$ and $\mathcal{E}_{{\mathscr{H}^+}}^{T}$.  Solutions to (\ref{RW}) arising from data on $\overline{\Sigma}$ of compact support give rise to smooth radiation fields in $\mathcal{E}_{\mathscr{I}^+}^{T}$ and $\mathcal{E}_{\overline{\mathscr{H}^+}}^{T}$
\end{corollary}
\begin{proof}
The evolution of $\Psi$ on $ J^+({\Sigma^*})\cap J^-(\Sigma)$ can be handled locally. A $T$-energy estimate on $ J^+({\Sigma})\cap J^-(\Sigma^*)$ gives the result. An identical statement applies to $\overline{\Sigma}$.
\end{proof}
\Cref{RWfcp,,RWfcpSigma} allow us to define the forwards maps $\mathscr{F}^+$ from dense subspaces of $\mathcal{E}^{T}_{\Sigma^*}$, $\mathcal{E}^{T}_{\Sigma}$, $\mathcal{E}^{T}_{\overline{\Sigma}}$.
\begin{defin}
    Let $(\uppsi,\uppsi')$ be a smooth data set to the Regge--Wheeler equation \bref{RW} on $\Sigma^*$ as in \Cref{RWwpCauchy}. Define the map $\mathscr{F}^+$ by 
    \begin{align}
        \mathscr{F}^+:\Gamma_c(\Sigma^*)\times\Gamma_c(\Sigma^*)\longrightarrow \Gamma(\mathscr{H}^+_{\geq0})\times\Gamma(\mathscr{I}^+), (\uppsi,\uppsi')\longrightarrow (\uppsi_{\mathscr{H}^+},\uppsi_{\mathscr{I}^+}),
    \end{align}
    where $(\uppsi_{\mathscr{H}^+},\uppsi_{\mathscr{I}^+})$ are as in the proof of \Cref{RWfcp}.\\
    The map $\mathscr{F}^+$ is defined analogously for data on $\Sigma, \overline{\Sigma}$:
    \begin{align}
        \mathscr{F}^+:\Gamma_c(\Sigma)\times\Gamma_c(\Sigma)\longrightarrow \Gamma(\mathscr{H}^+)\times\Gamma(\mathscr{I}^+), (\uppsi,\uppsi')\longrightarrow (\uppsi_{\mathscr{H}^+},\uppsi_{\mathscr{I}^+}),\\
        \mathscr{F}^+:\Gamma_c(\overline{\Sigma})\times\Gamma_c(\overline{\Sigma})\longrightarrow \Gamma(\overline{\mathscr{H}^+})\times\Gamma(\mathscr{I}^+), (\uppsi,\uppsi')\longrightarrow (\uppsi_{\mathscr{H}^+},\uppsi_{\mathscr{I}^+}).
    \end{align}
\end{defin}
The map $\mathscr{F}^+$ uniquely extends to the forward scattering map of \Cref{RWforwardmap}:
\begin{corollary} \label{RWforwardmap}
The map defined by the forward evolution of data in $\Gamma_c(\Sigma^*)\times\Gamma_c(\Sigma^*)$ as in \Cref{RWfcp} uniquely extends to a map
\begin{align}
        \mathscr{F}^{+}: \mathcal{E}^{T}_{\Sigma^*} \longrightarrow \mathcal{E}_{\mathscr{H}^+_{\geq0}}^{T}\oplus \mathcal{E}_{\mathscr{I}^+}^{T},
\end{align}
which is bounded:
\begin{align}
    ||(\uppsi,\uppsi')||_{\mathcal{E}^{T}_{\Sigma^*}}^2=||\bm{\uppsi}_{\mathscr{H}^+}||_{\mathcal{E}^{T}_{\mathscr{H}^+_{\geq0}}}^2+||\bm{\uppsi}_{\mathscr{I}^+}||_{\mathcal{E}^{T}_{\mathscr{I}^+}}^2 .
\end{align}
We similarly obtain bounded maps
\begin{align}
    \mathscr{F}^{+}: \mathcal{E}^{T}_{\Sigma} \longrightarrow \mathcal{E}_{\mathscr{H}^+}^{T}\oplus \mathcal{E}_{\mathscr{I}^+}^{T},\\
     \mathscr{F}^{+}: \mathcal{E}^{T}_{\overline{\Sigma}} \longrightarrow \mathcal{E}_{\overline{\mathscr{H}^+}}^{T}\oplus \mathcal{E}_{\mathscr{I}^+}^{T}.
\end{align}
The map $\mathscr{F}^+$ is injective on $\Gamma_c(\Sigma^*)\times\Gamma_c(\Sigma^*)$ and therefore extends to a unitary Hilbert-space isomorphism on its image.
\end{corollary}
\subsection{The backwards scattering map}\label{subsubsection 5.4 the backwards scattering map}
This section contains the proof of \Cref{backwardRW,,RW isomorphisms}. We define backwards evolution from data on the event horizon and null infinity in \Cref{RWbackwardsexistence}, and this defines the map $\mathscr{B}^-$ which inverts $\mathscr{F}^+$. \Cref{RW isomorphisms} follows immediately by \Cref{time inversion of RW}.\\
\indent We begin by constructing a solution to the equation on $ J^-(\mathscr{I}^+\cup\mathscr{H}^+_{\geq0})$ out of compactly supported future scattering data.
\begin{proposition}\label{RWbackwardsexistence}
Let $\bm{\uppsi}_{\mathscr{H}^+}\in\Gamma_c(\mathscr{H}^+_{\geq0})$ be supported on $v<v_+<\infty$ such that $\|\bm{\uppsi}_{\mathscr{H}^+}\|_{\mathcal{E}^T_{\mathscr{H}^+_{\geq0}}}<\infty$, $\bm{\uppsi}_{\mathscr{I}^+}\in\Gamma_c(\mathscr{I}^+)$ be supported on $u<u_+<\infty$ such that $\|\bm{\uppsi}_{\mathscr{I}^+}\|_{\mathcal{E}^T_{\mathscr{I}^+}}<\infty$. Then there exists a unique smooth $\Psi$ defined on $ J^+(\Sigma^*)$ that satisfies \cref{RW} and realises $\bm{\uppsi}_{\mathscr{I}^+}$, $\bm{\uppsi}_{\mathscr{H}^+}$ as its radiation fields. Moreover, $(\Psi|_{\Sigma^*},\slashednabla_{n_{\Sigma}}\Psi|_{\Sigma^*})\in \mathcal{E}^T_{\Sigma}$.
\end{proposition}
\begin{proof}
Assume $\bm{\uppsi}_{\mathscr{H}^+}$ is supported on $\{(v,\theta^A), v\in[v_-,v_+]\}\subset\mathscr{H}^+_{\geq0}$ and $\bm{\uppsi}_{\mathscr{I}^+}$ is supported on $[u_-,u_+]$, with $-\infty<u_-,u_+,v_-,v_+<\infty$. Let $\widetilde{\Sigma}$ be a spacelike surface connecting $\mathscr{H}^+$ at a finite $v_*>v_+$ to $\mathscr{I}^+$ at a finite $u_*>u_+$. Fix $\mathcal{R}_{\mathscr{I}^+}>3M$ and let $v^\infty$ be sufficiently large so that $\underline{\mathscr{C}}_{v^\infty}\cap [u_-,u_+]\subset J^+(\Sigma^*)$ and  $r(u,v^\infty)>\mathcal{R}_{\mathscr{I}^+}$ for $u\in[u_-,u_+]$. Denote by $\mathscr{D}$ the region bounded by $\mathscr{H}^+_{\geq 0}\cap\{v\in[v_-,v_*]\}$, $\widetilde{\Sigma}$,$
\underline{\mathscr{C}}_{v^\infty}$, $\Sigma^*$ and $\mathscr{C}_{u_-}$.
%   Local approximating problem
We can find $\Psi$ that solves the "finite" backwards problem for \bref{RW} in $\mathscr{D}$ with the following data: 
\begin{itemize}
    \item $\bm{\uppsi}_{\mathscr{H}^+}$ on $\mathscr{H}^+\cap\{v\in[v_-,v_+]\}$,
    \item $(0,0)$ on $\widetilde{\Sigma}$,
    \item $\bm{\uppsi}_{\mathscr{I}^+}$ on $\underline{\mathscr{C}}_{v^\infty}$.
\end{itemize}
\noindent From \bref{RW} we derive 
\begin{align}\label{RW first transverse derivative in the 3 direction}
    \nablau\left[\frac{r^2}{\Omega^2}|\nablav\Psi|^2\right]+\frac{3\Omega^2-1}{r}\frac{r^2}{\Omega^2}|\nablav\Psi|^2=-\nablav\left[|\mathring{\slashednabla}\Psi|^2+(3\Omega^2+1)|\Psi|^2\right]+\frac{6M\Omega^2}{r^2}|\Psi|^2.
\end{align}
Let $\tilde{v}<v^\infty$ be large enough that $r(u,\tilde{v})>\mathcal{R}_{\mathscr{I}^+}$ for $u\in[u_-,u_+]$. For $\tilde{v}\leq v<v^\infty$ integrate \bref{RW first transverse derivative in the 3 direction} in the region $\mathscr{D}_{v}=\mathscr{D}\cap J^+(\underline{\mathscr{C}}_v)$ with measure $dudvd\omega$ to derive
\begin{align}
\begin{split}
    \int_{\mathscr{C}_u\cap[v,v^\infty]}d\bar{v}d\omega\frac{r^2}{\Omega^2}|\nablav\Psi|^2\leq &\int_{u}^{u_+}d\bar{u}\int_{\mathscr{C}_{\bar{u}}[v,v^\infty]}d\bar{v}d\omega\frac{2\Omega^2}{r}\frac{r^2}{\Omega^2}|\nablav\Psi|^2\\&+\|\Psi\|_{\mathcal{E}^T_{\mathscr{I}^+}}^2+\|\Psi\|_{\mathcal{E}^T_{\mathscr{H}^+}}^2+\int_{u_-}^{u_+}d\bar{u}\int_{S^2}d\omega|\mathring{\slashednabla}\bm{\uppsi}_{\mathscr{I}^+}|^2_{S^2}+4|\bm{\uppsi}_{\mathscr{I}^+}|^2_{S^2}.
\end{split}
\end{align}
Applying Gr\"onwall's inequality to the above gives
\begin{align}\label{this}
    \int_{\mathscr{C}_u\cap[v,v^\infty]}d\bar{v}d\omega\; r^2|\nablav\Psi|^2\leq\frac{r(u,v)^2}{r(u_+,v)^2}\left[\|\Psi\|_{\mathcal{E}^T_{\mathscr{I}^+}}^2+\|\Psi\|_{\mathcal{E}^T_{\mathscr{H}^+}}^2+\int_{[u_-,u_+]\times S^2}d\bar{u}d\omega\;|\mathring{\slashednabla}\bm{\uppsi}_{\mathscr{I}^+}|_{S^2}^2+4|\bm{\uppsi}_{\mathscr{I}^+}|_{S^2}^2\right].
\end{align}
Using \bref{this} we can modify the argument of \Cref{RWradscri} to conclude that for $v>\tilde{v}$
\begin{align}\label{ptwise infinity}
\begin{split}
    \left|\Psi|_{(u,v)}-\bm{\uppsi}_{\mathscr{I}^+}\right|\;\lesssim_{M,u_-,\mathcal{R}_{\mathcal{I}^+}}  \frac{1}{v}\Bigg[\sum_{|\gamma|\leq2}\int_{[u_-,u_+]\times S^2}d\bar{u}d\omega\;&\left[|\slashed{\mathcal{L}}_{\Omega_i}^\gamma\bm{\uppsi}_{\mathscr{I}^+}|_{S^2}^2+|\mathring{\slashednabla}\slashed{\mathcal{L}}^\gamma_{\Omega_i}\bm{\uppsi}_{\mathscr{I}^+}|_{S^2}^2+|\slashed{\mathcal{L}}^\gamma_{\Omega_i}\partial_u\bm{\uppsi}_{\mathscr{I}^+}|_{S^2}^2\right]\\&+\|\slashed{\mathcal{L}}^\gamma_{\Omega_i}\Psi\|_{\mathcal{E}^T_{\mathscr{H}^+}}^2\Bigg].
\end{split}
\end{align}
Analogously, let $\tilde{u}$ be such that $\mathcal{R}_{\mathscr{H}^+}<r(\tilde{u},v)<3M$ for $v\in[v_-,v_+]$, where $\mathcal{R}_{\mathscr{H}^+}<3M$ is fixed. We can multiply the equation by $\frac{1}{\Omega^2}\nablau\Psi$ and integrate by parts over a region $\mathscr{D}_{u}=\mathscr{D}\cap J^+(\mathscr{C}_u)$ to get
\begin{align}
\begin{split}
    \int_{\underline{\mathscr{C}}_v\cap[u,\infty]}dud\omega&\frac{1}{\Omega^2}|\nablau\Psi|^2+\int_{\mathscr{C}_u\cap[v,v_+]}dvd\omega\left[\frac{1}{r^2}|\slashednabla\Psi|^2+\frac{1}{r^2}|\Psi|^2\right]+\int_{\mathscr{D}_{u}}\Omega^2dudv\left[|\mathring{\slashednabla}\Psi|^2+|\Psi|^2\right]\\&\lesssim \int_{\mathscr{H}^+\cap[v,v_+]}dvd\omega\; \left[|\mathring{\slashednabla}\bm{\uppsi}_{\mathscr{H}^+}|^2+|\bm{\uppsi}_{\mathscr{H}^+}|^2\right]+\int_{v}^{v_+}d\bar{u}\int_{\underline{\mathscr{C}}_{\bar{v}}\cap[u,\infty]}d\omega\;\frac{2M}{r^2}\frac{1}{\Omega^2}|\nablau\Psi|^2.
\end{split}
\end{align}
Gr\"onwall's inequality implies
\begin{align}\label{RW exponential backwards near H+}
\begin{split}
     \int_{\underline{\mathscr{C}}_v[u,\infty]}d\bar{u}d\omega\;\frac{1}{\Omega^2}|\nablau\Psi|^2&\lesssim  e^{\frac{1}{2M}(v_+-v)}\left\{\int_{\mathscr{H}^+\cap[v,v_+]} \left[|\mathring{\slashednabla}\bm{\uppsi}_{\mathscr{H}^+}|^2+|\bm{\uppsi}_{\mathscr{H}^+}|^2\right]dvd\omega+\|\Psi\|_{\mathcal{E}^T_{\mathscr{I}^+}}^2+\|\Psi\|_{\mathcal{E}^T_{\mathscr{I}^+}}^2\right\}.
\end{split}
\end{align}
%Exponentially diverging blueshift estimate still implies pointwise limit agrees with radiation field at H
In turn, this implies pointwise control of $\Psi$ near $\mathscr{H}^+$:
\begin{align}\label{ptwise horizon}
    |\Psi(u,v,\theta^A)&-\bm{\uppsi}_{\mathscr{H}^+}(v,\theta^A)|^2\lesssim  \int_{u}^\infty e^{\frac{v-\bar{u}}{2M}}d\bar{u}\times \int_{\underline{\mathscr{C}}_v\cap[u,\infty]}dud\omega\sum_{|\gamma|\leq2}\frac{1}{\Omega^2}\left|\slashed{\mathcal{L}}_{\Omega_i}^\gamma\nablau\Psi\right|^2\\&\lesssim_{M} r\Omega^2(u,v_+) \left[\sum_{|\gamma|\leq2}\int_{u_-}^{u_+}d\bar{u}\int_{S^2}d\omega\left[|\slashed{\mathcal{L}}_{\Omega_i}^\gamma\bm{\uppsi}_{\mathscr{I}^+}|^2+|\mathring{\slashednabla}\slashed{\mathcal{L}}^\gamma_{\Omega_i}\bm{\uppsi}_{\mathscr{I}^+}|^2+|\slashed{\mathcal{L}}^\gamma_{\Omega_i}\partial_u\bm{\uppsi}_{\mathscr{I}^+}|^2\right]+\|\slashed{\mathcal{L}}^\gamma_{\Omega_i}\Psi\|_{\mathcal{E}^T_{\mathscr{H}^+}}^2\right].
\end{align}
In the region $\mathscr{D}\textbackslash (\mathscr{D}_{\tilde{u}}\cap\mathscr{D}_{\tilde{v}})$, $r$ is bounded and energy conservation is sufficient to control $\Psi$ in $L^\infty$. In conclusion, we find that $\Psi$ is controlled in $L^\infty(\mathscr{D})$.\\
%   Taking the limit to Scri
\indent Let $\{v_n^\infty\}_{n=0}^\infty$ be a monotonically increasing sequence tending to $\infty$ with $v_0^\infty=v^\infty$ and define $\mathscr{D}$ in terms of $v_n^\infty$ analogously to $\mathscr{D}$. Denote $\underline{\mathscr{C}}_n=\underline{\mathscr{C}}_{v_n^\infty}\cap\{u\in[u_-,u_+]\}$. We can repeat the above on the region $\mathscr{D}_n$ with data $\bm{\uppsi}_{\mathscr{I}^+}$ on $\underline{\mathscr{C}}_{n}$ to obtain a sequence $\{\Psi_n\}_{n=0}^\infty$. $\Psi_n$ is bounded uniformly in $n$ in the region $\mathscr{D}_k$ for any $k<n$ and we can show uniform boundedness of the derivatives by commuting $\slashednabla_T, \slashednabla_{\Omega_i}$ and using the equation to obtain higher order versions of the estimates above. By Arzela-Ascoli we can extract a convergent subsequence in $C^k(\mathscr{D}_l)$ for any $k,l$ with a limit $\Psi$ that satisfies \bref{RW}. Note that this procedure can be used to uniquely define $\Psi$ everywhere on $ J^-(\widetilde{\Sigma})\cap J^+(\Sigma^*)$. Clearly, $\Psi|_{\mathscr{H}^+}=\bm{\uppsi}_{\mathscr{H}^+}$ and \bref{ptwise infinity} implies $\Psi\longrightarrow \bm{\uppsi}_{\mathscr{I}^+}$.
%   restriction of $\Psi$ to $\Sigma^*$ has finite energy
Finally, a $T$-energy estimate implies that
\begin{align}\label{subunitarity of B-}
    \|(\Psi|_{\Sigma^*}, \slashednabla_{n_{\Sigma^*}}\Psi|_{\Sigma^*})\|_{\mathcal{E}^T_{\Sigma^*}}^2\leq \|\bm{\uppsi}_{\mathscr{H}^+}\|_{\mathcal{E}^T_{\mathscr{H}^+}}^2+\|\bm{\uppsi}_{\mathscr{I}^+}\|_{\mathcal{E}^T_{\mathscr{I}^+}}^2,
\end{align}
so $(\Psi|_{\Sigma^*}, \slashednabla_{n_{\Sigma^*}}\Psi|_{\Sigma^*})\in \mathcal{E}^T_{\Sigma^*}$.
\end{proof}
\begin{defin}\label{RW definition of B-}
    Let $\uppsi_{\mathscr{H}^+},  \uppsi_{\mathscr{I}^+}$ be as in \Cref{RWbackwardsexistence}. Define the map $\mathscr{B}^-$ by
    \begin{align}
        \mathscr{B}^-:\Gamma_c(\mathscr{H}^+_{\geq0})\times\Gamma_c(\mathscr{I}^+)\longrightarrow\Gamma(\Sigma^*)\times\Gamma(\Sigma^*), (\uppsi_{\mathscr{H}^+},\uppsi_{\mathscr{I}^+})\longrightarrow (\Psi|_{\Sigma^*},\slashednabla_{n_{\Sigma^*}}\Psi|_{\Sigma^*}),
    \end{align}
    where $\Psi$ is the solution to \bref{RW} arising from scattering data $(\uppsi_{\mathscr{H}^+},\uppsi_{\mathscr{I}^+})$ as in \Cref{RWbackwardsexistence}.
\end{defin}
\begin{corollary}\label{B- inverts F+}
    The maps $\mathscr{F}^+$, $\mathscr{B}^-$ extend uniquely to unitary Hilbert space isomorphisms on their respective domains, such that $\mathscr{F}^+\circ\mathscr{B}^-=Id$, $\mathscr{B}^-\circ\mathscr{F}^+=Id$.
\end{corollary}
\begin{proof}
We will prove the statement for the map define on data on $\Sigma^*$. We already know that $\mathscr{F}^+$ is a unitary isomorphism and that $\mathscr{F}^+\left[\mathcal{E}^T_{\Sigma^*}\right]\subset\mathcal{E}^{T}_{\mathscr{H}^+_{\geq0}}\oplus \mathcal{E}^{T}_{\mathscr{I}^+}$. 
Let $\uppsi_{\mathscr{H}^+}\in\Gamma_c(\mathscr{H}^+_{\geq0})$, $\uppsi_{\mathscr{I}^+}\in\Gamma_c(\mathscr{I}^+)$. \Cref{RWbackwardsexistence} yields a solution $\Psi$ on $J^+(\Sigma^*)$ to \cref{RW}. Since $\Psi$ realises $\uppsi_{\mathscr{I}^+}$, $\uppsi_{\mathscr{H}^+}$ as its radiation fields as in \Cref{RW future rad field scri,,RWonH} and since $\mathscr{B}^-(\uppsi_{\mathscr{H}^+},\uppsi_{\mathscr{I}^+})\in\left[\Gamma(\Sigma^*)\times\Gamma(\Sigma^*)\right]\cap\mathcal{E}^T_{\Sigma^*}$ (see \Cref{RW enough to be in space}), we have that $\mathscr{F}^+\circ\mathscr{B}^-=Id$ on $\Gamma_c(\mathscr{H}^+_{\geq0})\times\Gamma_c(\mathscr{I}^+)$, which is dense in $\mathcal{E}^{T}_{\mathscr{H}^+_{\geq0}}\oplus \mathcal{E}^{T}_{\mathscr{I}^+}$. Therefore, since $\mathscr{F}^+\left[\mathcal{E}^T_{\Sigma^*}\right]$ is complete, we have that $\mathscr{F}^+\left[\mathcal{E}^T_{\Sigma^*}\right]=\mathcal{E}^{T}_{\mathscr{H}^+_{\geq0}}\oplus \mathcal{E}^{T}_{\mathscr{I}^+}$. The fact that $\mathscr{B}^-$ is bounded means that its unique extension to $\mathcal{E}^{T}_{\mathscr{H}^+_{\geq0}}\oplus \mathcal{E}^{T}_{\mathscr{I}^+}$ must be the inverse of $\mathscr{F}^+$ and we have that $\mathscr{B}^-\circ\mathscr{F}^+=Id_{\mathcal{E}^{T}_{\Sigma^*}}$.
\end{proof}

\begin{remark}\label{unitarity of B- is trivial}
Note that the proof of \Cref{RWbackwardsexistence} only establishes the boundedness of $\mathscr{B}^-$, but showing that $\mathscr{B}^-$ inverts $\mathscr{F}^+$ as was done \Cref{B- inverts F+} turns \bref{subunitarity of B-} to an equality:
\begin{align}\label{unitarity of B- formula}
    \|\mathscr{B}^-(\uppsi_{\mathscr{H}^+},\uppsi_{\mathscr{I}^+})\|_{\mathcal{E}^T_{\Sigma^*}}^2=\|\uppsi_{\mathscr{H}^+}\|^2_{\mathcal{E}^T_{\mathscr{H}^+_{\geq0}}}+\|\uppsi_{\mathscr{I}^+}\|^2_{\mathcal{E}^{T}_{\mathscr{I}^+}}.
\end{align}
\end{remark}
Since the region $J^+(\overline{\Sigma})\cap J^-(\Sigma^*)$ can be handled locally via \Cref{RWwp local statement near B}, \Cref{RWwpSigmabar} and $T$-energy conservation, we can immediately deduce the following:
\begin{corollary}
The map $\mathscr{B}^-$ can be defined on the following domains:
\begin{align}
    \mathscr{B}^{-}:\mathcal{E}^{T}_{\mathscr{H}^+}\oplus \mathcal{E}^{T}_{\mathscr{I}^+}\longrightarrow \mathcal{E}^{T}_{\Sigma},\\
    \mathscr{B}^{-}:\mathcal{E}^{T}_{\overline{\mathscr{H}^+}}\oplus \mathcal{E}^{T}_{\mathscr{I}^+}\longrightarrow \mathcal{E}^{T}_{\overline{\Sigma}},
\end{align}
and we have
\begin{align}
\mathscr{F}^{+}\circ\mathscr{B}^{-}=Id_{\mathcal{E}^T_{\mathscr{H}^+}\oplus\;\mathcal{E}^T_{\mathscr{I}^+}},\qquad
\mathscr{B}^{-}\circ\mathscr{F}^{+}=Id_{\mathcal{E}^T_{\Sigma}},\\
\mathscr{F}^{+}\circ\mathscr{B}^{-}=Id_{\mathcal{E}^T_{\overline{\mathscr{H}^+}}\oplus\;\mathcal{E}^T_{\mathscr{I}^+}},\qquad
\mathscr{B}^{-}\circ\mathscr{F}^{+}=Id_{\mathcal{E}^T_{\overline{\Sigma}}}.
\end{align}
\end{corollary}
We have just completed the proof of \Cref{backwardRW}.\\
\indent Since the Regge--Wheeler equation \bref{RW} is invariant under time inversion, the existence of the maps $\mathscr{F}^-, \mathscr{B}^+$ is immediate: 
\begin{proposition}\label{RW past scattering}
Solutions to (\ref{RW}) arising from smooth data of compact support on $\Sigma$ (or $\overline{\Sigma}$) give rise to smooth radiation fields $\uppsi_{\mathscr{I}^-}\in\mathcal{E}_{\mathscr{I}^-}^{T}$ on $\mathscr{I}^-$ and $\uppsi_{\mathscr{H}^-}\in\mathcal{E}_{\mathscr{H}^-}^{T}$ (or $\mathcal{E}_{\overline{\mathscr{H}^-}}^{T}$) on $\mathscr{H}^-$ (or $\overline{\mathscr{H}^-}$), such that 
\begin{align}\label{919191}
    ||\bm{\uppsi}_{\mathscr{I}^-}||_{\mathcal{E}^T_{\mathscr{I}^-}}^2+||\bm{\uppsi}_{\mathscr{H}^-}||_{\mathcal{E}^T_{\mathscr{H}^-}}^2=||(\Psi|_{\Sigma},\slashednabla_{n_{\Sigma}}\Psi|_{\Sigma}) ||_{\mathcal{E}^T_{\Sigma}}^2.\\
    ||\bm{\uppsi}_{\mathscr{I}^-}||_{\mathcal{E}^T_{\mathscr{I}^-}}^2+||\bm{\uppsi}_{\mathscr{H}^-}||_{\mathcal{E}^T_{\overline{\mathscr{H}^-}}}^2=||(\Psi|_{\Sigma},\slashednabla_{n_{\Sigma}}\Psi|_{\Sigma}) ||_{\mathcal{E}^T_{\overline{\Sigma}}}^2.
\end{align}
As in the case of $\mathscr{F}^+$, there exist Hilbert space isomorphisms
\begin{align}
    \mathscr{F}^{-}:\mathcal{E}^{T}_{\mathscr{H}^+}\oplus \mathcal{E}^{T}_{\mathscr{I}^+}\longrightarrow \mathcal{E}^{T}_{\Sigma},\\
    \mathscr{F}^{-}:\mathcal{E}^{T}_{\overline{\mathscr{H}^+}}\oplus \mathcal{E}^{T}_{\mathscr{I}^+}\longrightarrow \mathcal{E}^{T}_{\overline{\Sigma}},
\end{align}
Let $\bm{\uppsi}_{\mathscr{H}^-}\in\Gamma_c(\mathscr{H}^-)$ be supported on $u>u_+>-\infty$ such that $\|\bm{\uppsi}_{\mathscr{H}^-}\|_{\mathcal{E}^T_{\mathscr{H}^-}}<\infty$, $\bm{\uppsi}_{\mathscr{I}^-}\in\Gamma_c(\mathscr{I}^-)$ be supported on $v>v_+>-\infty$ such that $\|\bm{\uppsi}_{\mathscr{I}^-}\|_{\mathcal{E}^T_{\mathscr{I}^-}}<\infty$. Then there exists a unique smooth $\Psi$ defined on $ J^-(\Sigma)$ that satisfies \cref{RW} and realises $\bm{\uppsi}_{\mathscr{I}^-}$, $\bm{\uppsi}_{\mathscr{H}^-}$ as its radiation fields. Moreover, $(\Psi|_{\Sigma},\slashednabla_{n_{\Sigma}}\Psi|_{\Sigma})\in \mathcal{E}^T_{\Sigma}$ and \bref{919191} is satisfied. A similar statement applies in the case of compactly supported smooth scattering data on $\overline{\mathscr{H}^-}, \mathscr{I}^-$ mapping into $\mathcal{E}^T_{\overline{\Sigma}}$.\\
\indent Therefore, as in the case of $\mathscr{B}^-$, there exist Hilbert space isomorphisms
\begin{align}
    \mathscr{B}^{+}:\mathcal{E}^{T}_{\mathscr{H}^+}\oplus \mathcal{E}^{T}_{\mathscr{I}^+}\longrightarrow \mathcal{E}^{T}_{\Sigma},\\
    \mathscr{B}^{+}:\mathcal{E}^{T}_{\overline{\mathscr{H}^+}}\oplus \mathcal{E}^{T}_{\mathscr{I}^+}\longrightarrow \mathcal{E}^{T}_{\overline{\Sigma}},
\end{align}
which satisfy
\begin{align}
\mathscr{F}^{-}\circ\mathscr{B}^{+}=Id_{\mathcal{E}^T_{\mathscr{H}^-}\oplus\;\mathcal{E}^T_{\mathscr{I}^-}},\qquad
\mathscr{B}^{+}\circ\mathscr{F}^{-}=Id_{\mathcal{E}^T_{\Sigma}},\\
\mathscr{F}^{-}\circ\mathscr{B}^{+}=Id_{\mathcal{E}^T_{\overline{\mathscr{H}^-}}\oplus\;\mathcal{E}^T_{\mathscr{I}^-}},\qquad
\mathscr{B}^{+}\circ\mathscr{F}^{-}=Id_{\mathcal{E}^T_{\overline{\Sigma}}}.
\end{align}
\end{proposition}
With \Cref{RW past scattering}, \Cref{RW isomorphisms} is immediate.
\begin{remark}
It is possible to realise the map $\mathscr{S}$ by directly studying the future radiation fields $\mathscr{I}^+$, $\overline{\mathscr{H}^+}$ on of a solution to the Regge--Wheeler equation \bref{RW} arising all the way from past scattering data on $\mathscr{I}^-$, $\overline{\mathscr{H}^-}$, instead of obtaining it by formally composing $\mathscr{F}^+, \mathscr{B}^+$. The proof uses a subset of the ideas needed to prove \Cref{Corollary 1} of the introduction, so we will state the result here.
\end{remark}
\begin{proposition}
Given smooth, compactly supported past scattering data $(\uppsi_{\mathscr{H}^-},\uppsi_{\mathscr{I}^-})$ for the Regge--Wheeler equation \bref{RW}, there exists a unique solution $\Psi$ realising $\uppsi_{\mathscr{H}^-},\uppsi_{\mathscr{I}^-}$ as its radiation fields on $\overline{\mathscr{H}^-}, \mathscr{I}^-$ respectively. The solution $\Psi$ induces future radiation fields $(\uppsi_{\mathscr{H}^+},\uppsi_{\mathscr{I}^+})\in \mathcal{E}^{T}_{\overline{\mathscr{H}^+}}\oplus\mathcal{E}^{T}_{{\mathscr{I}^+}}$ such that
\begin{align}
    \|\uppsi_{{\mathscr{H}^-}}\|^2_{\mathcal{E}^{T}_{\overline{\mathscr{H}^-}}}+\|\uppsi_{{\mathscr{I}^-}}\|^2_{\mathcal{E}^{T}_{{\mathscr{I}^-}}}= \|\uppsi_{{\mathscr{H}^+}}\|^2_{\mathcal{E}^{T}_{\overline{\mathscr{H}^+}}}+\|\uppsi_{{\mathscr{I}^+}}\|^2_{\mathcal{E}^{T}_{{\mathscr{I}^+}}}
\end{align}
The same result applies with scattering data restricted to $\mathcal{E}^{T}_{{\mathscr{H}^\pm}}$.
\end{proposition}
\subsection{Auxiliary results on backwards scattering}\label{subsection 5.5 auxiliary results}
\subsubsection{Radiation fields of transverse null derivative near $\mathscr{I}^+$}\label{subsubsection 5.5.1 convergence of transverse null derivative}
We can recover the formulae of \Cref{Phi 1 forward,,Phi 2 forward} in backwards scattering from scattering data that is supported away from the future ends of $\mathscr{I}^+,\mathscr{H}^+$:
\begin{corollary}\label{Phi 1 backwards}
Let $(\bm{\uppsi}_{\mathscr{H}^+},\bm{\uppsi}_{\mathscr{I}^+})$ be smooth, compactly supported scattering data for \cref{RW} with corresponding solution $\Psi$. Then 
\begin{align}
    \lim_{v\longrightarrow\infty}\frac{r^2}{\Omega^2}\nablav\Psi=\int^{u_+}_u d\bar{u}(\mathcal{A}_2-2)\bm{\uppsi}_{\mathscr{I}^+}.
\end{align}
\end{corollary}
\begin{proof}
     In a similar fashion to \Cref{Phi 1 forward}, we integrate \bref{RW first transverse derivative in the 3 direction} on a hypersurface $\underline{\mathscr{C}}_v$ from $u_+$ to $u$ to find
    \begin{align}
        \Phi^{(1)}=\frac{r^2}{\Omega^2}\int_u^{u_+}d\bar{u}\; \frac{\Omega^2}{r^2}\left[\mathring{\slashed{\Delta}}\Psi-(3\Omega^2+1)\Psi\right].
    \end{align}
    Repeating the argument leading to \Cref{RW transverse derivatives converge} gives the result:
    \begin{align}
       \bm{\upphi}^{(1)}_{\mathscr{I}^+}= \lim_{v\longrightarrow\infty}\Phi^{(1)}=\int_u^{u_+}d\bar{u}\left(\mathcal{A}_2-2\right)\bm{\uppsi}_{\mathscr{I}^+}.
    \end{align}
\end{proof}
\Cref{Phi 2 forward} can also be recovered in backwards scattering for compactly supported data:
\begin{corollary}\label{Phi 2 backwards}
Let $\Psi$ be a solution to \cref{RW} arising from smooth, compactly supported scattering data $(\bm{\uppsi}_{\mathscr{H}^+},\bm{\uppsi}_{\mathscr{I}^+})$,  then
\begin{align}
\begin{split}
    \lim_{v\longrightarrow\infty}\left(\frac{r^2}{\Omega^2}\nablav\right)^2\Psi&=\int_{u}^\infty\int_{u_1}^\infty du_1 du_2 \left[\mathcal{A}(\mathcal{A}_2-2)-6M\partial_u\right]\bm{\uppsi}_{\mathscr{I}^+}(u_2,\theta^A)\\
    &=\int_u^{u_+}d\bar{u}(\bar{u}-u_-)\left[\mathcal{A}_2(\mathcal{A}_2-2)-6M\partial_u\right]\bm{\uppsi}_{\mathscr{I}^+}(\bar{u},\theta^A).
\end{split}
\end{align}
\end{corollary}
Note that we do not need compact support in the direction of $u\longrightarrow-\infty$ on $\mathscr{I}^+$ for the above results to hold:
\begin{corollary}
\Cref{Phi 1 backwards,,Phi 2 backwards} hold if $\bm{\uppsi}_{\mathscr{I}^+}$ is supported on $(-\infty,u]$, provided $\|\bm{\uppsi}_{\mathscr{I}^+}\|_{\mathcal{E}^T_{\mathscr{I}^+}}<~\infty$.
\end{corollary}
\subsubsection{Backwards $r^p$-estimates}\label{backwards rp estimates}
It is possible to use energy conservation to develop $r$-weighted estimates in the backwards direction that are uniform in $u$, provided $\bm{\uppsi}_{\mathscr{I}^+}$ is compactly supported in $u$. These estimates will help us show that $\mathscr{B}^-$ satisfies \bref{unitarity of B- formula} without reference to $\mathscr{F}^+$ or forwards scattering. We will also use them to show that $\Psi|_{\Sigma^*}\longrightarrow0$ towards $i^0$, and later to obtain similar statements for $\alpha,\underline\alpha$. These estimates first appeared in \cite{AAG19}.\\
\indent Let $u_-,u_+,v_-,v_+$ be as in the proof of \Cref{RWbackwardsexistence}, so that $\mathscr{C}_{u_+}\cap\{r>R\}$ is beyond the support of $\Psi$. Let $u<u_+$, then repeating the proof of \Cref{RWrp} in the region $\mathscr{D}_{u,v_+}^{u_+,\infty}$ for $p=1,2$ gives us (using $d\omega=\sin\theta d\theta d\phi$)
\begin{align}
\begin{split}
    \int_{\mathscr{C}_u\cap\{v>v_+\}}dvd\omega\; r|\nablav\Psi|^2\lesssim& \int_{\mathscr{I}^+\cap\{u\in[u_-,u_+]\}}dud\omega\;r(|\slashednabla\Psi|^2+V|\Psi|^2)\\&+\int_{\mathscr{D}_{u,v_+}^{u_+,\infty}}dudvd\omega\; \left[|\nablav\Psi|^2+|\slashednabla\Psi|^2+V|\Psi|^2\right],
\end{split}
\end{align}
\begin{align}\label{136}
    \int_{\mathscr{C}_u\cap\{v>v_+\}}dvd\omega\;r^2|\nablav\Psi|^2\lesssim \int_{\mathscr{I}^+\cap\{u\in[u_-,u_+]\}}dud\omega\; r^2(|\slashednabla\Psi|^2+V|\Psi|^2)+\int_{\mathscr{D}_{u,v_+}^{u_+,\infty}} dudv d\omega \;r|\nablav\Psi|^2.
\end{align}
We estimate the bulk terms on the right hand side as follows: An energy estimate applied in $\mathscr{D}_{u,v_+}^{u_+,\infty}$ gives for all $u<u_+$:
\begin{align}\label{backwards p=1}
    \int_{\mathscr{C}_u\cap\{v>v_+\}}dvd\omega\;\left[|\nablav\Psi|^2+|\slashednabla\Psi|^2+V|\Psi|^2\right]\leq \int_{\mathscr{I}^+\cap\{u\in[u_-,u_+]\}}dud\omega\; |\partial_u\Psi|^2.
\end{align}
Integrating in $u$ gives
\begin{align}\label{backwards p=1 integrated}
    \int_{\mathscr{D}_{u,v_+}^{u_+,\infty}}dudvd\omega\;\left[|\nablav\Psi|^2+|\slashednabla\Psi|^2+V|\Psi|^2\right]&\leq \int_{u_-}^{u_+}du_1\int_{\mathscr{I}^+\cap\{u_2\in[u_1,u_+]\}}du_2d\omega\; |\partial_u\Psi|^2\\&=\int_{\mathscr{I}^+\cap\{u\in[u_-,u_+]\}}dud\omega\; (u_+-u)|\partial_u\Psi|^2,
\end{align}
knowing that $\nablagml\Psi=0$ at $u=u_+,v>v_+$. Returning to the above we have
\begin{align}
    \int_{\mathscr{C}_u\cap\{v>v_+\}}dvd\omega\;r|\nablav\Psi|^2\lesssim \int_{\mathscr{I}^+\cap\{u\in[u_-,u_+]\}}dud\omega\;r(|\slashednabla\Psi|^2+V|\Psi|^2)+(u_+-u)|\partial_u\Psi|^2.
\end{align}
Integrating once more in $u$ and substituting in (\ref{136}) gives us
\begin{align}\label{RWbackwardsboundedness}
    \int_{\mathscr{C}_u\cap\{v>v_+\}}dvd\omega\;r^2|\nablav\Psi|^2\lesssim \int_{\mathscr{I}^+\cap\{u\in[u_-,u_+]\}}dud\omega\;r(u_+-u)(|\slashednabla\Psi|^2+V|\Psi|^2)+\frac{1}{2}(u-u_+)^2|\partial_u\Psi|^2.
\end{align}
We can integrate in $u$ once more:
\begin{align}\label{RWbackwardsdecay}
    \int_{\mathscr{D}_{u,v_+}^{u_+,\infty}}dudvd\omega\;r^2|\nablav\Psi|^2\lesssim \int_{\mathscr{I}^+\cap\{u\in[u_-,u_+]\}}dud\omega\;\frac{1}{2}r(u-u_+)^2(|\slashednabla\Psi|^2+V|\Psi|^2)+\frac{1}{6}(u_+-u)^3|\partial_u\Psi|^2.
\end{align}
Note that all of the bulk integrals above could be done over $\mathscr{D}=\mathscr{D}_{u,v_+}^{u_+,\infty}\cup\{ J^-(\mathscr{C}_{u_-})\cap J^+(\Sigma^*)\}$ provided that $\partial_u\bm{\uppsi}_{\mathscr{I}^+}$ decays sufficiently fast, such that $\int_{-\infty}^u dud\omega\; |\partial_u\bm{\uppsi}_{\mathscr{I}^+}|^2$ is integrable on $(-\infty,u_+]$.
The first application will be to show that the $\mathscr{B}^+$ is unitary:
\begin{proposition}\label{RW unitary backwards}
    Let $\Psi$ arise from smooth scattering data $\bm{\uppsi}_{\mathscr{I}^+}\in \mathcal{E}^T_{\mathscr{I}^+}, \bm{\uppsi}_{\mathscr{H}^+}\in \mathcal{E}^T_{\mathscr{H}^+}$ as in \Cref{RWbackwardsexistence}. Assume that $\bm{\uppsi}_{\mathscr{I}^+}$ is supported on $u\leq u_+<\infty$, $\bm{\uppsi}_{\mathscr{H}^+}$ is supported on $v \leq v_+ < \infty$, and that $\int_{-\infty}^u dud\omega |\partial_u\bm{\uppsi}_{\mathscr{I}^+}|^2$ is integrable on $(-\infty, u_+]$. Then 
    \begin{align}
        \lim_{u\longrightarrow-\infty} F^T_{\mathscr{C}_u\cap J^+(\Sigma^*)}[\Psi]=0.
    \end{align}
\end{proposition}
\begin{proof}
     The energy estimate
     \begin{align}
         \mathbb{F}_{\Sigma^*}^T[\Psi\cdot\theta_u]+F^T_{\mathscr{C}_u\cap J^+(\Sigma^*)}= \|\bm{\uppsi}_{\mathscr{H}^+}\|^2_{\mathcal{E}^T_{\mathscr{H}^+_{\geq0}}}+\|\bm{\uppsi}_{\mathscr{I}^+}\|^2_{\mathcal{E}^T_{\mathscr{I}^+}}
     \end{align}
    implies that $F^T_{\mathscr{C}_u\cap J^+(\Sigma^*)}[\Psi]$ decays monotonically as $u\longrightarrow\infty$ (here $\theta_u$ is the characteristic function of the subset $\Sigma^*\textbackslash J^-(\mathscr{C}_u)$ of $\Sigma^*$). Combining this with \bref{backwards p=1 integrated} gives the result.
\end{proof}
\begin{corollary}\label{RW unitary backwards corollary}
Let $\Psi$ be as in \Cref{RW unitary backwards}, then
\begin{align}
    \|(\Psi|_{\Sigma^*},\slashednabla_{n_{\Sigma^*}}\Psi|_{\Sigma^*})\|_{\mathcal{E}^T_{\Sigma^*}}^2=\|\bm{\uppsi}_{\mathscr{H}^+}\|^2_{\mathcal{E}^T_{\mathscr{H}^+_{\geq0}}}+\|\bm{\uppsi}_{\mathscr{I}^+}\|^2_{\mathcal{E}^T_{\mathscr{I}^+}}.
\end{align}
\end{corollary}
\indent In the following, we show that if $\bm{\uppsi}_{\mathscr{I}^+}$ is compactly supported on $\mathscr{I}^+$ then we have pointwise decay for $\Psi$ towards $i^0$:
\begin{proposition}\label{RW backwards decay at sigma}
    Let $\Psi$ arise from scattering data $(\bm{\uppsi}_{\mathscr{I}^+},\bm{\uppsi}_{\mathscr{H}^+})\in\Gamma_c(\mathscr{I}^+)\times\Gamma_c(\mathscr{H}^+)$ as in \Cref{RWbackwardsexistence}, then $\Psi|_{\Sigma}\longrightarrow 0$ as $r\longrightarrow \infty$.
\end{proposition}
\begin{proof}
For $R$ large enough, we can estimate
\begin{align}
    \int_{S^2} \left|\Psi|_{\Sigma^*\cap\{r=R\}}-\bm{\uppsi}_{\mathscr{I}^+}\right|\lesssim \int_{v=\frac{1}{2}R^*}^{\infty} \int_{S^2}\sin\theta d{\bar{v}}d\theta d\phi|\nablav\Psi|\lesssim \frac{1}{\sqrt{R}}\int_{\mathscr{C}_{-\frac{1}{2}R^*}\cap\{v>\frac{1}{2}R^*\}}r^2|\nablav\Psi|^2.
\end{align}
The result follows noting that $\bm{\uppsi}_{\mathscr{I}^+}$ is compactly supported and that the integral on the right hand side is bounded according to (\ref{RWbackwardsboundedness}).
\end{proof}
\begin{proposition} Let $\Psi$ arise from the backwards evolution of scattering data $(\bm{\uppsi}_{\mathscr{I}^+},\bm{\uppsi}_{\mathscr{H}^+})$ in $\Gamma_c (\mathscr{I}^+)\times \Gamma_c (\mathscr{H}^+_{\geq0})$ as in \Cref{RWbackwardsexistence}, then 
\begin{align}
    \lim_{R\longrightarrow\infty} \int_{\underline{\mathscr{C}}_{v=\frac{1}{2}R^*} \cap J^+(\Sigma^*)} \Psi= \int_{\mathscr{I}^+} \bm{\uppsi}_{\mathscr{I}^+}.
\end{align}
\end{proposition}
\begin{proof}
Assume the support of $\bm{\uppsi}_{\mathscr{I}^+}$ is in $\mathscr{I}^+\cap \{u\in[u_-,u_+]\}, -\infty<u_-<u_+<\infty$. Let $R$ be such that $u|_{t=0,r=R}=-\frac{1}{2}R^*<u_-$ and let $\tilde{v}=v(t=0,r=R)=\frac{1}{2}R^*$, $\tilde{u}>u_+$. We have
\begin{align}
    \left|\int_{\mathscr{C}_{u=-\frac{1}{2}R^*} \cap J^+(\Sigma)} \Psi-\int_{\mathscr{I}^+} \bm{\uppsi}_{\mathscr{I}^+}\right|^2\leq \left[\int_{\mathscr{D}} |\nablav \Psi|\right]^2\lesssim \frac{1}{{R}}\int_{\mathscr{D}} r^2|\nablav\Psi|^2,
\end{align}
where $\mathscr{D}= J^+(\Sigma^*\cap\{r\geq R\})\cap  J^-(\mathscr{C}_{\tilde{u}})$. The result follows as (\ref{RWbackwardsdecay}) gives us that $\int_{\mathscr{D}} r^2|\nablav\Psi|^2<\infty$.
\end{proof}
\subsubsection{Backwards scattering for data of noncompact support}\label{subsubsection 5.5.3 backwards scattering data of noncompact support}
Estimates \bref{ptwise infinity} and \bref{ptwise horizon} are uniform in the future cutoffs of $\bm{\uppsi}_{\mathscr{I}^+}, \bm{\uppsi}_{\mathscr{H}^+}$ if the relevant fluxes on $\mathscr{I}^+, \mathscr{H}^+_{\geq0}$ are finite, in which case we can remove these cutoffs altogether and work with non-compactly supported scattering data. This follows by a simple modification of the argument leading to the limit $\Psi$ in the proof of \Cref{RWbackwardsexistence}.
\begin{proposition}\label{RW backwards noncompact}
    The results of \Cref{RWbackwardsexistence} hold when $\bm{\uppsi}_{\mathscr{I}^+}, \bm{\uppsi}_{\mathscr{H}^+}$ are not compactly supported, provided
    \begin{align}%\label{this hypothesis}
    &\int_{[u_-,\infty)\times S^2} du\sin\theta d\theta d\phi \sum_{|\gamma|\leq2}| \slashed{\mathcal{L}}^\gamma_{S^2}\partial_u\bm{\uppsi}_{\mathscr{I}^+}|^2+|\slashed{\mathcal{L}}^\gamma_{S^2}\bm{\uppsi}_{\mathscr{I}^+}|^2+|\slashed{\mathcal{L}}^\gamma_{S^2}\mathring{\slashednabla}\bm{\uppsi}_{\mathscr{I}^+}|^2 <\infty,\label{tthis_hypothesis_1}\\
    &\int_{[v_-,\infty)\times S^2} dv\sin\theta d\theta d\phi\sum_{|\gamma|\leq2}| \slashed{\mathcal{L}}^\gamma_{S^2}\partial_v\bm{\uppsi}_{\mathscr{H}^+}|^2+|\slashed{\mathcal{L}}^\gamma_{S^2}\bm{\uppsi}_{\mathscr{H}^+}|^2+|\slashed{\mathcal{L}}^\gamma_{S^2}\mathring{\slashednabla}\bm{\uppsi}_{\mathscr{H}^+}|^2 <\infty.\label{thist_hypothesis_2}
\end{align}
Corollaries \ref{Phi 1 backwards} and \ref{Phi 2 backwards} also hold provided the fluxes of \bref{tthis_hypothesis_1}, \bref{thist_hypothesis_2} are finite with the sums running up to $|\gamma|\leq 4$.
\end{proposition}
\begin{proof}
Let $R>3M$ be fixed, $\{u_{+,n}\}_{n=1}^\infty$ a monotonically increasing sequence and $\{v_{+,n}\}_{n=1}^\infty$ such that $v_{+,n}-u_{+,n}=R^*$. Let $\xi_n^u,\xi_n^v$ be smooth cutoff functions cutting off at $u_{+,n}$ and $v_{+,n}$ respectively. Using $\xi_n^u\bm{\uppsi}_{\mathscr{I}^+}, \xi_n^v\bm{\uppsi}_{\mathscr{H}^+}$ as scattering data, we can apply \Cref{RWbackwardsexistence} to obtain solutions $\Psi_n$ to \cref{RW}, each defined on $\mathscr{D}_n:=J^+(\Sigma^*)\cap\{\{u<u_{+,n}\}\cup\{v<v_{+,n}\}\}$. On $\mathscr{D}_k$, the sequence $\{\Psi_n\}$ for $n>k$ is bounded and equicontinuous, so repeating the argument of \Cref{RWbackwardsexistence} we can find a subsequence converging to $\Psi$ in the topology of compact convergence. The estimate \bref{tthis_hypothesis_1} and the estimates \bref{ptwise horizon}, \bref{ptwise infinity} imply that $\Psi\longrightarrow \bm{\uppsi}_{\mathscr{I}^+}$ towards $\mathscr{I}^+$ and  $\Psi\longrightarrow \bm{\uppsi}_{\mathscr{H}^+}$ towards $\mathscr{H}^+$. The solution $\Psi$ can be extended to the future by repeating the above argument for each $\mathscr{D}_k$ as $k\longrightarrow\infty$. The remaining statements follow by analogous arguments.
\end{proof}
\section{Future asymptotics of the +2 Teukolsky equation}\label{section 6}
\Cref{section 6} is devoted to the study of future radiation fields induced by solutions to the $+2$ Teukolsky equation arising from smooth, compactly supported data on $\Sigma^*$, as was done for the Regge--Wheeler equation in \Cref{subsection 5.2 subsection Radiation fields}.\\
\indent We first gather the estimates we need in \Cref{T+2estimates}. We collect in \Cref{subsubsection 6.1.1 transport estimates} results from \cite{DHR16} estimating $\alpha$ from $\Psi$ defined via (\ref{hier+}) and the estimates of \Cref{subsection 5.1 Basic integrated boundedness and decay estimates} for $\Psi$. Building upon these estimates we then use the methods of \cite{DRrp} and \cite{AAG16a} to obtain $r$-weighted estimates for $\alpha, \psi$ in \Cref{rp+2}. We apply these results to study the future radiation fields and their fluxes in \Cref{+2 radiation}.
\subsection{Integrated boundedness and decay for $\alpha$ via $\Psi$}\label{T+2estimates}
We begin with the following basic proposition, already proven in \Cref{Chandra}:
\begin{proposition}\label{+2 implies RW}
Let $(\upalpha,\upalpha')$ be data on $\Sigma^*$, $\Sigma$ or $\overline{\Sigma}$ giving rise to a solution $\alpha$ to \cref{T+2} as in \Cref{WP+2Sigma*} or  \Cref{WP+2Sigmabar} respectively. Then $\Psi$ defined via \bref{hier+} out of the solution $\alpha$ on $ J^+(\Sigma^*)$, $ J^+(\Sigma)$ or $ J^+(\overline{\Sigma})$ satisfies \cref{RW}.
\end{proposition}
\subsubsection{Transport estimates for $\alpha$}\label{subsubsection 6.1.1 transport estimates}
In what follows assume a small fixed $0<\epsilon<1/8$.
\begin{proposition}\label{psiILED}
Let $\alpha, \psi, \Psi$ be as in  \bref{hier+} and \Cref{+2 implies RW}, Then for any $u$ and any $v>0$ such that $(u,v,\theta^A)\in J^+(\Sigma^*)$, the following estimate holds for sufficiently small $\epsilon>0$ \footnote{All integrals on $\underline{\mathscr{C}}_v$ here are done with respect to the measure $\Omega^2\sin\theta dvd\theta d\phi$}
 \begin{align}\label{locallabel1}
        \int_{\mathscr{C}_{u}\cap J^+(\Sigma^*)\cap J^-(\underline{\mathscr{C}}_v)}d\bar{v}d\omega\; r^{8-\epsilon}\Omega^2|\psi|^2+\int_{\mathscr{D}^{u,v}_{\Sigma^*}}d\bar{u}d\bar{v}d\omega\; r^{7-\epsilon}\Omega^4|\psi|^2 \lesssim \mathbb{F}_{\Sigma^*}[\Psi]+\int_{\Sigma^*\cap J^-(\mathscr{C}_u)\cap J^-(\underline{\mathscr{C}}_v)}drd\omega\; r^{8-\epsilon}\Omega^2|\psi|^2.
    \end{align}
\end{proposition}
\begin{proof}
Here we repeat the argument of Proposition 12.1.1 of \cite{DHR16}. Using the definition of $\psi$ in (\ref{hier+}) we can derive 
\begin{align}
    \partial_u \left[r^{6+n}\Omega^4|\psi|^2\right]+nr^{n+5}\Omega^4|\psi|^2=2r^{n-1}\frac{\Omega^2}{r^2}\Psi \cdot r^3\Omega\psi\leq \frac{1}{2}nr^{n+5}\Omega^4|\psi|^2+\frac{2}{n}r^{n-3}\Omega^2|\Psi|^2.
\end{align}
The result follows by integrating over $\mathscr{D}^{u,v}_{\Sigma^*}$ for $0<n<2$ and using \Cref{RWILED,,RWrp}.
\end{proof}
%Check that you can remove the epsilon when not estimating the angular derivatives
\begin{proposition}\label{alphaILED}
Let $\alpha, \psi, \Psi$ be as in  \bref{hier+} and \Cref{+2 implies RW}, Then for any $u$ and any $v>0$ such that $(u,v,\theta^A)\in J^+(\Sigma^*)$, the following estimate holds for sufficiently small $\epsilon>0$
 \begin{align}\label{locallabel2}
 \begin{split}
        \int_{\mathscr{C}_{u}\cap J^+(\Sigma^*)\cap J^-(\underline{\mathscr{C}}_v)}d\bar{v}d\omega\; r^{6-\epsilon}\Omega^4|\alpha|^2&+\int_{\mathscr{D}^{u,v}_{\Sigma^*}}d\bar{u}d\bar{v}d\omega\; r^{5-\epsilon}\Omega^6|\alpha|^2 \\ &\lesssim \mathbb{F}_{\Sigma^*}[\Psi]+\int_{\Sigma^*\cap J^-(\mathscr{C}_u)\cap J^-(\underline{\mathscr{C}}_v)}drd\omega\; r^{8-\epsilon}\Omega^2|\psi|^2+r^{6-\epsilon}\Omega^4|\alpha|^2.
        \end{split}
    \end{align}
provided the right hand side is finite.
\end{proposition}
\begin{proof}
Similar to \Cref{psiILED}. See Propositions 12.1.2, 12.2.6 and 12.2.7 of \cite{DHR16}.
\end{proof}
\begin{proposition}\label{ILED alpha 2nd angular}
Let $\alpha, \psi, \Psi$ be as in  \bref{hier+} and \Cref{+2 implies RW}, Then for any $u$ and any $v>0$ such that $(u,v,\theta^A)\in J^+(\Sigma^*)$, the following estimate holds for sufficiently small $\epsilon>0$
\begin{align}\label{2ndderivativeofpsi}
 \begin{split}
        \int_{\mathscr{C}_{u}\cap J^+(\Sigma^*)\cap J^-(\underline{\mathscr{C}}_v)}d\bar{v}d\omega\; r^{8-\epsilon}|-2r^2\fancydstar_2\fancyd_2&(r^3\Omega\psi)|^2+\int_{\mathscr{D}^{u,v}_{\Sigma^*}}d\bar{u}d\bar{v}d\omega\; \frac{\Omega^2}{r^3}\left(1-\frac{3M}{r}\right)^2 |-2r^2\fancydstar_2\fancyd_2(r^3\Omega\psi)|^2 \\ &\lesssim \mathbb{F}_{\Sigma^*}[\Psi]+\int_{\Sigma^*\cap J^-(\mathscr{C}_u)\cap J^-(\underline{\mathscr{C}}_v)}dr d\omega\; r^{8-\epsilon}\Omega^2|\psi|^2+r^{6-\epsilon}\Omega^4|\alpha|^2,
        \end{split}
    \end{align}
    provided the right hand side is finite.
\end{proposition}
\begin{proof}
Control of $\psi,\alpha$ as in \Cref{psiILED,,alphaILED} allows us to directly control the flux of $-2r^2\fancydstar_2\fancyd_2(r^3\Omega\psi)$ on $\mathscr{C}_u$ using (\ref{eq:d4Psi}) and the flux bound of \Cref{RWrp}, while the spacetime integral can be controlled via \Cref{RWILED}. 
\end{proof}
Commuting (\ref{hier+}) with $r\fancyd_2$ and using the flux bound of the previous proposition allows us to obtain an integrated decay statement for $r\fancyd_2\psi$:
\begin{proposition}
Let $\alpha, \psi, \Psi$ be as in  \bref{hier+} and \Cref{+2 implies RW}, Then for any $u$ and any $v>0$ such that $(u,v,\theta^A)\in J^+(\Sigma^*)$, the following estimate holds for sufficiently small $\epsilon>0$
\begin{align}
    \int_{\mathscr{D}^{u,v}_{\Sigma^*}} d\bar{u}d\bar{v} d\omega\; r^{7-\epsilon} \Omega^4|r\fancyd_2\psi|^2 &\lesssim \mathbb{F}_{\Sigma^*}[\Psi]+\int_{\Sigma^*\cap J^-(\mathscr{C}_u)\cap J^-(\underline{\mathscr{C}}_v)}drd\omega\;r^{8-\epsilon}\Omega^2|r\fancyd_2\psi|^2+r^{6-\epsilon}\Omega^4|\alpha|^2.
\end{align}
provided the right hand side is finite.
\end{proposition}
Finally, commuting the equation for $\psi$ in (\ref{hier+}) with $\slashednabla_{R^*}$ gives us control over the remaining $\nablav\psi$ using the estimates for $\Psi$ and the nondegenerate control of $\slashednabla_{R^*} \psi$ in \Cref{RWILED}. We can optimise the weights near the event horizon and null infinity by commuting further with $\Omega^{-1}\nablagml$ and $r\nablav$ respectively:
\begin{proposition}\label{ILED psi higherorder}
Let $\alpha, \psi, \Psi$ be as in  \bref{hier+} and \Cref{+2 implies RW}, Then for any $u$ and any $v>0$ such that $(u,v,\theta^A)\in J^+(\Sigma^*)$, the following estimate holds for sufficiently small $\epsilon>0$
\begin{align}
\begin{split}
    \int_{\mathscr{C}_{u}\cap J^+(\Sigma^*)\cap J^-(\underline{\mathscr{C}}_v)} d\bar{v}d\omega\;r^{4-\epsilon} |\nablav(r^3\Omega&\psi)|^2 + \int_{\mathscr{D}^{u,v}_{\Sigma^*}}d\bar{u}d\bar{v}d\omega\;r^{7-\epsilon}\left[|(\Omega^{-1}\nablagml(\Omega\psi)|^2+|r\nablav\Omega\psi|^2\right]\\&\lesssim \mathbb{F}_{\Sigma^*}[\Psi]+\int_{\Sigma^*\cap\J^-(\mathscr{C}_u)\cap J^-(\underline{\mathscr{C}}_v)}dr d\omega\;r^{8-\epsilon}\left[|\Omega\psi|^2+|\Omega^{-1}\nablagml\psi|^2+|r\nablav\psi|^2\right]
\end{split}
\end{align}
provided the right hand side is finite.
\end{proposition}
Similar estimates can be obtained for $\alpha$ by applying these ideas one more time to (\ref{hier+}), see section 12.3 of \cite{DHR16}. 
\begin{proposition}\label{alphaILED higher order}
Let $\alpha, \psi, \Psi$ be as in  \bref{hier+} and \Cref{+2 implies RW}, Then for any $u$ and any $v>0$ such that $(u,v,\theta^A)\in J^+(\Sigma^*)$, the following estimate holds for sufficiently small $\epsilon>0$
   \begin{align}
       \begin{split}
        &\int_{\mathscr{C}_{u}\cap J^+(\Sigma^*)\cap J^-(\underline{\mathscr{C}}_v)}d\bar{v}d\omega\;r^{6-\epsilon}\left[|r\fancyd_2\Omega^2\alpha|^2+|\Omega^{-1}\nablagml\Omega^2\alpha|^2+|r\nablav\Omega^2\alpha|^2\right]\\&+\int_{\mathscr{D}^{u,v}_{\Sigma^*}}d\bar{u}d\bar{v}d\omega\;r^{5-\epsilon}\left[|r\fancyd_2\Omega^2\alpha|^2+|\Omega^{-1}\nablagml\Omega^2\alpha|^2+|r\nablav\Omega^2\alpha|^2\right] \\ &\lesssim \mathbb{F}_{\Sigma^*}[\Psi]+\int_{\Sigma^*\cap\J^-(\mathscr{C}_u)\cap J^-(\underline{\mathscr{C}}_v)}drd\omega\;\Bigg\{r^{8-\epsilon}\left[|r\fancyd_2\Omega\psi|^2+|\Omega^{-1}\nablagml\Omega\psi|^2+|r\nablav\Omega\psi|^2\right]\\&+r^{6-\epsilon}\left[|r\fancyd_2\Omega^2\alpha|^2+|\Omega^{-1}\nablagml\Omega^2\alpha|^2+|r\nablav\Omega^2\alpha|^2\right]\Bigg\},
        \end{split}
    \end{align}
    provided the right hand side is finite.
\end{proposition}
\subsubsection{An $r^p$-estimate for $\alpha,\psi$}\label{rp+2}
The structure of the $+2$ Teukolsky equation allows us to apply the method of \cite{DRrp} and \cite{AAG16a} to \Cref{T+2} in the same way it was applied in \Cref{subsection 5.1 Basic integrated boundedness and decay estimates}.
\begin{proposition}\label{T+2rp}
 Let $\alpha$ be a solution to the +2 equation (\ref{T+2}), then for $p\in[0,2], u>u_0$ and $\mathscr{D}=\{(u,v,\theta,\phi): \bar{u}\in[u_0,u], r>R\}$ we have the following:
 \begin{align}
 \begin{split}
     \int_{\mathscr{C}_{u}\cap\{r>R\}}d\bar{v}d\omega\;r^p|\nablav r^5\Omega^{-2}\alpha|^2+\int_{\mathscr{D}}d\bar{u}d\bar{v}d\omega\; (p+8)r^{p-1}|\nablav r^5\Omega^{-2}\alpha|^2+(2-p)r^{p-1}|\slashednabla r^5\Omega^{-2}\alpha|^2\\ \lesssim \mathbb{F}_{\Sigma^*}[\Psi]+\int_{\Sigma^*}r^{8-\epsilon}\Omega^2|\psi|^2+r^{6-\epsilon}\Omega^2|\alpha|^2+\int_{\Sigma^*\cap\{r>R\}}drd\omega\; r^p|\nablav r^5\Omega^{-2}\alpha|^2.
 \end{split}
 \end{align}
\end{proposition}
\begin{proof}
Rewrite the +2 equation in terms of $r^5\Omega^2\alpha$:
\begin{align}\label{+2 equation for radiation field}
    \nablav\nablau r^5\Omega^{-2}\alpha+2\frac{3\Omega^2-1}{r}\nablav r^5\Omega^{-2}\alpha-\Omega^2\slashed{\Delta}r^5\Omega^{-2}\alpha-\frac{\Omega^2}{r^2}(15\Omega^2-13)r^5\Omega^{-2}\alpha=0.
\end{align}
Multiply by $r^p\nablav r^5\Omega^{-2}\alpha$ and integrate by parts:
\begin{align}
\begin{split}
    &\nablau\left[r^p|\nablav r^5\Omega^{-2}\alpha|^2\right]+\nablav\left[r^p\Omega^2\left(|\slashednabla r^5\Omega^{-2}\alpha|^2-(15\Omega^2-13)\frac{1}{r^2}|r^5\Omega^{-2}\alpha|^2\right)\right]\\
    &+\left\{4(3\Omega^2-1)+p\Omega^2\right\}r^{p-1}|\nablav r^5\Omega^{-2}\alpha|^2+\left[2-p-\frac{2M}{r}\right]r^{p-1}\left|\slashednabla r^5\Omega^{-2}\alpha\right|^2\\
    &-\left[\frac{2M}{r}(30\Omega^2-13)+(2-p)(15\Omega^4-13\Omega^2)\right]r^{p-3}\Omega^2|r^5\Omega^{-2}\alpha|^2=0.
\end{split}
\end{align}
Integrating in $\mathscr{D}$, the Poincar\'e inequality (\ref{poincare}) ensures that the leading order terms in the $\mathscr{I}^+$ flux term are positive, and we similarly use (\ref{poincare}) to absorb the last term in the previous equation into the term containing the angular derivative. Finally we can deal with the $r=R$ flux term by averaging over $R$ and using the integrated decay statement of \Cref{alphaILED}.
\end{proof}
Similarly, we have
\begin{proposition}\label{T+1rp}
    Let $\psi$ arise from $\alpha$ according to (\ref{hier+}), then we have
    \begin{align}
         \int_{\mathscr{C}_{u}\cap\{r>R\}}d\bar{v}d\omega\;r^p|\nablav r^5\Omega^{-1}\psi|^2+\int_{\mathscr{D}}d\bar{v}d\omega\; (p+4)r^{p-1}|\nablav r^5\Omega^{-1}\psi|^2+(2-p)r^{p-1}|\slashednabla r^5\Omega^{-1}\psi|^2\\ \lesssim \mathbb{F}_{\Sigma^*}[\Psi]+\int_{\Sigma^*}drd\omega\;r^{8-\epsilon}\Omega^2|\psi|^2+r^{6-\epsilon}\Omega^2|\alpha|^2+ \int_{\Sigma^*\cap\{r>R\}}drd\omega\;r^p|\nablav r^5\Omega^{-1}\psi|^2.
    \end{align}
\end{proposition}
\begin{proof}
Rewrite the definition of $\psi$ in terms of $r^5\Omega^{-1}\psi$ and differentiate via $\nablau$ to get
\begin{align}
    \nablau\nablav r^5\Omega^{-1}\psi+\frac{3\Omega^2-1}{r}\nablav r^5\Omega^{-1}\psi-\Omega^2\slashed{\Delta} r^5\Omega^{-1}\psi+\frac{\Omega^2}{r^2}(3\Omega^2-5)r^5\Omega^{-1}\psi=-12M^2\frac{\Omega^4}{r^4} r^5\Omega^{-2}\alpha.
\end{align}
We repeat the argument employed in \Cref{T+2rp} using Cauchy--Schwarz to estimate the $\alpha$ term on the right hand side. 
\end{proof}
\begin{remark}\label{transversealphapsi}
We have similar statements to \Cref{T+2rp,,T+1rp} for $\frac{r^2}{\Omega^2}\nablav$ derivatives of $r^5\Omega^{-1}\psi$ and $r^5\Omega^{-2}\alpha$ .
\end{remark}
\subsection{Future radiation fields and fluxes}\label{+2 radiation}
In this section the notion of future radiation fields of solutions to the +2 Teukolsky equation \bref{T+2} is defined, and some of the properties of these radiation fields are studied, in particular obtaining their $\mathcal{E}^{T,+2}_{\mathscr{H}^+}$, $\mathcal{E}^{T,+2}_{\mathscr{I}^+}$ fluxes when they belong to solutions of \bref{T+2} arising from smooth data of compact support.
\subsubsection{Radiation on $\mathscr{H}^+$}\label{+2 radiation on H+}
\begin{defin}\label{+2 radiation alpha definition H}
    Let $\alpha$ be a solution to (\ref{T+2}) arising from smooth data as in \Cref{WP+2Sigma*} or \Cref{WP+2Sigmabar}. The radiation field of $\alpha$ along $\mathscr{H}^+$, denoted $\upalpha_{\mathscr{H}^+}$ is defined to be the restriction of $2M\Omega^2\alpha$ to $\mathscr{H}^+$. 
\end{defin}
\begin{remark}
We will use the same notation for the radiation field on $\mathscr{H}^+_{\geq0}, \mathscr{H}+$ or $\overline{\mathscr{H}^+}$.
\end{remark}
As an easy consequence of the estimates of the previous section we have the following non-quantitative decay statements: (All statements here apply to $\overline{\mathscr{H}^+}$)
\begin{corollary}\label{psi+2ptwisedecay}
For smooth data of compact support for the +2 on $\Sigma^*$, $\Sigma$ or $\overline{\Sigma}$, $\psi$ decays along any hypersurface $r=R$
\begin{align}
    \lim_{v\longrightarrow \infty} \left|\left|\Omega\psi\right|\right|_{L^2(S^2_{R})}=0.
\end{align}
\end{corollary}
\begin{proof}
\Cref{psiILED} applied to $\psi$ and $\slashednabla_T\psi$ implies 
\begin{align}
    \lim_{v\longrightarrow \infty} \int_{\underline{\mathscr{C}}_v\cap\{r\in[2M,R]\}} \Omega^2|\psi|^2 du \sin\theta d\theta d\phi =0 .
\end{align}
Repeating this for $\Omega^{-1}\slashednabla_3 \Omega \psi$ using \Cref{ILED psi higherorder} gives the result.
\end{proof}
The same works for $\alpha$ using propositions \ref{alphaILED} and \ref{alphaILED higher order}:
\begin{corollary}\label{alpha+2ptwisedecay}
For smooth data of compact support on $\Sigma^*$, $\Sigma$ or $\overline{\Sigma}$, $\alpha$ decays along any hypersurface $r=R$:
\begin{align}
    \lim_{v\longrightarrow \infty} \left|\left|\Omega^2\alpha\right|\right|_{L^2(S^2_{R})}=0.
\end{align}
\end{corollary}
Commuting with the lie derivative along angular Killing fields $\slashed{\mathcal{L}}_{\Omega_i}^\gamma$ for $|\gamma|\leq2$ gives
\begin{corollary}\label{horizonpsidecay}
For smooth data of compact support for the +2 Teukolsky equation on $\Sigma^*$, $\Sigma$ or $\overline{\Sigma}$, $\Omega\psi|_{\mathscr{H}^+}$ and $\Omega^2\alpha|_{\mathscr{H}^+}$ decay towards $\mathscr{H}^+_+$.
\end{corollary}
\subsubsection{Radiation flux on $\mathscr{H}^+$}\label{+2 radiation flux on H+}
Assume $\alpha$ satisfies \bref{T+2} and arises from smooth, compactly supported data on $\Sigma^*$. The regularity of $\Psi$ implies that
on $\mathscr{H}^+$, the radiation flux in terms of $\Psi$ is given by \bref{RW def rad flux at H} %of \Cref{RW def of rad at H}:
\begin{align}
    \left\|\Psi\right\|_{\mathcal{E}^T_{\mathscr{H}^+}}^2= \left\|\nablav\Psi\right\|^2_{L^2(\mathscr{H}^+)}.
\end{align}
Recall that if $\alpha$ satisfies the +2 Teukolsky equation \cref{T+2} then $\alpha, \Psi$ also satisfy (\ref{eq:d4Psi}) and (\ref{eq:d4d4Psi}):
\begin{align}\label{psi out of alpha}
\begin{split}
\nablav \Psi=\mathcal{A}_2 \frac{r^2}{\Omega^2}\Omega\slashed{\nabla}_3r\Omega^2\alpha-6Mr\Omega^2\alpha-(3\Omega^2-1)\frac{r^2}{\Omega^2}\Omega\slashed{\nabla}_3r\Omega^2\alpha,
\end{split}
\end{align}
\begin{align}\label{Psi out of alpha}
\begin{split}
\frac{\Omega^2}{r^2}\Omega\slashed{\nabla}_4 \frac{r^2}{\Omega^2}\Omega\slashed{\nabla}_4 \Psi&=\mathcal{A}_2(\mathcal{A}_2-2) r\Omega^2\alpha-6M\left(\Omega\slashed{\nabla}_3+\Omega\slashed{\nabla}_4\right)r\Omega^2\alpha.
\end{split}
\end{align}
We find the limits towards $\mathscr{H}^+$: the left hand sides of (\ref{Psi out of alpha}) reads:
\begin{align}
(\nablav)^2\Psi+\frac{3\Omega^2-1}{r}\nablav\Psi\longrightarrow \left[\partial_v-\frac{1}{2M}\right]\partial_v\bm{\uppsi}_{\mathscr{H}^+} \text{\;towards\;} \mathscr{H}^+.
\end{align}
Now the right hand side reads:
\begin{align}
\mathcal{A}_2\left[\mathcal{A}_2-2\right]\upalpha_{\mathscr{H}^+}-6M\partial_v\upalpha_{\mathscr{H}^+},
\end{align}
so we must determine $\partial_v\Psi$ from the equation
\begin{align}\label{equation for Psi out of alpha on H}
\partial_v^2\bm{\uppsi}_{\mathscr{H}^+}-\frac{1}{2M}\partial_v\bm{\uppsi}_{\mathscr{H}^+}=\mathcal{A}_2\left[\mathcal{A}_2-2\right]\upalpha_{\mathscr{H}^+}-6M\partial_v\upalpha_{\mathscr{H}^+}.
\end{align}
In Kruskal coordinates, this reads
\begin{align}
\begin{split}
    \frac{1}{(2M)^2}\partial_V^2\Psi&=\mathcal{A}_2(\mathcal{A}_2-2)V^{-2}\upalpha_{\mathscr{H}^+}-3V^{-1}\partial_V \upalpha_{\mathscr{H}^+}\\
    &=\left[\mathcal{A}_2(\mathcal{A}_2-2)-6\right]V^{-2}\upalpha_{\mathscr{H}^+}-3V\partial_V V^{-2}\upalpha_{\mathscr{H}^+}.
\end{split}
\end{align}
With the condition that $\Psi,\nablav\Psi$ decay as $v\longrightarrow\infty$, we have
\begin{align}\label{eq:197}
    -\frac{1}{(2M)^2}\partial_V\Psi=\int_V^\infty\left\{\left[\mathcal{A}_2(\mathcal{A}_2-2)-6\right]V^{-2}\upalpha_{\mathscr{H}^+}-3V\partial_V V^{-2}\upalpha_{\mathscr{H}^+}\right\}d\widebar{V}
\end{align}
Integrating in again in $V$ and using the fact that $\upalpha_{\mathscr{H}^+}$ is compactly supported we get:
\begin{align}\label{eq:198}
    \frac{1}{(2M)^2}\Psi=\int_V^\infty (V-\bar{V})\left\{\left[\mathcal{A}_2(\mathcal{A}_2-2)-6\right]V^{-2}\upalpha_{\mathscr{H}^+}-3V\partial_V V^{-2}\upalpha_{\mathscr{H}^+}\right\}d\bar{V}.
\end{align}
In Eddington-Finkelstein coordinates this reads
\begin{lemma}\label{flux+2horizon}
Let $\alpha$ be a solution to the +2 Teukolsky equation \bref{T+2} arising from data of compact support on $\mathscr{H}^+_{\geq0}$, and let $\Psi$ be the corresponding solution to the Regge--Wheeler equation arising from $\alpha$ via \bref{hier+}. Then the radiation field $\uppsi_{\mathscr{H}^+}$ on $\mathscr{H}^+$ belonging to $\Psi$ is given by:
\begin{align}\label{eq:199}
\bm{\uppsi}_{\mathscr{H}^+}=2M\int_v^\infty \left[e^{\frac{1}{2M}(v-\bar{v})}-1\right]\left\{\mathcal{A}_2\left[\mathcal{A}_2-2\right]\upalpha_{\mathscr{H}^+}-6M\partial_v\upalpha_{\mathscr{H}^+}\right\},
\end{align}
\begin{align}\label{eq:200}
\partial_v \bm{\uppsi}_{\mathscr{H}^+}=\int^{\infty}_v e^{\frac{1}{2M}(v-\bar{v})}\{-\mathcal{A}_2\left[\mathcal{A}_2-2\right]\upalpha_{\mathscr{H}^+}+6M\partial_v\upalpha_{\mathscr{H}^+}\} d\overline{v}.
\end{align}
\end{lemma}
Equations \bref{eq:197}--\bref{eq:200} are the expressions for the radiation field and flux at $\mathscr{H}^+$ that we are able to compute directly out of data there. Note that this applies equally to radiation on $\mathscr{H}^+_{\geq0}, \mathscr{H}^+$ or $\overline{\mathscr{H}^+}$.\\
\indent Now let $F_{\mathscr{H}^+}=\int_v^\infty e^{\frac{1}{2M}(v-\bar{v})} \upalpha_{\mathscr{H}^+}d\bar{v}$, then $\partial_v F= \frac{1}{2M}F-\upalpha_{\mathscr{H}^+}$, which implies
\begin{align}
    -\partial_v \bm{\uppsi}_{\mathscr{H}^+}=\mathcal{A}_2(\mathcal{A}_2-2)F_{\mathscr{H}^+}-6M\partial_v F_{\mathscr{H}^+}.
\end{align}
Note that $F_{\mathscr{H}^+}$ decays towards the future end of $\mathscr{H}^+_{\geq0}$, since
\begin{align}
   \lim_{v\longrightarrow\infty}F_{\mathscr{H}^+}=\lim_{v\longrightarrow\infty} \int_v^\infty e^{\frac{1}{2M}(v-\bar{v})} \upalpha_{\mathscr{H}^+}d\bar{v}=\lim_{v\longrightarrow\infty} -2M \upalpha_{\mathscr{H}^+}=0.
\end{align}
Therefore, $L^2(\mathscr{H}^+_{\geq0})$ norm of $\partial_v \bm{\uppsi}_{\mathscr{H}^+}$ is given by
\begin{align}\label{+2 norm on H+ beyond B}
\begin{split}
    \left\|\partial_v \bm{\uppsi}_{\mathscr{H}^+}\right\|_{L^2(\mathscr{H}^+_{\geq0})}^2=&\left\|\mathcal{A}_2(\mathcal{A}_2-2)F_{\mathscr{H}^+}\right\|^2_{L^2(\mathscr{H}^+_{\geq0})}+\left\|6M\partial_v F_{\mathscr{H}^+}\right\|^2_{L^2(\mathscr{H}^+_{\geq0})}\\&+\int_{\Sigma^*\cap\mathscr{H}^+}\sin\theta d\theta d\phi \left(\left|\mathring{\slashed{\Delta}}F|_{\Sigma^*\cap\mathscr{H}^+}\right|^2+6\left|\mathring{\slashednabla}F|_{\Sigma^*\cap\mathscr{H}^+}\right|^2+8\Big|F|_{\Sigma^*\cap\mathscr{H}^+}\Big|^2\right).
\end{split}
\end{align}
Starting from initial data on $\Sigma$ or $\overline{\Sigma}$ and repeating the computation leading to \bref{+2 norm on H+ beyond B}, the boundary term drops out since we then have
\begin{align}
   \lim_{v\longrightarrow-\infty}F_{\mathscr{H}^+}=\lim_{v\longrightarrow-\infty} \int_v^\infty e^{\frac{1}{2M}(v-\bar{v})} \upalpha_{\mathscr{H}^+}d\bar{v}=\lim_{v\longrightarrow-\infty} -2M \upalpha_{\mathscr{H}^+}=0.
\end{align}
Therefore we have
\begin{align}\label{+2 norm on H+ up to B}
\begin{split}
    \left\|\partial_v \bm{\uppsi}_{\mathscr{H}^+}\right\|_{L^2(\mathscr{H}^+)}^2=&\left\|\mathcal{A}_2(\mathcal{A}_2-2)F_{\mathscr{H}^+}\right\|^2_{L^2(\mathscr{H}^+)}+\left\|6M\partial_v F_{\mathscr{H}^+}\right\|^2_{L^2(\mathscr{H}^+)}.
\end{split}
\end{align}
\begin{align}\label{+2 norm on overline H+ up to B}
\begin{split}
    \left\|\partial_v \bm{\uppsi}_{{\mathscr{H}^+}}\right\|_{L^2(\overline{\mathscr{H}^+})}^2=&\left\|\mathcal{A}_2(\mathcal{A}_2-2)F_{\mathscr{H}^+}\right\|^2_{L^2(\overline{\mathscr{H}^+})}+\left\|6M\partial_v F_{\mathscr{H}^+}\right\|^2_{L^2(\overline{\mathscr{H}^+})}.
\end{split}
\end{align}
\subsubsection{Radiation on $\mathscr{I}^+$}\label{+2 radiation on scri+}
The estimates of \Cref{rp+2} lead us to define a radiation field for $\alpha$ the same way it is defined for $\Psi$
\begin{corollary}\label{psi+2scri1}
    For smooth data of compact support for $\alpha$ on $\Sigma$, $r^5\psi$ has a finite pointwise limit on $\mathscr{I}^+$ which defines a smooth field there.
\end{corollary}
\begin{proof}
We follow step by step the argument of \Cref{RWradscri} and use the estimates of \Cref{T+1rp}.
\end{proof}
Similarly, using \Cref{T+2rp} we have
\begin{corollary}\label{alpha+2scri}
    For smooth data of compact support for $\alpha$ on $\Sigma$, $r^5\alpha$ has a finite pointwise limit on $\mathscr{I}^+$ which defines a smooth field there.
\end{corollary}
For computational convenience we define
\begin{defin}\label{+2 radiation alpha definition scri}
    For a solution $\alpha$ of (\ref{T+2}) arising from smooth data of compact support on $\Sigma^*$ as in \Cref{WP+2Sigma*} or on $\Sigma, \overline{\Sigma}$ as in (\ref{WP+2Sigmabar}), the radiation field of $\alpha$ along $\mathscr{I}^+$ is defined to be the limit $\upalpha_{\mathscr{I}^+}(u,\theta^A)=\lim_{v\longrightarrow\infty} r^5\Omega^{-2}\alpha(u,v,\theta^A)$.\\
    \indent Let $\psi$ be as in \bref{hier+}. We define $\psi_{\mathscr{I}^+}$ to be the limit of $r^5\Omega^{-1}\psi$ as $v\longrightarrow\infty$.
\end{defin}
Repeating the argument of \Cref{RWdecayscri} we have
\begin{proposition}\label{T+1+2scridecay}
For a solution $\alpha$ of (\ref{T+2}) arising from smooth data of compact support on $\Sigma^*$ as in \Cref{WP+2Sigma*} or on $\Sigma, \overline{\Sigma}$ as in (\ref{WP+2Sigmabar}), the radiation fields $\upalpha_{\mathscr{I}^+}$, $\psi_{\mathscr{I}^+}$ and $\bm{\uppsi}_{\mathscr{I}^+}$ decay along $\mathscr{I}^+$ as $u\longrightarrow \infty$.\\
\end{proposition}
\begin{remark}\label{psi+2scrialternative}
We can appeal to an alternative argument that gives the existence of the limits of $r^5\psi$ and $r^5\alpha$ at $\mathscr{I}^+$ without resorting to the hierarchy of $r^p$-estimates as follows:\\
\indent Let $u\geq u_0$. From \Cref{RWradscri} we know that $\Psi$ induces a smooth radiation field $\bm{\uppsi}_{\mathscr{I}^+}$ on $\mathscr{I}^+$. For large enough $v$ the definition of $\psi$ gives
\begin{align}
    r^5\Omega^{-1}\psi=\frac{r^2}{\Omega^2}\Big|_u\int_{u_0}^u \frac{\Omega^2}{r^2}\Psi d\bar{u}.
\end{align}
Therefore
\begin{align}
\begin{split}
\Big|r^5\Omega^{-1}\psi\Big|_{(u,v)}&\leq \sup_{\bar{u}\in[u_0,u]}\left|\Psi|_{(\bar{u},v)}\right|\frac{r^2}{\Omega^2}\int_{u_0}^u\frac{\Omega^2}{r^2} d\bar{u}.
\end{split}
\end{align}
Note that $ \frac{r^2}{\Omega^2}\int_{u_0}^u\frac{\Omega^2}{r^2}$ is uniformly bounded in $v$ for finite $u_0,u$. Since $\Psi$ is also uniformly bounded in $v$ on $[u_0,u]$ we can conclude (say by Lebesgue's bounded convergence theorem) that the pointwise limit $ \lim_{v\longrightarrow\infty} r^5\psi$
exists for any fixed $u$. Note now that (\ref{hier+}) also implies 
\begin{align}\label{+2 Gronwall ingredient}
    \nablau r^5\Omega^{-1}\psi+\frac{3\Omega^2-1}{r} r^5\Omega^{-1}\psi=\Psi.
\end{align}
Then we have
\begin{align}
    \Big|r^5\Omega^{-1}\psi\Big|_{u,v}\leq \int_{u_0}^u d\bar{u} \left|\Psi\right|+\int_{u_0}^ud\bar{u}\left(\frac{3\Omega^2-1}{r}\right)\left|r^5\Omega^{-1}\psi\right|.
\end{align}
We can apply Gr\"onwall's inequality to find:
\begin{align}\label{backwards estimate +2 Gronwall}
    \Big|r^5\Omega^{-1}\psi\Big|_{u,v}\leq\int_{u_0}^ud\bar{u}\left|\Psi\right|\exp\left[\int_{u_0}^u \frac{3\Omega^2-1}{r} ds\right]\lesssim\left|\int_{u_0}^u d\bar{u}\Psi\right|\left(\frac{r(u,v)}{r(u_0,v)}\right)^2.
\end{align}
Thus $r^5\Omega^{-1}\psi$ is uniformly bounded in $v$ on $[u_0,u]$. 
Existence of the $\nablau$ derivatives of the limit of $r^5\psi$ is immediate. Repeating the argument for $r\slashednabla r^5\psi$ gives differentiability in the angular directions.\\
\indent The benefit of the preceding argument is that it allows for a  characterisation of the radiation fields at null infinity that is local in $u$.
\end{remark}
\subsubsection{Radiation flux on $\mathscr{I}^+$}\label{+2 radiation flux on scri+}
The radiation flux on $\mathscr{I}^+$ is easy enough to write down being already in a form that can be computed from the radiation field $\upalpha_{\mathscr{I}^+}$ given the uniform convergence of $r^5\alpha$, $r^5\psi$ and $\Psi$ towards $\mathscr{I}^+$:
\begin{align}\label{Psi out of alpha at scri}
\begin{split}
\bm{\uppsi}_{\mathscr{I}^+}&=(\partial_u)^2 \upalpha_{\mathscr{I}^+},\\
\partial_u\bm{\uppsi}_{\mathscr{I}^+}&=(\partial_u)^3\upalpha_{\mathscr{I}^+}.
\end{split}
\end{align}
\section{Future asymptotics of the $-2$ Teukolsky equation}\label{section 7}
\Cref{section 7} is devoted to the study of future radiation fields induced by solutions to the $+2$ Teukolsky equation arising from smooth, compactly supported data on $\Sigma^*$, as was done for the $+2$ Teukolsky equation in \Cref{section 6} and to the Regge--Wheeler equation in \Cref{subsection 5.2 subsection Radiation fields}.\\
\indent We first gather the estimates we need in \Cref{subsection 7.1 integrated boundedness and decay estimates for -2}, where we collect results from \cite{DHR16} estimating $\underline\alpha$ from $\underline\Psi$ defined via (\ref{hier-}) and the estimates of \Cref{subsection 5.1 Basic integrated boundedness and decay estimates} for $\underline\Psi$. We apply these results to study the future radiation fields and their fluxes in \Cref{subsection 7.2 future radiation fields and fluxes}. The estimates of \cite{DHR16} collected in \Cref{subsection 7.1 integrated boundedness and decay estimates for -2} will be sufficient to construct and estimate the radiation fields on $\mathscr{H}^+$ and $\mathscr{I}^+$.
\subsection{Integrated boundedness and decay for $\underline\alpha$ via $\underline\Psi$}\label{subsection 7.1 integrated boundedness and decay estimates for -2}
We begin with the following basic proposition, already proven in \Cref{Chandra}:
\begin{proposition}\label{-2 implies RW}
Let $(\underline\upalpha,\underline\upalpha')$ be data for \cref{T-2} on $\Sigma^*$, $\Sigma$ or $\overline{\Sigma}$ as in \Cref{WP-2Sigma*,,WP-2Sigmabar} respectively. Then $\underline\Psi$ defined out of the solution $\underline\alpha$ on $ J^+(\Sigma^*)$, $ J^+(\Sigma)$ or $ J^+(\overline{\Sigma})$ satisfies \cref{RW}.
\end{proposition}
Throughout this section we focus on the case of data on $\Sigma^*$:
\begin{proposition}\label{-2psiILED}
Let $\underline\alpha$ be a solution to (\ref{T-2}) and $\underline\Psi, \underline\psi$ be as in (\ref{hier-}) and \Cref{-2 implies RW}. Then for any $u$ and any $v>0$ such that $(u,v,\theta^A)\in J^+(\Sigma^*)$, the following estimate holds: 
\begin{align}
\begin{split}
    \int_{\mathscr{D}^{u,v}_{\Sigma^*}} \Omega^2 d\bar{u}d\bar{v}d\omega\;r^{4}\Omega^{-2} |\underline\psi|^2+&\int_{\underline{\mathscr{C}}_v\cap J^+(\Sigma^*)\cap J^-(\mathscr{C}_u)}\Omega^2 d\bar{u}d\omega\;r^6\Omega^{-2}|\underline\psi|^2\\
    &\lesssim \mathbb{F}_{\Sigma^*}[\underline\Psi]+\int_{\Sigma^*\cap J^-(\mathscr{C}_u)\cap J^-(\underline{\mathscr{C}}_v)}drd\omega\; r^6\Omega^{-2}|\underline\psi|^2.
\end{split}
\end{align}
\end{proposition}
\begin{proof}
The definition of $\underline\psi$ (\ref{hier-}) and Cauchy--Schwarz imply
\begin{align}
   \partial_v [r^{6}\Omega^{-2}|\underline\psi|^2]+M r^{4}\Omega^{-2}|\underline\psi|^2\leq  \frac{1}{Mr^2}|\underline\Psi|^2.
\end{align}
The result follows by integrating over $\mathscr{D}^{u,v}_{\Sigma^*}$.
\end{proof}
\begin{proposition}\label{-2alphaILED}
Let $\underline\alpha$ be a solution to (\ref{T-2}) and $\underline\Psi, \underline\psi$ be as in (\ref{hier-}) and \Cref{-2 implies RW}. Then for any $u$ and any $v>0$ such that $(u,v,\theta^A)\in J^+(\Sigma^*)$, the following estimate holds for sufficiently small $\epsilon>0$: 
\begin{align}
\begin{split}
    \int_{\mathscr{D}^{u,v}_{\Sigma^*}}  \Omega^2 d\bar{u}d\bar{v}d\omega\;\Omega^{-4}|\underline\alpha|^2+&\int_{\underline{\mathscr{C}}_v\cap J^+(\Sigma^*)\cap J^-(\mathscr{C}_u)}\Omega^2 d\bar{u}d\omega\;r^2\Omega^{-4}|\underline\alpha|^2\\&\lesssim \mathbb{F}_{\Sigma^*}[\underline\Psi]+\int_{\Sigma^*\cap J^-(\mathscr{C}_u\cap J^-(\underline{\mathscr{C}}_v)}drd\omega\;r^6\Omega^{-2}|\underline\psi|^2+ r^2\Omega^{-4}|\underline\alpha|^2.
\end{split}
\end{align}
\end{proposition}
\begin{proof}
Similar to \Cref{-2psiILED}. See Propositions 12.1.2, 12.2.6 and 12.2.7 of \cite{DHR16}.
\end{proof}
\begin{proposition}
Let $\underline\alpha$ be a solution to (\ref{T-2}) and $\underline\Psi, \underline\psi$ be as in (\ref{hier-}) and \Cref{-2 implies RW}. Then for any $u$ and any $v>0$ such that $(u,v,\theta^A)\in J^+(\Sigma^*)$, the following estimate holds: 
\begin{align}\label{2ndderivativeofpsibar}
 \begin{split}
        \int_{\underline{\mathscr{C}}_v\cap J^+(\Sigma^*)\cap J^-(\mathscr{C}_u)}\Omega^2 d\bar{u}d\omega\;\left|-2r^2\fancydstar_2\fancyd_2(r^3\Omega^{-1}\underline\psi)\right|^2&+\int_{\mathscr{D}^{u,v}_{{\Sigma^*}}} d\bar{u}d\bar{v}d\omega\;\frac{\Omega^2}{r^3}\left(1-\frac{3M}{r}\right)^2 \left|-2r^2\fancydstar_2\fancyd_2(r^3\Omega\underline\psi)\right|^2 \\ &\lesssim \mathbb{F}_{{\Sigma^*}}[\underline\Psi]+\int_{{\Sigma^*}}drd\omega\;r^6\Omega^{-2}\left|\underline\psi\right|^2+r^2\Omega^{-4}\left|\underline\alpha\right|^2.
        \end{split}
    \end{align}
\end{proposition}
\begin{proposition}
Let $\underline\alpha$ be a solution to (\ref{T-2}) and $\underline\Psi, \underline\psi$ be as in (\ref{hier-}) and \Cref{-2 implies RW}. Then for any $u$ and any $v>0$ such that $(u,v,\theta^A)\in J^+(\Sigma^*)$. For sufficiently small $\epsilon>0$ the following estimate holds: 
\begin{align}
    \int_{\mathscr{D}^{u,v}_{{\Sigma^*}}} \Omega^2 d\bar{u}d\bar{v}d\omega\; r^{5-\epsilon} \Omega^{-2}|r\fancyd_2\underline\psi|^2 &\lesssim \mathbb{F}_{{\Sigma^*}}[\underline\Psi]+\int_{{\Sigma^*}}drd\omega\;r^{6-\epsilon}\Omega^{-2}\left[|r\fancyd_2\underline\psi|^2+|\underline\psi|^2\right]+r^{6-\epsilon}\Omega^{-4}|\underline\alpha|^2.
\end{align}
\end{proposition}
\begin{proposition}
Let $\underline\alpha$ be a solution to (\ref{T-2}) and $\underline\Psi, \underline\psi$ be as in (\ref{hier-}) and \Cref{-2 implies RW}. Then for any $u$ and any $v>0$ such that $(u,v,\theta^A)\in J^+(\Sigma^*)$, the following estimate holds: 
\begin{align}
\begin{split}
    \int_{\underline{\mathscr{C}}_v\cap J^+(\Sigma^*)\cap J^-(\mathscr{C}_u)}\Omega^2 d\bar{u}&d\omega\;r^6|\Omega^{-1}\nablagml(\Omega^{-1}\underline\psi)|^2+\int_{\mathscr{D}^{u,v}_{\Sigma^*}}\Omega^2 d\bar{u}d\bar{v}d\omega\; r^4\left[|\Omega^{-1}\nablagml(\Omega^{-1}\underline\psi)|^2+|r\nablav(\Omega^{-1}\underline\psi)|^2\right]\\& \lesssim \mathbb{F}_{{\Sigma^*}}[\underline\Psi]+\int_{{\Sigma^*}}drd\omega\;r^4\Omega^{-2}\left[|\underline\psi|^2+|r\fancyd_2\underline\psi|^2+|\Omega^{-1}\nablagml(\Omega^{-1}\underline\psi)|^2+|r\nablav(\Omega^{-1}\underline\psi)|^2\right].
\end{split}
\end{align}
\end{proposition}
\begin{proposition}
Let $\underline\alpha$ be a solution to (\ref{T-2}) and $\underline\Psi, \underline\psi$ be as in (\ref{hier-}) and \Cref{-2 implies RW}. Then for any $u$ and any $v>0$ such that $(u,v,\theta^A)\in J^+(\Sigma^*)$, the following estimate holds:
\begin{align}
\begin{split}
    &\int_{\underline{\mathscr{C}}_v\cap J^+(\Sigma^*)\cap J^-(\mathscr{C}_u)}\Omega^2 d\bar{u}d\omega\;|r^2\fancydstar_2\fancyd r\Omega^{-2}\underline\alpha|^2+\int_{\mathscr{D}_{{\Sigma^*}}^{u,v}} \Omega^2 d\bar{u}d\bar{v}d\omega\;|r^2\fancydstar_2\fancyd_2\Omega^{-2}\underline\alpha|^2
    \\& \lesssim \mathbb{F}_{{\Sigma^*}}[\underline\Psi]+\int_{{\Sigma^*}}drd\omega\;r^4\Omega^{-2}\left[|\underline\psi|^2+|r\fancyd_2\underline\psi|^2+|\Omega^{-1}\nablagml(\Omega^{-1}\underline\psi)|^2+|r\nablav(\Omega^{-1}\underline\psi)|^2\right]+\int_{{\Sigma^*}} drd\omega\;|r\Omega^{-2}\underline\alpha|^2.
\end{split}
\end{align}
\end{proposition}
\begin{proposition}
Let $\underline\alpha$ be a solution to (\ref{T-2}) and $\underline\Psi, \underline\psi$ be as in (\ref{hier-}) and \Cref{-2 implies RW}. Then for any $u$ and any $v>0$ such that $(u,v,\theta^A)\in J^+(\Sigma^*)$, the following estimate holds for sufficiently small $\epsilon>0$: 
\begin{align}
\begin{split}
    &\int_{\underline{\mathscr{C}}_v\cap J^+(\Sigma^*)\cap J^-(\mathscr{C}_u)}\Omega^2 d\bar{u}d\omega\; \left[|r\Omega^{-2}\underline\alpha|^2+|r\fancyd_2r\Omega^{-2}\underline\alpha|^2+|\Omega^{-1}\nablagml r\Omega^{-2}\underline\alpha|^2\right]\\&+\int_{\mathscr{D}_{\Sigma^*}^{u,v}}\Omega^2 d\bar{u}d\bar{v}d\omega\; \left[|\Omega^{-2}\underline\alpha|^2+|r\fancyd_2\Omega^{-2}\underline\alpha|^2+|\Omega^{-1}\nablagml \Omega^{-2}\underline\alpha|^2\right]\\& \lesssim \mathbb{F}_{{\Sigma^*}}[\underline\Psi]+\int_{{\Sigma^*}}drd\omega\;r^6\left[|\Omega^{-1}\underline\psi|^2+|r\fancyd_2\Omega^{-1}\underline\psi|^2+|\Omega^{-1}\nablagml(\Omega^{-1}\underline\psi)|^2\right]\\& +\int_{{\Sigma^*}}drd\omega\; r^2\left[|\Omega^{-2}\underline\alpha|^2+|r\fancyd_2\Omega^{-2}\underline\alpha|^2+|\Omega^{-1}\nablagml\Omega^{-2}\underline\alpha|^2\right].
\end{split}
\end{align}
\end{proposition}
\subsection{Future radiation fields and fluxes}\label{subsection 7.2 future radiation fields and fluxes}
In this section the notion of future radiation fields of solutions to the -2 Teukolsky equation \bref{T-2} is defined, and some of the properties of these radiation fields are studied, in particular obtaining their $\mathcal{E}^{T,-2}_{\mathscr{H}^+}$, $\mathcal{E}^{T,-2}_{\mathscr{I}^+}$ fluxes when they belong to solutions of \bref{T-2} arising from smooth data of compact support.
\subsubsection{Radiation on $\mathscr{H}^+$}\label{-2 radiation on H+}
\begin{defin}\label{-2 radiation alpha definition H}
    Let $\underline\alpha$ be a solution to \cref{T-2} arising from smooth data as in \Cref{WP-2Sigma*}. The radiation field of $\underline\alpha$ along $\mathscr{H}^+_{\geq0}$, denoted $\underline\upalpha_{\mathscr{H}^+}$, is defined to be the restriction of $2M\Omega^{-2}\underline\alpha$ to $\mathscr{H}^-$.
\end{defin}
\begin{defin}\label{-2 radiation alpha definition open H}
    Let $\underline\alpha$ be a solution to \cref{T-2} arising from smooth data which is compactly supported on $\Sigma$ according to \Cref{WP-2Sigmabar}. The radiation field of $\underline\alpha$ along $\mathscr{H}^+_{\geq0}$, denoted $\underline\upalpha_{\mathscr{H}^+}$, is defined to be the restriction of $2M\Omega^{-2}\underline\alpha$ to $\mathscr{H}^-$.
\end{defin}
\begin{defin}\label{-2 radiation alpha definition overline H}
    Let $\underline\alpha$ be a solution to \cref{T-2} arising from smooth data as in \Cref{WP-2Sigmabar}. The radiation field of $\underline\alpha$ along $\overline{\mathscr{H}^+}$, denoted $\underline\upalpha_{{\mathscr{H}^+}}$, is defined by $V^2\underline\upalpha_{{\mathscr{H}^+}}=2MV^2\Omega^{-2}\underline\alpha|_{{\mathscr{H}^+}}$.
\end{defin}
\begin{remark}
We will use the same notation for the radiation field on $\mathscr{H}^+_{\geq0}, \mathscr{H}^+$ or $\overline{\mathscr{H}^+}$.
\end{remark}
The following applies equally to radiation fields on $\mathscr{H}^+_{\geq0}$, $\mathscr{H}^+$ and $\overline{\mathscr{H}^+}$.
\begin{proposition}\label{-2 radiation ptwise decay H}
Assume $\underline\alpha$ arises from data which is supported away from $i^0$, then $\lim_{v\longrightarrow\infty}\underline{\bm{\uppsi}}_{\mathscr{H}^+}=~0$.
\end{proposition}
\begin{proof}
The flux estimate of \Cref{-2psiILED} commuted with $\mathcal{L}_T$ implies 
\begin{align}
    \int_{v_0}^\infty d\bar{v}d\omega\; \left|\Omega^{-1}\underline\psi\right|^2+ \left|\slashednabla_T\Omega^{-1}\underline\psi\right|^2 <\infty.
\end{align}
This implies $||\Omega^{-1}\underline\psi||_{S^2_{\infty,v}}\longrightarrow0$ as $v\longrightarrow\infty$. A further Sobolev embedding on the sphere gives the result.
\end{proof}
Similarly, we have
\begin{proposition}
Assume $\underline\alpha$ arises from data which is supported away from $i^0$, then $\lim_{v\longrightarrow\infty}\underline\upalpha_{\mathscr{H}^+}=0$.
\end{proposition}
\subsubsection{Radiation flux on $\mathscr{H}^+$}\label{-2 radiation flux on H+}
Now we can calculate the radiation energies in terms of $\underline\alpha$. We want to rewrite
\begin{align}
\Omega\slashed{\nabla}_4\underline\Psi=\nablav\left(\frac{r^2}{\Omega^2}\nablav\right)^2r\Omega^2\underline\alpha
\end{align}
in terms of $\Omega^{-2}\underline\alpha$ and $\Omega^{-1}\underline\psi$. We have for $\underline\psi$
\begin{align}
\begin{split}
r^3\Omega^{-1}\underline\psi&=\frac{r^2}{\Omega^4}\nablav r\Omega^2\underline\alpha=\frac{r^2}{\Omega^4}\nablav r\Omega^4 \Omega^{-2}\underline\alpha
\\&=r^2(2-\Omega^2)\Omega^{-2}\underline\alpha+r^3\nablav \Omega^{-2}\underline\alpha.
\end{split}
\end{align}
We can write for $\underline\Psi$
\begin{align}\label{Psi H+}
\begin{split}
\underline\Psi&=\frac{r^2}{\Omega^2}\nablav r^3\Omega\underline\psi=2M r^3\Omega^{-1}\underline\psi+r^2\nablav r^3\Omega^{-1}\underline\psi
\\&=2r^3\Omega^{-2}\underline\alpha+r^4(3+\Omega^2)\nablav\Omega^{-2}\underline\alpha+r^5(\nablav)^2\Omega^{-2}\underline\alpha.
\end{split}
\end{align}
We can write for $\nablav\underline\Psi$
\begin{align}\label{nablav Psi H+}
\begin{split}
\nablav \underline\Psi=&6r^2\Omega^2\Omega^{-2}\underline\alpha+r^3(2+13\Omega^2+3\Omega^4)\nablav\Omega^{-2}\underline\alpha
\\&+3r^4(1+2\Omega^2)(\nablav)^2\Omega^{-2}\underline\alpha+r^5(\nablav)^3\Omega^{-2}\underline\alpha.
\end{split}
\end{align}
At $\mathscr{H}^+$ \bref{Psi H+}, \bref{nablav Psi H+} become
\begin{align}\label{-2 Psi out of alpha H+}
\underline{\bm{\uppsi}}_{\mathscr{H}^+}=(2M)^2\left[2\underline\upalpha_{\mathscr{H}^+}+6M\partial_v\underline\upalpha_{\mathscr{H}^+}+(2M)^2\partial_v^2\underline\upalpha_{\mathscr{H}^+}\right],
\end{align}
\begin{align}\label{-2 expression on H is regular}
\nablav\underline{\bm{\uppsi}}_{\mathscr{H}^+}=(2M)\left[4M\partial_v\underline\upalpha_{\mathscr{H}^+}+3(2M)^2\partial_v^2\underline\upalpha_{\mathscr{H}^+}+(2M)^3\partial_v^3\underline\upalpha_{\mathscr{H}^+}\right].
\end{align}
\begin{remark}
On $\mathcal{E}^{T,+2}_{\mathscr{H}^-}$, the norm $\|\;\|_{\mathcal{E}^{T,+2}_{\mathscr{H}^+}}$ is equal to
\begin{align}
    \|A\|_{\mathcal{E}^{T,+2}_{\mathscr{H}^+}}^2=\|2(2M\partial_v) A\|^2_{L^2(\mathscr{H}^+)}+\|3(2M\partial_v)^2 A\|^2_{L^2(\mathscr{H}^+)}+\|(2M\partial_v)^3 A\|^2_{L^2(\mathscr{H}^+)}.
\end{align}
while for $\|\;\|_{\mathcal{E}^{T,+2}_{\mathscr{H}^+_{\geq0}}}$ we have
\begin{align}
\begin{split}
    \|A\|_{\mathcal{E}^{T,+2}_{\mathscr{H}^+_{\geq0}}}^2&=\|2(2M\partial_v) A\|^2_{L^2(\mathscr{H}^+_{\geq0})}+\|3(2M\partial_v)^2 A\|^2_{L^2(\mathscr{H}^+_{\geq0})}+\|(2M\partial_v)^3 A\|^2_{L^2(\mathscr{H}^+_{\geq0})}\\
    &-6\|(2M\partial_v)A\|_{L^2(S^2_{\infty,0})}^2-3\|(2M\partial_v)^2A\|_{L^2(S^2_{\infty,0})}^2.
\end{split}
\end{align}
If the same computation for $\|\;\|_{\mathcal{E}^{T,+2}_{\overline{\mathscr{H}^+}}}$ is done with terms expressed in the Eddington--Finkelstein coordinates, it produces boundary terms that are not regular near $\mathcal{B}$. The expression \bref{-2 expression on H is regular} for $\underline\Psi$ remains well-defined over $\mathscr{H}^+$ for data on $\overline{\Sigma}$ and has a finite limit at $\mathcal{B}$, as we can see by writing it in terms of the regular Kruskal coordinates:
\begin{align}
    \|\underline\upalpha_{\mathscr{H}^+}\|_{\mathcal{E}^{T,+2}_{\overline{\mathscr{H}^+}}}^2=\|V^{1/2}\partial_V^3 V^{-2}\underline\upalpha_{\mathscr{H}^+}\|_{L^2_VL^2(S^2_{\infty,v})}^2.
\end{align}
For smooth initial data on $\overline{\Sigma}$, \Cref{WP-2Sigmabar} guarantees the continuity of $V^2\Omega^{-2}\underline\alpha$ in a neighborhood of $\mathcal{B}$, and in the backwards direction we can show the same with \Cref{backwards wellposedness -2} and \Cref{WP-2Sigma*}.
\end{remark}                 
\subsubsection{Radiation on $\mathscr{I}^+$}\label{-2 radiation asymptotics on scri+}
\begin{proposition}\label{-2psiscri}
Let $\underline\alpha$ be a solution to (\ref{T-2}) arising from smooth compactly supported data on $\Sigma^*$ and let $\underline\psi,\underline\Psi$ be as in (\ref{hier-}). Then $r^3\underline\psi$ has a uniform smooth limit towards $\mathscr{I}^+$
\end{proposition}
\begin{proof}
We can integrate the definition of $\underline\Psi$ from (\ref{hier-}) from $r=R$ towards $\mathscr{I}^+$:
\begin{align}\label{above136 psi}
    r^3\Omega\underline\psi|_{u,v}=r^3\Omega\underline\psi|_{u,v(u,R)}+\int_{v(u,R)}^v \frac{\Omega^2}{r^2} \underline\Psi.
\end{align}
Note that Cauchy--Schwarz and Hardy's inequality applied to the integral term give
\begin{align}
    \left[\int_{S^2}d\omega\int_{v(u,R)}^vd\bar{v} \left|\frac{\Omega^2}{r^2} \underline\Psi\right|\right]^2\leq \frac{1}{R} \int_{\mathscr{C}_u\cap\{r>R\}}d\bar{v}d\omega\;\frac{\Omega^2}{r^2} \left|\underline\Psi\right|^2\lesssim\frac{1}{R}\int_{\mathscr{C}_u\cap\{r>R\}}d\bar{v}d\omega\; |\nablav\underline\Psi|^2,
\end{align}
which is finite for data of compact support. We can repeat this estimate for $r\slashed{\nabla}\underline\Psi$ conclude with a Sobolev embedding on the sphere that the integral on the right hand side of (\ref{above136 psi}) is bounded. The dominated convergence theorem gives the result. \Cref{RWrp} tells us that the convergence is uniform in $u$. Finally, we can repeat the argument having commuted with $\mathcal{L}_T, \mathcal{L}_{\Omega^i}$ to show that the limit is smooth.
\end{proof}
Similarly, \cref{eq:d3d3psibar} gives us
\begin{proposition}\label{-2 radiation at scri}
Let $\underline\alpha$ be a solution to (\ref{T-2}) arising from smooth compactly supported data on $\Sigma^*$ and let $\underline\psi$ be as in (\ref{hier-}). Then $r\underline\alpha$ has a uniform smooth limit $\underline{\upalpha}_{\mathscr{I}^+}$ towards $\mathscr{I}^+$.
\end{proposition}
\begin{proof}
We can again integrate the definition of $\underline\psi$ from (\ref{hier-}) from $r=R$ towards $\mathscr{I}^+$:
\begin{align}\label{above136 alpha}
    r\Omega^2\underline\alpha|_{u,v}=r\Omega^2\underline\alpha|_{u,v(u,R)}+\int_{v(u,R)}^vd\bar{v} \frac{\Omega^2}{r^2}r^3\Omega\underline\psi.
\end{align}
Hardy's inequality gives us
\begin{align}
    \int_{\mathscr{C}_u\cap\{r>R\}}d\bar{v}d\omega\;\frac{\Omega^2}{r^2}\left|r^3\Omega\underline\psi\right|^2\lesssim \int_{\mathscr{C}_u\cap\{r>R\}}d\bar{v}d\omega\;|\nablav r^3\Omega\underline\psi|^2=\int_{\mathscr{C}_u\cap\{r>R\}} d\bar{v}d\omega\;\frac{\Omega^2}{r^2}|\underline\Psi|^2.
\end{align}
We can conclude using the above and repeating the proof of \Cref{-2psiscri}.
\end{proof}
\begin{remark}
    In particular, $\slashednabla_T r\underline\alpha$ attains a limit towards $\mathscr{I}^+$ which is smooth and $\lim_{v\longrightarrow\infty} \slashednabla_T r\underline\alpha=\partial_u \underline\upalpha_{\mathscr{I}^+}$.
\end{remark}
\begin{remark}
Instead of resorting to commutation with $\mathcal{L}_T, \mathcal{L}_{\Omega^i}$ directly, one could employ the hierarchy of (\ref{eq:d3psibar}) and (\ref{eq:d3d3psibar}) to estimate the derivatives of $\underline\psi$ and $\underline\alpha$ one by one with a smaller loss of derivatives, see \cite{DHR16}.
\end{remark}
\begin{defin}\label{-2 definition radiation at scri}
    For a solution $\underline\alpha$ of (\ref{T-2}) arising from smooth data of compact support on $\Sigma^*$ according to \Cref{WP-2Sigma*} or on $\Sigma, \overline{\Sigma}$ as in \Cref{WP-2Sigmabar}, the radiation field of $\underline\alpha$ along $\mathscr{I}^+$ is defined by $\underline\upalpha_{\mathscr{I}^+}(u,\theta^A)=\lim_{v\longrightarrow \infty} r\underline\alpha(u,v,\theta^A)$.
\end{defin}
\begin{proposition}\label{-2 psi ptwisedecay}
Let $\underline\alpha$ be a solution to (\ref{T-2}) arising from smooth compactly supported data on $\Sigma^*$ and let $\underline\psi$ be as in (\ref{hier-}). Then $\psi|_{r=R}$ decays as $t\longrightarrow\infty$.
\end{proposition}
\begin{proof}
The estimate of \Cref{-2psiILED} applied to $r<R$ for some fixed $R<\infty$, commuted with $T$ gives
\begin{align}
    \lim_{v\longrightarrow \infty} \int_{\underline{\mathscr{C}}_v\cap\{2M<r<R\}}dud\omega\; \left|\Omega^{-1}\psi\right|=0.
\end{align}
Commuting with $\Omega^{-1}\nablagml$ gives the result.
\end{proof}
\begin{corollary}\label{-2 alpha ptwisedecay}
Let $\underline\alpha$ be a solution to (\ref{T-2}) arising from smooth compactly supported data on $\Sigma^*$ and let $\underline\Psi$ be as in (\ref{hier-}). Then $\alpha|_{r=R}$ decays as $t\longrightarrow\infty$.
\end{corollary}
\begin{proposition}\label{-2 psi radiation decay}
Let $\underline\alpha$ be a solution to (\ref{T-2}) arising from smooth compactly supported data on $\Sigma^*$ and let $\underline\psi$ be as in (\ref{hier-}). Then $\underline{\psi}_{\mathscr{I}^+}:=\lim_{v\longrightarrow\infty}r^3\underline\psi$ decays towards the future end of $\mathscr{I}^+$.
\end{proposition}
\begin{proof}
This follows from integrating (\ref{hier-}) between $r=R$ and $\mathscr{I}^+$:
\begin{align}
    \int_{S^2_R}\left|\frac{1}{r^2}r^3\Omega\underline{\psi}|_{(u,v)}-\underline{{\psi}}_{\mathscr{I}^+}\big|_{u}\right|_{S^2}^2\lesssim \frac{1}{R}\int_{\mathscr{C}_u\cap\{r>R\}} |\nablav\underline\Psi|^2.
\end{align}
This decays as $u\longrightarrow\infty$ by energy conservation. \Cref{-2 psi ptwisedecay} gives the result.
\end{proof}
\begin{corollary}\label{-2 alpha radiation decay}
Let $\underline\alpha$ be a solution to (\ref{T-2}) arising from smooth compactly supported data on $\Sigma^*$ and let $\underline\psi$ be as in (\ref{hier-}). Then the radiation field $\upalpha_{\mathscr{I}^+}$ of \Cref{-2 definition radiation at scri} decays towards $\mathscr{I}^+_+$
\end{corollary}
\subsubsection{Radiation flux on $\mathscr{I}^+$}\label{-2 radiation flux on scri+}
We want to find the limit towards $\mathscr{I}^+$ of
\begin{align}\label{eq:116}
\Omega\slashed{\nabla}_3\underline\Psi=-(3\Omega^2-1)\frac{r^2}{\Omega^2}\Omega\slashed{\nabla}_4r\Omega^2\underline\alpha+6Mr\Omega^2\underline\alpha-2r^2\slashed{\mathcal{D}}^*_2\slashed{\mathcal{D}}_2\frac{r^2}{\Omega^2}\Omega\slashed{\nabla}_4r\Omega^2\underline\alpha.
\end{align}
As $\underline\psi$ is related to the transverse derivative of $\underline\alpha$ near $\mathscr{I}^+$, we want to express $\frac{r^2}{\Omega^2}\nablav r\Omega^2\underline\alpha$ in terms of quantities that can be constructed intrinsically on $\mathscr{I}^+$ from data. We do this by integrating the Teukolsky equation: recall \cref{eq:d3d3psibar}
\begin{align}\label{this22}
\frac{\Omega^2}{r^2}\Omega\slashed{\nabla}_3\frac{r^2}{\Omega^2}\Omega\slashed{\nabla}_3\underline\Psi=6M\left[\Omega\slashed{\nabla}_4+\Omega\slashed{\nabla}_3\right]r\Omega^2\underline\alpha+\mathcal{A}_2(\mathcal{A}_2-2)r\Omega^2\underline\alpha.
\end{align}
The results of the previous section give us the asymptotics:
\begin{align}
\frac{\Omega^2}{r^2}\nablau\frac{r^2}{\Omega^2}\nablau\underline\Psi=(\nablau)^2\underline\Psi-\left(\frac{3\Omega^2-1}{r}\right)\nablau \underline\Psi\longrightarrow (\partial_u)^2\underline\Psi \text{\;\;towards\;} \mathscr{I}^+.
\end{align}
The right hand side gives:
\begin{align}
6M\partial_u \underline\upalpha_{\mathscr{I}^+}+\mathcal{A}_2\left(\mathcal{A}_2-2\right) \underline\upalpha_{\mathscr{I}^+}.
\end{align}
whereas the left hand side becomes $\partial_u^2\underline{\bm{\uppsi}}_{\mathscr{I}^+}$. \bref{this22} then becomes at $\mathscr{I}^+$ 
\begin{align}
\partial_u^2\underline\Psi=6M\partial_u \underline\upalpha_{\mathscr{I}^+}+\mathcal{A}_2\left(\mathcal{A}_2-2\right) \underline\upalpha_{\mathscr{I}^+}.
\end{align}
We can integrate along $\mathscr{I}^+$:
\begin{align}\label{202}
\partial_u \underline\Psi|_u=\partial_u\underline\Psi|_{u_0}-6M\underline\upalpha_{\mathscr{I}^+}|_{u_0}+6M\underline\upalpha_{\mathscr{I}^+}|_{u}+\mathcal{A}_2\left(\mathcal{A}_2-2\right)\int_{u_0}^u \underline\upalpha_{\mathscr{I}^+} d\bar{u}.
\end{align}
The fact that $\lim_{u\longrightarrow\infty} \partial_u \underline{\bm{\uppsi}}_{\mathscr{I}^+}=0=\lim_{u\longrightarrow\infty}  \underline\upalpha_{\mathscr{I}^+}$ tells us that 
\begin{align}
\mathcal{A}_2\left(\mathcal{A}_2-2\right)\int_{u_0}^\infty r\underline\alpha=-\partial_u\underline{\bm{\uppsi}}_{\mathscr{I}^+}|_{u_0}+6M\underline\upalpha_{\mathscr{I}^+}|_{u_0}.
\end{align}
For data of compact support on $\Sigma$, we can take $u_0$ such that the right hand side vanishes. Knowing that $\mathcal{A}_2, \mathcal{A}_2-2$ are uniformly elliptic, we must have
\begin{align}\label{-2 mean is 0}
\int_{u_0}^\infty \underline\upalpha_{\mathscr{I}^+}=0.
\end{align}   
We can integrate (\ref{202}) once more to find a useful expression for $\underline{\bm{\uppsi}}_{\mathscr{I}^+}$ that can be computed from data on $\mathscr{I}^+$:
\begin{align}\label{-2 Psi out of alpha on scri+}
\underline{\bm{\uppsi}}_{\mathscr{I}^+}(u,\theta^A)=6M\int_{u_0}^u d\bar{u}\underline\upalpha_{\mathscr{I}^+}+\mathcal{A}_2\left(\mathcal{A}_2-2\right)\int_{u_0}^u d\bar{u}(u-\bar{u})\underline\upalpha_{\mathscr{I}^+}.
\end{align}
Again, seeing that $\underline\Psi|_{\mathscr{I}^+}$ decays towards $\mathscr{I}^+_+$ we have:
\begin{align}
\int_{u_0}^\infty\int_{u_1}^\infty du_1du_2\underline\upalpha_{\mathscr{I}^+}=\int_{u_0}^\infty d\bar{u}(u-\bar{u})r\underline\alpha=0.
\end{align}
%+(u-u_0)\left[\nablau\underline\Psi_{u_0}+3\mathcal{C}r\underline\alpha_{u_0}\right]
%
%and noting ($122$) suggests we must have $(M>0)$:
%\begin{align}
%\int_{u_0}^\infty d\bar{u} r\underline\alpha =-\frac{1}{6M}\underline\Psi|_{u_0}
%\end{align}
We can rewrite $\underline{\bm{\uppsi}}_{\mathscr{I}^+}$ and $\partial_u\underline{\bm{\uppsi}}_{\mathscr{I}^+}$:
\begin{align}\label{formula for -2 RW in backwards direction}
\underline{\bm{\uppsi}}_{\mathscr{I}^+}=-6M\int_u^{\infty}d\bar{u} \underline\upalpha_{\mathscr{I}^+}-\mathcal{A}_2\left(\mathcal{A}_2-2\right)\int_u^{\infty}d\bar{u}(u-\bar{u})\underline\upalpha_{\mathscr{I}^+}.
\end{align}
\begin{align}\label{formula for partialu -2 RW in backwards direction}
\partial_u\underline{\bm{\uppsi}}_{\mathscr{I}^+}=-\mathcal{A}_2\left(\mathcal{A}_2-2\right)\int_u^{\infty} d\bar{u} \underline\upalpha_{\mathscr{I}^+}+6M\underline\upalpha_{\mathscr{I}^+}|_u.
\end{align}
Using \bref{-2 mean is 0}, we can recover \bref{-2 tricky norm at scri}
\begin{align}
    \|\partial_u\underline{\bm{\uppsi}}_{\mathscr{I}^+}\|_{L^2(\mathscr{I}^-)}^2=\int_{\mathscr{I}^+} d{u}\sin\theta d\theta d\phi\left[ 6M|\underline{\alpha}_{\mathscr{I}^+}|^2+\left|\mathcal{A}_2(\mathcal{A}_2-2)\int_{\bar{u}}^\infty d\bar{u}\; \underline{\alpha}_{\mathscr{I}^+}\right| ^2\right].
\end{align}
\begin{remark}The fact that $\int_{-\infty}^\infty du_1\; \underline{\bm{\uppsi}}_{\mathscr{I}^+}=\int_{-\infty}^\infty\int_{u_1}^\infty du_1 du_2 \;\underline{\bm{\uppsi}}_{\mathscr{I}^+}=0$ implies
\begin{align}
    \int_{-\infty}^\infty\int_{u_1}^\infty\int_{u_2}^\infty du_1du_2du_3\;\underline\upalpha_{\mathscr{I}^+}=\int_{u_0}^\infty\int_{u_1}^\infty\int_{u_2}^\infty\int_{u_3}^\infty du_1du_2du_3du_4\;\underline\upalpha_{\mathscr{I}^+}=0.
\end{align}
\end{remark}
\section{Constructing the scattering maps for $\alpha, \underline\alpha$}\label{section 8 constructing the scattering maps}
We gather the results of Sections \ref{section 6} and \ref{section 7} to finally construct the scattering theory for the Teukolsky equations \bref{T+2}, \bref{T-2}. \Cref{subsection 8.1 future scattering +2} is devoted to the +2 Teukolsky equation \bref{T+2}, where \Cref{subsubsection 8.1.1 forwards scattering +2} handles forwards scattering and \Cref{subsubsection 8.1.2 backwards scattering +2} handles backwards scattering. \Cref{subsection 8.2 future scattering -2} is devoted to the -2 Teukolsky equation \bref{T-2}, where \Cref{subsubsection 8.2.1 forwards scattering -2} handles forwards scattering and \Cref{subsubsection 8.2.2 backwards scattering -2} handles backwards scattering. Taking into account \Cref{time inversion}, results concerning scattering towards the past are immediate and they are collected in \Cref{subsection 8.3 past scattering +2-2}.
\subsection{Future scattering for $\alpha$}\label{subsection 8.1 future scattering +2}
Forwards scattering for the +2 Teukolsky equation \bref{T+2} is worked out entirely analogously to the case of the Regge--Wheeler equation \bref{RW}, using the results of Section \ref{+2 radiation}.\\
\indent For backwards scattering, we make use of the transport equations \bref{hier+} and the backwards scattering theory of \Cref{subsection 5.2 subsection Radiation fields} for the Regge--Wheeler equation \bref{RW}, instead of directly appealing to a limiting argument that repeats the proof of \Cref{RWbackwardsexistence}. Throughout this process, the uniform $T$-energy estimates of $\Psi$ are vital in controlling the backwards evolution of $\alpha$, but we note here that it is possible to derive uniform, nondegenerate energy estimates for $\alpha$ near $\mathscr{H}^+$ in contrast with the case of $\Psi$. In this sense, $\alpha$ is "red-shifted" in the backwards direction.
\subsubsection{Forwards scattering for $\alpha$}\label{subsubsection 8.1.1 forwards scattering +2}
We put together the ingredients worked out in \Cref{+2 radiation} to construct the forwards scattering map.
\begin{proof}[Proof of \Cref{+2 future forward scattering}]
Let $\alpha$ be the solution to \cref{T+2} on $ J^+(\Sigma^*)$ arising out of a compactly supported data set $(\upalpha,\upalpha')$ on $\Sigma^*$ as in \Cref{WP+2Sigma*}. The radiation field $\upalpha_{\mathscr{H}^+}$ exists as in \Cref{+2 radiation alpha definition H}. \Cref{alpha+2ptwisedecay} applied for $R=2M$ says that $\upalpha_{\mathscr{H}^+}\longrightarrow 0$ towards $\mathscr{H}^+_+$. Let $\Psi$ be the solution to \cref{RW} associated to $\alpha$ via (\ref{hier+}). The fact that $(\Psi|_{\Sigma^*},\slashednabla_{T}\Psi|_{\Sigma^*})$ are compactly supported means that the results of \Cref{+2 radiation flux on H+} apply. In particular, we find that
\begin{align}
     \left|\int_{v}^\infty d\bar{v}\; e^{\frac{1}{2M}(v-\bar{v})}\upalpha_{\mathscr{H}^+}(\bar{v},\theta^A)\right|\leq \frac{1}{2M} |\upalpha_{\mathscr{H}^+}(v,\theta^A)|,
\end{align}
and since $\|\uppsi_{\mathscr{H}^+}\|_{\mathcal{E}^T_{\mathscr{H}^+_{\geq0}}}<\infty$, this shows that $\|\upalpha_{\mathscr{H}^+}\|_{\mathcal{E}^{T,+2}_{\mathscr{H}^+_{\geq0}}}<\infty$ and $\upalpha_{\mathscr{H}^+}\in \mathcal{E}^{T,+2}_{\mathscr{H}^+_{\geq0}}$. Similarly by \Cref{alpha+2scri}, $r^5\alpha$ has a pointwise limit as $v\longrightarrow \infty$ which induces a smooth $\upalpha_{\mathscr{I}^+}$ on $\mathscr{I}^+$. \Cref{T+1+2scridecay} implies that $\upalpha_{\mathscr{I}^+}$ decays towards $\mathscr{I}^+_+$. As $\Psi_{\mathscr{I}^+}\in \mathcal{E}^T_{\mathscr{I}^+}$, we have that $\upalpha_{\mathscr{I}^+}\in\mathcal{E}^{T,+2}_{\mathscr{I}^+}$.
\end{proof}
\begin{corollary}\label{+2 future forward scattering Sigma Sigmabar}
    Solutions to (\ref{T+2}) arising from data on ${\Sigma}$ of compact support give rise to smooth radiation fields in $\mathcal{E}_{\mathscr{I}^+}^{T,+2}$ and $\mathcal{E}_{{\mathscr{H}^+}}^{T,+2}$. Solutions to (\ref{T+2}) arising from data on $\overline{\Sigma}$ of compact support give rise to smooth radiation fields in $\mathcal{E}_{\mathscr{I}^+}^{T,+2}$ and $\mathcal{E}_{\overline{\mathscr{H}^+}}^{T,+2}$
\end{corollary}
\begin{proof}
Identical to the proof of \Cref{RWfcpSigma} using \Cref{WP+2Sigmabar,,backwards wellposedness +2}.
\end{proof}
The proof of \Cref{+2 future forward scattering} above and \Cref{+2 future forward scattering Sigma Sigmabar} allow us to define the forwards maps ${}^{(+2)}\mathscr{F}^+$ from dense subspaces of $\mathcal{E}^{T+2}_{\Sigma^*}$, $\mathcal{E}^{T,+2}_{\Sigma}$, $\mathcal{E}^{T,+2}_{\overline{\Sigma}}$.
\begin{defin}
    Let $(\upalpha,\upalpha')$ be a smooth data set of compact support to the +2 Teukolsky equation \bref{T+2} on $\Sigma^*$ as in \Cref{WP+2Sigma*}. Define the map ${}^{(+2)}\mathscr{F}^+$ by 
    \begin{align}
        {}^{(+2)}\mathscr{F}^+:\Gamma_c(\Sigma^*)\times\Gamma_c(\Sigma^*)\longrightarrow \Gamma(\mathscr{H}^+_{\geq0})\times\Gamma(\mathscr{I}^+), (\upalpha,\upalpha')\longrightarrow (\upalpha_{\mathscr{H}^+},\upalpha_{\mathscr{I}^+}),
    \end{align}
    where $(\upalpha_{\mathscr{H}^+},\upalpha_{\mathscr{I}^+})$ are as in the proof of \Cref{+2 future forward scattering}.\\
    \indent Using \Cref{+2 future forward scattering Sigma Sigmabar}, the map ${}^{(+2)}\mathscr{F}^+$ is defined analogously for data on $\Sigma, \overline{\Sigma}$:
    \begin{align}
       {}^{(+2)}\mathscr{F}^+:\Gamma_c(\Sigma)\times\Gamma_c(\Sigma)\longrightarrow \Gamma(\mathscr{H}^+)\times\Gamma(\mathscr{I}^+), (\upalpha,\upalpha')\longrightarrow (\upalpha_{\mathscr{H}^+},\upalpha_{\mathscr{I}^+}),\\
       {}^{(+2)}\mathscr{F}^+:\Gamma_c(\overline{\Sigma})\times\Gamma_c(\overline{\Sigma})\longrightarrow \Gamma(\overline{\mathscr{H}^+})\times\Gamma(\mathscr{I}^+), (\upalpha,\upalpha')\longrightarrow (\upalpha_{\mathscr{H}^+},\upalpha_{\mathscr{I}^+}).
    \end{align}
\end{defin}
\subsubsection{Backwards scattering for $\alpha$}\label{subsubsection 8.1.2 backwards scattering +2}
Now we construct the inverse ${}^{(+2)}\mathscr{B}^-$ of \Cref{+2 future backward scattering} on a dense subspace of $\mathcal{E}^{T,+2}_{\mathscr{H}^+_{\geq0}}\oplus\mathcal{E}^{T,+2}_{\mathscr{I}^+}$. The existence of a solution to the +2 Teukolsky equation \bref{T+2} out of compactly supported scattering data on $\mathscr{H}^+_{\geq0}, \mathscr{I}^+$ is shown in \Cref{+2 backwards existence}. Showing that this solution defines an element of $\mathcal{E}^{T,+2}_{\Sigma^*}$ is done in \Cref{+2 backwards inclusion 7/2}. %The proof of \Cref{scatteringthm+2} is concluded in \Cref{concluding corollary +2}. 
\begin{proposition}\label{+2 backwards existence}
For $\upalpha_{\mathscr{H}^+}\in\Gamma_c(\mathscr{H}^+_{\geq0})\cap\mathcal{E}^{T,+2}_{\mathscr{H}^+_{\geq0}}$ supported on $\mathscr{H}^+_{\geq0}\cap\{v<v_+\}$ for $v_+<\infty$, ${\alpha}_{\mathscr{I}^+}\in\Gamma_c(\mathscr{I}^+)\cap\mathcal{E}^{T,+2}_{\mathscr{I}^+}$ supported on on $\mathscr{I}^+\cap\{u<u_+\}$ for $u_+<\infty$, there exists a unique solution $\alpha$ to \bref{T+2} in $J^+(\Sigma^*)$ that realises $\upalpha_{\mathscr{H}^+}$ and $\upalpha_{\mathscr{I}^+}$ as its radiation fields on $\mathscr{H}^+_{\geq0}, \mathscr{I}^+$.
\end{proposition}
\begin{proof}
%Let $[u_-,u_+]$ be the support of $\upalpha_{\mathscr{I}^+}$ in $u$, $[v_-,v_+]$ the support of $\upalpha_{\mathscr{H}^+}$ in $v$. 
Define
\begin{align}
    {\psi}_{\mathscr{H}^+}&=\frac{1}{(2M)^3}\int_v^\infty d\bar{v}\; e^{\frac{1}{2M}(v-\bar{v})}(\mathcal{A}_2-3)\upalpha_{\mathscr{H}^+},\\
    \bm{\uppsi}_{\mathscr{H}^+}&=2M\int_v^\infty d\bar{v} \left[e^{\frac{1}{2M}(v-\bar{v})}-1\right]\left\{\mathcal{A}_2\left[\mathcal{A}_2-2\right]\upalpha_{\mathscr{H}^+}-6M\partial_v\upalpha_{\mathscr{H}^+}\right\},\\
    \psi_{\mathscr{I}^+}&=\partial_u\upalpha_{\mathscr{I}^+},\\
    \bm{\uppsi}_{\mathscr{I}^+}&=\partial_u^2\upalpha_{\mathscr{I}^+}.
\end{align}
With scattering data $\bm{\uppsi}_{\mathscr{I}^+}, \bm{\uppsi}_{\mathscr{H}^+}$, there is a unique solution $\Psi$ to \cref{RW} on $J^+(\Sigma^*)$. Define $\psi, \alpha$ by
\begin{align}
    r^3\Omega\psi=(2M)^3\psi_{\mathscr{H}^+}-\int_u^\infty\frac{\Omega^2}{r^2}\Psi d\bar{u},\qquad\qquad
    r\Omega^2\alpha=\upalpha_{\mathscr{H}^+}-\int_u^\infty r{\Omega^3}\psi d\bar{u},
\end{align}
then $\psi, \alpha$ satisfy the transport relations \bref{hier+}:
\begin{align}
    \Psi=\frac{r^2}{\Omega^2}\nablau r^3\Omega\psi=\left(\frac{r^2}{\Omega^2}\nablau\right)^2r\Omega^2\alpha.
\end{align}
(note that we are working with $(1,1)$-tensor fields throughout). The boundedness of $F_v^T[\Psi](u,\infty)$ implies that $\Omega^2\alpha\longrightarrow\upalpha_{\mathscr{H}^+}$, $\Omega\psi\longrightarrow{\psi}_{\mathscr{H}^+}$ as $u\longrightarrow\infty$. 
Since $\Psi$ satisfies \cref{RW}, the commutation relation \bref{commutation relation} implies
\begin{align}
    \left(\frac{r^2}{\Omega^2}\nablau\right)^2\mathcal{T}^{+2}r\Omega^2\alpha=0.
\end{align}
where $\mathcal{T}^{+2}$ is the $+2$ Teukolsky operator. We have:
\begin{align}
\begin{split}
    \mathcal{T}^{+2}r\Omega^2\alpha&=\frac{3\Omega^2-1}{r}r^3\Omega\psi+\nablav r^3\Omega\psi-\left(\mathcal{A}_2-\frac{6M}{r}\right)r\Omega^2\alpha\\
    \frac{r^2}{\Omega^2}\nablau\mathcal{T}^{+2}r\Omega^2\alpha&=-(\mathcal{A}_2-3\Omega^2+1)r^3\Omega\psi-\nablav\Psi+6Mr\Omega^2\alpha
\end{split}
\end{align}
On $\mathscr{H}^+$ this evaluates to
\begin{align}
    \mathcal{T}^{+2}r\Omega^2\alpha|_{\mathscr{H}^+}&=(2M)^3\left(\partial_v-\frac{1}{2M}\right)\psi_{\mathscr{H}^+}-(\mathcal{A}_2-3)\upalpha_{\mathscr{H}^+},\\
     \frac{r^2}{\Omega^2}\nablau\mathcal{T}^{+2}r\Omega^2\alpha|_{\mathscr{H}^+}&=-(2M)^3(\mathcal{A}_2+1)\psi_{\mathscr{H}^+}+6M\upalpha_{\mathscr{H}^+}-\partial_v\bm{\uppsi}_{\mathscr{H}^+}.
\end{align}
It is clear that with our construction of initial data, 
$\mathcal{T}^{+2}r\Omega^2\alpha|_{\mathscr{H}^+}= \frac{r^2}{\Omega^2}\nablau\mathcal{T}^{+2}r\Omega^2\alpha|_{\mathscr{H}^+}=0$, therefore $\alpha$ satisfies $\mathcal{T}^{+2}r\Omega^2\alpha=0$. Note that as $\Psi(u,v)$ vanishes for $u>u_+,v>v_+$, the same applies to $\alpha,\psi$. Let $\mathcal{R}>3M$, we can estimate $\psi(u,v)$ for $r(u_+,v)>\mathcal{R}$ by:
\begin{align}
    |r^5\Omega\psi|\leq\int_u^{u_+}\Omega^2|\Psi|+\int_u^{u_+}\frac{2}{r}|r^5\Omega\psi|
\end{align}
Gr\"onwall's inequality implies
\begin{align}
    |r^5\Omega\psi|\lesssim \left(\frac{r(u,v)}{r(u_+,v)}\right)^2\int_u^{u_+}|\Psi|.
\end{align}
As $\Psi$ converges uniformly to $\bm{\uppsi}_{\mathscr{I}^+}$, this implies that $\partial_u r^5\Omega\psi$ converges uniformly to $\partial_u\bm{\uppsi}_{\mathscr{H}^+}$, which in turn says that $r^5\Omega\psi$ converges to $\psi_{\mathscr{H}^+}$. An identical argument shows that $r^5\alpha$ converges to $\upalpha_{\mathscr{I}^+}$.
\end{proof}
In the following we explicitly show that $\alpha$ of \Cref{+2 backwards existence} defines a member of $\mathcal{E}^{T,+2}_{\mathscr{I}^+}$: 
%Let $\theta_u$ be a smooth cutoff function that is identically vanishing on $[u+1,\infty)$ and is equal to 1 on $(-\infty,u]$, and assume $\|\theta_u\cdot\upalpha_{\mathscr{I}^+}\|^2_{\mathcal{E}^{T,+2}_{\mathscr{I}^+}}$ is integrable in $u$ over $(-\infty,u_+]$.
\begin{corollary}\label{+2 backwards inclusion 7/2}
Let $\upalpha_{\mathscr{H}^+}, \upalpha_{\mathscr{I}^+}$ be as in \Cref{+2 backwards existence}. Let $\alpha$ be the solution to \cref{T+2} arising from $\upalpha_{\mathscr{H}^+}, \upalpha_{\mathscr{I}^+}$. Then $(\Omega^2\alpha|_{\Sigma^*},\slashednabla_{n_{\Sigma^*}}\Omega^2\alpha|_{\Sigma^*})\in\mathcal{E}^{T,+2}_{\Sigma^*}$ 
\end{corollary}
\begin{proof}
Let $\xi$ be a smooth cutoff function over $\mathbb{R}$ with $\xi=1$ for $r\leq0$, $\xi=0$ for $r\geq1$ such that all derivatives of $\xi$ are uniformly bounded. Let $\{R_n\}_{n=1}^\infty$ with $R_{1}$ large and $R_{n+1}=2R_n$ and define $\xi_n (r)=\xi\left(\frac{r-R_{n}}{R_{n+1}-R_n}\right)$. We want to show that the sequence $\alpha_n=\xi_n\alpha$ is such that $(\Omega^2\alpha_n,\slashednabla_{n_{\Sigma^*}}\Omega^2\alpha_n)$ converges to $(\Omega^2\alpha,\slashednabla_{n_{\Sigma^*}}\Omega^2\alpha)$ in $\mathcal{E}^{T,+2}_{\Sigma^*}$. Denoting by $\Psi_{n}=\left(\frac{r^2}{\Omega^2}\nablau\right)^2r\Omega^2\alpha_n$ the solution to the Regge--Wheeler equation arising from $\alpha_n$, we calculate
\begin{align}
\begin{split}
    \Psi_n&=\left(\frac{r^2}{\Omega^2}\nablau\right)^2 r\Omega^2\alpha_n=\left(\frac{r^2}{\Omega^2}\nablau\right)^2\xi_n r\Omega^2\alpha\\
    &=r^2(r^2\xi_n')'r\Omega^2\alpha-2r^2\xi_n'r^3\Omega\psi+\xi_n\Psi.
\end{split}
\end{align}
We know that $\xi_n\Psi\longrightarrow \Psi$ in $\mathcal{E}^{T}_{\Sigma^*}$ (see \Cref{RW enough to be in space}). Seeing that $r^2\xi_n'\sim r, r^2(r^2\xi_n')'\sim r^2$ on $[R_n,R_{n+1}]$, we can estimate the remainder via
\begin{align}
\begin{split}
    \|\Psi_n-\xi_n\Psi\|_{\mathcal{E}^T_{\Sigma^*}}^2\lesssim\int_{R_n}^{R_{n+1}} dr\sin\theta d\theta d\phi\;& \left[|r^3\Omega\psi|^2+|\mathring{\slashednabla}r^3\Omega\psi|^2+|r\nablav r^3 \Omega\psi|^2\right]\\&+\left[|r^3\Omega\alpha|^2+|\mathring{\slashednabla}r^3\Omega\alpha|^2+|r\nablav r^3 \Omega\alpha|^2\right]\\&+\left[\frac{1}{r^2}(|\Psi|^2+|\mathring{\slashednabla}\Psi|^2)+|\nablav\Psi|^2\right].
\end{split}
\end{align}
The result follows if we can show that $r^{\frac{7}{2}}\Omega\psi|_{\Sigma^*}, r^{\frac{7}{2}}\Omega^2\alpha|_{\Sigma^*}, r^{\frac{3}{2}}\nablav r^3\Omega\psi, r^{\frac{3}{2}}\nablav r^3\Omega^2\alpha$ decay as $r\longrightarrow\infty$. Let $u<u'<u_-$ and take $r=r(u',v), R=r(u,v)$ and $(u,v,\theta^A):=(R,\theta^A)\in\Sigma^*$. We estimate $R^{\frac{7}{2}}\Omega\psi|_{\Sigma^*}$ by integrating the definition of $\Psi$ (\ref{hier+}):
\begin{align}
\begin{split}
    \int_{S^2}R^{\frac{7}{2}}\Omega|\psi(R,\theta^A)|d\omega&\leq\sqrt{r}\int_u^{u'}d\bar{u}\int_{S^2}d\omega\; \frac{\Omega^2}{r^2}|\Psi|+\sqrt{R}r^{3}\Omega|\psi(u',v,\theta^A)|\\&\lesssim_{u'}\sqrt{r}\int_u^{u'}d\bar{u}\int_{S^2}d\omega\; \frac{\Omega^2}{r^2}|\Psi|+r^{\frac{7}{2}}\Omega|\psi(u',v,\theta^A)|\\&\lesssim_{u'}\sqrt{F^T_v[\Psi](u,u')}+r^{\frac{7}{2}}\Omega|\psi(u',v,\theta^A)|.
\end{split}
\end{align}
We used Cauchy--Schwarz to get to the last step. The right hand side decays as $v\longrightarrow\infty$ since $F^T_v[\Psi](u,u')$ decays, $F^T_{u'}[\Psi](v,\infty)<\infty$ and $\bm{\uppsi}_{\mathscr{I}^+}$ vanishes for $u<u_-$, so that
\begin{align}
    |r^{3}\Omega\psi(u',v,\theta^A)|_{L^2(S^2_{u',v})}\leq\int_v^\infty d\bar{v}\int_{S^2_{u',\bar{v}}} \frac{\Omega^2}{r^2}|\Psi|\leq\frac{1}{\sqrt{r(u',v)}}\sqrt{F^T_{u'}[\Psi](v,\infty)}.
\end{align}
and commuting with $\slashed{\mathcal{L}}_{S^2}^\gamma$ for $|\gamma|\leq 3$ gives that $R^{\frac{7}{2}}\Omega\psi|_{\Sigma^*}$ decays as $R\longrightarrow\infty$. This can be repeated to show the same for $R^{\frac{7}{2}}\Omega^2\alpha|_{\Sigma^*}$. Furthermore, we have
\begin{align}
\begin{split}
    \nablau r\nablav r^3\Omega\psi=-\frac{\Omega^2}{r} r\nablav r^3\Omega\psi+(3\Omega^2-1)\frac{\Omega^2}{r^2}\Psi+\frac{\Omega^2}{r}\nablav \Psi.
\end{split}
\end{align}
We estimate
\begin{align}
    \left|r\nablav r^3\Omega\psi|_{\Sigma^*}\right|\leq \left|r\nablav r^3\Omega\psi(u',v,\theta^A)\right|+\int_u^{u'}d\bar{u}\left[\frac{\Omega^2}{r}|r\nablav r^3\Omega\psi|+(3\Omega^2-1)\frac{\Omega^2}{r^2}|\Psi|+\frac{\Omega^2}{r}|\nablav\Psi|\right].
\end{align}
Gr\"onwall's inequality implies
\begin{align}
    \left|r\nablav r^3\Omega\psi|_{\Sigma^*}\right|\lesssim\frac{r(u',v)}{r(u,v)}\left[\left|r\nablav r^3\Omega\psi(u,v,\theta^A)\right|+\frac{1}{\sqrt{R}}\sqrt{F^T_v[\Psi](u,u')}\right],
\end{align}
which in turn implies that $r^{\frac{3}{2}}\nablav r^3\Omega\psi|_{\Sigma^*}\longrightarrow 0$ as $R\longrightarrow\infty$. The same can be repeated to show $r^{\frac{3}{2}}\nablav r^3\Omega^2\alpha|_{\Sigma^*}\longrightarrow 0$ as $R\longrightarrow\infty$. 
%Unitarity is a consequence of \Cref{RW unitary backwards corollary}.
\end{proof}
\begin{defin}\label{+2 definition of B-}
    Let $\upalpha_{\mathscr{H}^+}, \upalpha_{\mathscr{I}^+}$ be as in \Cref{+2 backwards existence}. Define the map ${}^{(+2)}\mathscr{B}^-$ by
    \begin{align}
        {}^{(+2)}\mathscr{B}^-:\Gamma_c(\mathscr{H}^+_{\geq0})\times\Gamma_c(\mathscr{I}^+)\longrightarrow\Gamma(\Sigma^*)\times\Gamma(\Sigma^*), (\upalpha_{\mathscr{H}^+},\upalpha_{\mathscr{I}^+})\longrightarrow (\Omega^2\alpha|_{\Sigma^*},\slashednabla_{n_{\Sigma^*}}\Omega^2\alpha|_{\Sigma^*}),
    \end{align}
    where $\alpha$ is the solution to \bref{T+2} arising from scattering data $(\upalpha_{\mathscr{H}^+},\upalpha_{\mathscr{I}^+})$ as in \Cref{+2 backwards existence}.
\end{defin}
\begin{corollary}\label{+2B- inverts +2F+}
    The maps ${}^{(+2)}\mathscr{F}^+$, ${}^{(+2)}\mathscr{B}^-$ extend uniquely to unitary Hilbert space isomorphisms on their respective domains, such that ${}^{(+2)}\mathscr{F}^+\circ{}^{(+2)}\mathscr{B}^-=Id$, ${}^{(+2)}\mathscr{B}^-\circ{}^{(+2)}\mathscr{F}^+=Id$.
\end{corollary}
\begin{proof}
Identical to the proof of \Cref{B- inverts F+}.
\end{proof}
\begin{remark}\label{unitarity of +2B- is trivial}
As in the case of \Cref{unitarity of B- is trivial}, \Cref{+2B- inverts +2F+} implies
\begin{align}\label{unitarity of +2B- formula}
    \|{}^{(+2)}\mathscr{B}^-(\upalpha_{\mathscr{H}^+},\upalpha_{\mathscr{I}^+})\|_{\mathcal{E}^{T,+2}_{\Sigma^*}}^2=\|\upalpha_{\mathscr{H}^+}\|^2_{\mathcal{E}^{T,+2}_{\mathscr{H}^+_{\geq0}}}+\|\upalpha_{\mathscr{I}^+}\|^2_{\mathcal{E}^{T,+2}_{\mathscr{I}^+}}.
\end{align}
As in the case of \Cref{RW unitary backwards}, we can use the backwards $r^p$-estimates of \Cref{backwards rp estimates} to directly show \bref{unitarity of +2B- formula} without reference to the forwards map ${}^{(+2)}\mathscr{F}^+$.
\end{remark}
Since the region $J^+(\overline{\Sigma})\cap J^-(\Sigma^*)$ can be handled locally via \Cref{WP+2Sigmabar}, \Cref{backwards wellposedness +2} and $T$-energy conservation, we can immediately deduce the following:
\begin{corollary}
The map ${}^{(+2)}\mathscr{B}^-$ can be defined on the following domains:
\begin{align}
    {}^{(+2)}\mathscr{B}^{-}:\mathcal{E}^{T,+2}_{\mathscr{H}^+}\oplus \mathcal{E}^{T,+2}_{\mathscr{I}^+}\longrightarrow \mathcal{E}^{T,+2}_{\Sigma},\\
    {}^{(+2)}\mathscr{B}^{-}:\mathcal{E}^{T,+2}_{\overline{\mathscr{H}^+}}\oplus \mathcal{E}^{T,+2}_{\mathscr{I}^+}\longrightarrow \mathcal{E}^{T,+2}_{\overline{\Sigma}},
\end{align}
and we have
\begin{align}
{}^{(+2)}\mathscr{F}^{+}\circ{}^{(+2)}\mathscr{B}^{-}=Id_{\mathcal{E}^{T,+2}_{\mathscr{H}^+}\oplus\;\mathcal{E}^{T,+2}_{\mathscr{I}^+}},\qquad
{}^{(+2)}\mathscr{B}^{-}\circ{}^{(+2)}\mathscr{F}^{+}=Id_{\mathcal{E}^{T,+2}_{\Sigma}},\\
{}^{(+2)}\mathscr{F}^{+}\circ{}^{(+2)}\mathscr{B}^{-}=Id_{\mathcal{E}^{T,+2}_{\overline{\mathscr{H}^+}}\oplus\;\mathcal{E}^{T,+2}_{\mathscr{I}^+}},\qquad
{}^{(+2)}\mathscr{B}^{-}\circ{}^{(+2)}\mathscr{F}^{+}=Id_{\mathcal{E}^{T,+2}_{\overline{\Sigma}}}.
\end{align}
\end{corollary}
This concludes the proof of \Cref{+2 future backward scattering}.
\begin{remark}[A nondegenerate estimate near $\mathscr{H}^+$]\label{nondegenerate estimate near H+} 
Note that the transport hierarchy \bref{hier+} implies (integrating in the measure $du\sin\theta d\theta d\phi$)
\begin{align}
\begin{split}
    \int_{\underline{\mathscr{C}}_v\cap[u,\infty)} \frac{1}{\Omega^2} |\nablau r^3\Omega\psi|^2&= \int_{\underline{\mathscr{C}}_v\cap[u,\infty)} \frac{\Omega^2}{r^2}|\Psi|^2\leq \underline{F}_v^T[\Psi](u,\infty),\\
    \int_{\underline{\mathscr{C}}_v\cap[u,\infty)}\frac{1}{\Omega^2}|\nablau r\Omega^2\alpha|^2&\lesssim \frac{1}{(2M)^2} \int_{\underline{\mathscr{C}}_v\cap[u,\infty)}\frac{1}{r^2}|\nablau r^3\Omega\psi|^2\lesssim \Omega^2(u,v)  \underline{F}_v^T[\Psi](u,\infty).
\end{split}
\end{align}
These estimates hold uniformly in $v$, in contrast to \bref{RW exponential backwards near H+}. This can be traced to the sign of the first order term in
\begin{align}
    \nablau\nablav r\Omega^2\alpha+\frac{2(3\Omega^2-1)}{r}\nablau r\Omega^2\alpha-\Omega^2\slashed{\Delta} r\Omega^2 \alpha+\frac{6M\Omega^2}{r^2} r\Omega^2\alpha=0.
\end{align}
for $r<3M$.\\
\indent Near $\mathscr{I}^+$ we can use \bref{+2 equation for radiation field} and follow the same steps leading to \bref{this} to derive for $R>\mathcal{R}_{\mathscr{I}^+}$:
\begin{align}\label{this+2}
    \int_{\mathscr{C}_u\cap\{r>R\}}r^2|\nablav r^5\Omega^{-2}\alpha|^2\lesssim_{u_-,M}\left[\|\upalpha_{\mathscr{I}^+}\|_{\mathcal{E}^{T,+2}_{\mathscr{I}^+}}^2 +\|\upalpha_{\mathscr{H}^+}\|_{\mathcal{E}^{T,+2}_{\mathscr{H}^+}}^2+\int_{\mathscr{I}^+\cap[u,u_+]}|\upalpha_{\mathscr{I}^+}|^2+|\mathring{\slashednabla}\upalpha_{\mathscr{I}^+}|^2\right].
\end{align}
With these estimates we can conclude as for the Regge--Wheeler equation:
\begin{corollary}\label{+2 noncompact}
    The results of \Cref{+2 backwards existence} hold when $\upalpha_{\mathscr{H}^+}$, $\upalpha_{\mathscr{I}^+}$ are not compactly supported, provided
\begin{align}\label{noncompact estimate}
    \sum_{|\gamma|\leq2}\|\slashed{\mathcal{L}}^\gamma_{S^2}\upalpha_{\mathscr{H}^+}\|_{\mathcal{E}^{T,+2}_{\mathscr{H}^+}}^2+\|\slashed{\mathcal{L}}^\gamma_{S^2}\upalpha_{\mathscr{I}^+}\|_{\mathcal{E}^{T,+2}_{\mathscr{I}^+}}^2+\int_{\mathscr{I}^+}|\slashed{\mathcal{L}}^\gamma_{S^2}\upalpha_{\mathscr{I}^+}|^2+|\slashed{\mathcal{L}}^\gamma_{S^2}\mathring{\slashednabla}\upalpha_{\mathscr{I}^+}|^2<\infty.
\end{align}
\end{corollary}
\end{remark}
The results above can be extended to scattering from $\Sigma, \overline\Sigma$, since the region $J^+(\overline\Sigma)\cap J^-(\Sigma^*)$ can be handled locally with \Cref{WP+2Sigmabar} and \Cref{RWfcpSigma}.
\begin{corollary}
    Let $\upalpha_{\mathscr{H}^+}\in\Gamma(\mathscr{H}^+)\cap\;\mathcal{E}^{T,+2}_{\mathscr{H}^+}$,  $\upalpha_{\mathscr{I}^+}\in\Gamma (\mathscr{I}^+)\cap\;\mathcal{E}^{T,+2}_{\mathscr{I}^+}$, such that \bref{noncompact estimate} is satisfied. Then there exists a unique solution $\alpha$ to \cref{T+2} in $J^+({\Sigma})$ such that $\lim_{v\longrightarrow\infty}r^5\alpha=\upalpha_{\mathscr{I}^+}$, $2M\Omega^2\alpha\big|_{\mathscr{H}^+}=\upalpha_{\mathscr{H}^+}$. Moreover, $(\alpha\big|_{{\Sigma}},\slashednabla_{n_\Sigma}\alpha|_{{\Sigma}})\in \mathcal{E}^{T,+2}_{{\Sigma}}$ and 
    \begin{align}
    \left\|\left(\alpha|_{{\Sigma}},\slashednabla_{n_\Sigma}\alpha|_{{\Sigma}}\right)\right\|^2_{\mathcal{E}^{T,+2}_{{\Sigma}}}=\left|\left|\upalpha_{\mathscr{I}^+}\right|\right|^2_{\mathcal{E}^{T,+2}_{\mathscr{I}^+}}+\left|\left|\upalpha_{\mathscr{H}^+}\right|\right|^2_{\mathcal{E}^{T,+2}_{{\mathscr{H}^+}}}.
\end{align}
\end{corollary}
\begin{corollary}
    Let $\upalpha_{\mathscr{H}^+}\in\mathcal{E}^{T,+2}_{\overline{\mathscr{H}^+}}$ be such that $V^{-2}\upalpha\in \Gamma({\overline{\mathscr{H}^+}})$ and let $\upalpha_{\mathscr{I}^+}\in\Gamma(\mathscr{I}^+)\cap\;\mathcal{E}^{T,+2}_{\mathscr{I}^+}$. Then there exists a unique solution $\alpha$ to \cref{T+2} in $J^+(\overline{\Sigma})$ such that $\lim_{v\longrightarrow\infty}r^5\alpha=\upalpha_{\mathscr{I}^+}$, $2MV^{-2}\Omega^2\alpha\big|_{\mathscr{H}^+}=V^{-2}\upalpha_{\mathscr{H}^+}$. Moreover, $(\alpha\big|_{\overline{\Sigma}},\slashednabla_{n_{\overline{\Sigma}}}\alpha|_{\overline{\Sigma}})\in \mathcal{E}^{T,+2}_{\overline{\Sigma}}$ and 
    \begin{align}
    \left\|\left(\alpha|_{\overline{\Sigma}},\slashednabla_{n_{\overline{\Sigma}}}\alpha|_{\overline{\Sigma}}\right)\right\|^2_{\mathcal{E}^{T,+2}_{\overline{\Sigma}}}=\left|\left|\upalpha_{\mathscr{I}^+}\right|\right|^2_{\mathcal{E}^{T,+2}_{\mathscr{I}^+}}+\left|\left|\upalpha_{\mathscr{H}^+}\right|\right|^2_{\mathcal{E}^{T,+2}_{\overline{\mathscr{H}^+}}}.
\end{align}
\end{corollary}
\subsubsection{A pointwise estimate near $i^0$ in backwards scattering}\label{subsubsection 8.1.3 pointwise estimate near i0}
As an aside, if $\upalpha_{\mathscr{I}^+}$ is compactly supported we can use the backwards $r^p$-estimates of \Cref{backwards rp estimates} to obtain better decay for $\alpha,\psi$ towards $i^0$. We illustrate this point in what follows:
\begin{proposition}
Let $\alpha$ be the solution to (\ref{T+2}) arising from scattering data $\upalpha_{\mathscr{H}^+}\in \Gamma_c (\mathscr{H}^+_{\geq0}), \upalpha_{\mathscr{I}^+}\in \Gamma_c (\mathscr{I}^+_{\geq0})$ as in \Cref{+2 backwards existence}. Then $r^5\psi|_{\Sigma^*}, r^5\alpha|_{\Sigma^*}\longrightarrow 0$. The same applies when $\Sigma^*$ is replaced by $\Sigma$ or $\overline\Sigma$.
\end{proposition}
\begin{proof}
Given that $\bm{\uppsi}_{\mathscr{I}^+}=\partial_u^2\upalpha_{\mathscr{I}^+}$ is compactly supported, we already know that $\Psi|_{\Sigma^*,r=R}\longrightarrow 0$ as $R\longrightarrow 0$. We first work with $r^5\psi$, for which we can derive a similar estimate to (\ref{backwards estimate +2 Gronwall}): Let $u<u'<u_-$ and take $(u,v,\theta^A)\in\Sigma^*$, $v-u:=R^*$. Integrating \cref{+2 Gronwall ingredient} in $u$ on $\underline{\mathscr{C}}_v$ between $u,u'$, we obtain:
\begin{align}
     \Big|r^5\Omega^{-1}\psi(u,v)-r^5\Omega^{-1}\psi(u',v)\Big|\leq\int_{u}^{u'}\left|\Psi\right|\exp\left[\int_{u}^{u'} \frac{3\Omega^2-1}{r} d\bar{u}\right]\lesssim\left[\int_{u}^{u'}\left|\Psi\right|\right]\left(\frac{r(u',v)}{r(u,v)}\right)^2.
\end{align}
We further compare $\int_{u}^{u'}\left|\Psi\right|d\bar{u}$ to $\int_{-\infty}^{u'} |\Psi|_{\mathscr{I}^+}$: via the backwards $r^p$-estimates of \Cref{backwards rp estimates}:
\begin{align}
    \left|\int_{u}^{u'}du \left|\Psi\right|-\int_{-\infty}^{u'}du \left|\bm{\uppsi}_{\mathscr{I}^+}\right|\right|^2\leq\left[ \int_{\mathscr{D}}du dv|\nablav\Psi|\right]^2\leq \frac{1}{\sqrt{R}}\int_{\mathscr{D}}du dv \;r^2|\nablav\Psi|^2,
\end{align}
where $\mathscr{D}=J^+(\Sigma^*)\cap J^+(\underline{\mathscr{C}}_{v})\cap J^-(\mathscr{C}_{u'})$. As in \Cref{backwards rp estimates}, we can bound the last integral by the right hand side of (\ref{RWbackwardsdecay}). As $R\longrightarrow\infty$, $\int_{u}^{u'}du \left|\Psi\right|\longrightarrow \int_{-\infty}^{u'}du \left|\bm{\uppsi}_{\mathscr{I}^+}\right|=0$. Consequently $\left|r^5\Omega^{-1}\psi(u,v)-r^5\Omega^{-1}\psi(u',v)\right|\ $ decays as $R\longrightarrow\infty$ and 
\begin{align}
    \lim_{R\longrightarrow\infty} r^5\psi|_{\Sigma^*,r=R}=0.
\end{align}
We can prove the same for $r^5\alpha|_{\Sigma,r=R}$ by repeating the above argument for $\int_{u_-}^{u_+} du (u-u_-)\Psi$ and noticing that $\int_{u_-}^{u_+} du (u-u_-)\bm{\uppsi}_{\mathscr{I}^+}$ also vanishes since $\bm{\uppsi}_{\mathscr{I}^+}$ is the 2nd derivative of compactly supported fields on $\mathscr{I}^+$.
\end{proof}
\subsection{Future scattering for $\underline\alpha$}\label{subsection 8.2 future scattering -2}
Forwards and backwards scattering for the $-2$ Teukolsky equation are worked out entirely analogously to the case of the $+2$ Teukolsky equation, using the scattering theory of the Regge--Wheeler equation and the results of \Cref{subsection 7.2 future radiation fields and fluxes}. In contrast to the $+2$ equation, the transport equation \bref{hier-} relating $\underline\alpha$ and $\underline\Psi$ is sufficient to obtain an estimate for the radiation field near $\mathscr{I}^+$ that is uniform in the future end of the support of $\underline\upalpha_{\mathscr{I}^+}$, while near $\mathscr{H}^+$ $\underline\alpha$ experiences an \textit{enhanced blueshift}, and it is necessary for scattering data to decay exponentially at a sufficiently fast rate towards the future in order to obtain a solution in backwards scattering that is smooth near $\mathscr{H}^+$.
\subsubsection{Forwards scattering for $\underline\alpha$}\label{subsubsection 8.2.1 forwards scattering -2}
We put together the ingredients worked out in \Cref{subsection 7.2 future radiation fields and fluxes} to construct the forwards scattering map.
\begin{proof}[Proof of \Cref{-2 future forward scattering}]
Let $\underline\alpha$ be the solution to \cref{T-2} on $J^+(\Sigma^*)$ arising out of a compactly supported data set $(\underline\upalpha,\underline\upalpha')$ on $\Sigma^*$ as in \Cref{WP+2Sigma*}. \Cref{WP-2Sigma*} guarantees the existence of the radiation field $\underline\upalpha_{\mathscr{H}^+}$ as in \Cref{-2 radiation alpha definition H}. \Cref{-2 radiation ptwise decay H} says that $\underline\upalpha_{\mathscr{H}^+}\longrightarrow 0$ towards the future end of $\mathscr{H}^+$. Let $\underline\Psi$ be the solution to \cref{RW} associated to $\underline\alpha$ via (\ref{hier-}) The fact that $(\underline\Psi|_{\Sigma^*},\slashednabla_{T}\underline\Psi|_{\Sigma^*})$ are compactly supported means that the results of \Cref{-2 radiation flux on H+} apply and $\underline\upalpha_{\mathscr{H}^+}\in \mathcal{E}^{T,-2}_{\mathscr{H}^+_{\geq0}}$. Similarly, by \Cref{-2 radiation at scri}, $r\underline\alpha$ has a pointwise limit as $v\longrightarrow \infty$ which induces a smooth $\underline\upalpha_{\mathscr{I}^+}$ on $\mathscr{I}^+$. \Cref{-2 alpha radiation decay} implies that $\underline\upalpha_{\mathscr{I}^+}$ decays towards the future end of $\mathscr{I}^+$. As $\underline\uppsi_{\mathscr{I}^+}\in \mathcal{E}^T_{\mathscr{I}^+}$, we have that
\begin{align}\label{-2 term in L2 on scri+}
    \mathcal{A}_2(\mathcal{A}_2-2)\int_v^\infty d\bar{u}\underline\upalpha_{\mathscr{I}^+}-6M\underline\upalpha_{\mathscr{I}^+}\in L^2(\mathscr{I}^+).
\end{align}
The fact that $\underline\alpha$ arises from data of compact support means that \bref{-2 mean is 0} applies. This implies upon evaluating the $L^2(\mathscr{I}^+)$ norm of the left hand side of \bref{-2 term in L2 on scri+} that $\underline\upalpha_{\mathscr{I}^+}\in\mathcal{E}^{T,-2}_{\mathscr{I}^+}$.
% The proof goes through for data on $\overline{\Sigma}$. This defines the bounded map $\mathscr{F}^+$ on a dense subspace of $\mathcal{E}^{T,-2}_{\Sigma^*}$ (or $\mathcal{E}^{T,-2}_{\Sigma}$,  $\mathcal{E}^{T,-2}_{\overline{\Sigma}}$), which then extends uniquely to the full space.
\end{proof}
\begin{corollary}\label{-2 future forward scattering Sigma Sigmabar}
    Solutions to (\ref{T-2}) arising from data on ${\Sigma}$ of compact support give rise to smooth radiation fields in $\mathcal{E}_{\mathscr{I}^+}^{T,-2}$ and $\mathcal{E}_{{\mathscr{H}^+}}^{T,-2}$. Solutions to (\ref{T-2}) arising from data on $\overline{\Sigma}$ of compact support give rise to smooth radiation fields in $\mathcal{E}_{\mathscr{I}^+}^{T,-2}$ and $\mathcal{E}_{\overline{\mathscr{H}^+}}^{T,-2}$.
\end{corollary}
\begin{proof}
Identical to the proof of \Cref{RWfcpSigma} using \Cref{WP-2Sigmabar,,backwards wellposedness -2}.
\end{proof}
The proof of \Cref{-2 future forward scattering} above and \Cref{-2 future forward scattering Sigma Sigmabar} allow us to define the forwards maps ${}^{(-2)}\mathscr{F}^+$ from dense subspaces of $\mathcal{E}^{T,-2}_{\Sigma^*}$, $\mathcal{E}^{T,-2}_{\Sigma}$, $\mathcal{E}^{T,-2}_{\overline{\Sigma}}$.
\begin{defin}
    Let $(\underline\upalpha,\underline\upalpha')$ be a smooth data set of compact support to the -2 Teukolsky equation \bref{T-2} on $\Sigma^*$ as in \Cref{WP-2Sigma*}. Define the map ${}^{(-2)}\mathscr{F}^+$ by 
    \begin{align}
        {}^{(-2)}\mathscr{F}^+:\Gamma_c(\Sigma^*)\times\Gamma_c(\Sigma^*)\longrightarrow \Gamma(\mathscr{H}^+_{\geq0})\times\Gamma(\mathscr{I}^+), (\upalpha,\upalpha')\longrightarrow (\upalpha_{\mathscr{H}^+},\upalpha_{\mathscr{I}^+}),
    \end{align}
    where $(\underline\upalpha_{\mathscr{H}^+},\underline\upalpha_{\mathscr{I}^+})$ are as in the proof of \Cref{-2 future forward scattering}.\\
    \indent Using \Cref{-2 future forward scattering Sigma Sigmabar}, the map ${}^{(-2)}\mathscr{F}^+$ is defined analogously for data on $\Sigma, \overline{\Sigma}$:
    \begin{align}
       {}^{(-2)}\mathscr{F}^+:\Gamma_c(\Sigma)\times\Gamma_c(\Sigma)\longrightarrow \Gamma(\mathscr{H}^+)\times\Gamma(\mathscr{I}^+), (\upalpha,\upalpha')\longrightarrow (\upalpha_{\mathscr{H}^+},\upalpha_{\mathscr{I}^+}),\\
       {}^{(-2)}\mathscr{F}^+:\Gamma_c(\overline{\Sigma})\times\Gamma_c(\overline{\Sigma})\longrightarrow \Gamma(\overline{\mathscr{H}^+})\times\Gamma(\mathscr{I}^+), (\underline\upalpha,\underline\upalpha')\longrightarrow (\underline\upalpha_{\mathscr{H}^+},\underline\upalpha_{\mathscr{I}^+}).
    \end{align}
\end{defin}
\subsubsection{Backwards scattering for $\underline\alpha$}\label{subsubsection 8.2.2 backwards scattering -2}
Now we construct the inverse ${}^{(-2)}\mathscr{B}^-$ of \Cref{-2 future backward scattering} on a dense subspace of $\mathcal{E}^{T,-2}_{\mathscr{H}^+_{\geq0}}\oplus\mathcal{E}^{T,-2}_{\mathscr{I}^+}$. The existence of a solution to \bref{T-2} out of compactly supported scattering data on $\mathscr{H}^+_{\geq0}, \mathscr{I}^+$ is shown in \Cref{-2 backwards existence}. Showing that this solution defines an element of $\mathcal{E}^{T,-2}_{\Sigma^*}$ is done in \Cref{-2 backwards inclusion 7/2}.
\begin{proposition}\label{-2 backwards existence}
For $\underline\upalpha_{\mathscr{H}^+}\in\Gamma(\mathscr{H}^+_{\geq0})\cap\mathcal{E}^{T,-2}_{\mathscr{H}^+_{\geq0}}$ supported on $\mathscr{H}^+_{\geq0}\cap\{v<v_+\}$ for $v_+<\infty$, ${\underline\alpha}_{\mathscr{I}^+}\in\Gamma(\mathscr{I}^+)\cap\mathcal{E}^{T,-2}_{\mathscr{I}^+}$ supported on on $\mathscr{I}^+\cap\{u<u_+\}$ for $u_+<\infty$, there exists a unique solution $\alpha$ to \bref{T-2} in $J^+(\Sigma^*)$ that realises $\underline\upalpha_{\mathscr{H}^+}$ and $\underline\upalpha_{\mathscr{I}^+}$ as its radiation fields on $\mathscr{H}^+_{\geq0}$, $\mathscr{I}^+$ respectively.
\end{proposition}
\begin{remark}
The fact that $\underline\upalpha_{\mathscr{I}^+}\in\mathcal{E}^{T,-2}_{\mathscr{I}^+}$ automatically implies that $\int_{-\infty}^\infty d\bar{u}\; \underline\upalpha_{\mathscr{I}^+}=0$.
\end{remark}
\begin{proof}
Let $\widetilde{\Sigma}$ be a spacelike surface connecting $\mathscr{H}^+$ at a finite $v_*>v_+$ to $\mathscr{I}^+$ at a finite $u_*>u_+$. Denote by $\mathscr{D}$ the region bounded by $\mathscr{H}^+_{\geq 0}\cap\{v<v_+\}$, $\widetilde{\Sigma}$, $
\mathscr{I}^+\cap[u_-,u_+]$, $\Sigma^*$ and $\mathscr{C}_{u_-}$ for $u_->-\infty$. We define
\begin{align}
    \underline{{\psi}}_{\mathscr{H}^+}&=\frac{2}{(2M)^2}\partial_v \underline\upalpha_{\mathscr{H}^+} +\frac{1}{2M}\partial_v \underline\upalpha_{\mathscr{H}^+}, \label{backwards psibar H} \\
    \underline{\bm{\uppsi}}_{\mathscr{H}^+}&=2(2M)^2\underline\upalpha_{\mathscr{H}^+}+2(2M)^3\partial_v \underline\upalpha_{\mathscr{H}^+}+(2M)^4\partial_v^2\underline\upalpha_{\mathscr{H}^+}, \label{backwards P sibar H}\\
    \underline{{\psi}}_{\mathscr{I}^+}&=-\int_u^\infty d\bar{u}\; \mathcal{A}_2 \;\underline\upalpha_{\mathscr{I}^+},\label{backwards psibar I}\\
    \underline{\bm{\uppsi}}_{\mathscr{I}^+}&=\int_u^\infty d\bar{u}\; (u_+-u) \left[\mathcal{A}_2(\mathcal{A}_2-2)\underline\upalpha_{\mathscr{I}^+}+6M\partial_u \underline\upalpha_{\mathscr{I}^+}\right]. \label{backwards P sibar I}
\end{align}
We can find a unique solution $\underline\Psi$ to \bref{RW} with radiation fields $\underline{\bm{\uppsi}}_{\mathscr{I}^+}$, $\underline{\bm{\uppsi}}_{\mathscr{H}^+}$. Let
\begin{align}
    r^3\Omega\underline\psi=(2M)^3\underline{{\psi}}_{\mathscr{I}^+}-\int_v^\infty d\bar{v}\;\frac{\Omega^2}{r^2}\underline\Psi,\qquad\qquad r\Omega^2\underline\alpha=\underline\upalpha_{\mathscr{I}^+}-\int_v^\infty d\bar{v}\; r\Omega^3\underline\psi.
\end{align}
Then $\underline\psi, \underline\alpha$ satisfy:
\begin{align}
    \underline\Psi=\frac{r^2}{\Omega^2}\nablav r^3\Omega\underline\psi=\left(\frac{r^2}{\Omega^2}\nablav\right)^2r\Omega^2\underline\alpha.
\end{align}
Moreover, we can see that $\lim_{v\longrightarrow\infty}r^3\Omega\underline\psi(u,v,\theta^A)=\underline{{\psi}}_{\mathscr{I}^+}(u,\theta^A)$ uniformly in $u$, as
\begin{align}
    \int_{S^2}|r^3\Omega\underline\psi-(2M)^3\underline{{\psi}}_{\mathscr{I}^+}|^2=\int_{S^2}\left[\int_v^\infty \frac{\Omega^2}{r^2}\underline\Psi d\bar{v}\right]^2\lesssim \frac{1}{r}F^T_u[\underline\Psi](v,\infty),
\end{align}
and similarly $\lim_{v\longrightarrow\infty}r\Omega^2\underline\alpha(u,v,\theta^A)=\underline\upalpha_{\mathscr{I}^+}(u,\theta^A)$ uniformly in $u$. We can repeat the same for $\slashednabla_T,\mathring{\slashednabla}$-derivatives of $r\Omega^2\underline\alpha, r^3\Omega\underline\psi$, which immediately implies that $\partial_u r^3\Omega\underline\psi\longrightarrow \partial_u \underline{{\psi}}_{\mathscr{I}^+}$, $\partial_u r\Omega^2\underline\alpha\longrightarrow \partial_u \underline\upalpha_{\mathscr{I}^+}$ as $v\longrightarrow\infty$. \\ 
The commutation relation \bref{commutation relation 2} implies 
\begin{align}
    \left(\frac{r^2}{\Omega^2}\nablav\right)^2\mathcal{T}^{-2}r\Omega^2\underline\alpha=0.
\end{align}
We find $\mathcal{T}^{-2}r\Omega^2\underline\alpha$ and $\frac{r^2}{\Omega^2}\nablav \mathcal{T}^{-2}r\Omega^2\underline\alpha$:
\begin{align}
    \mathcal{T}^{-2}r\Omega^2\underline\alpha&=\nablau r^3\Omega\underline\psi-\frac{3\Omega^2-1}{r}r^3\Omega\underline\psi-\left(\mathcal{A}_2-\frac{6M}{r}\right)r\Omega^2\underline\alpha,\\
    \frac{r^2}{\Omega^2}\nablav\mathcal{T}^{-2}r\Omega^2\underline\alpha&=\nablau\underline\Psi-\left[\mathcal{A}_2-(3\Omega^2-1)\right]r^3\Omega\underline\psi-6M r\Omega^2\underline\alpha.
\end{align}
It is not hard to see from \bref{backwards psibar H}, \bref{backwards psibar I}, \bref{backwards P sibar H}, \bref{backwards P sibar I}, that in the limit $v\longrightarrow\infty$,  $\mathcal{T}^{-2}r\Omega^2\underline\alpha$ and $\frac{r^2}{\Omega^2}\nablav \mathcal{T}^{-2}r\Omega^2\underline\alpha$ vanish. This implies that $\underline\alpha$ satisfies $\mathcal{T}^{-2}r\Omega^2{\underline\alpha}=0$ on $\mathscr{D}$. It is also clear that $\Omega^{-2}\underline\alpha|_{\mathscr{H}^+}=\underline\upalpha_{\mathscr{H}^+}$. Finally, we can repeat the above to extend $\underline\alpha$ to $J^+(\Sigma^*)\cap\{u\geq \tilde{u}\}$ for arbitrarily small $\tilde{u}$.
\end{proof}
\indent Note that energy conservation translates to the following $r$-weighted estimates that are uniform in $u$ as $u\longrightarrow -\infty$:
\begin{align}
    \int_{{\mathscr{C}}_u} \frac{r^2}{\Omega^2}|\nablav r^3\Omega\underline\psi|^2&\leq {F}^T_u[\underline\Psi](v,\infty),\label{334}\\
    \int_{{\mathscr{C}}_u} \frac{r^2}{\Omega^2}|\nablav r\Omega^2\underline\alpha|^2&\lesssim \int_{{\mathscr{C}}_u} |\nablav r^3\Omega\underline\psi|^2 \lesssim \frac{1}{r^2}{F}^T_u[\underline\Psi](v,\infty).\label{335}
\end{align}
This can be traced to the good sign of the first order term in \cref{T-2} near $\mathscr{I}^+$ when evolving backwards, and similar estimates can in fact be derived directly from \cref{T-2}. We can deduce
\begin{proposition}\label{-2 backwards inclusion 7/2}
Let $\underline\upalpha_{\mathscr{H}^+}, \underline\upalpha_{\mathscr{I}^+}$ be as in \Cref{-2 backwards existence}. Let $\underline\alpha$ be the corresponding solution to \cref{T-2}. Then we have that $(\Omega^{-2}\underline\alpha|_{\Sigma^*},\slashednabla_{n_{\Sigma^*}}\Omega^{-2}\underline\alpha|_{\Sigma^*})\in\mathcal{E}^{T,-2}_{\Sigma^*}$.% and
\end{proposition}
\begin{proof}
Using \bref{334}, \bref{335} it is easy to use the argument of \cref{+2 backwards inclusion 7/2} to show that $\lim_{r\longrightarrow\infty}\left|r^{\frac{7}{2}}\underline\psi|_{\Sigma^*}\right|=\lim_{r\longrightarrow\infty}\left|r^{\frac{7}{2}}\underline\alpha|_{\Sigma^*}\right|=0$, so we can repeat what was done to prove \Cref{+2 backwards inclusion 7/2} to obtain the result.
\end{proof}

\begin{defin}\label{-2 definition of B-}
    Let $\underline\upalpha_{\mathscr{H}^+}, \underline\upalpha_{\mathscr{I}^+}$ be as in \Cref{-2 backwards existence}. Define the map ${}^{(-2)}\mathscr{B}^-$ by
    \begin{align}
        {}^{(-2)}\mathscr{B}^-:\Gamma_c(\mathscr{H}^+_{\geq0})\times\Gamma_c(\mathscr{I}^+)\longrightarrow\Gamma(\Sigma^*)\times\Gamma(\Sigma^*), (\underline\upalpha_{\mathscr{H}^+},\underline\upalpha_{\mathscr{I}^+})\longrightarrow (\Omega^{-2}\underline\alpha|_{\Sigma^*},\slashednabla_{n_{\Sigma^*}}\Omega^{-2}\underline\alpha|_{\Sigma^*}),
    \end{align}
    where $\underline\alpha$ is the solution to \bref{T-2} arising from scattering data $(\underline\upalpha_{\mathscr{H}^+},\underline\upalpha_{\mathscr{I}^+})$ as in \Cref{-2 backwards existence}.
\end{defin}
\begin{corollary}\label{-2B- inverts -2F+}
    The maps ${}^{(-2)}\mathscr{F}^+$, ${}^{(-2)}\mathscr{B}^-$ extend uniquely to unitary Hilbert space isomorphisms on their respective domains, such that ${}^{(-2)}\mathscr{F}^+\circ{}^{(-2)}\mathscr{B}^-=Id$, ${}^{(-2)}\mathscr{B}^-\circ{}^{(-2)}\mathscr{F}^+=Id$.
\end{corollary}

\begin{remark}\label{unitarity of -2B- is trivial}
As in the case of \Cref{unitarity of B- is trivial,,unitarity of +2B- is trivial}, \Cref{-2B- inverts -2F+} implies
\begin{align}\label{unitarity of -2B- formula}
    \|{}^{(-2)}\mathscr{B}^-(\underline\upalpha_{\mathscr{H}^+},\underline\upalpha_{\mathscr{I}^+})\|_{\mathcal{E}^{T,-2}_{\Sigma^*}}^2=\|\underline\upalpha_{\mathscr{H}^+}\|^2_{\mathcal{E}^{T,-2}_{\mathscr{H}^+_{\geq0}}}+\|\underline\upalpha_{\mathscr{I}^+}\|^2_{\mathcal{E}^{T,-2}_{\mathscr{I}^+}}.
\end{align}
As in the case of \Cref{RW unitary backwards}, we can use the backwards $r^p$-estimates of \Cref{backwards rp estimates} to directly show \bref{unitarity of -2B- formula} without reference to the forwards map ${}^{(-2)}\mathscr{F}^+$.
\end{remark}
Since the region $J^+(\overline{\Sigma})\cap J^-(\Sigma^*)$ can be handled locally via \Cref{WP-2Sigmabar}, \Cref{backwards wellposedness -2} and $T$-energy conservation, we can immediately deduce the following:
\begin{corollary}
The map ${}^{(-2)}\mathscr{B}^-$ can be defined on the following domains:
\begin{align}
    {}^{(-2)}\mathscr{B}^{-}:\mathcal{E}^{T,-2}_{\mathscr{H}^+}\oplus \mathcal{E}^{T,-2}_{\mathscr{I}^+}\longrightarrow \mathcal{E}^{T,-2}_{\Sigma},\\
    {}^{(-2)}\mathscr{B}^{-}:\mathcal{E}^{T,-2}_{\overline{\mathscr{H}^+}}\oplus \mathcal{E}^{T,-2}_{\mathscr{I}^+}\longrightarrow \mathcal{E}^{T,-2}_{\overline{\Sigma}},
\end{align}
and we have
\begin{align}
{}^{(-2)}\mathscr{F}^{+}\circ{}^{(-2)}\mathscr{B}^{-}=Id_{\mathcal{E}^{T,-2}_{\mathscr{H}^+}\oplus\;\mathcal{E}^{T,-2}_{\mathscr{I}^+}},\qquad
{}^{(-2)}\mathscr{B}^{-}\circ{}^{(-2)}\mathscr{F}^{+}=Id_{\mathcal{E}^{T,-2}_{\Sigma}},\\
{}^{(-2)}\mathscr{F}^{+}\circ{}^{(-2)}\mathscr{B}^{-}=Id_{\mathcal{E}^{T,-2}_{\overline{\mathscr{H}^+}}\oplus\;\mathcal{E}^{T,-2}_{\mathscr{I}^+}},\qquad
{}^{(-2)}\mathscr{B}^{-}\circ{}^{(-2)}\mathscr{F}^{+}=Id_{\mathcal{E}^{T,-2}_{\overline{\Sigma}}}.
\end{align}
\end{corollary}
This concludes the proof of \Cref{-2 future backward scattering}.
\subsubsection{Non-compact future scattering data and the blueshift effect}\label{subsubsection 8.3 blueshift -2}
\indent In contrast to \bref{334}, \bref{335} (and to the estimates of \Cref{nondegenerate estimate near H+}), estimates for $\Omega^{-2}\underline\alpha$ near $\mathscr{H}^+$ in the backwards direction suffer from an enhanced blueshift, which can be readily seen in the transport equations \bref{hier-}:
\begin{align}
    \nablav r^3\Omega^{-1}\underline\psi+\frac{2M}{r^2} r^3\Omega^{-1}\underline\psi=\frac{\underline\Psi}{r^2}.
\end{align}
For $r<\mathcal{R}_{\mathscr{H}^+}<3M$, we can derive
\begin{align}\label{this 3}
\begin{split}
    &\int_{S^2_{u,v}}|r^3\Omega^{-1}\underline\psi - (2M)^3\underline{\bm{\uppsi}}_{\mathscr{H}^+}|^2\lesssim \underbrace{\int_{S^2_{u,v_+}}|r^3\Omega^{-1}\underline\psi - (2M)^3\underline{\bm{\uppsi}}_{\mathscr{H}^+}|^2}_{=0}\\
    &+\frac{1}{M}\int_v^{v_+}d\bar{v}\int_{S^2_{u,\bar{v}}}|r^3\Omega^{-1}\underline\psi - (2M)^3\underline{\bm{\uppsi}}_{\mathscr{H}^+}|^2+\frac{1}{(2M)^2}\int_{v}^{v_+}d\bar{v}\int_{S^2_{u,\bar{v}}}|\underline\Psi-\underline{\bm{\uppsi}}_{\mathscr{H}^+}|^2.
\end{split}
\end{align}
Gr\"onwall's inequality and \bref{ptwise horizon} imply
\begin{align}
\begin{split}
    \int_{S^2_{u,v}}|r^3\Omega^{-1}\underline\psi - (2M)^3\underline{\bm{\uppsi}}_{\mathscr{H}^+}|^2\lesssim_{v_+} e^{\frac{1}{M}(v_+-v)}\left[\|\underline{\bm{\uppsi}}_{\mathscr{I}^+}\|_{\mathcal{E}^T_{\mathscr{I}^+}}^2+\|\underline{\bm{\uppsi}}_{\mathscr{I}^+}\|_{\mathcal{E}^T_{\mathscr{H}^+}}^2+\int_{\mathscr{H}^+\cap[v,v_+]}|\underline{\bm{\uppsi}}_{\mathscr{H}^+}|^2+|\mathring{\slashednabla}\underline{\bm{\uppsi}}_{\mathscr{H}^+}|^2\right].
\end{split}
\end{align}
The equation
\begin{align}
    \nablav r\Omega^{-2}\underline\alpha +\frac{4M}{r^2}r\Omega^{-2}\underline\alpha=r\Omega^{-1}\underline\psi
\end{align}
implies a similar estimate with a worse exponential factor
\begin{align}
\begin{split}
    \int_{S^2_{u,v}}|r\Omega^{-2}\underline\alpha - 2M\underline\upalpha_{\mathscr{H}^+}|^2\lesssim_{v_+} e^{\frac{2}{M}(v_+-v)}\left[\|\underline{\bm{\uppsi}}_{\mathscr{I}^+}\|_{\mathcal{E}^T_{\mathscr{I}^+}}^2+\|\underline{\bm{\uppsi}}_{\mathscr{I}^+}\|_{\mathcal{E}^T_{\mathscr{H}^+}}^2+\int_{\mathscr{H}^+\cap[v,v_+]}|\underline{\bm{\uppsi}}_{\mathscr{H}^+}|^2+|\mathring{\slashednabla}\underline{\bm{\uppsi}}_{\mathscr{H}^+}|^2\right].
\end{split}
\end{align}
\indent We can conclude that the statement of the backwards existence theorem holds when scattering data is not compactly supported, but the solution will not be smooth unless data decays exponentially, which we can then show with the following applied to \bref{this 3}:
\begin{lemma}
Let $f(v)>0$ and assume 
\begin{align}
    f(v)\leq \Lambda \int_v^{v_+} f(v) + e^{-Pv}
\end{align}
for all $v<v_+$. Then if $P>\Lambda$ we have
\begin{align}
    f(v)< \frac{P}{P-\Lambda}e^{-Pv}.
\end{align}
\end{lemma}
With this, we see that if $\underline\upalpha_{\mathscr{H}^+}$, $\underline\upalpha_{\mathscr{I}^+}$ decay exponentially at a rate faster than $\frac{1}{M}$ then the we are guaranteed that
\begin{align}
    \int_{S^2_{u,v}}|r\Omega^{-2}\underline\alpha - 2M\underline\upalpha_{\mathscr{H}^+}|^2\lesssim \left[\|\underline{\bm{\uppsi}}_{\mathscr{I}^+}\|_{\mathcal{E}^T_{\mathscr{I}^+}}^2+\|\underline{\bm{\uppsi}}_{\mathscr{I}^+}\|_{\mathcal{E}^T_{\mathscr{H}^+}}^2+\int_{\mathscr{H}^+\cap[v,v_+]}|\underline{\bm{\uppsi}}_{\mathscr{H}^+}|^2+|\mathring{\slashednabla}\underline{\bm{\uppsi}}_{\mathscr{H}^+}|^2\right].
\end{align}
\begin{corollary}\label{-2 noncompact}
Let $\underline\upalpha_{\mathscr{H}^+}$ be a smooth symmetric traceless $S^2_{\infty,v}$ 2-tensor field with domain $\mathscr{H}^+$, $\underline\upalpha_{\mathscr{I}^+}$ a smooth symmetric traceless $S^2_{\infty,v}$ 2-tensor field with domain $\mathscr{I}^+$. Then there exists a unique $\underline\alpha$ that is smooth on the interior of $J^+(\Sigma^*)$ and satisfies \bref{T-2}. If $\underline\upalpha_{\mathscr{H}^+}, \underline\upalpha_{\mathscr{I}^+}$ decay exponentially towards the future at rate faster than $\frac{1}{M}$ then $\Omega^{-2}\underline\alpha$ is smooth up to and including $\mathscr{H}^+$.
\end{corollary}
Since the region $J^+(\overline\Sigma)\cap J^-(\Sigma^*)$ can be handled locally with \Cref{WP+2Sigmabar} and \Cref{RWfcpSigma}, 
the results above can be extended to scattering from $\Sigma, \overline\Sigma$.
\begin{corollary}
    Let $\underline\upalpha_{\mathscr{H}^+}\in\Gamma(\mathscr{H}^+)\cap\;\mathcal{E}^{T,-2}_{\mathscr{H}^+}$,  $\underline\upalpha_{\mathscr{I}^+}\in\Gamma (\mathscr{I}^+)\cap\;\mathcal{E}^{T,-2}_{\mathscr{I}^+}$. Assume $\underline\upalpha_{\mathscr{H}^+}$, $\underline\upalpha_{\mathscr{I}^+}$  decay exponentially at a rate faster than $\frac{1}{M}$. Then there exists a unique solution $\underline\alpha$ to \cref{T-2} in $J^+({\Sigma})$ such that $\lim_{v\longrightarrow\infty}r\underline\alpha=\underline{\upalpha}_{\mathscr{I}^+}$, $2M\Omega^{-2}\underline\alpha\big|_{\mathscr{H}^+}=\underline\upalpha_{\mathscr{H}^+}$. Moreover, $(\underline\alpha\big|_{{\Sigma}},\slashednabla_T\underline\alpha|_{{\Sigma}})\in \mathcal{E}^{T,-2}_{{\Sigma}}$ and 
    \begin{align}
    \left\|\left(\underline\alpha|_{{\Sigma}},\slashednabla_{n_{\Sigma}}\underline\alpha|_{{\Sigma}}\right)\right\|^2_{\mathcal{E}^{T,-2}_{{\Sigma}}}=\left|\left|\underline\upalpha_{\mathscr{I}^+}\right|\right|^2_{\mathcal{E}^{T,-2}_{\mathscr{I}^+}}+\left|\left|\underline\upalpha_{\mathscr{H}^+}\right|\right|^2_{\mathcal{E}^{T,-2}_{{\mathscr{H}^+}}}.
\end{align}
\end{corollary}
\begin{corollary}
     Let $\underline\upalpha_{\mathscr{H}^+}\in\mathcal{E}^{T,-2}_{\overline{\mathscr{H}^+}}$ be such that $V^{2}\underline\upalpha\in \Gamma({\overline{\mathscr{H}^+}})$ and let $\underline\upalpha_{\mathscr{I}^+}\in\Gamma(\mathscr{I}^+)\cap\;\mathcal{E}^{T,-2}_{\mathscr{I}^+}$. Assume $\underline\upalpha_{\mathscr{H}^+}$, $\underline\upalpha_{\mathscr{I}^+}$ decay exponentially at a rate faster than $\frac{1}{M}$, then there exists a unique solution $\underline\alpha$ to \cref{T-2} in $J^+(\overline{\Sigma})$ such that $\lim_{v\longrightarrow\infty}r\underline\alpha=\underline\upalpha_{\mathscr{I}^+}$, $V^{2}\Omega^{-2}\underline\alpha\big|_{\mathscr{H}^+}=V^{2}\underline\upalpha_{\mathscr{H}^+}$. Moreover, $(\underline\alpha\big|_{\overline{\Sigma}},\slashednabla_T\underline\alpha|_{\overline{\Sigma}})\in \mathcal{E}^{T,-2}_{\overline{\Sigma}}$ and 
    \begin{align}
    \left\|\left(\underline\alpha|_{\overline{\Sigma}},\slashednabla_{n_{\overline{\Sigma}}}\underline\alpha|_{\overline{\Sigma}}\right)\right\|^2_{\mathcal{E}^{T,-2}_{\overline{\Sigma}}}=\left|\left|\underline\upalpha_{\mathscr{I}^+}\right|\right|^2_{\mathcal{E}^{T,-2}_{\mathscr{I}^+}}+\left|\left|\underline\upalpha_{\mathscr{H}^+}\right|\right|^2_{\mathcal{E}^{T,-2}_{\overline{\mathscr{H}^+}}}.
\end{align}
\end{corollary}
\subsection{Past scattering for $\alpha, \underline\alpha$}\label{subsection 8.3 past scattering +2-2}
Taking into account \Cref{time inversion}, \Cref{+2 past forward scattering,,-2 past forward scattering} are immediate. We state the results regarding scattering on $J^-(\overline{\Sigma})$. 
\begin{corollary}\label{past scattering of +2}
    Given smooth data of compact support $(\upalpha,\upalpha')\in \mathcal{E}^{T,+2}_{\overline{\Sigma}}$, there exists a unique solution $\alpha$ to the +2 Teukolsky equation \bref{T+2} on $J^-(\overline{\Sigma})$ that induces smooth radiation fields 
    \begin{itemize}
        \item $\upalpha_{\mathscr{I}^-}\in \mathcal{E}^{T,+2}_{\mathscr{I}^-}$ given by  $\upalpha_{\mathscr{I}^-}(v,\theta^A)=\lim_{u\longrightarrow -\infty} r\alpha(u,v,\theta^A)$,
        \item $\upalpha_{\mathscr{H}^-}\in \mathcal{E}^{T,+2}_{\overline{\mathscr{H}^-}}$ given by  $U^{2}\upalpha_{\mathscr{I}^-}=2MU^2\Omega^{-2}\alpha|_{\mathscr{H}^-}$. 
    \end{itemize}
    such that 
     \begin{align}\label{109091}
    \left\|\left(\alpha|_{\overline{\Sigma}},\slashednabla_T\alpha|_{\overline{\Sigma}}\right)\right\|^2_{\mathcal{E}^{T,+2}_{\overline{\Sigma}}}=\left|\left|\upalpha_{\mathscr{I}^-}\right|\right|^2_{\mathcal{E}^{T,+2}_{\mathscr{I}^-}}+\left|\left|\upalpha_{\mathscr{H}^-}\right|\right|^2_{\mathcal{E}^{T,+2}_{\overline{\mathscr{H}^-}}}.
\end{align}
    Let $\upalpha_{\mathscr{H}^-}\in\mathcal{E}^{T,+2}_{\overline{\mathscr{H}^-}}$ be such that $U^{2}\upalpha\in \Gamma({\overline{\mathscr{H}^-}})$ and let $\upalpha_{\mathscr{I}^-}\in\Gamma(\mathscr{I}^-)\cap\;\mathcal{E}^{T,+2}_{\mathscr{I}^-}$. Assume $\upalpha_{\mathscr{H}^-}$, $\upalpha_{\mathscr{I}^-}$ decay exponentially at a rate faster than $\frac{1}{M}$, then there exists a unique solution $\alpha$ to \cref{T+2} in $J^-(\overline{\Sigma})$ such that $\lim_{u\longrightarrow-\infty}r\alpha=\upalpha_{\mathscr{I}^-}$, $2MU^{2}\Omega^{-2}\alpha\big|_{\mathscr{H}^-}=U^{2}\upalpha_{\mathscr{H}^-}$. Moreover, $(\alpha\big|_{\overline{\Sigma}},\slashednabla_T\alpha|_{\overline{\Sigma}})\in \mathcal{E}^{T,+2}_{\overline{\Sigma}}$ and \bref{109091}.\\
    Therefore, as in the case of ${}^{(+2)}\mathscr{F}^+,{}^{(+2)}\mathscr{B}^-$ we can define the unitary isomorphisms
    \begin{align}
        {}^{(+2)}\mathscr{F}^-:\mathcal{E}^{T,+2}_{\overline\Sigma}\longrightarrow\mathcal{E}^{T,+2}_{\overline{\mathscr{H}^-}}\oplus \mathcal{E}^{T,+2}_{\mathscr{I}^-},\qquad\qquad {}^{(+2)}\mathscr{B}^+:\mathcal{E}^{T,+2}_{\overline{\mathscr{H}^-}}\oplus \mathcal{E}^{T,+2}_{\mathscr{I}^-}\longrightarrow\mathcal{E}^{T,+2}_{\overline\Sigma},
    \end{align}
    with
    \begin{align}
        {}^{(+2)}\mathscr{F}^-\circ {}^{(+2)}\mathscr{B}^+=Id_{\mathcal{E}^{T,+2}_{\overline{\Sigma}}},\qquad\qquad{}^{(+2)}\mathscr{B}^+\circ{}^{(+2)}\mathscr{F}^-\circ=Id_{\mathcal{E}^{T,+2}_{\overline{\mathscr{H}^-}}\oplus \mathcal{E}^{T,+2}_{\mathscr{I}^-}}.
    \end{align}
    An identical statement holds with $\mathcal{E}^{T,+2}_{{\Sigma}}, \mathcal{E}^{T,+2}_{{\mathscr{H}^-}}$ instead.
\end{corollary}
\begin{corollary}\label{past scattering of -2}
     Given smooth data of compact support $(\underline\upalpha,\underline\upalpha')\in \mathcal{E}^{T,-2}_{\overline{\Sigma}}$, there exists a unique solution $\underline\alpha$ to the -2 Teukolsky equation \bref{T-2} on $J^-(\overline\Sigma)$ that induces radiation fields 
    \begin{itemize}
        \item $\underline\upalpha_{\mathscr{I}^-}\in \mathcal{E}^{T,-2}_{\mathscr{I}^-}$ given by  $\underline\upalpha_{\mathscr{I}^-}(v,\theta^A)=\lim_{u\longrightarrow -\infty} r^5\underline\alpha(u,v,\theta^A)$,
        \item $\underline\upalpha_{{\mathscr{H}^-}}\in \mathcal{E}^{T,-2}_{\overline{\mathscr{H}^-}}$ given by  $U^{-2}\underline\upalpha_{\mathscr{H}^-}=2MU^{-2}\Omega^{2}\underline\alpha|_{\mathscr{H}^-}$. 
    \end{itemize}
    such that 
     \begin{align}\label{190190190}
    \left\|\left(\underline\alpha|_{\overline{\Sigma}},\slashednabla_T\underline\alpha|_{\overline{\Sigma}}\right)\right\|^2_{\mathcal{E}^{T,-2}_{\overline{\Sigma}}}=\left|\left|\underline\upalpha_{\mathscr{I}^-}\right|\right|^2_{\mathcal{E}^{T,-2}_{\mathscr{I}^-}}+\left|\left|\underline\upalpha_{\overline{\mathscr{H}^-}}\right|\right|^2_{\mathcal{E}^{T,+2}_{\overline{\mathscr{H}^-}}}.
\end{align}
    Let $\underline\upalpha_{\mathscr{H}^-}\in\mathcal{E}^{T,-2}_{\overline{\mathscr{H}^-}}$ be such that $U^{-2}\underline\upalpha\in \Gamma({\overline{\mathscr{H}^-}})$ and let $\underline\upalpha_{\mathscr{I}^-}\in\Gamma(\mathscr{I}^-)\cap\;\mathcal{E}^{T,-2}_{\mathscr{I}^-}$. Then there exists a unique solution $\underline\alpha$ to \cref{T-2} in $J^+(\overline{\Sigma})$ such that $\lim_{u\longrightarrow-\infty}r^5\underline\alpha=\underline\upalpha_{\mathscr{I}^-}$, $2MU^{-2}\Omega^2\underline\alpha\big|_{\mathscr{H}^-}=U^{-2}\underline\upalpha_{\mathscr{H}^-}$. Moreover, $(\underline\alpha\big|_{\overline{\Sigma}},\slashednabla_T\underline\alpha|_{\overline{\Sigma}})\in \mathcal{E}^{T,-2}_{\overline{\Sigma}}$ and \bref{190190190} is satisfied. An identical statement holds with $\mathcal{E}^{T,-2}_{{\Sigma}}, \mathcal{E}^{T,-2}_{{\mathscr{H}^-}}$ instead.
\end{corollary}
Finally, note that using \Cref{past scattering of +2,,past scattering of -2}, the proof of \Cref{scatteringthm+2} and \Cref{scatteringthm-2} is immediate. 
\numberwithin{lemma}{section}
\numberwithin{proposition}{section}
\numberwithin{corollary}{section}
\numberwithin{remark}{section}
\section{Teukolsky--Starobinsky Correspondence}\label{section 9 TS correspondence}
We now turn to the proof of \Cref{Theorem 3} of the introduction, whose detailed statement is contained in \Cref{Theorem 3 detailed statement}. We start by stating in Section 9.1 some useful algebraic relations satisfied by the constraints \bref{eq:227intro1}, \bref{eq:228intro1}. We then study the constraints on scattering data in Section 9.2 to construct the maps $\mathcal{TS}_{\mathscr{H}^\pm}, \mathcal{TS}_{\mathscr{I}^\pm}$, and then we use the results of Section 9.1 and Section 9.2 to show that the constraints are propagated by solutions arising from scattering data consistent with the constraints, culminating in the proof of \Cref{Corollary 1} of the introduction in Section 9.4.
\subsection{Some algebraic properties of the Teukolsky--Starobinsky identities}\label{subsection 9.1 algebraic properties of TS}
\indent Let $\alpha$ be a solution to the $+2$ Teukolsky equation and let  $\Psi=\left(\frac{r^2}{\Omega^2}\nablau\right)^2r\Omega^2\alpha$, then the commutation relation \bref{commutation relation} implies that 
\begin{align}\label{TS fact 1}
    &\mathcal{T}^{-2}\left[ \frac{\Omega^2}{r^2}\nablau\frac{r^2}{\Omega^2}\nablau\Psi\right]=0.
\end{align}
Similarly, if  $\underline\alpha$ satisfies the $-2$ Teukolsky equation and $\underline\Psi=\left(\frac{r^2}{\Omega^2}\nablav\right)^2r\Omega^2\underline\alpha$, \bref{commutation relation 2} implies
\begin{align}\label{TS fact 2}
    &\mathcal{T}^{+2}\left[ \frac{\Omega^2}{r^2}\nablav\frac{r^2}{\Omega^2}\nablav\underline\Psi\right]=0.
\end{align}
Note that were $(\overone{\alpha},\overone{\underline\alpha})$ to belong to a solution to the full system of equations (\ref{start of full system})-(\ref{Bianchi 0*}) then in fact we would have equations \bref{eq:TS1}, \bref{eq:TS2}:
\begin{align}
\frac{\Omega^2}{r^2}\Omega\slashed{\nabla}_3 \frac{r^2}{\Omega^2}\Omega\slashed{\nabla}_3\overone\Psi-2r^4\slashed{\mathcal{D}}^*_2\slashed{\mathcal{D}}^*_1\overline{\slashed{\mathcal{D}}}_1\slashed{\mathcal{D}}_2 r\Omega^2\overone{\underline\alpha}-6M\left[\Omega\slashed{\nabla}_4+\Omega\slashed{\nabla}_3\right]r\Omega^2\overone{\underline\alpha}=0, \label{eq:227}\\
\frac{\Omega^2}{r^2}\Omega\slashed{\nabla}_4 \frac{r^2}{\Omega^2}\Omega\slashed{\nabla}_4\overone{\underline\Psi}-2r^4\slashed{\mathcal{D}}^*_2\slashed{\mathcal{D}}^*_1\overline{\slashed{\mathcal{D}}}_1\slashed{\mathcal{D}}_2 r\Omega^2\overone\alpha+6M\left[\Omega\slashed{\nabla}_4+\Omega\slashed{\nabla}_3\right]r\Omega^2\overone\alpha=0.\label{eq:228}
\end{align}
\indent Combining \bref{TS fact 1} and \bref{TS fact 2} with the fact that $-2r^4\fancydstar_2\fancydstar_1\overline{\fancyd}_1\fancyd_2,\slashednabla_T$ commute with both (\ref{T+2}) and (\ref{T-2}) leads to the following: denote by $\TSm[\alpha,\underline\alpha]$ the expression on the left hand side of (\ref{eq:227}) acting on  $\alpha, \underline\alpha$, such that the constraint becomes
\begin{align}\label{TS constraint -}
    \TSm[\alpha,\underline\alpha]:=\frac{1}{r^3}\nablau \frac{r^2}{\Omega^2}\nablau\Psi-2r^4\slashed{\mathcal{D}}^*_2\slashed{\mathcal{D}}^*_1\overline{\slashed{\mathcal{D}}}_1\slashed{\mathcal{D}}_2 {\underline\alpha}+6M\left[\Omega\slashed{\nabla}_4+\Omega\slashed{\nabla}_3\right]{\underline\alpha}=0.
\end{align}
Similarly denote by $\TSm[\alpha,\underline\alpha]$ the expression on the left hand side of (\ref{eq:228}) so that the constraint becomes
\begin{align}\label{TS constraint +}
    \TSp[\alpha,\underline\alpha]:=\frac{1}{r^3}\nablav \frac{r^2}{\Omega^2}\nablav\underline\Psi-2r^4\slashed{\mathcal{D}}^*_2\slashed{\mathcal{D}}^*_1\overline{\slashed{\mathcal{D}}}_1\slashed{\mathcal{D}}_2 {\alpha}-6M\left[\Omega\slashed{\nabla}_4+\Omega\slashed{\nabla}_3\right]{\alpha}=0.
\end{align}
\begin{lemma}\label{propagation lemma}
For $\alpha$ satisfying the $+2$ Teukolsky equation (\ref{T+2}) and $\underline\alpha$ satisfying the $-2$ equation (\ref{T-2}), $\TSp[\alpha,\underline\alpha]$ also satisfies the $+2$ Teukolsky equation (\ref{T+2}) and $\TSm[\alpha,\underline\alpha]$ satisfies the $-2$ equation \bref{T-2}
\end{lemma}
This implies that if we impose both constraints \bref{eq:227},\bref{eq:228} on initial or scattering data for both the $+2$ and $-2$ Teukolsky equations then the constraints will be propagated by the solutions in evolution. More specifically, if we have scattering data for $\alpha, \underline\alpha$ such that the \textit{radiation fields} belonging to the quantities $\TSp[\alpha,\underline\alpha]$, $\TSm[\alpha,\underline\alpha]$ (in the sense of the definitions stated in \Cref{+2 radiation} and \Cref{subsection 7.2 future radiation fields and fluxes}) are vanishing, then we must have that $\TSp[\alpha,\underline\alpha]=0$, $\TSm[\alpha,\underline\alpha]=0$ by \Cref{+2 future backward scattering} and \Cref{-2 future backward scattering}.\\
\indent We would like to know the extent to which data for $\alpha$, $\underline\alpha$ are constrained by \cref{TS constraint -} and \cref{TS constraint +}. Doing this for data on a Cauchy surface is complicated, but if we restrict to data consistent with the scattering theory developed so far in this paper then we can alternatively attempt to address this question for scattering data on $\mathscr{I}^+, \mathscr{H}^+$. This is the subject of the remainder of this section.\\
\indent To start with, we can show the following by a straightforward computation
\begin{lemma}\label{not independent}
For $\alpha$ satisfying the $+2$ Teukolsky equation (\ref{T+2}) and $\underline\alpha$ satisfying the $-2$ Teukolsky equation (\ref{T-2})
\begin{align}\label{ ts- to parabolic ts+}
    \frac{\Omega^2}{r^2}\nablav\left(\frac{r^2}{\Omega^2}\nablav\right)^3 r\Omega^2\TSm[\alpha,\underline\alpha]=-\left[2r^4\fancydstar_2\fancydstar_1\overline{\fancyd}_1\fancyd_2+12M\slashednabla_T\right]r\Omega^2\TSp[\alpha,\underline\alpha],
\end{align}
\begin{align}\label{ ts+ to parabolic ts-}
    \frac{\Omega^2}{r^2}\nablau\left(\frac{r^2}{\Omega^2}\nablau\right)^3 r\Omega^2\TSp[\alpha,\underline\alpha]=\left[2r^4\fancydstar_2\fancydstar_1\overline{\fancyd}_1\fancyd_2-12M\slashednabla_T\right]r\Omega^2\TSm[\alpha,\underline\alpha].
\end{align}
In other terms,
\begin{align}
    \TSp\left[\TSp[\alpha,\underline\alpha],-\TSm[\alpha,\underline\alpha]\right]=0,\qquad\qquad\qquad\TSm\left[-\TSp[\alpha,\underline\alpha],\TSm[\alpha,\underline\alpha]\right]=0,
\end{align}
regardless of whether or not the constraints $\TSp[\alpha,\underline\alpha]=0,\TSm[\alpha,\underline\alpha]=0$ are satisfied.
\end{lemma}
\Cref{not independent} implies that \Cref{eq:227}, \Cref{eq:228} are not independent. We will use \Cref{not independent} in \Cref{subsection 9.3 propagating the identities} to show that imposing only of the constraints on $\mathscr{I}^+$ and imposing only the other constraint on $\overline{\mathscr{H}^+}$ is enough to propagate the constraints on the solutions $\alpha, \underline\alpha$. 

\subsection{Inverting the identities on $\mathscr{I}^+, \overline{\mathscr{H}^+}$}\label{subsection 9.2 inverting the identities}
\subsubsection*{Constraint \bref{eq:228} at $\mathscr{I}^+$}
We know that there are dense subspaces of $\mathcal{E}^{T,+2}_{\overline{\Sigma}}, \mathcal{E}^{T,-2}_{\overline{\Sigma}}$ consisting of smooth data for \cref{T+2}, \cref{T-2} such that 
\begin{align}
    \lim_{v\longrightarrow\infty} r\Omega^2\TSm[\alpha,\underline\alpha]=\partial_u^4\upalpha_{\mathscr{I}^+}-2\fourthorder\underline\upalpha_{\mathscr{I}^+}+6M\partial_u\underline{\alpha}_{\mathscr{I}^+},
\end{align}
so we consider
\begin{align}\label{constraint null infinity}
    \partial_u^4\upalpha_{\mathscr{I}^+}-2\fourthorder\underline\upalpha_{\mathscr{I}^+}-6M\partial_u\underline{\alpha}_{\mathscr{I}^+}=0
\end{align}
as a constraint on scattering data $\underline\upalpha_{\mathscr{I}^+}, \upalpha_{\mathscr{I}^+}$ at $\mathscr{I}^+$. We now show the following: if $\upalpha_{\mathscr{I}^+}$ is smooth and compactly supported, then there is a unique $\underline\upalpha_{\mathscr{I}^+}$ that decays towards $\mathscr{I}^+_\pm$ and satisfies \bref{constraint null infinity}:
\begin{proposition}\label{alphabar out of alpha on scri}
Let $\upalpha_{\mathscr{I}^+}\in \Gamma_c(\mathscr{I}^+)$. Then there exists a unique smooth $\underline\upalpha_{\mathscr{I}^+}$ such that
\begin{align}\label{scalarise this}
     \partial_u^4\upalpha_{\mathscr{I}^+}-2\fourthorder\underline\upalpha_{\mathscr{I}^+}-6M\partial_u\underline\upalpha_{\mathscr{I}^+}=0,
\end{align}
with $\underline\upalpha_{\mathscr{I}^+}\longrightarrow0$ as $u\longrightarrow \pm \infty$.
\end{proposition}
\begin{proof}
To make sense of (\ref{scalarise this}) we scalarise it: we associate to $\underline\upalpha_{\mathscr{I}^+}$ scalar fields $(\fbar,\gbar)$ on $\mathscr{M}$ with vanishing $\ell=0,1$ modes such that $\underline\upalpha_{\mathscr{I}^+}=r^2 \fancydstar_2\fancydstar_1(\fbar,\gbar)$. Similarly, we associate to $\upalpha_{\mathscr{I}^+}$ the two fields $(f,g)$ such that $\upalpha_{\mathscr{I}^+}=r^2\fancydstar_2\fancyd_2(f,g)$. Define further $F=\frac{\Omega^2}{r^2}\nablau(\frac{r^2}{\Omega^2}\nablau)^3f$ and $G=\frac{\Omega^2}{r^2}\nablau(\frac{r^2}{\Omega^2}\nablau)^3g$. In the absence of $\ell=0,1$ modes, $r^2\fancydstar_2\fancydstar_1$ is injective and thus (\ref{eq:227}) becomes:
\begin{align}\label{eq:231}
\begin{split}
(F,G)&=2r^4\bar{\fancyd_1}\fancyd_2\fancydstar_2\fancydstar_1(\fbar,\gbar)+6M\nablau(\fbar,\gbar)
\\&=2r^4\fancyd_1\fancyd_2\fancydstar_2\fancydstar_1 (\fbar,-\gbar)+6M\nablau(\fbar,\gbar).
\end{split}
\end{align}
Note that $r^4\fancyd_1\fancyd_2\fancydstar_2\fancydstar_1 = \frac{1}{2} r^2\fancyd_1[-\mathring{\slashed{\Delta}}-1]\fancydstar_1$ and $r^2\fancydstar_1\fancyd_1=-\mathring{\slashed{\Delta}}+1$, so $r^4\fancyd_1\fancyd_2\fancydstar_2\fancydstar_1=\frac{1}{2} r^4\fancyd_1\fancydstar_1$ \\ $\times\{\fancyd_1\fancydstar_1-2\}=\frac{1}{2}\mathring{\slashed{\Delta}}(\mathring{\slashed{\Delta}}+2)$. Equations (\ref{eq:231}) become
\begin{align}\label{eq:232} 
\partial_u\fbar-\frac{1}{6M}\mathring{\slashed{\Delta}}(\mathring{\slashed{\Delta}}+2)\fbar=F,
\end{align}
\begin{align}\label{eq:233}
\partial_u \gbar+\frac{1}{6M}\mathring{\slashed{\Delta}}(\mathring{\slashed{\Delta}}+2)\gbar=G.
\end{align}
Equations (\ref{eq:232}) and (\ref{eq:233}) are two $4^{th}$order parabolic equations which are well-behaved in opposite directions in time; a unique smooth solution exists for (\ref{eq:232}) when evolving in the direction of increasing $u$ whereas (\ref{eq:233}) admits a unique smooth solution in the direction of decreasing $u$. Therefore, assuming the boundary condition $f\longrightarrow 0$ as $u\longrightarrow -\infty$ we will have a unique solution $f$ to (\ref{eq:232}) and this solution will decay for $u\longrightarrow\infty$. Similarly, there is a unique smooth $g$ solving (\ref{eq:233}) with $g\longrightarrow0$ when $u\longrightarrow \pm \infty$. Thus there is a unique smooth $\underline\upalpha_{\mathscr{I}^+}$ solving (\ref{scalarise this}) and decays towards $\mathscr{I}^+_\pm$.
\end{proof}
\begin{corollary}\label{iterated integrals}
Let $\upalpha_{\mathscr{I}^+},\underline\upalpha_{\mathscr{I}^+}$ be as in \Cref{alphabar out of alpha on scri}, then
\begin{align}\label{-2 further constraint}
\begin{split}
    \int_{-\infty}^\infty  \underline{\alpha}_{\mathscr{I}^+}du_1=0
\end{split}
\end{align}
\end{corollary}
\begin{proof}
\cref{scalarise this} and the decay of $\upalpha_{\mathscr{I}^+},\underline\upalpha_{\mathscr{I}^+}$ implies
\begin{align}\label{energy r ts-}
    \partial_u\upalpha_{\mathscr{I}^+}=2r^4\slashed{\mathcal{D}}^*_2\slashed{\mathcal{D}}^*_1\overline{\slashed{\mathcal{D}}}_1\slashed{\mathcal{D}}_2\int_{-\infty}^{u}du\; \underline\upalpha_{\mathscr{I}^+}+6M\underline\upalpha_{\mathscr{I}^+}.
\end{align}
Taking $u\longrightarrow \infty$ gives $2r^4\slashed{\mathcal{D}}^*_2\slashed{\mathcal{D}}^*_1\overline{\slashed{\mathcal{D}}}_1\slashed{\mathcal{D}}_2\int_{-\infty}^{\infty}du \underline\upalpha_{\mathscr{I}^+}=0$ which implies $\int_{-\infty}^{\infty}du \underline\upalpha_{\mathscr{I}^+}=0$ as in \Cref{alphabar out of alpha on scri}.
\end{proof}
Conversely we have the following lemma which follows immediately by inspecting \bref{constraint null infinity}:
\begin{proposition}\label{alpha out of alphabar on scri}
Given $\underline\upalpha_{\mathscr{I}^+}\in \Gamma_c(\mathscr{I}^+)$, there exists a unique $\upalpha_{\mathscr{H}^+}$ that is smooth and supported away from $\mathscr{H}^+_+$, such that \bref{constraint null infinity} is satisfied by $\upalpha_{\mathscr{I}^+}, \underline\upalpha_{\mathscr{I}^+}$. Furthermore, if $\int_{-\infty}^\infty du\; \underline\upalpha_{\mathscr{I}^+}=0$ then $\upalpha_{\mathscr{I}^+}\in\mathcal{E}^{T,+2}_{\mathscr{I}^+}$.
\end{proposition}\label{TS-2 on I+}
This completes the construction of the map $\mathcal{TS}_{\mathscr{I}^+}$:
\begin{corollary}\label{TS scri +}
\Cref{alphabar out of alpha on scri} defines the map
\begin{align}
    \mathcal{TS}_{\mathscr{I}^+}:\mathcal{E}^{T,+2}_{\mathscr{I}^+}\longrightarrow \mathcal{E}^{T,-2}_{\mathscr{I}^+}.
\end{align}
The map $\mathcal{TS}_{\mathscr{I}^+}$ is surjective on a dense subspace of $\mathcal{E}^{T,-2}_{\mathscr{I}^+}$ by \Cref{alpha out of alphabar on scri}. Therefore it extends to a unitary Hilbert-space isomorphism.
\end{corollary}
\begin{remark}
The argument leading to \cref{iterated integrals} can be used to show that 
\begin{align}
    \begin{split}
        \int_{-\infty}^\infty \int_{-\infty}^{u_1} &\underline{\alpha}_{\mathscr{I}^+}du_1 du_2=\int_{-\infty}^\infty\int_{-\infty}^{u_1}\int_{-\infty}^{u_2}\underline{\alpha}_{\mathscr{I}^+}du_1 d u_2\\
    &=\int_{-\infty}^\infty\int_{-\infty}^{u_1}\int_{-\infty}^{u_2}\int_{-\infty}^{u_3}\underline{\alpha}_{\mathscr{I}^+}du_1 du_2 du_3=0.
    \end{split}
\end{align}
\end{remark}
\subsubsection*{Constraint \bref{eq:227} at $\overline{\mathscr{H}^+}$}
\indent Similar considerations apply to constraint $\TSp[\alpha,\underline\alpha]=0$, which in Kruskal coordinates looks like
\begin{align}\label{constraint horizon}
    \partial_V^4 V^2\underline\upalpha_{\mathscr{H}^+}=\Big[2\fourthorder-3V\partial_V-6\Big]V^{-2}\upalpha_{\mathscr{H}^+}.
\end{align}
\begin{proposition}\label{alphabar out of alpha on H}
Given $\upalpha_{\mathscr{H}^+}$ such that  $V^{-2}\upalpha_{\mathscr{H}^+}\in\Gamma_c(\overline{\mathscr{H}^+})$, solving \bref{constraint horizon} as a transport equation for $V^2\underline\upalpha_{\mathscr{H}^+}$ with decay conditions towards $\mathscr{H}^+_+$:
\begin{align}
    V^2\underline\upalpha_{\mathscr{H}^+}, \partial_V  V^2\underline\upalpha_{\mathscr{H}^+}, \partial_V^2 V^2\underline\upalpha_{\mathscr{H}^+}, \partial_V^3 V^2\underline\upalpha_{\mathscr{H}^+} \longrightarrow 0 \text{ as } V \longrightarrow \infty,
\end{align}
gives a unique solution such that $V^2\underline\upalpha_{\mathscr{H}^+}\in\Gamma_c(\overline{\mathscr{H}^+})$ and $\underline\upalpha_{\mathscr{H}^+}, \upalpha_{\mathscr{H}^+}$ satisfy \bref{constraint horizon}.
\end{proposition}
Conversely, we have the following:
\begin{proposition}\label{alpha out of alphabar on H}
Let $\underline\upalpha_{\mathscr{H}^+}$ be such that  $V^{2}\underline\upalpha_{\mathscr{H}^+}\in\Gamma_c(\overline{\mathscr{H}^+})$, then there exists a unique $\upalpha_{\mathscr{H}^+}$ with $V^{-2}\upalpha_{\mathscr{H}^+}$ such that \bref{constraint horizon} is satisfied with $V^{-2}\upalpha_{\mathscr{H}^+}\longrightarrow 0$ as $V\longrightarrow \infty$
\end{proposition}
\begin{proof}
As in the proof of \Cref{alphabar out of alpha on scri}, we scalarise \bref{constraint horizon}: Let $V^2\underline\upalpha_{\mathscr{H}^+}=(2M)^2\fancydstar_2\fancydstar_1(\underline{f},\gbar)$, $V^{-2}\upalpha_{\mathscr{H}^+}=(2M)^2\fancydstar_2\fancydstar_1({f},g)$ and let $\underline{F}=-\partial_V^4 \fbar, \underline{G}=-\partial_V^4 \gbar$. Then $f, g, \underline{F}, \underline{G}$ satisfy
\begin{align}
    \underline{F}&=\left[3V\partial_V+6-\mathring{\slashed{\Delta}}(\mathring{\slashed{\Delta}}+2)\right]f,\label{eq:444}\\
    \underline{G}&=\left[3V\partial_V+6+\mathring{\slashed{\Delta}}(\mathring{\slashed{\Delta}}+2)\right]g.\label{eq:555}
\end{align}
Equations \bref{eq:444}, \bref{eq:555} are degenerate at $V=0$. If $f,g$ satisfy \bref{eq:444} and \bref{eq:555} then at $V=0$ we must have
\begin{align}\label{elliptic}
    \underline{F}|_{V=0}&=\left[6-\mathring{\slashed{\Delta}}(\mathring{\slashed{\Delta}}+2)\right]f|_{V=0},\\
    \underline{G}|_{V=0}&=\left[6+\mathring{\slashed{\Delta}}(\mathring{\slashed{\Delta}}+2)\right]g|_{V=0}.
\end{align}
The above are elliptic identities that determine $(f,g)|_{V=0}$ from $F|_{V=0}, G|_{V=0}$. Denote $(f_0,g_0):=(f,g)|_{V=0}$.\\
\indent As was done in the proof of \Cref{alphabar out of alpha on scri}, we evolve \bref{eq:444} and \bref{eq:555} in opposite directions in $V$.  Working with \bref{eq:555} is straightforward: let $V_\infty$ lie beyond the support of $\underline{F}$, then there is a unique $f$ satisfying \bref{eq:555} with $f|_{V_\infty}=0$ and we set $f$ to vanish for $V>V_\infty$.\\
\indent To find a solution to \bref{eq:444}, note that for $V_0>0$, there is a unique $g$ that satisfies \bref{eq:444} on $V\geq V_0$ and $g|_{V_0}=g_0$. Multiply \bref{eq:444} by $g$, integrate by parts to get:
\begin{align}
    \frac{3}{2}\left[g(V)^2-g(V_0)^2\right]+\int_{V_0}^V\frac{1}{\widetilde{V}}6g^2+|f\mathring{\slashed{\Delta}}(\mathring{\slashed{\Delta}}+2)g|^2=\int_{V_0}^V\frac{1}{\widetilde{V}}g\cdot \underline{G}
\end{align}
Poincar\'e's inequality and Cauchy--Schwarz imply:
\begin{align}\label{estimate}
    g(V)^2+\int_{V_0}^V\frac{5}{\widetilde{V}}g^2\lesssim\int_{V_0}^V \underline{G}^2+g_0^2
\end{align}
We obtain similar estimates for $\partial_V g$ by commuting \bref{eq:444} with $\partial_V$. We can use \bref{estimate} commuted with $\partial_V, \mathring{\slashednabla}$ to conclude that taking $V_0\longrightarrow 0$, we can find $g$ that satisfies \bref{eq:444} with $g|_{V=0}=g_0$.
\end{proof}
\begin{remark}
Were we to apply the constraint \bref{eq:227} on a smaller portion of the future event horizon, we would have needed more data to specify $\underline\upalpha_{\mathscr{H}^+}$ completely. In considering the problem on the entirety of $\overline{\mathscr{H}^+}$ no such additional data is necessary, since \bref{elliptic} determines the $f|_{\mathcal{B}}$ in terms of $\underline\upalpha_{\mathscr{H}^+}$.
\end{remark}
\begin{corollary}\label{TS+2 on H+}
\Cref{alphabar out of alpha on H} defines the map
\begin{align}
    \mathcal{TS}_{\mathscr{H}^+}:\mathcal{E}^{T,+2}_{\overline{\mathscr{H}^+}}\longrightarrow \mathcal{E}^{T,-2}_{\overline{\mathscr{H}^+}}.
\end{align}
The map $\mathcal{TS}_{\mathscr{H}^+}$ is surjective on a dense subspace of $\mathcal{E}^{T,-2}_{\overline{\mathscr{H}^+}}$ by \Cref{alpha out of alphabar on H}. Therefore it extends to a unitary Hilbert-space isomorphism.
\end{corollary}\label{TS on H-}
We can analogously consider the constraints on $\overline{\mathscr{H}^-}, \mathscr{I}^-$. In light of \Cref{time inversion} we can immediately deduce the appropriate statements:
\begin{corollary}\label{TS-2 on H-}
Given $\upalpha_{\mathscr{H}^-}$ such that $U^2\upalpha_{\mathscr{H}^-}\in\Gamma_c(\overline{\mathscr{H}^-})$, there exists a unique solution $\underline\upalpha_{\mathscr{H}^-}$ to the equation
\begin{align}\label{TS-2 on H- equation}
    \partial_U^4 U^2\upalpha_{\mathscr{H}^-}=\Big[2\fourthorder-3U\partial_U-6\Big]U^{-2}\underline\upalpha_{\mathscr{H}^-}.
\end{align}
such that $U^{-2}\underline\upalpha_{\mathscr{H}^-}\in\Gamma(\overline{\mathscr{H}^-})$. The solution $\underline\upalpha_{\mathscr{H}^-}(u,\theta^A)$ and its $\partial_U,\mathring{\slashednabla}$ derivatives decay exponentially as $u\longrightarrow-\infty$ at a rate $\frac{4}{M}$.\\
\indent Given $\underline\upalpha_{\mathscr{H}^-}$ such that $U^{-2}\underline\upalpha_{\mathscr{H}^-}\in\Gamma_c(\overline{\mathscr{H}^-})$, there exists a unique solution $\upalpha_{\mathscr{H}^-}$ such that $U^2\upalpha_{\mathscr{H}^-}\in\Gamma_c(\overline{\mathscr{H}^-})$.\\
\indent As in \Cref{TS+2 on H+}, we can combine the statements above to define a unitary Hilbert-space isomorphism via \bref{TS-2 on H- equation}:
\begin{align}
    \mathcal{TS}_{\mathscr{H}^-}:\mathcal{E}^{T,+2}_{\overline{\mathscr{H}^-}}\longrightarrow\mathcal{E}^{T,-2}_{\overline{\mathscr{H}^-}}.
\end{align}
\end{corollary}
\begin{corollary}\label{TS on scri-}
Let $\upalpha_{\mathscr{I}^-}\in \Gamma_c(\mathscr{I}^-)$. Then there exists a unique smooth $\underline\upalpha_{\mathscr{I}^-}$ such that
\begin{align}\label{TS+2 on I-}
     \partial_v^4\upalpha_{\mathscr{I}^-}-2\fourthorder\underline\upalpha_{\mathscr{I}^+}-6M\partial_v\underline\upalpha_{\mathscr{I}^+}=0,
\end{align}
with $\underline\upalpha_{\mathscr{I}^+}\longrightarrow0$ as $u\longrightarrow \pm \infty$. The solution $\underline\upalpha_{\mathscr{I}^-}$ and its derivatives decay exponentially as $v\longrightarrow\pm\infty$. \\
\indent Given $\underline\upalpha_{\mathscr{I}^-}$, there exists a unique solution $\upalpha_{\mathscr{I}^-}$ to \bref{TS+2 on I-} that is supported away from the past end of $\mathscr{I}^-$. Moreover, $\int_{-\infty}^\infty d\bar{v}\;\upalpha_{\mathscr{I}^-}=0$.\\
\indent As in \bref{TS scri +}, the statements above can be combined to define via \bref{TS+2 on I-} a unitary Hilbert space isomorphism:
\begin{align}
    \mathcal{TS}_{\mathscr{I}^-}:\mathcal{E}^{T,+2}_{\mathscr{I}^-}\longrightarrow\mathcal{E}^{T,-2}_{\mathscr{I}^-}.
\end{align}
\end{corollary}
\begin{corollary}\label{Both TS+ and TS-}
There exist Hilbert space isomorphisms
\begin{align}
    &\mathcal{TS}^+:=\mathcal{TS}_{\mathscr{H}^+}\oplus\mathcal{TS}_{\mathscr{I}^+}:\mathcal{E}^{T,+2}_{\overline{\mathscr{H}^+}}\oplus\mathcal{E}^{T,+2}_{\mathscr{I}^+}\longrightarrow\mathcal{E}^{T,-2}_{\overline{\mathscr{H}^+}}\oplus\mathcal{E}^{T,-2}_{\mathscr{I}^+},\\
    &\mathcal{TS}^-:=\mathcal{TS}_{\mathscr{H}^-}\oplus\mathcal{TS}_{\mathscr{I}^-}:\mathcal{E}^{T,+2}_{\overline{\mathscr{H}^-}}\oplus\mathcal{E}^{T,+2}_{\mathscr{I}^-}\longrightarrow\mathcal{E}^{T,-2}_{\overline{\mathscr{H}^-}}\oplus\mathcal{E}^{T,-2}_{\mathscr{I}^-}.
\end{align}
\end{corollary}
\subsection{Propagating the identities}\label{subsection 9.3 propagating the identities}
\indent We can summarise the contents of the previous section as follows: given scattering data for either $\alpha$ or $\underline\alpha$ on $\mathscr{I}^+$ and $\overline{\mathscr{H}^+}$, there exist unique scattering data for the other that is consistent with \bref{constraint null infinity} and \bref{constraint horizon} and \cref{+2 noncompact,,-2 noncompact}.\\
\indent For $\alpha$ and $\underline\alpha$ arising from scattering data related by \bref{constraint null infinity} and \bref{constraint horizon}, if we can verify that
\begin{align}
    \lim_{v\longrightarrow\infty} r^5\TSp[\alpha,\underline\alpha]&=0,\\
    V^2\Omega^{-2}\TSm[\alpha,\underline\alpha]\Big|_{\overline{\mathscr{H}^+}}&=0,
\end{align}
then \Cref{propagation lemma} together with \Cref{+2 future backward scattering}, \Cref{-2 future backward scattering} imply that $\TSm[\alpha,\underline\alpha]=\TSp[\alpha,\underline\alpha]=0$ everywhere. \\
\indent Assume future scattering data with  $(V^{-2}\upalpha_{\mathscr{H}^+},\upalpha_{\mathscr{I}^+})\in \Gamma_c(\overline{\mathscr{H}^+})\times\Gamma_c(\mathscr{I}^+)$ for the +2 Teukolsky equation \cref{T+2}. We can obtain $\underline\upalpha_{\mathscr{H}^+}$ that is supported away from $\mathscr{H}^+_+$ by solving \bref{constraint horizon} as a transport equation, and we can use \Cref{alphabar out of alpha on scri} to find a smooth $\underline{\alpha}_{\mathscr{I}^+}$ decays exponentially towards $\mathscr{I}^+_\pm$ at rate faster than $\frac{4}{M}$. Therefore, there exists a unique solution $\underline\alpha$ that realises scattering data $(\underline\upalpha_{\mathscr{H}^+},\underline{\alpha}_{\mathscr{I}^+})$ with $V^2\Omega^{-2}\underline\alpha$ smooth everywhere on $J^+(\overline\Sigma)$ up to and including $\overline{\mathscr{H}^+}$. In particular, since $\TSp[\alpha,\underline\alpha]\Big|_{\overline{\mathscr{H}^+}}=0$, \cref{ ts- to parabolic ts+} implies
\begin{align}
    \partial_V^4 \left\{\partial_U^4V^{-2}\upalpha_{\mathscr{H}^+}+\Big(2\fourthorder-3V\partial_V-6\Big)V^2\underline\upalpha_{\mathscr{H}^+}\right\}=0.
\end{align}
Since $V^{-2}\upalpha_{\mathscr{H}^+}, V^2\underline\upalpha_{\mathscr{H}^+}$ and their derivatives decay as $v\longrightarrow\infty$, we conclude that $V^2\Omega^{-2}\TSm[\alpha,\underline\alpha]\Big|_{\overline{\mathscr{H}^+}}=0$.\\
\indent Towards $\mathscr{I}^+$, $(\underline\upalpha_{\mathscr{H}^+},\underline\upalpha_{\mathscr{I}^+})$ decay at a sufficiently fast rate that we can use \Cref{-2 noncompact}, \Cref{RW backwards noncompact}, and \Cref{Phi 2 backwards} to deduce.
\begin{align}
\begin{split}
    \lim_{u\longrightarrow\infty}\lim_{v\longrightarrow\infty}\left(\frac{r^2}{\Omega^2}\nablav\right)^2\underline\Psi&=\lim_{u\longrightarrow\infty} \int_u^\infty (u-\bar{u})\left[\mathcal{A}_2(\mathcal{A}_2-2)-6M\partial_u\right]\underline{\bm{\uppsi}}_{\mathscr{I}^+}\\
    &=\lim_{u\longrightarrow\infty}\int_u^\infty (u-\bar{u}) \left[\mathcal{A}_2^2(\mathcal{A}_2-2)^2-(6M\partial_u)^2\right]\underline\upalpha_{\mathscr{I}^+}=0.
\end{split}
\end{align}
We also have
\begin{align}
    \lim_{u\longrightarrow\infty}\lim_{v\longrightarrow\infty}\partial_u^i\left(\frac{r^2}{\Omega^2}\nablav\right)^2\underline\Psi=0.
\end{align}
for $0\leq i\leq3$. Taking the limit of \bref{ ts+ to parabolic ts-} as $v\longrightarrow\infty$ implies
\begin{align}
    \partial_u^4\left[\lim_{v\longrightarrow\infty}\left(\frac{r^2}{\Omega^2}\nablav\right)^2\underline\Psi-\Big(2\fourthorder+6M\partial_u\Big)\underline\upalpha_{\mathscr{I}^+}\right]=0.
\end{align}
Altogether, we see that $\lim_{v\longrightarrow\infty} r^5\TSp[\alpha,\underline\alpha]=0$. We have shown
\begin{proposition}
Assume $\alpha$ is a solution to \cref{T+2} arising from smooth scattering data $(\upalpha_{\mathscr{H}^+},\upalpha_{\mathscr{I}^+})$ such that $\upalpha_{\mathscr{I}^+} \in \Gamma_c({\mathscr{I}^+})$, $V^{-2}\upalpha_{\mathscr{H}^+}\in\Gamma_c(\overline{\mathscr{H}^+})$. There exists unique smooth scattering data $\underline\upalpha_{\mathscr{H}^+}\in\mathcal{E}^{T,-2}_{\overline{\mathscr{H}^+}}, \underline\upalpha_{\mathscr{I}^+}\in\mathcal{E}^{T,-2}_{\mathscr{I}^+}$ giving rise to a solution $\underline\alpha$ to \cref{T-2}. Moreover, $\alpha$ and $\underline\alpha$ satisfy $\TSp[\alpha,\underline\alpha]=\TSm[\alpha,\underline\alpha]=0$ everywhere on $J^+(\overline\Sigma)$.
\end{proposition}
We can repeat the above arguments starting from smooth, compactly supported scattering data for the $-2$ equation to arrive at
\begin{proposition}\label{proof of corollary}
Assume $\underline\alpha$ is a solution to \cref{T-2} arising from smooth scattering data $(\underline\upalpha_{\mathscr{H}^+},\underline\upalpha_{\mathscr{I}^+})$ such that $\underline\upalpha_{\mathscr{I}^+} \in \Gamma_c({\mathscr{I}^+})$, $V^2\underline\upalpha_{\mathscr{H}^+}\in\Gamma_c(\overline{\mathscr{H}^+})$. There exists unique smooth scattering data $\upalpha_{\mathscr{H}^+}\in\mathcal{E}^{T,+2}_{\overline{\mathscr{H}^+}}, \upalpha_{\mathscr{I}^+}\in\mathcal{E}^{T,+2}_{\mathscr{I}^+}$ giving rise to a solution $\alpha$ to \cref{T+2}. Moreover, $\alpha$ and $\underline\alpha$ satisfy $\TSp[\alpha,\underline\alpha]=\TSm[\alpha,\underline\alpha]=0$ everywhere on $J^+(\overline\Sigma)$.
\end{proposition}
This concludes the proof of \Cref{Theorem 3 detailed statement}, i.e.~\Cref{Theorem 3} of the introduction.
\subsection{A mixed scattering theory: proof of Corollary 1}\label{subsection 9.4 mixed scattering}
We are in a position to prove Corollary 1 of the introduction, i.e.~\Cref{corollary to be proven} of \Cref{subsection 4.4 Corollary 1: mixed scattering}:
\begin{proof}[Proof of Corollary 1]
We will construct the map $\mathscr{S}^{+2,-2}$ only in the forward direction on a dense subset of $\mathcal{E}^{T,-2}_{\overline{\mathscr{H}^-}}\oplus\;\mathcal{E}^{T,+2}_{\mathscr{I}^-}$. Let $\upalpha_{\mathscr{I}^-}\in \Gamma_c({\mathscr{I}^-})$, $\underline\upalpha_{\mathscr{H}^-}$ be such that $V^2\underline\upalpha_{\mathscr{H}^-}\in\Gamma_c(\overline{\mathscr{H}^-})$ and $\int_{-\infty}^\infty d\bar{v}\upalpha_{\mathscr{I}^-}=0$. The map $\mathcal{TS}^-$ of \Cref{Both TS+ and TS-} defines a scattering data set consisting of a smooth field $\underline\upalpha_{\mathscr{I}^-}$ on $\mathscr{I}^-$ which is supported away from the past end of $\mathscr{I}^-$, $\underline\upalpha_{\mathscr{H}^-}$ on $\overline{\mathscr{H}^-}$ which is supported away from the past end of $\overline{\mathscr{H}^-}$.\\
\indent The map ${}^{(+2)}\mathscr{B}^-$ of \Cref{+2 past forward scattering} gives rise to a smooth solution $\alpha$ on $J^-(\overline{\Sigma})$ such that
\begin{align}\label{545454}
\left\|\left(\alpha|_{\overline{\Sigma}}, \slashednabla_{n_{\overline\Sigma}}\alpha|_{n_{\overline\Sigma}}\right)\right\|_{\mathcal{E}^{T,+2}_{\overline{\Sigma}}}^2=\left\|\upalpha_{\mathscr{H}^-}\right\|^2_{\mathcal{E}^{T,+2}_{\overline{\mathscr{H}^-}}}+\left\|\upalpha_{\mathscr{I}^-}\right\|^2_{\mathcal{E}^{T,+2}_{\mathscr{I}^-}},
\end{align}
and the map ${}^{(+2)}\mathscr{F}^{+}$ extends $\alpha$ to a smooth solution of \bref{T+2} on $J^+(\overline{\Sigma})$. Combining \bref{545454} with the fact that $\alpha|_{{\Sigma^*}}, \slashednabla_{n_{{\Sigma^*}}}\alpha|_{\Sigma^*}$ are smooth implies that the estimates of \Cref{psiILED,,alphaILED,,ILED psi higherorder,,alphaILED higher order} apply, and we can apply \Cref{psi+2ptwisedecay,,alpha+2ptwisedecay,,horizonpsidecay} together with \Cref{WP+2Sigmabar} to conclude that $\alpha$ realises the image of ${}^{(+2)}\mathscr{F}^+$ on $\overline{\mathscr{H}^+}$ as its radiation field there.\\
\indent The scattering data set $(\underline\upalpha_{\mathscr{H}^-},\underline\upalpha_{\mathscr{I}^-})$ give rise to a unique smooth solution $\underline\alpha$ according to \Cref{past scattering of -2}, which in particular realises $\underline\upalpha_{\mathscr{H}^-},\underline\upalpha_{\mathscr{I}^-}$ as its radiation fields on $\mathscr{H}^-$, $\mathscr{I}^-$ respectively. The quantity $\underline\Psi=\left(\frac{r^2}{\Omega^2}\nablav\right)^2r\Omega^2\underline\alpha$ satisfies the Regge--Wheeler equation \bref{RW} and induces a radiation field on $\mathscr{I}^-$ that is given by $\uppsi_{\mathscr{I}^-}=\partial_v^2\underline\upalpha_{\mathscr{I}^-}$. Note that in particular, $\partial_v\uppsi_{\mathscr{I}^-}$ vanishes whenever $\underline\upalpha_{\mathscr{I}^-}$ vanishes on  $\mathscr{I}^-$.\\
\indent Assume the support of $\underline\upalpha_{\mathscr{I}^-}$ on $\mathscr{I}^-$ in $v$ is contained in $[v_-,v_+]$. Since $\alpha$ arises from scattering data of compact support, we can follow the steps leading to estimate \bref{this+2} taking into account \Cref{time inversion} to obtain the following: let $R$ be sufficiently large, then
\begin{align}\label{this-2}
\begin{split}
     \int_{\underline{\mathscr{C}}_v\cap\{r>R\}}d\bar{u}d\omega\; r^2|\nablau\underline\Psi|^2\lesssim_{v_+} R^2\Bigg[\|\underline\upalpha_{\mathscr{I}^-}&\|_{\mathcal{E}^{T,-2}_{\mathscr{I}^-}}^2+\|\underline\upalpha_{\mathscr{H}^-}\|_{\mathcal{E}^{T,-2}_{\overline{\mathscr{H}^-}}}^2\\&+\int_{[v_-,v_+]\times S^2}d\bar{v}d\omega\;|\mathring{\slashednabla}\partial_v^2{\underline\upalpha}_{\mathscr{I}^-}|_{S^2}^2+4|\partial_v^2{\underline\upalpha}_{\mathscr{I}^-}|_{S^2}^2\Bigg].
\end{split}
\end{align}
Let $v_1>v_+$, then we can use \bref{this-2} to show that $\sqrt{r}\partial_v\Psi|_{u,v}\longrightarrow 0$ as $u\longrightarrow-\infty$: 
\begin{align}
\begin{split}
    |\nablav\Psi|&\leq \int_{-\infty}^u d\bar{u} |\nablav\nablau\Psi|\lesssim\int_{-\infty}^u d\bar{u}\frac{1}{r^2}|\mathring{\slashed{\Delta}}\underline\Psi+\underline\Psi|\lesssim \frac{1}{\sqrt{r(u,v)}}\sqrt{\int_{-\infty}^u d\bar{u}\frac{1}{r^2}|\mathring{\slashed{\Delta}}\underline\Psi|^2+|\mathring{\slashednabla}\underline\Psi|^2+\underline\Psi|^2}\\& \lesssim \frac{1}{\sqrt{r(u,v)}}\sqrt{\sum_{|\gamma|\leq2} F^T_v[\slashed{\mathcal{L}}_{\Omega^\gamma}\underline\Psi}].
\end{split}
\end{align}
where $\Omega^\gamma=\Omega_1^{\gamma_1}\Omega_2^{\gamma_2}\Omega_3^{\gamma_3}$ denotes Lie differentiation with respect to the $so(3)$ algebra of $S^2$ Killing fields. Now take $u_1<u_2, v_2>v_1$ such that $(u_2,v_1,\theta^A)\in J^-(\overline\Sigma)$ and $r(u_2,v_1)>R$. We can repeat the procedure leading to \Cref{RWrp} in the region $\mathscr{D}^{u_2,v_2}_{u_1,v_1}$ to get for $p\in[0,2]$:
\begin{align}\label{747474}
\begin{split}
    &\int_{\mathscr{C}_{u_2}\cap[v_1,v_2]} d\bar{v}\sin\theta d\theta d\phi \;r^p|\nablav\underline\Psi|^2+\int_{\underline{\mathscr{C}}_{v_2}\cap[u_1,u_2]}d\bar{u}\sin\theta d\theta d\phi\;r^p\left[|\slashednabla\underline\Psi|^2+\frac{1}{r^2}|\underline\Psi|^2\right]\\&+ \int_{\mathscr{D}^{u_2,v_2}_{u_1,v_1}}d\bar{u}d\bar{v}\sin\theta d\theta d\phi\;r^{p-1}\left[p|\nablav\underline\Psi|^2+(2-p)|\slashednabla\underline\Psi|^2+r^{p-3}|\underline\Psi|^2\right]\\&\lesssim\int_{\mathscr{C}_{u_1}\cap[v_1,v_2]} d\bar{v}\sin\theta d\theta d\phi\; r^p|\nablav\underline\Psi|^2+\int_{\underline{\mathscr{C}}_{v_1}\cap[u_1,u_2]}d\bar{u}\sin\theta d\theta d\phi\;r^p\left[|\slashednabla\underline\Psi|^2+\frac{1}{r^2}|\underline\Psi|^2\right].
\end{split}
\end{align}
Set $p=1$ in \bref{747474}. Keeping $v_1,v_2$ fixed and taking $u_1\longrightarrow-\infty$, the first term on the right hand side of \bref{747474} decays. The remaining term can be estimated by \bref{this-2} and applying Hardy's inequality, knowing that $\underline\Psi$ and its angular derivatives converge pointwise towards $\mathscr{I}^-$. In conclusion we have
\begin{align}\label{rp estimate from past null infinity}
\begin{split}
    &\int_{\mathscr{D}^{u_2,\infty}_{-\infty,v_1}}d\bar{u}d\bar{v}\sin\theta d\theta d\phi\;r^{p-1}\left[p|\nablav\underline\Psi|^2+(2-p)|\slashednabla\underline\Psi|^2+r^{p-3}|\underline\Psi|^2\right]\\
    &\lesssim_{R}\sum_{|\gamma|\leq2} \Bigg[\|\slashed{\mathcal{L}}_{\Omega^\gamma}\underline\upalpha_{\mathscr{I}^-}\|_{\mathcal{E}^{T,-2}_{\mathscr{I}^-}}^2+\|\slashed{\mathcal{L}}_{\Omega^\gamma}\underline\upalpha_{\mathscr{H}^-}\|_{\mathcal{E}^{T,-2}_{\overline{\mathscr{H}^-}}}^2+\int_{[v_-,v_+]\times S^2}d\bar{v}d\omega\;|\mathring{\slashednabla}\partial_v^2{\slashed{\mathcal{L}}_{\Omega^\gamma}\underline\upalpha}_{\mathscr{I}^-}|_{S^2}^2+4|\partial_v^2{\slashed{\mathcal{L}}_{\Omega^\gamma}\underline\upalpha}_{\mathscr{I}^-}|_{S^2}^2\Bigg].
\end{split}
\end{align}
We can extend the region $\mathscr{D}^{u_2,\infty}_{-\infty,v_1}$ to obtain \bref{rp estimate from past null infinity} over a region $\mathscr{D}^{\infty,\infty}_{-\infty,v_1}\cap\{r>R\}$. In view of the monotonicity of $F^T_u[\underline\Psi]\cap\{r>R\}$, this implies in particular that
\begin{align}
    \lim_{u\longrightarrow\infty} \int_{\mathscr{C}_u\cap\{r>R\}}d\bar{v}\sin\theta d\theta d\phi \;\frac{\Omega^2}{r^2}|\underline\Psi|^2=0.
\end{align}
Now we show that $\underline\alpha$ induces a radiation field $\underline\upalpha_{\mathscr{I}^+}$ on $\mathscr{I}^+$ which is in $\mathcal{E}^{T,+2}_{\mathscr{I}^+}$. First, note that energy conservation is sufficient to show that $\underline\alpha, \underline\psi$ attains radiation fields on $\mathscr{I}^+$: Fix $u$ and take $v_2>v_1$:
\begin{align}\label{corollary to be proven, radiation field exists}
    |r^3\Omega\underline\psi(u,v_2,\theta^A)-r^3\Omega\underline\psi(u,v_1,\theta^A)|\leq\int_{v_1}^{v_2} d\bar{v} \frac{\Omega^2}{r^2}|\underline\Psi|\leq\frac{1}{\sqrt{r(u,v_1)}}\;\sqrt{\int_{v_1}^{v_2} d\bar{v} \frac{\Omega^2}{r^2}|\underline\Psi|^2}.
\end{align}
by commuting with angular derivatives and using a Sobolev estimate as in the proof of \Cref{RWradscri}, this shows that for any sequence $\{v_n\}$ with $v_n\longrightarrow\infty$ we have that $r^3\Omega\underline\psi(u,v_n,\theta^A)$ is a Cauchy sequence, and an identical argument yields the same for $\underline\alpha$. Denote the limit of $r^3\underline\psi$ near $\mathscr{I}^+$ by $\underline\psi_{\mathscr{I}^+}$.\\
\indent Since $r^5\underline\psi$ converges near $\mathscr{I}^-$, estimate \bref{corollary to be proven, radiation field exists} can be easily modified to show that $\underline\psi_{\mathscr{I}^+}$ decays towards the past end of $\mathscr{I}^+$. As for the future end of $\mathscr{I}^+$, we repeat the estimate \bref{corollary to be proven, radiation field exists} estimating $\underline\psi_{\mathscr{I}^+}$ in terms of $\underline\psi$ along a hypersurface $\{r=R\}$ for a fixed $R$. Since $\underline\alpha$ is smooth and $\|(\underline\alpha|_{\overline{\Sigma}},\slashednabla_{n_{\overline{\Sigma}}}\underline\alpha|_{\overline{\Sigma}})\|_{\mathcal{E}^{T,-2}_{\overline{\Sigma}}}<\infty$, the results of \Cref{-2 radiation on H+} apply and we can deduce that $\underline\psi|_{r=R}$ decays as $t\longrightarrow\infty$, and this says that $\underline\psi_{\mathscr{I}^+}$ decays towards the future end of $\mathscr{I}^+$.\\
\indent We now show that $\int_{-\infty}^\infty d\bar{u} \;\underline\upalpha_{\mathscr{I}^+}=0$. Consider the $-2$ Teukolsky equations \bref{T-2}, which we write as follows:
\begin{align}
    \frac{\Omega^2}{r^2}\nablau r^5\Omega^{-1}\underline\psi=\left(\mathcal{A}_2-\frac{6M}{r}\right)r\Omega^2\underline\alpha.
\end{align}
It can be shown that the limit towards $\mathscr{I}^+$ produces
\begin{align}\label{858585}
    \partial_u \underline\psi_{\mathscr{I}^+}=\mathcal{A}_2\;\underline\upalpha_{\mathscr{I}^+}
\end{align}
We can conclude by observing that $\underline\psi_{\mathscr{I}^+}$ decays towards both ends of $\mathscr{I}^+$. With this we can also conclude that $\underline\upalpha_{\mathscr{I}^+}\in \mathcal{E}^{T,-2}_{\mathscr{I}^+}$ and that
\begin{align}
    \|\underline\upalpha_{\mathscr{I}^+}\|_{\mathcal{E}^{T,-2}_{\mathscr{I}^+}}^2+ \|\upalpha_{\mathscr{H}^+}\|_{\mathcal{E}^{T,+2}_{\overline{\mathscr{H}^+}}}^2= \|\upalpha_{\mathscr{I}^-}\|_{\mathcal{E}^{T,+2}_{\mathscr{I}^-}}^2+\|\underline\upalpha_{\mathscr{H}^-}\|_{\mathcal{E}^{T,-2}_{\overline{\mathscr{H}^-}}}^2.
\end{align}
\end{proof}
\begin{remark}
The result above subsumes a restricted map to scattering data in $\mathcal{E}^{T,-2}_{\mathscr{H}^-}$, $\mathcal{E}^{T,+2}_{\mathscr{H}^+}$, which leads to an isomorphism
\begin{align}
    \mathscr{S}^{+2,-2}:\mathcal{E}^{T,-2}_{\mathscr{I}^-}\oplus\mathcal{E}^{T,-2}_{\mathscr{H}^-}\longrightarrow \mathcal{E}^{T,-2}_{\mathscr{H}^+}\oplus \mathcal{E}^{T,-2}_{\mathscr{I}^+}.
\end{align}
\end{remark}
\appendix
\section{Robinson--Trautman spacetimes}\label{Appendix A Robinson--Trautman}
In \cite{DHR16}, it was shown that solutions to the linearised Einstein equations where $\overone{{\Psi}}=\overone{\underline\Psi}=0$ can be identified with the Robinson--Trautman family of spacetimes near Schwarzschild. This family of spacetimes is defined by the condition that they admit a null geodesic congruence that is shear-free and twist-free, and as such these spacetimes are algebraically special of Petrov type D in vacuum. These conditions lead to the reduction of the Einstein equations to a nonlinear parabolic equation, and this leads to interesting properties, such as the fact that for positive mass $M$, a generic member of this family can not be smoothly extended through the event horizon \cite{ChruscielSingleton}. \\
\indent It is easy to see that linearised Robinson--Trautman solutions cannot arise with data for $\alpha$ in $\mathcal{E}^{T+2}_{\mathscr{I}^\pm}\oplus\mathcal{E}^{T+2}_{\mathscr{H}^\pm}$ and data for $\overone{\underline\alpha}$ in $\mathcal{E}^{T-2}_{\mathscr{I}^\pm}\oplus\mathcal{E}^{T-2}_{\mathscr{H}^\pm}$. If $\overone{\Psi}=\overone{\underline\Psi}=0$ everywhere then $\|\overone{\upalpha}_{\mathscr{I}^\pm}\|_{\mathcal{E}^{T,+2}_{\mathscr{I}^\pm}}=\|\overone{\upalpha}_{\mathscr{H}^\pm}\|_{\mathcal{E}^{T,+2}_{\mathscr{H}^\pm}}=0$ and $\|\overone{\upalpha}_{\mathscr{I}^\pm}\|_{\mathcal{E}^{T,-2}_{\mathscr{I}^\pm}}=\|\overone{\upalpha}_{\mathscr{I}^\pm}\|_{\mathcal{E}^{T,-2}_{\mathscr{I}^\pm}}=0$, which means $\overone{\alpha}=\overone{\underline\alpha}=0$.
\section{The double null gauge and the Einstein vacuum equations}\label{Appendix B Double null guage}
The following is a synopsis of sections 3, 4 of \cite{DHR16}.
Let $(\mathscr{M},g)$ be a Lorentzian manifold. A coordinate system $(u,v,\theta^A)$ is said to define a double null gauge if the loci of $u,v,$ denoted by $\mathscr{C}_u, \underline{\mathscr{C}}_v$ respectively, constitute foliations of spacetime by null hypersurfaces with respect to $g$. The metric $g$ in a double null gauge takes the form
\begin{align}\label{DNGmetric}
    \bm{ds}^2=-4\bm{\Omega}^2\bm{dudv}+\bm{\slashed{g}}_{AB}(\bm{d\theta}^A-\bm{b}^A\bm{dv})(\bm{d\theta}^B-\bm{b}^B\bm{dv}).
\end{align}
Here, $(\theta^A)$ are coordinates on the 2-manifolds $\mathcal{S}_{u,v}=\mathscr{C}_u\cap\underline{\mathscr{C}}_v$ that are intersections of constant $u,v$ hypersurfaces, $\Omega$ is a scalar, $b^A$ is a vector field that is tangent to $\mathcal{S}_{u,v}$. This gauge comes with a null frame $(e_3,e_4,e_1,e_2)$:
\begin{align}
    \bm{e}_3=\frac{1}{\bm{\Omega}}\bm{\partial}_{\bm{u}}\qquad\qquad \bm{e}_4=\frac{1}{\bm{\Omega}}(\bm{\partial}_{\bm{v}}-\bm{b}^A\bm{\partial}_A),
\end{align}
$\{\bm{e}_A,A=1,2\}$ is a frame associated to the coordinates $(\theta^A)$ on $\mathcal{S}_{u,v}$, such that $\bm{e}_A\cdot \bm{e}_B=\bm{\slashed{g}}_{AB}$.\\
\indent Let $\bm{\nabla}$ be the Levi--Civita connection associated with the metric $\bm{g}$. In a double null gauge the connection and curvature and organised into $\mathcal{S}_{u,v}$-tangent tensor fields. The following are the \textit{connection coefficients}:
\begin{align}
\begin{split}
\bm{\chi}_{AB}&=\bm{g}(\bm{\nabla}_A \bm{e}_4,\bm{e}_B)\;\;\;,\;\;\;\bm{\underline{\chi}}_{AB}=\bm{g}(\bm{\nabla}_A \bm{e}_3,\bm{e}_B)\\
\bm{\eta}_A&=-\frac{1}{2}\bm{g}(\bm{\nabla}_{3}\bm{e}_A,\bm{e}_4)\;\;\;,\;\;\;\bm{\underline\eta}_A=-\frac{1}{2}\bm{g}(\bm{\nabla}_{4}\bm{e}_A,\bm{e}_3)\\
\bm{\hat{\omega}}&=\frac{1}{2}\bm{g}(\bm{\nabla}_{4}\bm{e}_3,\bm{e}_4)\;\;\;,\;\;\;\bm{\hat{\underline{\omega}}}=\frac{1}{2}\bm{g}(\bm{\nabla}_{3}\bm{e}_4,\bm{e}_3)\\
\bm{\zeta}&=\frac{1}{2}\bm{g}(\bm{\nabla}_A \bm{e}_4,\bm{e}_3)
\end{split}
\end{align}
We further decompose $\bm{\chi}$ and $\bm{\underline\chi}$ into their trace $\frac{1}{2}\bm{g}(\bm{tr\chi}),\frac{1}{2}\bm{g}(\bm{tr\underline{\chi}}) $ and traceless symmetric parts $\bm{\hat{\chi}},\bm{\hat{\underline{\chi}}} $. In the Schwarzschild background only $\bm{\hat{\omega}},\bm{\hat{\underline{\omega}}}$ and the traces of $\bm{\chi}, \bm{\underline{\chi}}$ survive:
\begin{align}\label{SchwValues1}
\begin{split}
\bm{\chi}_{AB}&=\frac{\Omega}{r}\slashed{\bm{g}}_{AB}\qquad\qquad\qquad \bm{\underline{\chi}}_{AB}=-\frac{\Omega}{r}\slashed{\bm{g}}_{AB}\\
\bm{\hat{\omega}}=\frac{M}{r^2\Omega}&=-\frac{\Omega x}{2r}\qquad\qquad\qquad \bm{\hat{\underline{\omega}}}=-\frac{M}{r^2\Omega}=\frac{\Omega x}{2r}
\end{split}
\end{align}
The curvature components are organised as follows:
\begin{align}
\begin{split}
\bm{\alpha}_{AB}&=\bm{R}(\bm{e}_A,\bm{e}_4,\bm{e}_B,\bm{e}_4)\qquad\qquad\qquad\qquad\bm{\underline\alpha}_{AB}=\bm{R}(\bm{e}_A,\bm{e}_3,\bm{e}_B,\bm{e}_3)\\
\bm{\beta}_A&=\frac{1}{2}\bm{R}(\bm{e}_A,\bm{e}_4,\bm{e}_3,\bm{e}_4)\qquad\qquad\qquad\qquad\bm{\underline\beta}_A=\frac{1}{2}\bm{R}(\bm{e}_A,\bm{e}_3,\bm{e}_3,\bm{e}_4)\\
\bm{\rho}&=\frac{1}{4}\bm{R}(\bm{e}_4,\bm{e}_3,\bm{e}_4,\bm{e}_3)\qquad\qquad\qquad\qquad\bm{\sigma}=\frac{1}{4}*\bm{R}(\bm{e}_4,\bm{e}_3,\bm{e}_4,\bm{e}_3)
\end{split}
\end{align}
with $*\bm{R}_{abcd}=\bm{\epsilon}_{abef}\bm{R}^{ef}_{\;\;\;cd}$ denoting the Hodge dual on $(\mathscr{M},\bm{g})$ of $\bm{R}$.\\
\indent For the Schwarzschild metric, the only non-vanishing component is
\begin{align}\label{SchwValue2}
    \bm{\rho}=-\frac{2M}{r^3}.
\end{align}
For a tensor field $\bm{\xi}$ that is tangent to $\mathcal{S}_{u,v}$ for all $u,v$, the expressions $\bm{\slashednabla}_3\bm{\xi}, \bm{\slashednabla}_4\bm{\xi}$ denote the projections of the covariant derivatives $\bm{\nabla}_3\bm{\xi}, \bm{\nabla}_4\bm{\xi}$ onto the tangent space of $\mathcal{S}_{u,v}$. Thus $\bm{\slashednabla}_3\bm{\xi}, \bm{\slashednabla}_4\bm{\xi}$ are also tangent to $\mathcal{S}_{u,v}$ for all $u,v$. Denote by $\bm{D}\bm{\xi}, \bm{\underline{D}}\bm{\xi}$ the projections of the Lie derivatives of $\bm{\xi}$ in the 4,3 directions respectively, then if $\xi$ is a 1-form we have
\begin{align}
    \bm{\Omega\slashednabla}_4 \bm{\xi}_A=(\bm{D}\bm{\xi})_A+\bm{\Omega{\chi}}_A{}^B\bm{\xi}_B,\\
    \bm{\Omega\slashednabla}_3 \bm{\xi}_A=(\bm{\underline{D}}\bm{\xi})_A+\bm{\Omega{\underline{\chi}}}_A{}^B\bm{\xi}_B.\label{connection formula appendix}
\end{align}
and so on for higher order tensor fields. Let $\bm{\mathcal{D}}_1 $ be the operator acting on a $\mathcal{S}_{u,v}$-tangent 1-form $\bm\xi$ by $\bm{\mathcal{D}}_1 \bm\xi=(\slashed{\textbf{div}}\xi,\slashed{\textbf{curl}}\xi)$ and denote its $L^2(\mathcal{S}_{u,v})$-dual by $\bm{\mathcal{D}}_1^*$. Let $\bm{\mathcal{D}}_2$ be the operator acting on a $\mathcal{S}_{u,v}$-tangent 2-form $\bm\Xi$ by $(\bm{\mathcal{D}}_2 \bm\Xi)_A=\slashednabla^B\bm\Xi{BA}$ and denote its $L^2(\mathcal{S}_{u,v})$ dual by $\bm{\mathcal{D}}_2^*$.  \\
\indent The vacuum Einstein equations read
\begin{align}\label{EinsteinVac}
\bm{R}_{ab}[\bm{g}]=0.
\end{align}
For a metric satisfying \bref{EinsteinVac} we know that only the conformal part of the curvature tensor survives: $\bm{R}_{abcd}=\bm{W}_{abcd}$. In particular, $\bm\alpha_{AB}$ and $\underline{\bm\alpha}_{AB}$ are traceless. Furthermore, combining the vanishing of the Ricci curvature with the Bianchi identities implies:
\begin{align}\label{Bianchibody}
\nabla^a \bm{W}_{abcd}=\nabla^a*\bm{W}_{abcd}=0
\end{align}
Equations \bref{Bianchibody} are called the Bianchi equations and it is easy to see that they are in fact equivalent to the vacuum Einstein equations.\\
\indent We now describe how to linearise the Einstein equations against a fixed background in this gauge. We denote the background values of the quantities involved by unbolding their symbols, and their linearised versions are further distinguished by the superscript $\overone{{ }}$. For example:
\begin{align}
\bm{\Omega}=\Omega+\epsilon \overone{\Omega}.
\end{align}
Similarly,
\begin{align}
\bm{\slashed{g}}_{AB}=\slashed{g}_{AB}+\epsilon \overone{\slashed{g}}_{AB}\qquad\qquad \bm{b}{}_A=0+\epsilon\overone{b}_A
\end{align}
And so on for the connection and curvature components. Note that we further decompose the linearised metric by separating out its trace with respect to $\slashed{g}$:
\begin{align}
    \overone{\slashed{g}}_{AB}\;=\;\overone{\hat{\slashed{g}}}_{AB}+\frac{1}{2}\slashed{g}_{AB}\tr\overone{\slashed{g}}
\end{align}
We decompose $\bm{\chi}$ and $\underline{\bm{\chi}}$ to their traceless and pure trace parts:
\begin{align}
\bm{\chi}_{AB}=\hat{\bm{\chi}}_{AB}+\tr\bm\chi \;\bm{\slashed{g}}_{AB}\qquad\qquad\qquad\qquad \bm{\underline\chi}_{AB}=\hat{\bm{\underline\chi}}_{AB}+\tr\bm{\underline\chi}\; \bm{\slashed{g}}_{AB}
\end{align}
and we linearise $\bm{\hat{\chi}}, \bm{\hat{\underline\chi}}$ and $\bm{\Omega}\tr\bm{\chi}, \bm{\Omega}\tr\bm{\underline\chi}$ separately:
\begin{align}
\bm{\hat{\chi}}=\hat{\chi}+\epsilon \overone{\hat{\chi}}\qquad\qquad\qquad\qquad \bm{\Omega}\tr\bm{\chi}=\Omega\tr\chi+\epsilon \;\overone{\Omega tr\chi}\\
\bm{\hat{\underline\chi}}=\hat{\underline\chi}+\epsilon \overone{\hat{\underline\chi}}\qquad\qquad\qquad\qquad \bm{\Omega}\tr\bm{\underline\chi}=\Omega\tr\underline\chi+\epsilon \;\overone{\Omega tr\underline\chi}
\end{align}
For an example of linearisation against a Schwarzschild background, we can look at \bref{Bianchibody} taking $(b,c,d)=(4,A,4)$:
\begin{align}\label{example}
    \bm{\slashednabla}_3\bm{\alpha}+\frac{1}{2}{\textbf{tr}}\bm{\underline\chi}\bm{\alpha}+2\bm{\hat{\underline\omega}}\bm{\alpha}=-2\bm{\slashed{\mathcal{D}}}_2^*\bm{\beta}-3\bm{\hat\chi}\bm{\rho}-3{}^*\bm{\hat\chi}\sigma+(5\bm{\eta}-\bm{\Omega}^{-1}\bm{\slashednabla}_A\bm{\Omega})\widehat{\otimes}\bm\beta.
\end{align}
Since we want to keep only the leading order terms in $\epsilon$, we can use the Eddington--Finkelstein coordinates of the Schwarzschild background to write the perturbed metric in the form \bref{DNGmetric}. Using \bref{connection formula appendix}, the fact that $\alpha=0$ and keeping only leading order terms in $\epsilon$ yields
\begin{align}
    \bm{\slashednabla}_3\bm{\alpha}=\epsilon\cdot\nablau\overone\alpha+O(\epsilon^2).
\end{align}
where $\nablagml$ is the background Schwarzschild covariant derivative in the 3-direction. Following this recipe for the remaining terms of \bref{example} taking into account the background Schwarzschild values \bref{SchwValues1}, \bref{SchwValue2} yields the equation governing $\overone{\alpha}$ in \bref{Bianchi +2}. For the full details of the linearisation leading to equations \bref{start of full system}--\bref{Bianchi 0*} see \cite{DHR16}.

\addcontentsline{toc}{section}{References}

\printbibliography

\end{document}